\documentclass[11pt]{jmlr}
\usepackage[T1]{fontenc}
\usepackage{times}
\usepackage{mathtools,booktabs,array,graphicx}
\usepackage{tikz,pgfplots}
\usepackage{placeins,needspace}
\pgfplotsset{compat=1.18,table/search path={.,..}}
\graphicspath{{figures/}}
\usepackage[capitalize,nameinlink]{cleveref}
\newcommand{\R}{\mathbb R}
\newcommand{\E}{\mathbb E}
\newcommand{\D}{\mathcal D}
\newcommand{\K}{\mathcal K}
\newcommand{\supp}{\operatorname{supp}}
\newcommand{\OPT}{\mathrm{OPT}}
\newcommand{\eps}{\varepsilon}
\newcommand{\ip}[2]{\langle #1,#2\rangle}
\newcommand{\one}{\mathbf 1}
\newcommand{\ee}{\mathbf e}

\newcommand{\Cs}{C_*}
\newcommand{\Ws}{W_*}

\newcommand{\al}{\alpha_*}
\newcommand{\athree}{\alpha_3}

\newcommand{\Kbar}{\overline K}
\newcommand{\poly}{\operatorname{poly}}
\newcommand{\alphacert}{\underline{\alpha}}
\newcommand{\acf}{{\alpha^{\mathrm{cf}}_2}}
\DeclareMathOperator*{\argmax}{arg\,max}
\DeclareMathOperator*{\argmin}{arg\,min}

\makeatletter
\@for\jmlr@env:=lemma,corollary,remark,proposition,definition\do{%
  \expandafter\let\csname\jmlr@env\endcsname\relax
  \expandafter\let\csname end\jmlr@env\endcsname\relax
  \@for\jmlr@part:=c@,cl@,p@,the,theH\do{%
    \expandafter\let\csname\jmlr@part\jmlr@env\endcsname\@undefined}}
\makeatother
\newcounter{lemmactr}\newcounter{corollaryctr}\newcounter{remarkctr}
\newcounter{propositionctr}\newcounter{definitionctr}
\newtheorem{lemma}[lemmactr]{Lemma}
\newtheorem{corollary}[corollaryctr]{Corollary}
\newtheorem{remark}[remarkctr]{Remark}
\newtheorem{proposition}[propositionctr]{Proposition}
\newtheorem{definition}[definitionctr]{Definition}
\makeatletter
\@for\jmlr@env:=lemma,corollary,remark,proposition,definition\do{%
  \@for\jmlr@part:=body@font,header@font,sep,postheader\do{%
    \expandafter\let\csname jmlr@thm@\jmlr@env ctr@\jmlr@part\expandafter\endcsname
      \csname jmlr@thm@\jmlr@env @\jmlr@part\endcsname}}
\makeatother
\crefname{lemmactr}{Lemma}{Lemmas}
\crefname{corollaryctr}{Corollary}{Corollaries}
\crefname{remarkctr}{Remark}{Remarks}
\crefname{propositionctr}{Proposition}{Propositions}
\crefname{definitionctr}{Definition}{Definitions}
\crefname{lemma}{Lemma}{Lemmas}
\crefname{corollary}{Corollary}{Corollaries}
\crefname{remark}{Remark}{Remarks}
\crefname{proposition}{Proposition}{Propositions}
\crefname{definition}{Definition}{Definitions}
\crefname{construction}{Construction}{Constructions}
\crefname{algocf}{Algorithm}{Algorithms}
\crefname{appendix}{Appendix}{Appendices}
\Crefname{appendix}{Appendix}{Appendices}
\crefname{subappendix}{Appendix}{Appendices}
\Crefname{subappendix}{Appendix}{Appendices}
\crefname{subsubappendix}{Appendix}{Appendices}
\Crefname{subsubappendix}{Appendix}{Appendices}
\crefname{subsubsubappendix}{Appendix}{Appendices}
\Crefname{subsubsubappendix}{Appendix}{Appendices}
\jmlryear{2026}
\makeatletter
\renewcommand\ps@jmlrtps{%
  \let\@mkboth\@gobbletwo
  \def\@oddhead{}%
  \let\@evenhead\@oddhead
  \def\@oddfoot{\@titlefoot}%
  \let\@evenfoot\@oddfoot
}
\makeatother
\title[Geometry-Dependent $k$-Submodular Maximization]{Geometry-Dependent Approximation for\\Non-Monotone $k$-Submodular Maximization}
\author{\Name{Vaneet Aggarwal}\\\addr Purdue University}
\hypersetup{pdftitle={Geometry-Dependent Approximation for Non-Monotone k-Submodular Maximization},pdfauthor={Vaneet Aggarwal},pdfsubject={Analytical approximation profiles, comparison certificates, and online guarantees}}
\begin{document}
\maketitle
\begin{abstract}
We study the maximization of a nonnegative $k$-submodular function, not assumed monotone, in which each of $n$ elements receives one of $k\ge2$ labels or stays unselected and the set of selected elements must satisfy a constraint. For a compact convex down-closed region $P\subseteq[0,1]^n$ of fractional selection vectors, the diagonal level $\zeta(P)=\max\{t\in[0,1]:t(1,\ldots,1)\in P\}$ is the largest selection probability that can be given to all elements at once; it equals $r/n$ for a uniform matroid of rank $r$, and $\zeta(P)=1$ means that the constraint is vacuous. We give approximation coefficients that increase continuously with $\zeta$.
The coefficients come from one inequality for the multilinear extension $F$ of the objective. At each fractional point the algorithm computes an action without querying the objective; objective feedback then gives a comparator-independent update direction, and the inequality bounds $F$ at every feasible comparison point by a multiple of $F$ at the action plus a term linear in the comparison point. For $k\ge3$, feasibility of the action reduces to three polynomial inequalities in one variable, and the parameters follow from tangency equations; for $k=2$ the profile is closed form.
At $\zeta=0$ the coefficients are $0.4456\ldots$ for $k=2$ and $0.4541\ldots$ for every $k\ge3$. For every finite $k\ge3$ they reach $1/2$ above an explicit level, which is $0.4026\ldots$ for $k=3,4,5$, and at $\zeta=1$ they reach $1/2$ for $k=2$, $(\sqrt{17}-3)/2$ for $k=3,4$, and $k/(2k-1)$ for $k\ge5$.
Because the update direction does not depend on the comparison point, online linear optimization gives $O(\sqrt T)$ approximate regret over $T$ rounds, with gradient feedback or with value queries made after each action. The coefficients carry over to matroid independence and exact bases, with expected and high-probability guarantees, and to any fixed number of knapsacks: the zero-level coefficients hold on every instance, and the larger ones when items are small or when prefix enumeration keeps the diagonal level on one branch. For one budget the enumeration exponent is $O(1/\eps)$; objective-independent rounding of small items also transfers the online guarantees.
\end{abstract}
\clearpage
\tableofcontents
\clearpage
\section{Introduction}
A $k$-submodular objective $f$ evaluates labelings in which each selected element receives one of $k\ge2$ labels and the remaining elements stay unselected \citep{HuberKolmogorov12,WardZivny16}. Such objectives model multi-type sensor placement and multi-topic influence maximization \citep{OhsakaYoshida15,NieEtAl23}; we study the nonnegative class without assuming monotonicity. A non-monotone example is the \emph{multi-class cut}: for nonnegative pair weights $w_{uv}$ and label classes $S_1,\ldots,S_k$,
\[
 f(s)=\sum_{i=1}^k\ \sum_{u\in S_i,\ v\notin S_i}w_{uv}.
\]
It rewards placing interacting elements in different classes, so giving an element the label already carried by its neighbours can decrease the value. \Cref{sec:discrete} verifies that this objective is nonnegative and $k$-submodular and that adding a labeled-coverage objective keeps it in the class.

As in earlier work on constrained $k$-submodular maximization \citep{SunLiuLi22,ZhouHuangWang25}, feasibility depends only on the \emph{support}, the set of selected elements, and not on their labels. A deployment or seeding budget whose costs do not depend on the label is a knapsack constraint on the support, and upper quotas on the number of selected elements in disjoint groups form a partition matroid. Such constraints range from allowing only a small fraction of the candidates to be selected to allowing all of them. Prior guarantees cover unrestricted support \citep{IwataTanigawaYoshida16,Oshima21} and matroid or knapsack constraints \citep{SunLiuLi22,ZhouHuangWang25}, but their coefficients do not depend on how much of the ground set the constraint allows to be selected. We give coefficients that depend on one scalar of the constraint and move continuously from the constrained to the unconstrained case.

Let $n$ be the number of elements, let $\one\in\mathbb R^n$ be the all-ones vector, and let $P\subseteq[0,1]^n$ be a nonempty compact convex down-closed region of fractional support vectors, whose coordinate $e$ is the probability that element $e$ is selected, whatever its label. The \emph{diagonal level} of $P$ is
\[
 \zeta(P)=\max\{t\in[0,1]:t\one\in P\},
\]
the largest selection probability that all elements can receive simultaneously. The certified lower level $\zeta=0$ is available for every down-closed region, $\zeta(P)=1$ forces $P=[0,1]^n$, and a uniform matroid of rank $r$ has diagonal level $r/n$; \cref{sec:access} gives the formulas for general matroids and knapsack regions.

Our main result is a function $\alphacert_k(\zeta)$, the \emph{approximation profile}, continuous and nondecreasing in $\zeta$, such that an algorithm whose parameters depend only on $k$ and $\zeta$ attains coefficient $\alphacert_k(\zeta)$, up to an arbitrarily small additive error, for every objective in the class and every region with diagonal level at least $\zeta$ (\cref{thm:twoclosed,cor:everyk}). The fractional guarantee compares with the optimum of the multilinear relaxation defined in \cref{sec:access}, and the integral results compare with the best feasible labeling. For two labels the profile has a closed form; for three or more labels it consists of an algebraic initial piece followed by a piece defined by tangent lines.

At $\zeta=0$, the profile gives $\alphacert_2(0)=0.4456\ldots$ and $\alphacert_k(0)=0.4541\ldots$ for every $k\ge3$, for fractional and integral approximation and for online regret. Larger diagonal-dependent coefficients apply directly to matroids. For knapsacks they apply when the enumeration of large items has one branch, the prefix of an optimal labeling, whose residual region keeps the diagonal level (\cref{sec:multiknapsack}); when every cost is small, the coefficient at the original diagonal level applies without enumeration, up to a factor $1-\delta$ for one budget and $1-\delta-\kappa_d\sqrt\delta$ for several budgets with costs at most $\delta$ (\cref{prop:smallitems}). Measuring the diagonal after rescaling each coordinate by its largest feasible value can only increase it, and the fractional and online guarantees hold at the rescaled level (\cref{prop:normalized}). Without constraints, \cref{thm:exactendpoint} gives $\alphacert_k(1)\ge b_k$, where
\begin{equation}\label{eq:introfree}
 b_k=\begin{cases}
 1/2,&k=2,\\
 (\sqrt{17}-3)/2,&k=3,4,\\
 k/(2k-1),&k\ge5.
 \end{cases}
\end{equation}

For every finite $k\ge3$, \cref{thm:halfthreshold} gives a diagonal level from which the coefficient is at least $1/2$. This level is $0.4026818391\ldots$ for $k=3,4,5$ and increases with $k$ from there, to $0.4081613347\ldots$ for six labels and $0.4878056924\ldots$ for twelve. It is defined by explicit algebraic equations, a quartic when $k\ge12$, and approaches one at rate $1-\Theta(1/k)$. Simple rational parameters give the coefficient $1/2$ for $3\le k\le12$ at $\zeta\ge49/100$ and $45/86$ for $3\le k\le10$ at $\zeta\ge9/10$ (\cref{thm:half,thm:high}).

\begin{table}[thbp]
\centering\small
\caption{Geometry-independent non-monotone offline coefficients under the listed support constraints, using the literature versions available at the first posting of this work. The comparison with \citet{ChenEtAl26} uses only their September v1, which assumes $f(0)=0$; our guarantees allow $f(0)\ge0$. Our entries and the $1/3$ entry of \citet{XiaoEtAl26} are approached to arbitrary accuracy. Decimal entries evaluate exact algebraic constants. The same coefficients yield online approximate regret for matroids; one-knapsack online transfer is stated separately for small items in \cref{prop:smallitems}. Objective-query, projection, and rounding costs are in \cref{sec:stochastic,sec:rounding}.}\label{tab:comparison}
\begin{tabular}{@{}lccc@{}}\toprule
Result & One matroid & One knapsack & Fixed $d$ budgets\\\midrule
\citet{SunLiuLi22} & $1/3$ & --- & ---\\
\citet{ZhouHuangWang25} & $1/3$ & $1/3$ & $1/3$\\
\citet{Wang25} & --- & $1/3$ & ---\\
\citet{XiaoEtAl26} & --- & $1/3$ & ---\\
\citet{ChenEtAl26} & $\sqrt2-1$ & $\sqrt2-1$ & ---\\
This work, $k=2$ & $0.4456016834$ & $0.4456016834$ & $0.4456016834$\\
This work, $k\ge3$ & $0.4541264676$ & $0.4541264676$ & $0.4541264676$\\
\bottomrule\end{tabular}
\end{table}

\Cref{tab:freecomparison} compares our unconstrained coefficients with those of \citet{IwataTanigawaYoshida16,Oshima21}. For $k\ge5$ the excess above $1/2$ is $1/(4k-2)$, compared with the earlier $1/(4k^2+2)$, and the coefficient equals the monotone coefficient $k/(2k-1)$ of \citet{IwataTanigawaYoshida16}. \Citet{IwataTanigawaYoshida16} also show that a coefficient above $(k+1)/(2k)$ requires exponentially many value queries, even for monotone objectives.

\begin{table}[thbp]
\centering\small
\caption{Unconstrained non-monotone coefficients. Prior results are \citet{IwataTanigawaYoshida16} for $k=2$ and the general or special three-label coefficient of \citet{Oshima21}. Decimal entries are truncated.}\label{tab:freecomparison}
\begin{tabular}{@{}cccc@{}}\toprule
Labels & Prior coefficient & This work, closed form & This work, decimal\\\midrule
$2$ & $1/2$ & $1/2$ & $0.5000000000$\\
$3$ & $(\sqrt{17}-3)/2$ & $(\sqrt{17}-3)/2$ & $0.5615528128$\\
$4$ & $17/33$ & $(\sqrt{17}-3)/2$ & $0.5615528128$\\
$5$ & $26/51$ & $5/9$ & $0.5555555555$\\
$6$ & $37/73$ & $6/11$ & $0.5454545454$\\
$k\ge5$ & $(k^2+1)/(2k^2+1)$ & $k/(2k-1)$ & ---\\
\bottomrule\end{tabular}
\end{table}

We solve a continuous relaxation and round its solution. Each element's label choice is represented by a vector of label probabilities with sum at most one, the remaining probability leaving the element unselected, and the \emph{multilinear extension} $F$ of $f$ is the expected value of $f$ when every element draws its label independently from its vector \citep{ZhouHuangWang25}. We use three properties of $F$, proved in \cref{lem:extension}: it is affine in the label vector of each element, its derivatives decrease across elements \eqref{eq:hessiansigns}, and for each element the derivatives with respect to any two distinct labels have a nonnegative sum \eqref{eq:gradientpairwise}. The third property takes the place of the coordinatewise maximum used for diminishing-returns (DR) submodular functions on a box, which is unavailable here because the maximum of two label vectors of one element can have total probability above one.

Our main technical result shows that $F$ is \emph{upper-linearizable} in the sense of \citet{PedramfarAggarwal24}. At a fractional state $x$, the algorithm plays an action $p(x)$, a feasible fractional point computed from $x$ without querying the objective. We construct an update direction $g_F(x)$, the \emph{field}, that does not depend on the feasible fractional point $y$ against which the action is compared (the \emph{comparator}). With nonnegative constants $C,W$ and the value $f(0)$ of the empty labeling,
\begin{equation}\label{eq:introcert}
 F(y)+Cf(0)\le(1+C+W)F(p(x))+\langle g_F(x),y-x\rangle,
\end{equation}
which yields the coefficient $(1+C+W)^{-1}$ once a linear optimization method keeps the average of the last term small. The proof replaces the comparator by the action one element at a time. Each replacement changes $F$ by a gradient inner product that can be negative for a non-monotone objective; we bound it with coefficients placed on pairs of labels, which pair nonnegatively with every gradient by the third property. A second sequence of points, in which the elements left unselected by the comparator are filled in, reduces the cost of those elements. The diagonal level enters only through the requirement that the action stay feasible, which limits the total label probability of each action block. Since the same action and field serve every comparator, a property we call \emph{comparator-uniform}, \cref{thm:certificate} reduces the problem to linear optimization over the feasible region, and the inequality retains the term $Cf(0)$.

Feasibility of the action reduces to three polynomial inequalities in the total label probability $Z$ of one element (\cref{thm:finitepoly}), one for each of three worst-case inputs: probability spread uniformly over all labels, split equally between two labels, or placed on one label. Each inequality requires the action's mass at that input to be at most an allowed mass determined by the diagonal level; we call these lower bounds on the allowed mass the \emph{mass requirements}. An affine change of variables and a stationary-point argument show that the three tests imply the remaining two-variable condition, and a parameter choice valid for every number of labels needs only an endpoint inequality and a cubic. Without the filled-in sequence, the three requirements are concave in $Z$, so admissible parameters correspond to lines through $(\zeta,\zeta)$ lying above all three, and the selected line touches the requirements at $\zeta$, at two points on either side of $\zeta$, or at an endpoint. The initial pieces of the profile, which use the filled-in sequence, are obtained from rational parameterizations and polynomial factorizations.

The field has unbiased estimates from the $nk+1$ values of a prefix scan (\emph{full-prefix} feedback), from one extension gradient (\emph{one-gradient} feedback), from two discrete value queries (\emph{two-value} feedback), or from one value query multiplied by a random sign (\emph{one-value} feedback). Mirror or projected ascent on these estimates gives offline fractional approximations and, over $T$ online rounds, $O(\sqrt T)$ expected approximate regret (\cref{thm:continuous,thm:fewquery}), which compares the learner's cumulative reward with the corresponding fraction of the best fixed comparator's reward. With value feedback, the queries are made after the action is committed and may be at labelings other than the played one.

Rounding then produces feasible labelings. Sampling the labels conditionally reduces matroid and knapsack rounding to rounding the support of an ordinary submodular function; this is the rounding scheme of \citet{ZhouHuangWang25}, in which matroid rounding keeps the extension value in expectation and small-item knapsack rounding keeps it up to a factor $1-\eps$. Large items are handled by enumerating a short prefix of the solution; in our setting the term $Cf(0)$ in \eqref{eq:introcert}, applied to the objective conditioned on that prefix, compensates for the residual items that rounding discards. Rounding that does not query the objective can be carried out before feedback arrives, so the online guarantees carry over; \cref{sec:rounding} gives the integral results and their costs.

\paragraph{Related work.}
$k$-submodular functions were introduced by \citet{HuberKolmogorov12}; the characterization by orthant submodularity and pairwise monotonicity, \eqref{eq:orthant}--\eqref{eq:pairwise}, is due to \citet{WardZivny16} (their Theorem~7). The unconstrained algorithms of \citet{IwataTanigawaYoshida16,Oshima21} give the prior unconstrained coefficients. Under support constraints, the greedy rules of \citet{SunLiuLi22,Wang25} and the Frank--Wolfe method of \citet{ZhouHuangWang25} on the multilinear extension give $1/3$-type non-monotone guarantees, and for monotone objectives under a matroid constraint \citet{Sakaue17} gives a $1/2$-approximation. For one knapsack, \citet{XiaoEtAl26} obtain $\tfrac13(1-\mathrm e^{-3})$ with a combinatorial algorithm and $1/3-\eps$ with a multilinear-extension method, and the proportional top-2 rule of \citet{ChenEtAl26} raises the coefficient for one matroid or one knapsack to $\sqrt2-1$. In these combinatorial methods the label probabilities depend on current marginal gains of the objective, whereas our action is computed from the fractional state before the objective is queried and one field controls every fractional comparator; the online reduction needs both properties.

Our comparison is also related to continuous DR-submodular maximization over down-closed convex sets, for which \citet{BuchbinderFeldman24} prove a $0.401$-approximation in the non-monotone case. In the upper-linearizable framework of \citet{PedramfarAggarwal24}, \citet{LuEtAl26} prove $1/e$-upper-linearizability over such sets, and the September version of \citet{Aggarwal26geom} proves $4/9$-upper-linearizability together with geometry-dependent online bounds. A $k$-submodular multilinear extension is DR-submodular on its domain (\cref{lem:extension}), but these results do not apply directly, because that domain is a product of label simplices rather than a box.

For online maximization, \citet{Soma19} obtains non-monotone $1/2$-regret and monotone $k/(2k-1)$-regret through approachability and online linear optimization. Our unconstrained comparison improves the non-monotone online coefficient for each fixed $k\ge3$ and also covers support constraints. \Citet{NieAggarwalQuinn25} consider stochastic objectives with full-bandit feedback, in which only the reward of the played action is observed; our one-value and two-value results instead allow queries after the decision. The transfer from a comparator-uniform inequality to linear regret follows \citet{PedramfarAggarwal24}.

Matroid pipage rounding follows \citet{CalinescuEtAl11,ChekuriVondrakZenklusen10}, and \citet{ZhouHuangWang25} combine conditional label sampling with matroid and knapsack rounding for $k$-submodular objectives. General contention-resolution methods are developed by \citet{ChekuriVondrakZenklusen14}, including the random-order methods of \citet{AdamczykWlodarczyk18}, and prefix enumeration follows \citet{Sviridenko04,LeeEtAl10,KulikShachnaiTamir13}. For one budget we use the cost-preserving pipage method of \citet{AgeevSviridenko04}; for several budgets we follow the framework of \citet[Theorem~2.6]{KulikShachnaiTamir13}, which covers any constant number of budgets, with a random-priority variant of these contention-resolution schemes whose pruning step queries the objective and is therefore offline.

\paragraph{Contributions and organization.} \Cref{sec:setup} formulates the discrete problem, its multilinear extension, and the diagonal level. Our contributions are summarized as follows.
\begin{itemize}
 \item \Cref{sec:global,sec:geometry}: the comparator-uniform comparison (\cref{thm:certificate}), block rules for two and for at least three labels, certified by a quadratic test with a side condition for two labels (\cref{lem:twolabel}) and by three polynomial tests for three or more (\cref{thm:finitepoly}), and the resulting geometry-dependent inequality (\cref{cor:linearizable}).
 \item \Cref{sec:profile}: the profile $\alphacert_k(\zeta)$, in closed form for two labels (\cref{thm:twoclosed}) and by algebraic pieces and tangent lines for three or more (\cref{thm:kprofile,cor:everyk}), with explicit branch intervals (\cref{cor:explicitranges}), endpoint values (\cref{thm:exactendpoint}) and the levels above which it reaches $1/2$ (\cref{thm:halfthreshold}).
 \item \Cref{sec:stochastic}: fractional and online guarantees under full-prefix, one-gradient, two-value, and one-value feedback (\cref{thm:continuous,thm:fewquery,cor:localentropy}), and the guarantees at the rescaled diagonal level (\cref{prop:normalized}).
 \item \Cref{sec:rounding}: integral guarantees under matroid and knapsack constraints (\cref{thm:matroid,thm:single,thm:multi}), including high-probability matroid guarantees, an $O(1/\eps)$ enumeration exponent for one budget, and small-item guarantees at the original diagonal level, online for one budget (\cref{prop:smallitems}).
\end{itemize}
\Cref{sec:conclusion} concludes. Appendices A--E contain the proofs for \cref{sec:setup}, \cref{sec:global,sec:geometry}, \cref{sec:profile}, \cref{sec:stochastic}, and \cref{sec:rounding}, in that order, together with the oracle implementations; \cref{app:marginalproof} also bounds the marginals used by the algorithms.

\section{Problem Setup}\label{sec:setup}
\subsection{The discrete maximization problem}\label{sec:discrete}
For a positive integer $N$, write $[N]=\{1,\ldots,N\}$. Let $V$ be a finite ground set with $|V|=n$, and let $k\ge2$. If $V=\varnothing$, for instance after preprocessing or conditioning, the empty labeling is the only candidate and is optimal, so we take $n\ge1$ from now on. Labelings are elements of
\[
 \{0,1,\ldots,k\}^V:=\{s:V\to\{0,1,\ldots,k\}\}.
\]
A labeling $s$ assigns a label $s_e:=s(e)$ to each element $e\in V$, with label $0$ meaning unselected; after ordering $V$ we identify this space with $\{0,1,\ldots,k\}^n$. The objective is a nonnegative function $f:\{0,1,\ldots,k\}^V\to\R_{\ge0}$. The support of $s$ is $\supp(s)=\{e\in V:s_e\ne0\}$, and $0$ denotes the empty labeling, which maps every element to zero. For $i\in[k]$, the \emph{label class} of $s$ is $S_i=\{e\in V:s_e=i\}$; the classes are disjoint and their union is $\supp(s)$.

Write $s\preceq t$ if, for every $e\in V$, either $s_e=0$ or $t_e=s_e$; that is, $t$ is obtained from $s$ by labeling additional elements. For $e\in V$ with $s_e=0$ and $i\in[k]$, the labeling $s+e\mapsto i$ assigns label $i$ to $e$ and agrees with $s$ elsewhere:
\[
 (s+e\mapsto i)_{e'}=
 \begin{cases}
  i,&e'=e,\\
  s_{e'},&e'\in V\setminus\{e\}.
 \end{cases}
\]
We abbreviate the singleton labeling $0+e\mapsto i$ as $e\mapsto i$ and write the marginal gain as
\[
 \Delta_{e,i}f(s)=f(s+e\mapsto i)-f(s).
\]

\citet{HuberKolmogorov12} define $k$-submodularity through two label operations. For $a,b\in\{0,1,\ldots,k\}$, let $\operatorname{min}_0(a,b)=a$ if $a=b$ and $\operatorname{min}_0(a,b)=0$ otherwise; let $\operatorname{max}_0(a,b)=0$ if $a,b$ are distinct nonzero labels and $\operatorname{max}_0(a,b)=\max\{a,b\}$ otherwise. Applying these operations coordinatewise to labelings $s,t$, the function $f$ is \emph{$k$-submodular} if
\begin{equation}\label{eq:ksubdef}
 f(s)+f(t)\ge f(\operatorname{min}_0(s,t))+f(\operatorname{max}_0(s,t))\qquad\text{for all labelings }s,t.
\end{equation}
\citet[Theorem~7]{WardZivny16} prove that \eqref{eq:ksubdef} is equivalent to the following two marginal conditions, called \emph{orthant submodularity} and \emph{pairwise monotonicity}, respectively:
\begin{align}
 \Delta_{e,i}f(s)&\ge\Delta_{e,i}f(t) &&(s\preceq t,\ t_e=0,\ i\in[k]),\label{eq:orthant}\\
 \Delta_{e,i}f(s)+\Delta_{e,j}f(s)&\ge0 &&(s_e=0,\ i,j\in[k],\ i\ne j).\label{eq:pairwise}
\end{align}
We work with these marginal conditions throughout. Condition \eqref{eq:orthant} says that a fixed label assignment gains less when more elements are already labeled, and \eqref{eq:pairwise} compares alternative labels of one element, allowing a negative marginal provided that every other label compensates for it. Monotonicity would require $f(s)\le f(t)$ whenever $s\preceq t$; we do not assume it.

\begin{remark}[The initial value is retained]\phantomsection\label{rem:initialvalue}
We retain $f(0)$ because subtracting it can destroy nonnegativity. For example, take one element $e$ and $k\ge2$ labels, and set
\[
 f(0)=1,\qquad f(e\mapsto1)=0,\qquad f(e\mapsto i)=2\quad(i\ge2).
\]
This function is nonnegative. Condition \eqref{eq:orthant} is vacuous because there is only one element, and \eqref{eq:pairwise} holds because $\Delta_{e,1}f(0)+\Delta_{e,i}f(0)=-1+1=0$ and $\Delta_{e,i}f(0)+\Delta_{e,j}f(0)=2$ for $i,j\ge2$. However, $f-f(0)$ takes the value $-1$ at $e\mapsto1$, so $f(0)$ appears explicitly in our guarantees, as an additive term with a nonnegative coefficient.
\end{remark}

\begin{example}[A non-monotone objective: multi-class cut]\phantomsection\label{ex:multicut}
Let $w_{uv}=w_{vu}\ge0$ be weights on pairs of distinct elements, and define the multi-class cut objective
\begin{equation}\label{eq:multicut}
 f(s)=\sum_{i=1}^k\ \sum_{u\in S_i,\ v\notin S_i}w_{uv}.
\end{equation}
For a pair $\{e,v\}$, the number of label classes containing exactly one of $e,v$ is $0$ if $s_e=s_v$, is $1$ if exactly one of them is selected, and is $2$ if they carry distinct nonzero labels. Adding $e\mapsto i$ to a labeling with $s_e=0$ therefore changes the contribution of $\{e,v\}$ by $+1$ if $s_v\ne i$ and by $-1$ if $s_v=i$, so
\[
 \Delta_{e,i}f(s)=\sum_{v\ne e}w_{ev}\bigl(1-2\cdot\mathbf 1\{s_v=i\}\bigr).
\]
The objective is nonnegative because it is a sum of nonnegative weights. If $s\preceq t$ and $t_e=0$, then $s_v=i$ implies $t_v=i$, so each indicator can only increase from $s$ to $t$; hence $\Delta_{e,i}f(s)-\Delta_{e,i}f(t)=2\sum_{v\ne e}w_{ev}\bigl(\mathbf 1\{t_v=i\}-\mathbf 1\{s_v=i\}\bigr)\ge0$, which is \eqref{eq:orthant}. For $i\ne j$, at most one of $\mathbf 1\{s_v=i\}$ and $\mathbf 1\{s_v=j\}$ equals one, so $\Delta_{e,i}f(s)+\Delta_{e,j}f(s)=\sum_{v\ne e}w_{ev}(2-2\cdot\mathbf 1\{s_v=i\}-2\cdot\mathbf 1\{s_v=j\})\ge0$, which is \eqref{eq:pairwise}. Thus \eqref{eq:multicut} is nonnegative and $k$-submodular. It is not monotone: if $w_{uv}=1$ is the only nonzero weight, then $f(u\mapsto1)=1$ but labeling $v$ with the same label gives value $0$. For sets $A_{e,i}$ in a universe with nonnegative weights $\mathsf w$, the \emph{labeled-coverage} objective $h(s)=\mathsf w\bigl(\bigcup_{e\in\supp(s)}A_{e,s_e}\bigr)$, a coverage analogue of the multi-type sensor-placement objectives of \citet{OhsakaYoshida15}, has marginals $\Delta_{e,i}h(s)=\mathsf w\bigl(A_{e,i}\setminus\bigcup_{u\in\supp(s)}A_{u,s_u}\bigr)$, which are nonnegative and can only decrease as $s$ grows; hence $h$ satisfies \eqref{eq:orthant} and \eqref{eq:pairwise}. Both conditions are preserved under nonnegative combinations, so $f+h$ is another nonnegative $k$-submodular objective.
\end{example}

A label-independent support constraint is a nonempty family $\mathcal S\subseteq2^V$: a labeling $s$ is feasible exactly when $\supp(s)\in\mathcal S$. The discrete problem is
\begin{equation}\label{eq:discreteproblem}
 \OPT:=\max\{f(s):s\in\{0,1,\ldots,k\}^V,\ \supp(s)\in\mathcal S\}.
\end{equation}
An $\alpha$-approximation, for $\alpha\in(0,1]$, returns a feasible labeling $\widehat s$ with $\E f(\widehat s)\ge\alpha\OPT$. We focus on down-closed support families: $A\in\mathcal S$ and $B\subseteq A$ imply $B\in\mathcal S$. Exact matroid bases are treated separately in \cref{sec:rounding}.

For a matroid $\mathcal M=(V,\mathcal I)$, the family $\mathcal S=\mathcal I$ consists of its independent sets. We remove matroid loops and individually infeasible items, so all remaining singletons are feasible. For $d$ knapsack budgets, indexed by rows $m\in[d]$, with nonnegative costs $c_{m,e}$ and capacities normalized to one, it is
\[
 \mathcal S=\left\{A\subseteq V:\sum_{e\in A}c_{m,e}\le1\ \text{for every }m\in[d]\right\}.
\]
We reserve $m$ for budget rows and $i,j,\ell$ for labels.
The unconstrained problem has $\mathcal S=2^V$. Write $\one_A\in\{0,1\}^V$ for the indicator of a support $A$.

\subsection{The multilinear extension}\label{sec:extension}
Let $\Delta_k=\{z\in\R_+^k:\sum_i z_i\le1\}$ and $\D=\Delta_k^V$. For $x\in\D$, the \emph{block} $x_e\in\Delta_k$ holds the label probabilities of element $e$; write $|x_e|=\sum_i x_{e,i}$ for its mass and $|x|=(|x_e|)_{e\in V}$ for the support-mass vector. Bold $\ee_i$ denotes a standard coordinate vector, whose subscript is a label or an element as the context indicates; it is distinct from the element name $e$ and from the constant $\mathrm e=\exp(1)$, which is always written upright. All logarithms are natural. The one-hot encoding of $s$ is $\one_s\in\D$, with a zero block when $s_e=0$. For $x\in\D$, put $q_{e,0}=1-|x_e|$ and $q_{e,i}=x_{e,i}$ for $i\ge1$, and draw each $X_e$ independently according to $q_e$. The \emph{$k$-submodular multilinear extension} is
\begin{equation}\label{eq:extensionpoly}
 F(x)=\E_{X\sim x}f(X)
 =\sum_{s\in\{0,\ldots,k\}^V}f(s)\prod_{e\in V}q_{e,s_e}.
\end{equation}
The extension and its use in constrained maximization were developed by \citet{ZhouHuangWang25}. It is affine in each block when the other blocks are fixed, a property we call \emph{block-multiaffine}, and the algorithms estimate its derivatives by sampling instead of expanding the exponentially many terms of \eqref{eq:extensionpoly}.

\Cref{lem:extension} collects the properties of the extension used in the comparison and in the rounding; most are due to \citet{ZhouHuangWang25}, and \cref{app:extensionproof} gives the precise sources and short proofs. Gradients at boundary points are the derivatives of the polynomial in \eqref{eq:extensionpoly}; $\nabla_eF$ is its gradient in block $e$. A subscript $-e$ omits block $e$, and in a labeling argument of $f$ the omitted element is left unselected; thus $X_{-e}$ denotes the random labeling that agrees with $X$ outside $e$ and assigns label $0$ to $e$.
\begin{lemma}[Continuous structure of the extension]\phantomsection\label[lemma]{lem:extension}
The extension $F$ is nonnegative on $\D$ and satisfies $F(\one_s)=f(s)$. For all labels $i,j\in[k]$ and distinct elements $e,u\in V$, it has the following properties (the last inequality requires $i\ne j$):
\begin{align}
 \partial_{e,i}F(x)&=\E_{X_{-e}\sim x_{-e}}\Delta_{e,i}f(X_{-e}),\label{eq:expectedmarginal}\\
 \partial_{e,i}\partial_{e,j}F(x)&=0,\qquad
 \partial_{e,i}\partial_{u,j}F(x)\le0,\label{eq:hessiansigns}\\
 \partial_{e,i}F(x)+\partial_{e,j}F(x)&\ge0.\label{eq:gradientpairwise}
\end{align}
Consequently $\nabla_eF$ is independent of $x_e$, and $x\le y$ in $\D$ implies $\nabla F(x)\ge\nabla F(y)$. The restriction of $F$ to every segment in $\D$ whose direction is nonnegative or nonpositive is concave. In particular,
\begin{equation}\label{eq:ray}
 F(tx)\ge tF(x)+(1-t)f(0)\qquad(x\in\D,\ 0\le t\le1).
\end{equation}

For fixed labels $\ell_e\in[k]$ and $A\subseteq V$, let $\ell|_A$ assign these labels on $A$ and leave all other elements unselected. The set function $g_\ell(A)=f(\ell|_A)$ is nonnegative and submodular. Its ordinary multilinear extension $G_\ell(z)=\E g_\ell(R_z)$, where $R_z$ contains each $e$ independently with probability $z_e$, is convex along every segment $v+\tau(a\ee_e-b\ee_u)$ in $[0,1]^V$ with $a,b\ge0$. Independently choosing $\ell_e$ with probabilities $x_{e,i}/|x_e|$ gives
\begin{equation}\label{eq:labelconditioning}
 \E_\ell G_\ell(|x|)=F(x),
\end{equation}
with arbitrary conditional labels when $|x_e|=0$.
\end{lemma}
\subsection{Support geometry}\label{sec:access}
Let $P\subseteq[0,1]^V$ be a nonempty compact convex relaxation of feasible supports, containing $\one_A$ for every $A\in\mathcal S$. The fractional feasible region and its optimum are
\[
 \K(P)=\{x\in\D:|x|\in P\},\qquad
 \mathrm{OPT}_F(P)=\max_{x\in\K(P)}F(x).
\]
Every feasible discrete labeling has $\one_s\in\K(P)$ and $F(\one_s)=f(s)$, so $\mathrm{OPT}_F(P)\ge\OPT$. The diagonal-dependent results assume that $P$ is down-closed ($v\in P$ and $0\le u\le v$ imply $u\in P$) and take as input a level $\zeta\in[0,1]$ with $\zeta\one\in P$; the largest such level is the \emph{diagonal level}
\begin{equation}\label{eq:diagonal}
 \zeta(P)=\max\{t\in[0,1]:t\one\in P\}.
\end{equation}
The level $\zeta=0$ is valid for every down-closed $P$, and $\zeta=1$ forces $P=[0,1]^V$ by down-closedness.

For the matroid constraint, use the independence polytope $P(\mathcal M)=\operatorname{conv}\{\one_A:A\in\mathcal I\}$. Let $r$ be the matroid rank and $r_{\mathcal M}(A)$ the rank of $A\subseteq V$. By Edmonds' theorem, $P(\mathcal M)=\{z\in\R_+^V:\sum_{e\in A}z_e\le r_{\mathcal M}(A)\ \text{for all }A\subseteq V\}$. Hence $t\one\in P(\mathcal M)$ exactly when $t|A|\le r_{\mathcal M}(A)$ for every nonempty $A$, and
$\zeta(P(\mathcal M))=\min_{\varnothing\ne A\subseteq V}r_\mathcal M(A)/|A|$;
the minimum is at most one because singletons have rank one after loops are removed. For the uniform matroid of rank $r$, $r_{\mathcal M}(A)/|A|=\min\{1,r/|A|\}$ is minimized at $A=V$, giving $r/n$. For the $d$ knapsack budgets, use $P=\{v\in[0,1]^V:\sum_ec_{m,e}v_e\le1\ (m\in[d])\}$. Since $t\one\in P$ exactly when $t\le1$ and $t\sum_ec_{m,e}\le1$ for every row, this gives
$\zeta(P)=\min\{1,\min_{m:\sum_ec_{m,e}>0}(\sum_ec_{m,e})^{-1}\}$,
since a zero row imposes no restriction.
\section{Comparator-Uniform Linearization}\label{sec:certificate}
The comparison \eqref{eq:introcert} must hold for \emph{every} comparator $y\in\K(P)$ with the same action $p(x)$ and the same field $g_F(x)$. For a concave function the gradient at $x$ would suffice, since then $F(y)\le F(x)+\ip{\nabla F(x)}{y-x}$. The extension $F$ is concave only along nonnegative and nonpositive directions (\cref{lem:extension}), and for a non-monotone objective $F(y)$ can exceed $F(x)$ even when $y\le x$. We therefore compare $y$ with a modified action $p(x)$ and assemble the field from gradients of $F$ at points that depend only on $p(x)$. \Cref{sec:global} reduces the comparison to three inequalities on a single element block, \cref{sec:geometry} constructs rules satisfying them, and \cref{sec:profile} chooses their parameters.

\subsection{From block inequalities to a comparator-uniform comparison}\label{sec:global}
By \eqref{eq:gradientpairwise}, every block gradient of $F$ has a nonnegative sum over any two distinct labels. Call a vector $a\in\R^k$ \emph{pairwise monotone} if $a_i+a_j\ge0$ whenever $i\ne j$. Such a vector may have a negative coordinate (at most one), so a nonnegative coefficient vector need not pair nonnegatively with it; the coefficient vectors that pair nonnegatively with every such vector form the following cone.
\begin{lemma}[Pairing cone]\phantomsection\label[lemma]{lem:pairing}
The cone generated by the label pairs is
\begin{equation}\label{eq:cone}
 Q_k=\operatorname{cone}\{\ee_i+\ee_j:i\ne j\}
 =\Bigl\{v\ge0:2v_i\le\textstyle\sum_jv_j\text{ for every }i\Bigr\}.
\end{equation}
Its dual cone is the set of pairwise-monotone vectors: $v\cdot a\ge0$ for every $v\in Q_k$ exactly when $a$ is pairwise monotone.
\end{lemma}
Two of the inclusions follow from the generators: each generator $\ee_i+\ee_j$ satisfies the coordinate inequalities, which gives ``$\subseteq$'', and writing $v$ as a nonnegative combination of label pairs and applying $a_i+a_j\ge0$ to each pair shows that pairwise-monotone vectors lie in the dual cone. The reverse inclusion is proved in \cref{app:pairing}. By the coordinate description, no label may carry more coefficient mass than all the other labels combined; for $k=2$ the two coordinates must therefore be equal, whereas for $k\ge3$ the mass can be spread over several pairs, which the construction for $k\ge3$ uses.

Fix comparison constants $C,W\ge0$, which will give the coefficient $\alpha=1/(1+C+W)$.
\begin{definition}[Block rule]\phantomsection
A \emph{block rule} with constants $C,W\ge0$ receives one input block $z\in\Delta_k$, of mass $Z=\sum_iz_i$, together with fixed parameters. It returns an action block $p\in\Delta_k$, an auxiliary block $t\in\Delta_k$, a weight $0\le w\le W$, and nonnegative coefficient vectors $\varphi_0,\ldots,\varphi_k\in\R^k_+$, and it satisfies, for all pairwise-monotone $h,a\in\R^k$ with $h\le a$ coordinatewise,
\begin{align}
 -p\cdot a-w(p-t)\cdot h&\le\varphi_0\cdot a,\tag{B0}\label{eq:mainbracketzero}\\
 (\ee_\ell-p)\cdot(a+wh)&\le\varphi_\ell\cdot a,\qquad \ell\in[k],\tag{B1}\label{eq:mainbracketlabel}\\
 Cp-\Bigl((1-Z)\varphi_0+\textstyle\sum_\ell z_\ell\varphi_\ell\Bigr)&\in Q_k.\tag{B2}\label{eq:mainbudget}
\end{align}
The last condition is the \emph{coefficient budget}.
\end{definition}
The rule sees only $z$ and the parameters, whereas $h$ and $a$ stand for block gradients of $F$ at points that depend on the unknown comparator (see the proof of \cref{thm:certificate}), so the rule must work for all of them.

Suppose the comparator labeling is replaced by the action one element at a time; we call the resulting sequence of fractional points the \emph{main chain}. Because $F$ is affine in each block, replacing block $e$ changes the value by exactly an inner product with the current block gradient $a$. If the comparator assigns label $\ell$ to $e$, the value lost is $(\ee_\ell-p)\cdot a$; (B1) pays for this loss, together with an auxiliary term $w(\ee_\ell-p)\cdot h$, by the coefficient $\varphi_\ell$. If the comparator leaves $e$ unselected, the value lost is $-p\cdot a$, and (B0) pays for it by $\varphi_0$.

The auxiliary block $t$ and the weight $w$ come from a second sequence, the \emph{completed chain}, in which each element left unselected by the comparator is first filled with $t$ and whose block gradients are $h$. It adds the term $-w(p-t)\cdot h$ to the left side of (B0), at a total cost of at most $WF(p)$ over all elements. In the constructions below, $t$ is supported on a largest input label, so this term reduces the coefficient needed for an unselected comparator element.

The field accounts for the part of the coefficient that varies with the comparator's label; the remaining part is the average $(1-Z)\varphi_0+\sum_\ell z_\ell\varphi_\ell$ of the coefficients under the input probabilities, including the unselected probability $1-Z$. Condition (B2) requires this average demand to be at most $C$ times the action, up to a remainder in $Q_k$; by \cref{lem:pairing}, the remainder pairs nonnegatively with every pairwise-monotone gradient.

The action and the field are assembled by applying the block rule to every block $x_e$. Write $p_e,t_e,w_e,\varphi_{e,0},\ldots,\varphi_{e,k}$ for its outputs at $z=x_e$, and write $p(x)=(p_e)_{e\in V}$ for the resulting fractional action. Index the elements $1,\ldots,n$ in nondecreasing order of $w_e$, with ties broken by a fixed order on $V$; subscripts $<e$ and $\ge e$ refer to this order. Define the \emph{lower chain} and the field
\begin{equation}\label{eq:gdef}
 \Pi^e=(p_{<e},0_{\ge e}),\qquad A_e=\nabla_eF(\Pi^e),\qquad
 g_{e,\ell}(x)=(\varphi_{e,\ell}-\varphi_{e,0})\cdot A_e.
\end{equation}
The point $\Pi^e$ agrees with the action on the elements before $e$ and is empty from $e$ onward; its terminal point is $\Pi^{n+1}=p$. The action, the order, the coefficients, and the lower chain depend only on $x$, so the field is linear in $f$ and independent of the comparator. We write $g_F(x)$ when the objective must be explicit and $g(x)$ otherwise.

\begin{theorem}[Structural comparison]\phantomsection\label{thm:certificate}
Suppose measurable block rules satisfy \textup{(B0)--(B2)} for every input block, with common constants $C,W\ge0$. For every nonnegative $k$-submodular multilinear extension and every $x,y\in\D$, the action and field in \eqref{eq:gdef} satisfy
\begin{equation}\label{eq:certificate}
 \boxed{F(y)+Cf(0)\le(1+C+W)F(p(x))+\langle g_F(x),y-x\rangle.}
\end{equation}
If, for common $\mu,\zeta\in[0,1]$, the block rule obeys
\begin{equation}\label{eq:lifting}
 |p(z)|\le\mu|z|+(1-\mu)\zeta\qquad\text{for every input block }z\in\Delta_k,
\end{equation}
then $p(\K(P))\subseteq\K(P)$ for every compact convex down-closed $P$ containing $\zeta\one$. Mass-preserving rules satisfy this property with $\mu=1$ and also preserve arbitrary convex support regions.
\end{theorem}
The proof uses the main and completed chains, whose block gradients $a_e$ and $h_e$ depend on the comparator, and the lower chain \eqref{eq:gdef}, whose gradients $A_e$ do not and are the only ones the field uses.
\begin{proof}
The empty ground set is immediate, so assume $n\ge1$. Throughout, $p_e,t_e,w_e,\varphi_{e,0},\ldots,\varphi_{e,k}$ denote the outputs of the block rule at the input block $z=x_e$, and the elements are indexed $1,\ldots,n$ as in \eqref{eq:gdef}, so that $w_1\le\cdots\le w_n$. Steps~1--4 prove the inequality for an arbitrary fixed labeling $o$ in place of $y$, and Step~5 averages over labelings.

\emph{Step 1: three chains and their gradients.} Fix a labeling $o\in\{0,1,\ldots,k\}^V$, which need not be feasible. Write $\one_{o,e}\in\Delta_k$ for block $e$ of its one-hot encoding, so $\one_{o,e}=\ee_{o_e}$ if $o_e\ne0$ and $\one_{o,e}=0$ otherwise. Complete the unselected blocks of $o$ by the auxiliary blocks:
\[
 \widetilde o_e=\begin{cases}\ee_{o_e},&o_e\ne0,\\ t_e,&o_e=0.\end{cases}
\]
For $E=0,\ldots,n$, define the main chain and the completed chain
\[
 Y^E=(p_{\le E},\one_{o,>E}),\qquad T^E=(p_{\le E},\widetilde o_{>E}),
\]
meaning that blocks $1,\ldots,E$ equal the action and blocks $E+1,\ldots,n$ equal those of $\one_o$, respectively $\widetilde o$. All these points lie in $\D$. The endpoints are $Y^0=\one_o$, $T^0=\widetilde o$, and $Y^n=T^n=p$. The lower chain $\Pi^e$ of \eqref{eq:gdef} sets the blocks after $e-1$ to zero. Define the block gradients
\[
 a_e=\nabla_eF(Y^{e-1}),\qquad h_e=\nabla_eF(T^{e-1}),\qquad A_e=\nabla_eF(\Pi^e).
\]
By \cref{lem:extension}, $\nabla_eF$ does not depend on block $e$ itself, so these three gradients depend only on the blocks $u\ne e$. On these blocks the three points are ordered: blocks $u<e$ equal $p_u$ in all three, while for $u>e$,
\[
 0\le\one_{o,u}\le\widetilde o_u,
\]
because $\one_{o,u}=\widetilde o_u$ when $o_u\ne0$ and $\one_{o,u}=0\le t_u=\widetilde o_u$ when $o_u=0$. Replacing block $e$ by the same vector in all three points therefore gives three points of $\D$ ordered as $\Pi^e\le Y^{e-1}\le T^{e-1}$, and the antitone-gradient property of \cref{lem:extension} gives
\begin{equation}\label{eq:gradientorder}
 h_e\le a_e\le A_e.
\end{equation}
Each of the three vectors is pairwise monotone by \eqref{eq:gradientpairwise}.

\emph{Step 2: telescoping the main chain.} Consecutive points $Y^{e-1}$ and $Y^e$ differ only in block $e$, where they equal $\one_{o,e}$ and $p_e$. Since $F$ is affine in block $e$ with gradient $a_e$ there (block affinity, \cref{lem:extension}),
\[
 F(Y^{e-1})-F(Y^e)=\ip{\one_{o,e}-p_e}{a_e}.
\]
Summing over $e=1,\ldots,n$ and using $F(Y^0)=f(o)$ and $F(Y^n)=F(p)$ gives
\begin{equation}\label{eq:maintelescope}
 f(o)-F(p)=\sum_{e=1}^n\ip{\one_{o,e}-p_e}{a_e}.
\end{equation}
In the same way, consecutive points of the completed chain differ only in block $e$, and
\begin{equation}\label{eq:completeddelta}
 \delta_e:=F(T^e)-F(T^{e-1})=\ip{p_e-\widetilde o_e}{h_e}.
\end{equation}

\emph{Step 3: the completed chain costs at most $WF(p)$.} Summation by parts gives
\[
 \sum_{e=1}^nw_e\delta_e
 =\sum_{e=1}^nw_eF(T^e)-\sum_{e=0}^{n-1}w_{e+1}F(T^e)
 =w_nF(T^n)-w_1F(T^0)-\sum_{e=1}^{n-1}(w_{e+1}-w_e)F(T^e).
\]
Since $F\ge0$ on $\D$ and $w_{e+1}\ge w_e$, the last two terms are nonpositive. Because $T^n=p$ and $w_n\le W$,
\begin{equation}\label{eq:abelbound}
 \sum_{e=1}^nw_e\delta_e\le w_nF(p)\le WF(p).
\end{equation}
This is the only place where the order of the elements is used.

\emph{Step 4: applying the block inequalities.} Adding and subtracting $\sum_ew_e\delta_e$ in \eqref{eq:maintelescope} gives $f(o)-F(p)=\sum_e\bigl[\ip{\one_{o,e}-p_e}{a_e}-w_e\delta_e\bigr]+\sum_ew_e\delta_e$; substituting \eqref{eq:completeddelta} in the bracket and \eqref{eq:abelbound} in the last sum,
\begin{equation}\label{eq:mainplusw}
 f(o)\le(1+W)F(p)+\sum_{e=1}^n\Bigl[\ip{\one_{o,e}-p_e}{a_e}-w_e\ip{p_e-\widetilde o_e}{h_e}\Bigr].
\end{equation}
Consider one bracket. If $o_e=0$, then $\one_{o,e}=0$ and $\widetilde o_e=t_e$, so the bracket equals $-p_e\cdot a_e-w_e(p_e-t_e)\cdot h_e$. If $o_e=\ell\ne0$, then $\one_{o,e}=\widetilde o_e=\ee_\ell$, so the bracket equals $(\ee_\ell-p_e)\cdot(a_e+w_eh_e)$. By \eqref{eq:gradientorder}, $h_e\le a_e$ are pairwise monotone, so \eqref{eq:mainbracketzero} in the first case and \eqref{eq:mainbracketlabel} in the second case bound the bracket by $\varphi_{e,o_e}\cdot a_e$. Since $\varphi_{e,o_e}\ge0$ and $a_e\le A_e$ by \eqref{eq:gradientorder}, also $\varphi_{e,o_e}\cdot a_e\le\varphi_{e,o_e}\cdot A_e$. Hence
\begin{equation}\label{eq:phibound}
 f(o)\le(1+W)F(p)+\sum_{e=1}^n\ip{\varphi_{e,o_e}}{A_e}.
\end{equation}
To separate the comparator-dependent part, put
\[
 r_e=(1-|x_e|)\varphi_{e,0}+\sum_{\ell=1}^kx_{e,\ell}\varphi_{e,\ell}.
\]
If $o_e=\ell_0\ne0$, the definition of $g$ in \eqref{eq:gdef} gives
\begin{align*}
 \ip{g_e(x)}{\one_{o,e}-x_e}+\ip{r_e}{A_e}
 &=(\varphi_{e,\ell_0}-\varphi_{e,0})\cdot A_e-\sum_\ell x_{e,\ell}(\varphi_{e,\ell}-\varphi_{e,0})\cdot A_e\\
 &\quad+(1-|x_e|)\varphi_{e,0}\cdot A_e+\sum_\ell x_{e,\ell}\varphi_{e,\ell}\cdot A_e
 =\varphi_{e,\ell_0}\cdot A_e,
\end{align*}
because, apart from the term $\varphi_{e,\ell_0}\cdot A_e$, the coefficients of $\varphi_{e,0}\cdot A_e$ add to $-1+|x_e|+(1-|x_e|)=0$ and, for each $\ell\in[k]$, the coefficients of $\varphi_{e,\ell}\cdot A_e$ add to $-x_{e,\ell}+x_{e,\ell}=0$. If $o_e=0$, the first term on the right is absent and the same cancellation gives $\varphi_{e,0}\cdot A_e$. Thus in both cases
\[
 \ip{\varphi_{e,o_e}}{A_e}=\ip{g_e(x)}{\one_{o,e}-x_e}+\ip{r_e}{A_e},
\]
and \eqref{eq:phibound} becomes
\begin{equation}\label{eq:beforebudget}
 f(o)\le(1+W)F(p)+\ip{g(x)}{\one_o-x}+\sum_{e=1}^n\ip{r_e}{A_e}.
\end{equation}
The budget condition \eqref{eq:mainbudget} at the input $z=x_e$, $Z=|x_e|$, states $Cp_e-r_e\in Q_k$. Since $A_e$ is pairwise monotone, \cref{lem:pairing} gives $\ip{Cp_e-r_e}{A_e}\ge0$, that is, $\ip{r_e}{A_e}\le C\ip{p_e}{A_e}$. Consecutive lower-chain points $\Pi^e$ and $\Pi^{e+1}$ differ only in block $e$, where they equal $0$ and $p_e$; block affinity gives $F(\Pi^{e+1})-F(\Pi^e)=\ip{p_e}{A_e}$. Telescoping from $\Pi^1=0$ to $\Pi^{n+1}=p$,
\begin{equation}\label{eq:lowertelescope}
 \sum_{e=1}^n\ip{r_e}{A_e}\le C\sum_{e=1}^n\bigl(F(\Pi^{e+1})-F(\Pi^e)\bigr)=C\bigl(F(p)-f(0)\bigr).
\end{equation}
Substituting \eqref{eq:lowertelescope} into \eqref{eq:beforebudget} proves the discrete inequality
\begin{equation}\label{eq:discretecertificate}
 f(o)+Cf(0)\le(1+C+W)F(p)+\ip{g(x)}{\one_o-x}\qquad\text{for every labeling }o.
\end{equation}

\emph{Step 5: averaging over comparator labelings.} Let $O$ be a random labeling with the product law defining $F(y)$, so that $\Pr\{O_e=\ell\}=y_{e,\ell}$. Then $\E f(O)=F(y)$ by \eqref{eq:extensionpoly} and $\E\one_O=y$. For the fixed objective and state, the action $p=p(x)$ and the field $g_F(x)$ are independent of $O$, so \eqref{eq:discretecertificate} holds for every outcome with the same action and field. Since the right side of \eqref{eq:discretecertificate} is affine in $\one_o$, taking expectations proves \eqref{eq:certificate}.

\emph{Feasibility.} Let $x\in\K(P)$, so $|x|\in P$. By convexity, $\mu|x|+(1-\mu)\zeta\one\in P$. By \eqref{eq:lifting} at each input block $z=x_e$, $0\le|p(x)|\le\mu|x|+(1-\mu)\zeta\one$ coordinatewise, so down-closedness gives $|p(x)|\in P$. Since every $p_e$ lies in $\Delta_k$, also $p(x)\in\D$, and therefore $p(x)\in\K(P)$. If the rule preserves masses, then $|p(x)|=|x|\in P$ directly, without using down-closedness or the diagonal.
\end{proof}

In \eqref{eq:certificate} the coefficient budget contributes $C$ and the completed chain contributes $W$, and the field need not be the gradient of any function. Large coefficients therefore require block rules with small $C+W$ whose actions satisfy the mass bound \eqref{eq:lifting}.

\subsection{Weighted block rules and their polynomial certificates}\label{sec:geometry}
The block rules of this subsection depend on three scalar parameters $(c,\omega,\mu)$, with $c>0$, $0\le\omega<1$, and $0\le\mu\le1$, which are fixed before the objective is seen and are the same for every element, through the normalization
\begin{equation}\label{eq:normalizedparameters}
 u=\frac1{1-\omega},\qquad C=uc,\qquad W=u\omega,
 \qquad \alpha=\frac{1-\omega}{1+c},\qquad
 \chi(Z)=\mu Z+(1-\mu)\zeta.
\end{equation}
Since $1+C+W=(1+c)/(1-\omega)$, the coefficient $1/(1+C+W)$ in \eqref{eq:certificate} equals $\alpha=(1-\omega)/(1+c)$. The parameter $c$ sets the coefficient budget, a smaller $c$ giving a larger coefficient but forcing the rule to spend more support mass, and $\omega$ is the weight of the completed chain; together they determine the map $z\mapsto(p,t,w,\varphi)$. The parameter $\mu$ does not affect this map. With the diagonal level $\zeta$, it defines the affine target $\chi$, the support mass that the geometry allows: if $|p(z)|\le\chi(|z|)$ for every input block, then \eqref{eq:lifting} holds and the action stays in $\K(P)$. Thus $\mu$ interpolates between certifying against the input mass ($\mu=1$) and against the full diagonal level ($\mu=0$).

For $k\ge3$ labels, \cref{alg:levelrule} computes the map. Write $Z=|z|$ for the input mass and $q_i=z_i/(c+z_i)$ for the \emph{floors}. The rule starts from the floor action $p=q$, at which the budget (B2) with only the diagonal coefficients holds with equality but one label may dominate the others. It then raises every coordinate other than the largest that lies below a common level $\pi$ to $\pi$, until the budget surplus created by raising covers the excess of the dominant label, and the coefficient vectors record how this surplus is allocated. With the input sorted so that $z_1\ge\cdots\ge z_k$, the raised action and the \emph{resources} it creates are
\begin{align}\label{eq:weightedlevel}
 p_1=q_1,\quad p_j=\max\{q_j,\pi\}\ (j\ge2),\quad S=\sum_{j\ge2}p_j,\qquad
 \gamma_1=0,\quad \gamma_j=(c+z_j)p_j-z_j.
\end{align}
We call $p_1=q_1$ the \emph{head} of the action and $p_2,\ldots,p_k$, of total mass $S$, its \emph{tail}. With $(s)_+=\max\{s,0\}$, the \emph{demands} are
\begin{equation}\label{eq:weighteddemands}
 D_0=(q_1-S-\omega)_+,\qquad
 E_1=(2p_2-S)_+,\qquad E_\ell=(q_1-S+p_\ell)_+\quad(\ell\ge2),
\end{equation}
and the level is the least $\pi\ge0$ at which the total resource covers the total demand, weighted by the input probabilities:
\begin{equation}\label{eq:levelbalance}
 \sum_j\gamma_j\ \ge\ (1-Z)D_0+\sum_\ell z_\ell E_\ell.
\end{equation}
At this level, put $\theta=\gamma/\sum_j\gamma_j$ when the denominator is positive and $\theta=\ee_k$ otherwise. For $\omega>0$, set $t=\min\{1,(q_1-S)_+/\omega\}\ee_1$, and take $t=0$ at $\omega=0$. The outputs are
\begin{equation}\label{eq:weightedphi}
 w=W,\qquad \varphi_0=uD_0\theta,\qquad
 \varphi_\ell=u\bigl((1-p_\ell)\ee_\ell+E_\ell\theta\bigr).
\end{equation}

\begin{algorithm2e}[htb]
\DontPrintSemicolon
\KwIn{An input block $z\in\Delta_k$ with $k\ge3$; parameters $c>0$ and $0\le\omega<1$.}
\KwOut{Action $p$, auxiliary block $t$, weight $w$, and the data $\theta,D_0,E_1,\ldots,E_k$ that determine the coefficient vectors $\varphi_0,\ldots,\varphi_k$ through \eqref{eq:weightedphi}.}
\emph{Floors.} Sort $z$ decreasingly (fixed tie rule); put $Z=|z|$ and $q_i=z_i/(c+z_i)$ for every $i$\;
\emph{Raising.} For a level $\pi\ge0$, let $p_1=q_1$ and $p_j=\max\{q_j,\pi\}$ for $j\ge2$; write $S=\sum_{j\ge2}p_j$\;
\emph{Resources and demands.} Compute $\gamma_j$ by \eqref{eq:weightedlevel} and $D_0,E_1,\ldots,E_k$ by \eqref{eq:weighteddemands}\;
\emph{Balancing.} Choose the least $\pi\ge0$ satisfying \eqref{eq:levelbalance}, by a scan of the breakpoints of both sides\;
\emph{Coefficients.} Compute $\theta$, $t$, and $w=W$; the coefficient vectors are $\varphi_0,\ldots,\varphi_k$ of \eqref{eq:weightedphi}. Undo the sorting in every returned vector\;
\caption{Weighted level rule for one input block}\label{alg:levelrule}
\end{algorithm2e}

Through \cref{lem:pairing}, each of these objects has one role in (B0)--(B2):
\begin{itemize}
 \item \emph{Floors $q$.} With the diagonal coefficient $\varphi_\ell=u(1-p_\ell)\ee_\ell$ alone, coordinate $j$ of the budget (B2) equals $u[(c+z_j)p_j-z_j]=u\gamma_j$, which vanishes exactly at $p_j=q_j$. The floor is therefore the smallest action probability that covers this diagonal coefficient; raising it above the floor creates nonnegative surplus.
 \item \emph{Demands $D_0$ and $E_\ell$.} After the diagonal coefficient cancels the gain $\ee_\ell$ in (B1), what remains is the action with coordinate $\ell$ removed, $p-p_\ell\ee_\ell$. By \cref{lem:pairing}, this vector pairs nonnegatively with every pairwise-monotone gradient if no coordinate exceeds the sum of the others. The demand $E_\ell$ is exactly the excess of its largest coordinate over the others; this is the head $q_1$ when $\ell\ge2$, and $p_2$ when $\ell=1$. Likewise $D_0$ is the excess of the head over the tail in (B0), reduced by $\omega$ because the completed chain removes up to $\omega$ of the head's mass through $t$.
 \item \emph{Resources $\gamma$ and the level $\pi$.} Raising the common level $\pi$ creates budget surplus $\gamma_j>0$ in every raised coordinate of (B2). It also makes every demand nonincreasing in $\pi$, because the competitors of the head grow together (Step~2 in \cref{app:weightedlevel}). The balance \eqref{eq:levelbalance} chooses the least level that pays for all demands.
 \item \emph{Allocation $\theta$ and coefficients $\varphi$.} Adding $E_\ell\theta$ to $\varphi_\ell$, and $D_0\theta$ to $\varphi_0$, removes the excess. This works because $\theta$ never puts mass on the dominant coordinate: $\theta_1=0$ always, and $\theta_2=0$ whenever $E_1>0$. Since $\theta$ is proportional to $\gamma$, coordinate $j$ of (B2) becomes $u(\gamma_j-\mathrm{Dem}\,\theta_j)=0$, where $\mathrm{Dem}$ is the right side of \eqref{eq:levelbalance}.
\end{itemize}
Raising costs support mass: the action has mass $q_1+S$, the only quantity that the geometry constrains.

For a block $z=Z\ee_1$ concentrated on one label, the rule can be computed in closed form, and the computation produces the pure-label polynomial \eqref{eq:purepoly} below. The floors are $q_1=Z/(c+Z)$ and $q_j=0$ for $j\ge2$, so raising sets every tail coordinate to $\pi$, with $S=(k-1)\pi$ and total resource $cS$. Since $z_\ell=0$ for $\ell\ge2$ and $E_1=((3-k)\pi)_+=0$, only the empty-label demand counts. If $q_1\le\omega$, this demand vanishes, so $\pi=0$ and $p=q$, and the requirement $|p|=q_1\le\chi(Z)$ is then implied by the floor test $\mathcal F_k\ge0$ introduced below. Otherwise the balance reads $cS=(1-Z)(q_1-S-\omega)$, so
\[
 S=\frac{(1-Z)(q_1-\omega)}{c+1-Z},\qquad
 |p|=q_1+S=\frac{Z(c+2-2Z)-\omega(1-Z)(c+Z)}{(c+Z)(c+1-Z)}.
\]
On this branch, the requirement $|p|\le\chi(Z)$ is therefore exactly $\mathcal R(Z)\ge0$ for the pure-label polynomial $\mathcal R$ of \eqref{eq:purepoly} below. This demand, the only one here, comes from comparators that leave the element unselected, and the completed-chain weight $\omega$ reduces it directly.

Both sides of \eqref{eq:levelbalance} are continuous piecewise-linear functions of $\pi$; the left side is nondecreasing and the right side is nonincreasing. Their first intersection lies in $[0,q_1]$ and, after sorting, one scan of the breakpoints locates it, so \cref{alg:levelrule} runs in $O(k\log k)$ operations. The coefficient vectors are returned through $\theta$, $D_0$, and $E$, since writing all $k+1$ of them out costs $\Theta(k^2)$. This suffices for the field \eqref{eq:gdef}: for a gradient block $A$, \eqref{eq:weightedphi} gives
\[
 (\varphi_\ell-\varphi_0)\cdot A=u\bigl[(1-p_\ell)A_\ell+(E_\ell-D_0)\,\theta\cdot A\bigr],
\]
so after computing $\theta\cdot A$ once, all $k$ field coordinates of the block cost $O(k)$. Its outputs satisfy the block inequalities, and the mass requirement reduces to two scalar inequalities per block.
\begin{lemma}[Weighted level comparison]\phantomsection\label[lemma]{lem:weightedlevel}
Let $k\ge3$, $c>0$, and $0\le\omega<1$, and let the input block $z\in\Delta_k$ be sorted so that $z_1\ge\cdots\ge z_k$, as in \cref{alg:levelrule}. Whenever the action at the balancing level of \eqref{eq:levelbalance} has mass at most one, the outputs \eqref{eq:weightedphi} satisfy \textup{(B0)--(B2)} with the constants in \eqref{eq:normalizedparameters}, and the remainder in \textup{(B2)} is zero. Moreover,
\begin{equation}\label{eq:weightednorm}
 \|\varphi_\ell-\varphi_0\|_1\le
 \frac{1+1/(c+1)}{1-\omega}.
\end{equation}
To establish $|p(z)|\le\chi(Z)$, put $\sigma=Z-z_1$, $S_\chi=\chi(Z)-q_1$, and $d=q_1-S_\chi$. It suffices that
\begin{equation}\label{eq:weightedfloors}
 \sum_i\frac{z_i}{c+z_i}\le\chi(Z)
\end{equation}
and
\begin{align}\label{eq:weightedgamma}
 \Gamma^\omega_\chi(z):={}&cS_\chi-\sigma-\sigma d_+-(1-Z)(d-\omega)_+
 -z_1\bigl(2q_2-S_\chi\bigr)_+\\[-2pt]
 &\hspace{13mm}+\sum_{j\ge2}z_j\min\{q_j,(-d)_+\}\ \ge0.\nonumber
\end{align}
\end{lemma}
The quantity $\Gamma^\omega_\chi(z)$ is a lower bound on the surplus of resources over demands when the tail is raised just enough for the action to reach the target mass $\chi(Z)$. If this surplus is nonnegative, the least balancing level is no higher, so the action uses at most the target mass; \cref{app:weightedlevel} gives the proof. Since all elements use the same weight $w=W$, the order in \eqref{eq:gdef} can be fixed once, independently of the state.

For $k=2$, the pairing cone $Q_2=\{(v,v):v\ge0\}$ is a ray, so every cone condition forces two equal coordinates and the balancing problem has an explicit solution. Instead of a level, the rule chooses the imbalance $d$ between the two action probabilities and their total mass $r$, which the budget fixes; with $Z=z_1+z_2$ and $\kappa=|z_1-z_2|$, they solve
\begin{equation}\label{eq:tworule}
 (c+Z)d+(1-Z)(d-\omega)_+=\kappa,
 \qquad r=\frac{Z-\kappa d+(1-Z)(d-\omega)_+}{c}.
\end{equation}
\begin{lemma}[Two-label rule]\phantomsection\label[lemma]{lem:twolabel}
Let $k=2$, $c>0$, and $0\le\omega<1$. Equation \eqref{eq:tworule} has a unique solution $d\ge0$. Assign probabilities $(r+d)/2$ and $(r-d)/2$ to the larger and smaller input labels, and in these sorted coordinates set $t=\min\{d,\omega\}\ee_1/\omega$ (with $t=0$ if $\omega=0$), $w=W$, and
\[
 \varphi_0=u(d-\omega)_+\ee_2,\qquad
 \varphi_1=u(1-d)\ee_1,\qquad \varphi_2=u(1+d)\ee_2.
\]
Whenever $r\le1$, this rule satisfies \textup{(B0)--(B2)} with the constants in \eqref{eq:normalizedparameters} and zero budget remainder; in all cases $\|\varphi_\ell-\varphi_0\|_1\le(1+\omega)/(1-\omega)$. The mass bound $|p(z)|\le\chi(|z|)$ for every $z\in\Delta_2$ holds if $\mu+(1-\mu)\zeta\ge1/c$ and the quadratic
\begin{equation}\label{eq:twoweightedtest}
 c(c+1)(\mu Z+(1-\mu)\zeta)-Z(c+2-2Z)
                  +\omega(1-Z)(c+2Z)
\end{equation}
is nonnegative for $0\le Z\le1$.
\end{lemma}
On the branch $d\ge\omega$, the quadratic is $c(c+1)$ times the gap between $\chi(Z)$ and the action mass at the most unbalanced input, $\kappa=Z$; the side condition, which implies $\chi(Z)\ge Z/c$ because $\chi(0)\ge0$, covers the remaining inputs together with a convexity argument. Since $\mu+(1-\mu)\zeta\le1$, the side condition also forces $c\ge1$; \cref{app:twoweighted} gives the proof.

For $k\ge3$, \cref{lem:weightedlevel} still requires $\Gamma^\omega_\chi(z)\ge0$ on all of $\Delta_k$, and this requirement reduces to three one-variable tests. Write $v=(1-\mu)\zeta$, so that $\chi(Z)=\mu Z+v$, put $c_0=(\sqrt{17}-1)/4$, and define the univariate polynomials
\begin{align}
 \mathcal F_k(Z)&=(kc+Z)\chi(Z)-kZ,\label{eq:finitefloor}\\
 \mathcal B(Z)&=(c+Z)(c+Z/2)\chi(Z)-Z(c+3Z/2),\label{eq:balancedpoly}\\
 \mathcal R(Z)&=(c+Z)(c+1-Z)\chi(Z)-Z(c+2-2Z)
                   +\omega(1-Z)(c+Z),\label{eq:purepoly}
\end{align}
and the auxiliary bivariate polynomial
\begin{align}\label{eq:jointpoly}
 \mathcal P(a,b)={}&(c+a)(c+b)\bigl((c+1)\chi(a+b)-(a+b)+(1-a-b)\omega\bigr)\nonumber\\
 &-2a(1-a)(c+b)-2ab(c+a).
\end{align}
Each univariate polynomial corresponds to one extreme type of input block. The \emph{floor} polynomial $\mathcal F_k$ controls the floor mass $\sum_iq_i$ in \eqref{eq:weightedfloors}. By concavity of $s\mapsto s/(c+s)$, this mass is largest when the input is spread uniformly over the $k$ labels, and $\mathcal F_k\ge0$ says $\chi(Z)\ge kZ/(kc+Z)$. The \emph{balanced} polynomial handles blocks in which no label dominates, with two equally occupied labels as the worst case; it is equivalent to
\begin{equation}\label{eq:balancedmass}
 \chi(Z)\ge B_c(Z):=\frac{Z(c+3Z/2)}{(c+Z)(c+Z/2)}.
\end{equation}
The \emph{pure-label} polynomial $\mathcal R$ handles a block concentrated on a single label of mass $Z$. When one label dominates, aggregating the other labels into a single mass $b$ produces the \emph{joint} polynomial $\mathcal P$, which is nonnegative whenever the three univariate tests hold (\cref{lem:jointredundant}) and therefore imposes no additional constraint. The polynomial $\mathcal F_k$ has degree two and the others have degree at most three.

\begin{theorem}[Univariate polynomial certificate]\phantomsection\label[theorem]{thm:finitepoly}
Let $k\ge3$, $c>0$, $0\le\omega<1$, and $0\le\mu,\zeta\le1$. The weighted level rule satisfies $|p(z)|\le\chi(|z|)$ for every $z\in\Delta_k$ if
\begin{equation}\label{eq:finitecriterion}
 \mathcal F_k(Z),\mathcal B(Z),\mathcal R(Z)\ge0
 \qquad(0\le Z\le1).
\end{equation}
Whenever these conditions hold, the outputs satisfy \textup{(B0)--(B2)} with the constants in \eqref{eq:normalizedparameters}, and the same parameters work for every label count between $3$ and $k$.
\end{theorem}
The proof (\cref{app:finiteverification}) treats blocks with and without a dominant head separately and reduces both cases to the three tests and the joint polynomial, which an affine shift absorbing $\omega$ and a stationary-point argument on the triangle then eliminate.

Dividing the three polynomials by their positive factors turns them into lower bounds on $\chi$. In addition to $B_c$ from \eqref{eq:balancedmass}, define the floor and pure-label requirements
\begin{equation}\label{eq:unweightedobstacles}
 f_{k,c}(s)=\frac{ks}{kc+s},\qquad
 r_c(s)=\frac{s(c+2-2s)}{(c+s)(c+1-s)}.
\end{equation}
At $\omega=0$ the sufficient conditions in \cref{thm:finitepoly} ask for an affine majorant of three functions:
\begin{equation}\label{eq:threeobstaclecriterion}
 \chi(s)\ge\max\{r_c(s),B_c(s),f_{k,c}(s)\}
 \qquad(0\le s\le1).
\end{equation}
We write $k=\infty$ when one parameter choice must work for every finite label count; this replaces $f_{k,c}$ by $f_{\infty,c}(s)=s/c$.

Because only the floor polynomial depends on $k$, a parameter choice gives an explicit range of certified label counts.
\begin{lemma}[Floor certificate across label counts]\phantomsection\label[lemma]{lem:labelcutoff}
Let $c>0$, $\mu>0$, $v>0$, and $\chi(s)=\mu s+v$. Define $\Kbar(c,\mu,v)=+\infty$ when $c(\mu+v)\ge1$. Otherwise put
\[
 h=\sqrt{c\mu}<1,\qquad
 \tau=\frac{v\sqrt{c/\mu}}{1-h},\qquad
 \Kbar(c,\mu,v)=
 \begin{cases}
 \displaystyle\frac{v}{(1-h)^2},&\tau\le1,\\[2mm]
 \displaystyle\frac{\mu+v}{1-c(\mu+v)},&\tau\ge1.
 \end{cases}
\]
The two formulas agree when $\tau=1$. For every integer $k\ge1$,
\begin{equation}\label{eq:labelcutoff}
 k\le\Kbar(c,\mu,v)
 \quad\Longrightarrow\quad
 \mathcal F_k(s)\ge0\ \ (0\le s\le1).
\end{equation}
In particular, the certificate holds for all positive integers up to $\lfloor\Kbar(c,\mu,v)\rfloor$, with $\lfloor+\infty\rfloor=+\infty$.
\end{lemma}
The lemma follows from minimizing the ratio of the two terms of $\mathcal F_k$ over $s$; \cref{app:labelcutoff} gives the calculation. Once $\mathcal B$ and $\mathcal R$ have been verified along a parameter curve, \eqref{eq:labelcutoff} gives a certified range of label counts for that curve.

\begin{corollary}[Cubic certificate for every $k\ge3$]\phantomsection\label[corollary]{thm:uniformcubic}
Let $c>0$, $0\le\omega<1$, $0\le\mu,\zeta\le1$, and put $v=(1-\mu)\zeta$. The weighted level rule satisfies $|p(z)|\le\chi(|z|)$ for every $z\in\Delta_k$ and every finite $k\ge3$ if
\begin{equation}\label{eq:uniformside}
 \mu+v\ge\frac1c
\end{equation}
and
\begin{equation}\label{eq:uniformcubic}
 \mathcal R(s)\ge0\qquad(0\le s\le1).
\end{equation}
These conditions imply \textup{(B0)--(B2)} with $C=c/(1-\omega)$ and $W=\omega/(1-\omega)$. Hence \cref{thm:certificate} gives coefficient $(1-\omega)/(1+c)$, and the action preserves $\K(P)$ for every compact convex down-closed $P$ containing $\zeta\one$.
\end{corollary}
By affinity and $\chi(0)\ge0$, the endpoint condition is equivalent to $\chi(s)\ge s/c$ on $[0,1]$. This dominates every finite-label floor and $B_c$, so \cref{thm:finitepoly} proves sufficiency (\cref{app:uniformcubic}); the remaining cubic need only be checked at the endpoints and at its real stationary points.

\subsection{The certified approximation profile}\label{sec:profile}
Restrict $c\in[c_0,2]$, $\omega\in[0,1/4]$, and $\mu\in[0,1]$. Let $\mathcal A_k(\zeta)$ consist of the triples $(c,\omega,\mu)$ that pass \eqref{eq:finitecriterion} when $k\ge3$, or that pass \eqref{eq:twoweightedtest} together with the side condition $\mu+(1-\mu)\zeta\ge1/c$ when $k=2$; we call such triples \emph{admissible at level $\zeta$}.

\begin{proposition}[Monotonicity in the diagonal level]\phantomsection\label{prop:zetamonotone}
For every $k\ge2$ and $0\le\zeta\le\zeta'\le1$,
\[
 \mathcal A_k(\zeta)\subseteq\mathcal A_k(\zeta').
\]
Consequently, every parameter triple and every coefficient certified at level $\zeta$ remains valid at every larger diagonal level $\zeta'$.
\end{proposition}
\begin{proof}
Fix $(c,\omega,\mu)$; the constraints $c\in[c_0,2]$, $\omega\in[0,1/4]$, $\mu\in[0,1]$ do not involve $\zeta$. The target $\chi(Z)=\mu Z+(1-\mu)\zeta$ satisfies $\partial\chi/\partial\zeta=1-\mu\ge0$. Each test polynomial depends on $\zeta$ only through $\chi$, linearly, with the factor $(kc+Z)$ in $\mathcal F_k$, $(c+Z)(c+Z/2)$ in $\mathcal B$, $(c+Z)(c+1-Z)$ in $\mathcal R$, and $c(c+1)$ in \eqref{eq:twoweightedtest}; all four factors are nonnegative for $0\le Z\le1$ and $c>0$. Hence each test polynomial has $\zeta$-derivative equal to its factor times $1-\mu$, which is nonnegative, and so every test value at $Z\in[0,1]$ is nondecreasing in $\zeta$. The two-label side condition $\mu+(1-\mu)\zeta\ge1/c$ is likewise nondecreasing in $\zeta$ on its left side. Hence every admissibility inequality valid at $\zeta$ remains valid at $\zeta'\ge\zeta$.
\end{proof}

Combining \cref{sec:global,sec:geometry}, every admissible triple gives a comparison on $\K(P)$, valid for all comparators, with coefficient $(1-\omega)/(1+c)$.
\begin{theorem}[Geometry-dependent comparator-uniform linearization]\phantomsection\label{cor:linearizable}
For every $k\ge2$, every compact convex down-closed $P$ containing $\zeta\one$, every nonnegative $k$-submodular extension, and every triple in $\mathcal A_k(\zeta)$, the corresponding rule satisfies
\begin{equation}\label{eq:linearizable}
 \boxed{\alpha F(y)+\alpha C f(0)-F(p(x))
 \le\alpha\langle g_F(x),y-x\rangle\qquad(x,y\in\K(P)),}
\end{equation}
where $\alpha=(1-\omega)/(1+c)$ and $C=c/(1-\omega)$, and $p(\K(P))\subseteq\K(P)$. The field is linear in the objective and its action is objective-independent. The coefficient vectors satisfy $\|\varphi_\ell(z)-\varphi_0(z)\|_1\le3$ for every input block $z$ and label $\ell$.
\end{theorem}
\begin{proof}
Fix a triple $(c,\omega,\mu)\in\mathcal A_k(\zeta)$ and write $u=1/(1-\omega)$, $C=uc$, and $W=u\omega$. For $k\ge3$, \cref{thm:finitepoly} shows that the weighted level rule satisfies \textup{(B0)--(B2)} with these constants and $|p(z)|\le\chi(|z|)$ for every $z\in\Delta_k$. For $k=2$, \cref{lem:twolabel} gives the mass bound from \eqref{eq:twoweightedtest} and the side condition, and then $r=|p(z)|\le\chi(|z|)\le\mu+(1-\mu)\zeta\le1$, so the same lemma gives \textup{(B0)--(B2)}. \Cref{thm:certificate} therefore gives $F(y)+Cf(0)\le(1+C+W)F(p(x))+\ip{g_F(x)}{y-x}$ for all $x,y\in\D$, and, since the mass bound $|p(z)|\le\chi(|z|)=\mu|z|+(1-\mu)\zeta$ is \eqref{eq:lifting}, $p(x)\in\K(P)$ whenever $x\in\K(P)$. Since $1+C+W=1+(c+\omega)/(1-\omega)=(1+c)/(1-\omega)=1/\alpha$, multiplying by $\alpha>0$ gives $\alpha F(y)+\alpha Cf(0)\le F(p(x))+\alpha\ip{g_F(x)}{y-x}$, which is \eqref{eq:linearizable}; note that $\alpha C=\frac{1-\omega}{1+c}\cdot\frac c{1-\omega}=c/(1+c)$. The action depends only on $x$ and the parameters, and each field coordinate is a fixed linear combination of gradient coordinates of $F$, hence linear in $f$.

For the row bound with $k\ge3$, $c\ge c_0=(\sqrt{17}-1)/4>3/4$ gives $1/(c+1)<4/7$, and $\omega\le1/4$ gives $1/(1-\omega)\le4/3$; hence \eqref{eq:weightednorm} gives $(1+1/(c+1))/(1-\omega)<(11/7)(4/3)=44/21<3$. For $k=2$, \cref{lem:twolabel} and $\omega\le1/4$ give $(1+\omega)/(1-\omega)\le(5/4)/(3/4)=5/3$.
\end{proof}

By \cref{cor:linearizable}, explicit guarantees follow from admissible triples at the desired diagonal levels. At $\zeta=0$ the target is $\chi(Z)=\mu Z$, least restrictive at $\mu=1$, where the action may use at most the input mass. For $k\ge3$ the pure-label cubic \eqref{eq:purepoly} then becomes
\begin{equation}\label{eq:baselinecubic}
 N(s)=c^2s+\omega(c+s)(1-s)-s(1-s)(2-s).
\end{equation}
Requiring $N$ to have a double root at some $t\in(0,1)$ leaves a one-parameter family of pairs $(c,\omega)$; \cref{prop:zerodiag} fixes the algebraic member at which the weighted transition of \cref{thm:kprofile}(a) starts, so that the plateau and the transition join at $t_3$ with $\mu=1$, and the analogous double-root choice for the two-label quadratic \eqref{eq:twoweightedtest} at $\mu=1$ and $\zeta=0$.
\Needspace{8\baselineskip}
\begin{proposition}[Zero-diagonal parameters]\phantomsection\label{prop:zerodiag}
\begin{enumerate}
 \item[\textup{(a)}] \emph{Three or more labels.} Let $t_3=0.1844954157\ldots$ be the unique root of the sextic
\[
 \mathsf P_3(t)=9t^6+33t^5-43t^4-27t^3+73t^2-34t+4
\]
in the interval $(461/2500,\,923/5000)$, and set
\begin{equation}\label{eq:threeexact}
 \begin{aligned}
 c(t)&=(350-1417t+810t^2+900t^3-945t^4-243t^5)/119,\\
 \omega(t)&=(394-2324t+1790t^2+1086t^3-2043t^4-486t^5)/357,\\
 c_3&=c(t_3),\qquad \omega_3=\omega(t_3).
 \end{aligned}
\end{equation}
Then
\begin{equation}\label{eq:newdoublefactor}
 N(s)=(3-\omega_3-2t_3-s)(s-t_3)^2,
 \qquad \athree:=\frac{1-\omega_3}{1+c_3}=0.4541264676\ldots,
\end{equation}
the triple $(c_3,\omega_3,1)$ satisfies the hypotheses of \cref{thm:uniformcubic} at every $\zeta\in[0,1]$, and it lies in $\mathcal A_k(\zeta)$ for every $k\ge3$ and every $\zeta\in[0,1]$.
 \item[\textup{(b)}] \emph{Two labels.} Let $\Cs=1.1564593\ldots$ be the unique root in $(1,6/5)$ of the quartic $\mathsf Q_2(x)=x^4-6x^3-8x^2+14x+2$, and set
\begin{equation}\label{eq:twozerodiag}
 \begin{aligned}
 \Ws&=\frac{\Cs^3+6\Cs^2+7\Cs-16}{19},&
 t_2&=\frac{2(1+\Ws)-\Cs^2+\Cs\Ws}{4(1+\Ws)},\\
 c_2&=\frac{\Cs}{1+\Ws},&
 \omega_2&=\frac{\Ws}{1+\Ws}.
 \end{aligned}
\end{equation}
At $(c,\omega,\mu)=(c_2,\omega_2,1)$ and $\zeta=0$, the quadratic \eqref{eq:twoweightedtest} equals $2(1-\omega_2)(Z-t_2)^2$, with $t_2=0.2159183068\ldots$. The triple $(c_2,\omega_2,1)$ lies in $\mathcal A_2(\zeta)$ for every $\zeta\in[0,1]$, and certifies
\[
 \al:=\frac{1-\omega_2}{1+c_2}=\frac1{1+\Cs+\Ws}=0.4456016834\ldots.
\]
\end{enumerate}
\end{proposition}
In part~(b), $\Cs$ and $\Ws$ are the constants $C$ and $W$ of \eqref{eq:normalizedparameters} at $(c_2,\omega_2)$. \Cref{app:threeconstants} proves the proposition: it verifies both factorizations, the isolating intervals, and the side conditions. Since $\mu=1$, the target is $\chi(Z)=Z$ at every level, so none of the tests depends on $\zeta$ and admissibility holds at every diagonal level. The profile keeps these parameters, and hence a constant coefficient, for $\zeta\le t_2$ when $k=2$ and for $\zeta\le t_3$ when $k\ge3$; we call this range the \emph{plateau}.

For two labels, the quadratic mass test gives a closed form on all of $[0,1]$, with breakpoints $t_2$ (\cref{prop:zerodiag}(b)), $\zeta_{\mathrm a}$, and $\zeta_{\mathrm b}$, and the constant $\rho_2$:
\[
 \zeta_{\mathrm a}=\sqrt5-2,\qquad \zeta_{\mathrm b}=1-\frac1{\sqrt2},\qquad \rho_2=\sqrt2-1.
\]
\begin{theorem}[Two-label profile]\phantomsection\label[theorem]{thm:twoclosed}
For every $\zeta\in[0,1]$, the two-label rule of \cref{lem:twolabel} certifies the coefficient
\begin{equation}\label{eq:twoclosedprofile}
 \acf(\zeta)=
 \begin{cases}
 \al, & 0\le\zeta\le t_2,\\[1mm]
 \displaystyle
 \frac{\tilde c_\zeta(1-\zeta+\zeta \tilde c_\zeta)}
 {(1-\zeta)(1+\tilde c_\zeta)(\tilde c_\zeta+2\zeta)},
   & t_2\le\zeta\le \zeta_{\mathrm a},\\[3mm]
 \displaystyle \frac1{1+\sqrt{2(1-\zeta)}},
   & \zeta_{\mathrm a}\le\zeta\le \zeta_{\mathrm b},\\[3mm]
 \displaystyle \frac1{1+\sqrt{2\rho_2+\rho_2^2/\zeta}},
   & \zeta_{\mathrm b}\le\zeta\le1,
 \end{cases}
\end{equation}
where
\begin{equation}\label{eq:twotransitionc}
 \tilde c_\zeta=
 \frac{2\zeta^2+(1-2\zeta)\sqrt{2\zeta(1+\zeta)}}
 {1-2\zeta-2\zeta^2}.
\end{equation}
An admissible triple in $\mathcal A_2(\zeta)$ attains this coefficient. The four pieces have matching values and first derivatives at their junctions; hence $\acf$ is $C^1$ on $[0,1]$ and is nondecreasing. In particular,
\[
\begin{aligned}
 \acf(\zeta_{\mathrm a})&=1/\sqrt5, & \acf(0.3)&=0.4580100437\ldots,\\
 \acf(0.5)&=0.4802169350\ldots, & \acf(0.8)&=0.4947503439\ldots,
\end{aligned}
\]
and $\acf(1)=1/2$.
\end{theorem}
The first piece is the plateau; on the second and third pieces, the quadratic \eqref{eq:twoweightedtest} has a double root at $Z=\zeta$; along the second, $c=\tilde c_\zeta$ solves $\partial_c\alpha=0$, and $\omega$ goes from its plateau value at $t_2$ to zero at $\zeta_{\mathrm a}$. On the last piece, $\omega=0$, the double root stays at $Z=\zeta_{\mathrm b}$, and the side condition $\mu+(1-\mu)\zeta\ge1/c$ holds with equality. The proof is in \cref{app:twoclosedproof}.

For three or more labels, the three requirements in \eqref{eq:threeobstaclecriterion} are concave, so a supporting line either touches one requirement at one point, called a \emph{contact}, or touches two of them at two contacts with a common tangent.
\begin{lemma}[Concavity and ordered crossings]\phantomsection\label[lemma]{lem:concaveobstacles}
For $c>0$, the functions $r_c$, $B_c$, and $f_{k,c}$ are concave on $[0,1]$, strictly so for finite $k$, and decrease strictly with $c$ at every positive mass. For finite $k\ge3$ and $s>0$, the signs of their differences are
\begin{align}
 \operatorname{sgn}(r_c-B_c)&=\operatorname{sgn}\bigl(s^2-(4c+1)s+2c\bigr),\label{eq:obstaclecrossings}\\
 \operatorname{sgn}(r_c-f_{k,c})&=\operatorname{sgn}\bigl((k-2)s^2-((2k-1)c+k-2)s+kc\bigr),\nonumber\\
 \operatorname{sgn}(B_c-f_{k,c})&=\operatorname{sgn}\bigl(c+(3-k)s/2\bigr).\nonumber
\end{align}
The first two differences each change sign exactly once in $(0,1)$, from positive to negative. For $k=3,4$ and $c\ge c_0$, $B_c$ dominates $f_{k,c}$ throughout $[0,1]$. For the all-label floor, $B_c\le f_{\infty,c}$, and $r_c-f_{\infty,c}$ has the sign of $s^2-(2c+1)s+c$, again with one crossing in $(0,1)$.
\end{lemma}
The pure-label concavity follows from the decomposition
\begin{equation}\label{eq:purepartialfractions}
 r_c(s)=2-\frac{c(3c+2)}{(2c+1)(c+s)}
          -\frac{c(c+1)}{(2c+1)(c+1-s)}.
\end{equation}
The remaining calculations are in \cref{app:analyticprofile}.

For $H\in\{r,B,f_k\}$, write $H_c$ for the corresponding requirement at cost $c$. The intercept of its supporting line of slope $\mu$ and the remaining gap to the target are
\begin{equation}\label{eq:supportprofile}
 \begin{aligned}
 M_H(c,\mu)&=\max_{0\le s\le1}\{H_c(s)-\mu s\},\\
 \mathcal E_k(c,\zeta)&=\min_{0\le\mu\le1}
 \left\{\max_{H\in\{r,B,f_k\}}M_H(c,\mu)-(1-\mu)\zeta\right\}.
 \end{aligned}
\end{equation}
In the \emph{unweighted} case $\omega=0$, select the cost and coefficient by
\begin{equation}\label{eq:unweightedprofile}
 \begin{aligned}
 c_k^{\mathrm u}(\zeta)&=\min\{c\in[c_0,\sqrt2]:\mathcal E_k(c,\zeta)\le0\},\\
 \alpha_k^{\mathrm u}(\zeta)&=\frac1{1+c_k^{\mathrm u}(\zeta)}.
 \end{aligned}
\end{equation}
\begin{theorem}[Continuous supporting-line construction]\phantomsection\label[theorem]{thm:unweightedenvelope}
For each $k\in\{3,4,\ldots\}\cup\{\infty\}$ and $\zeta\in[0,1]$, the minimum in \eqref{eq:unweightedprofile} is attained and its coefficient is certified by \cref{alg:levelrule}. The function $c_k^{\mathrm u}$ is continuous and nonincreasing, so $\alpha_k^{\mathrm u}$ is continuous and nondecreasing. These profiles are nonincreasing in $k$, and $\alpha_3^{\mathrm u}=\alpha_4^{\mathrm u}$ everywhere.

The parameters are determined by finitely many algebraic contact types. An interior supporting contact of $H_c$ satisfies $H_c'(s)=\mu$; the equations for $r_c$ and $B_c$ have degree at most four. For $\mu>0$ the floor contact is
\begin{equation}\label{eq:floorcontact}
 s_f=\min\left\{1,\max\left\{0,k\left(\sqrt{c/\mu}-c\right)\right\}\right\};
\end{equation}
at $\mu=0$ it is one. For $f_{\infty,c}$, an endpoint maximizes the support, and both do when $\mu=1/c$.

At the selected cost $c=c_k^{\mathrm u}(\zeta)$, away from the boundaries $c=c_0$ and $\mu\in\{0,1\}$, a minimizing line in \eqref{eq:supportprofile} either has an active contact at $s=\zeta$, or two active contacts $s\le\zeta\le t$ satisfying
\begin{align}
 H_c(s)-\mu s&=J_c(t)-\mu t=(1-\mu)\zeta,\label{eq:contactvalues}\\
 H_c'(s)&=J_c'(t)=\mu,\qquad H,J\in\{r,B,f_k\}.\label{eq:contactslopes}
\end{align}
At an endpoint contact, the corresponding derivative equality is replaced by $H_c'(0)\le\mu$ or $H_c'(1)\ge\mu$. Enumerating the three single-contact types, the three pairs, and the cost and slope boundaries, then retaining the least cost that passes all support tests, determines the parameters in \eqref{eq:unweightedprofile}. The boundary $\mu=1$ has cost $c=\sqrt2$.
\end{theorem}
The proof in \cref{app:analyticprofile} uses convexity of the support functions in $\mu$: at an interior minimum in the slope variable, the active contact masses must bracket $\zeta$. The single-contact equations recover the pure and balanced formulas below and give a third branch,
\begin{equation}\label{eq:floorbranch}
 c=1-\zeta/k,\qquad \mu=c,\qquad v=\zeta^2/k,
 \qquad \mathcal F_k(s)=c(s-\zeta)^2.
\end{equation}
Indeed, a floor contact at $s=\zeta$ requires $f_{k,c}(\zeta)=\zeta$, that is, $kc+\zeta=k$, so $c=1-\zeta/k$; the slope is then $\mu=f_{k,c}'(\zeta)=k^2c/(kc+\zeta)^2=c$, and the line through $(\zeta,\zeta)$ has intercept $v=\zeta-c\zeta=\zeta^2/k$. With $\chi(s)=cs+\zeta^2/k$, direct expansion gives $(kc+s)\chi(s)-ks=c(s-\zeta)^2$, and \eqref{eq:floorcontact} gives the contact $s_f=k(1-c)=\zeta$. A branch is used exactly where its tangent also dominates the other two requirements, and common tangents fill the intervals between such branches. For $k\ge13$ the floor branch is unavailable even at $\zeta=1$: there its tangent $\chi(s)=(1-1/k)s+1/k$ violates the pure-label requirement, since
\[
 \chi(\tfrac14)-r_c(\tfrac14)=-\frac{5k^3-73k^2+128k-48}{4k(5k-4)(7k-4)}\qquad(c=1-1/k),
\]
and with $k=13+h$ the numerator equals $5h^3+122h^2+765h+264>0$ for $h\ge0$. The minimizing line of \cref{thm:unweightedenvelope} is then of another type.

For three or more labels, between the plateau and the level $\zeta_0$ at which $\omega$ reaches zero, the profile follows a \emph{transition} branch that does not depend on $k$; this branch and the two tangency curves have explicit algebraic formulas.
\begin{definition}[Transition and tangency parameters]\phantomsection\label{def:curves}
Let $\zeta_0$ be the unique root of
\begin{equation}\label{eq:junctionroot}
 4z^5-16z^4+12z^3+15z^2-13z+2=0
\end{equation}
in the interval $(1/5,\,43/200)$; this is the level at which the transition weight reaches zero. For $t_3\le\zeta\le\zeta_0$, define
\begin{equation}\label{eq:transitionpolynomials}
 \begin{aligned}
 a_\zeta&=\zeta^2+2\zeta-1,&
 b_\zeta&=2\zeta(\zeta^2-\zeta+1),\\
 h_\zeta&=\zeta^4-\zeta^3-2\zeta^2+2\zeta,&
 c_\zeta&=\frac{b_\zeta+\sqrt{b_\zeta^2-4a_\zeta h_\zeta}}{-2a_\zeta},
 \end{aligned}
\end{equation}
followed by
\begin{equation}\label{eq:transitionparameters}
 \begin{aligned}
 \omega_\zeta&=\frac{\zeta\bigl((1-\zeta)(2-\zeta)-c_\zeta^2\bigr)}
 {(c_\zeta+\zeta)(1-\zeta)},\\
 \mu_\zeta&=\frac{c_\zeta+2-5\zeta+2\zeta^2
               -\omega_\zeta(1-c_\zeta-2\zeta)}
 {c_\zeta(c_\zeta+1)+\zeta-\zeta^2},\qquad
 \alpha_{\mathrm{tr}}(\zeta)=\frac{1-\omega_\zeta}{1+c_\zeta}.
 \end{aligned}
\end{equation}
The two single-contact tangency curves are
\begin{align}
 c_P(z)&=\sqrt{(1-z)(2-z)},&
 \mu_P(z)&=\frac{c_P(z)+2-5z+2z^2}{c_P(z)(c_P(z)+1)+z-z^2},\label{eq:pureclosed}\\
 c_B(z)&=\frac{2-3z+\sqrt{z^2+12z+4}}4,&
 \mu_B(z)&=\left.\frac{\partial B_c(s)}{\partial s}\right|_{c=c_B(z),\ s=z}.\label{eq:balancedclosed}
\end{align}
Only parameters with $c_H(z)\ge c_0$ and $0<\mu_H(z)<1$ are considered. Both use $\omega=0$ and $\chi(s)=\mu_H(z)s+v_H(z)$ for $H\in\{P,B\}$; for $0<z\le1$, their certified label-count ranges are specified by
\begin{equation}
 v_H(z)=(1-\mu_H(z))z,\qquad
 K_H(z)=\left\lfloor\Kbar(c_H(z),\mu_H(z),v_H(z))\right\rfloor.
\end{equation}
\end{definition}
On $[t_3,\zeta_0]$ all denominators in \eqref{eq:transitionparameters} are positive, and the transition parameters factor the pure-label cubic as $\mathcal R(s)=[2-\omega_\zeta-\zeta+(1-\zeta)\mu_\zeta-\mu_\zeta s](s-\zeta)^2$ (\cref{app:transition}); along the pure tangency curve the same factorization holds with $\omega=0$.

In \cref{thm:kprofile}, the support tests of parts~(b)--(c) determine where each tangency curve applies.
\begin{theorem}[Profile for three or more labels]\phantomsection\label[theorem]{thm:kprofile}
Let $k\ge3$.
\begin{enumerate}
 \item[\textup{(a)}] \emph{Plateau and weighted transition.} The zero-diagonal parameters of \cref{prop:zerodiag}(a) certify $\athree$ for $0\le\zeta\le t_3$. For $t_3\le\zeta\le\zeta_0$, the parameters in \eqref{eq:transitionpolynomials}--\eqref{eq:transitionparameters} satisfy \cref{thm:uniformcubic} and certify $\alpha_{\mathrm{tr}}(\zeta)$. The parameters equal $(c_3,\omega_3,1)$ at $t_3$, and $\omega_{\zeta_0}=0$, $c_{\zeta_0}=\sqrt{(1-\zeta_0)(2-\zeta_0)}$. Moreover
\begin{equation}
 \alpha_{\mathrm{tr}}'(\zeta)
 =\frac{(c_\zeta+1-\zeta)(1-\mu_\zeta)}{(c_\zeta+1)(1-\zeta)},
 \qquad \mu_\zeta'<0.
\end{equation}
Thus $\alpha_{\mathrm{tr}}$ increases strictly above $t_3$, with $\alpha_{\mathrm{tr}}'(t_3)=0$ and $\alpha_{\mathrm{tr}}''(t_3)>0$.
 \item[\textup{(b)}] \emph{Pure-label tangency.} Suppose $\zeta_0\le\zeta<1$, $c_P(\zeta)\ge c_0$ and $0<\mu_P(\zeta)<1$. The pure-curve parameters are admissible if
\begin{equation}\label{eq:pureacceptance}
 M_B(c_P(\zeta),\mu_P(\zeta))\le v_P(\zeta),\qquad k\le K_P(\zeta).
\end{equation}
Where these conditions hold, they certify
\begin{equation}
 \alpha_{\mathrm{pu}}(\zeta)
 =\frac1{1+\sqrt{(1-\zeta)(2-\zeta)}}.
\end{equation}
At $\zeta_0$, the value and first derivative agree with those of $\alpha_{\mathrm{tr}}$.
 \item[\textup{(c)}] \emph{Balanced tangency.} Suppose $c_B(\zeta)\ge c_0$ and $0<\mu_B(\zeta)<1$. The balanced-curve parameters are admissible if
\begin{equation}\label{eq:balancedacceptance}
 M_r(c_B(\zeta),\mu_B(\zeta))\le v_B(\zeta),\qquad k\le K_B(\zeta).
\end{equation}
Where these conditions hold, they certify
\begin{equation}
 \alpha_{\mathrm{ba}}(\zeta)
 =\frac4{6-3\zeta+\sqrt{\zeta^2+12\zeta+4}}.
\end{equation}
 \item[\textup{(d)}] \emph{Unrestricted level.} At $\zeta=1$, the triple $(c,\omega,\mu)=(\max\{c_0,1-1/k\},0,0)$ is admissible and certifies
 \[
  \frac1{1+\max\{c_0,1-1/k\}}.
 \]
\end{enumerate}
The coefficient formulas in parts~\textup{(b)} and \textup{(c)} increase strictly on their accepted ranges.
\end{theorem}
Part~(a) is proved in \cref{app:transition}, parts~(b)--(c) in \cref{app:closedprofiles}, and part~(d) in \cref{app:exactendpointdetail}; the tangency tests use only the support calculations of \cref{thm:unweightedenvelope}.

On the following intervals the tangency curves pass every support test; \cref{app:explicitranges} proves this with rational identities and derivative inequalities.
\begin{corollary}[Explicit pure and balanced intervals]\phantomsection\label[corollary]{cor:explicitranges}
The pure-label parameters certify
\[
 \alpha_{\mathrm{pu}}(\zeta)=\frac1{1+\sqrt{(1-\zeta)(2-\zeta)}}
\]
for $3\le k\le5$ on $[\zeta_0,61/200]$, and for $3\le k\le6$ on $[\zeta_0,3/10]$.
The balanced parameters certify
\[
 \alpha_{\mathrm{ba}}(\zeta)=\frac4{6-3\zeta+\sqrt{\zeta^2+12\zeta+4}}
\]
for $k=3,4$ on $[17/25,1]$, and for $3\le k\le5$ on $[17/25,31/40]$.
The plateau, transition, and pure formulas join with matching first derivatives on the stated initial ranges.
\end{corollary}
The supporting-line construction of \cref{thm:unweightedenvelope} covers the remaining diagonal levels and every finite label count.

\begin{theorem}[Continuous achievable profile]\phantomsection\label[theorem]{cor:everyk}
For $k=2$, set $\alphacert_2=\acf$, the profile of \cref{thm:twoclosed}. For $k\in\{3,4,\ldots\}\cup\{\infty\}$, set
\begin{equation}\label{eq:certifiedprofile}
 \alphacert_k(\zeta)=
 \begin{cases}
  \athree,&0\le\zeta\le t_3,\\
  \alpha_{\mathrm{tr}}(\zeta),&t_3\le\zeta\le\zeta_0,\\
  \alpha_k^{\mathrm u}(\zeta),&\zeta_0\le\zeta\le1.
 \end{cases}
\end{equation}
Every value is attained by an admissible block rule. For $\zeta\ge\zeta_0$ the selection \eqref{eq:unweightedprofile} ranges over unweighted triples ($\omega=0$).

Each profile is continuous and nondecreasing. The profiles for $k\ge3$ are nested, and $\alphacert_3=\alphacert_4$ on $[0,1]$.
\end{theorem}
\begin{proof}
The individual pieces are feasible and continuous by \cref{thm:kprofile,thm:unweightedenvelope}. At $t_3$ the transition parameters equal $(c_3,\omega_3,1)$ (\cref{thm:kprofile}(a)), so $\alpha_{\mathrm{tr}}(t_3)=\athree$ and the first two pieces agree. At $\zeta_0$, \cref{thm:kprofile}(a) gives $\omega_{\zeta_0}=0$ and $c_{\zeta_0}=c_P(\zeta_0)=\sqrt{(1-\zeta_0)(2-\zeta_0)}$, an unweighted line admissible for every $k\ge3$; it is therefore included in the scalar parameter selection \eqref{eq:unweightedprofile}, so $c_k^{\mathrm u}(\zeta_0)\le c_P(\zeta_0)$. Conversely, every line in that selection has the form $\chi(s)=\mu s+(1-\mu)\zeta_0$, so $\chi(\zeta_0)=\zeta_0$, and evaluating the pure-label requirement at $s=\zeta_0$ gives $r_c(\zeta_0)\le\zeta_0$. Since
\[
 \zeta_0-r_c(\zeta_0)=\frac{\zeta_0\bigl(c^2-(1-\zeta_0)(2-\zeta_0)\bigr)}{(c+\zeta_0)(c+1-\zeta_0)}
\]
with a positive denominator and $\zeta_0>0$,
\[
 r_c(\zeta_0)\le\zeta_0
 \quad\Longleftrightarrow\quad
 c^2\ge(1-\zeta_0)(2-\zeta_0),
\]
so $c_k^{\mathrm u}(\zeta_0)\ge c_P(\zeta_0)$. Thus the shared tangent realizes the minimum in \eqref{eq:unweightedprofile} at this level, and $\alpha_{\mathrm{tr}}(\zeta_0)=1/(1+c_P(\zeta_0))=\alpha_k^{\mathrm u}(\zeta_0)$. The pieces therefore agree at both junctions, and continuity follows. The plateau is constant, $\alpha_{\mathrm{tr}}$ increases by \cref{thm:kprofile}(a), and $\alpha_k^{\mathrm u}$ is nondecreasing by \cref{thm:unweightedenvelope}; a function that is nondecreasing on each of finitely many consecutive closed intervals and continuous at the junctions is nondecreasing. For nesting, the first two pieces do not depend on $k$, and $f_{k,c}(s)=ks/(kc+s)$ has $\partial_kf_{k,c}(s)=s^2/(kc+s)^2\ge0$, with $f_{\infty,c}(s)=s/c$ as its limit; so the requirement in \eqref{eq:threeobstaclecriterion} grows with $k$, the feasible costs shrink, and $\alpha_k^{\mathrm u}$ is nonincreasing in $k$. For $k=3,4$, \cref{lem:concaveobstacles} gives $B_c\ge f_{k,c}$ on $[0,1]$ for $c\ge c_0$, so the floor requirement is redundant and $\alphacert_3=\alphacert_4$.
\end{proof}
\Cref{fig:profile} plots these profiles.

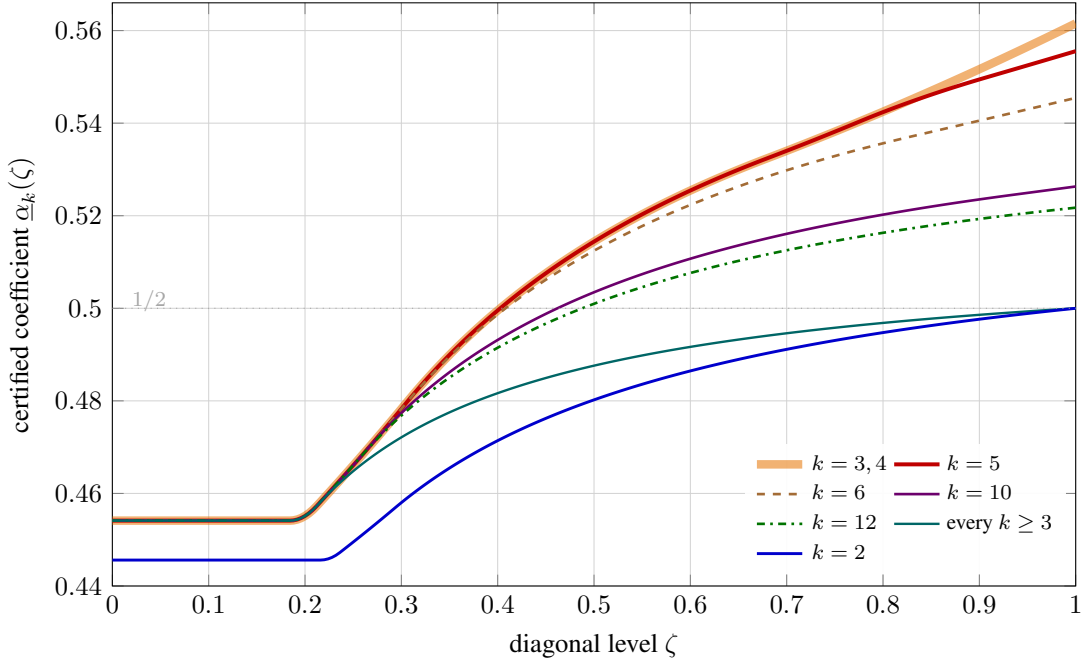
\begin{figure}[t]
\centering
\begin{tikzpicture}
\begin{axis}[
    width=0.94\textwidth,
    height=0.61\textwidth,
    xmin=0,xmax=1,
    ymin=0.44,ymax=0.566,
    xlabel={diagonal level $\zeta$},
    ylabel={certified coefficient $\alphacert_k(\zeta)$},
    grid=both,
    grid style={line width=.1pt, draw=gray!20},
    major grid style={line width=.2pt, draw=gray!35},
    tick label style={font=\small},
    label style={font=\small},
    yticklabel style={/pgf/number format/fixed,/pgf/number format/precision=2},
    legend style={font=\scriptsize, draw=none, fill=white, fill opacity=0.94,
                  text opacity=1, at={(0.99,0.02)}, anchor=south east},
    legend columns=2,
    legend cell align=left,
]
\addplot[color=orange!90!black,line width=3.2pt,opacity=0.55]
    table[x=z,y=a] {figures/profile_k34.dat};
\addlegendentry{$k=3,4$}
\addplot[color=red!75!black,line width=1.5pt]
    table[x=z,y=a] {figures/profile_k5.dat};
\addlegendentry{$k=5$}
\addplot[color=brown!85!black,line width=1.0pt,dashed]
    table[x=z,y=a] {figures/profile_k6.dat};
\addlegendentry{$k=6$}
\addplot[color=violet!85!black,line width=1.0pt]
    table[x=z,y=a] {figures/profile_k10.dat};
\addlegendentry{$k=10$}
\addplot[color=green!45!black,line width=1.0pt,dashdotted]
    table[x=z,y=a] {figures/profile_k12.dat};
\addlegendentry{$k=12$}
\addplot[color=teal!80!black,line width=0.9pt]
    table[x=z,y=a] {figures/profile_kinf.dat};
\addlegendentry{every $k\ge3$}
\addplot[color=blue!80!black,line width=1.1pt]
    table[x=z,y=a] {figures/profile_k2.dat};
\addlegendentry{$k=2$}
\addplot[gray!70,dotted,forget plot] coordinates {(0,0.5) (1,0.5)};
\node[anchor=west,font=\scriptsize,text=gray!70] at (axis cs:0.01,0.5017) {$1/2$};
\end{axis}
\end{tikzpicture}
\caption{Achievable comparison and online coefficients \eqref{eq:certifiedprofile}. The two-label curve is the closed form of \cref{thm:twoclosed}. For $k\ge3$, the plateau and the transition use the weighted curves of \cref{thm:kprofile}, and the later pieces use the supporting-line equations of \cref{thm:unweightedenvelope}. For the all-label curve, each parameter choice works for every finite $k\ge3$. The profiles for $k\ge3$ are nested; the two-label curve uses a different rule.}\label{fig:profile}
\end{figure}

The endpoint values in \cref{tab:comparison,tab:freecomparison} follow from \cref{thm:twoclosed,thm:kprofile}.
\begin{corollary}[Endpoint coefficients]\phantomsection\label[corollary]{thm:exactendpoint}
The certified profile starts at the zero-diagonal constants and attains at least the following unrestricted coefficients:
\begin{equation}
 \begin{aligned}
 \alphacert_2(0)&=\al=0.4456016834\ldots,
 &\alphacert_k(0)&=\athree=0.4541264676\ldots\quad(k\ge3),\\
 \alphacert_k(1)&\ge b_k:=
 \begin{cases}
 1/2,&k=2,\\
 (\sqrt{17}-3)/2,&k=3,4,\\
 k/(2k-1),&k\ge5.
 \end{cases}
 \end{aligned}
\end{equation}
\end{corollary}
\begin{proof}
For $k=2$, the first piece of \eqref{eq:twoclosedprofile} gives $\acf(0)=\al$, and the last piece gives $\acf(1)=1/(1+\sqrt{2\rho_2+\rho_2^2})=1/2$, because $2\rho_2+\rho_2^2=(\rho_2+1)^2-1=1$. For $k\ge3$, the plateau of \cref{thm:kprofile}(a) gives $\alphacert_k(0)=\athree$. At $\zeta=1$, \cref{thm:kprofile}(d) certifies $1/(1+c_k^{\mathrm{end}})$ with $c_k^{\mathrm{end}}=\max\{c_0,1-1/k\}$ and $\omega=0$, so this cost is feasible in \eqref{eq:unweightedprofile} and $\alphacert_k(1)=\alpha_k^{\mathrm u}(1)\ge1/(1+c_k^{\mathrm{end}})$. Since $4<\sqrt{17}<21/5$, we have $3/4<c_0<4/5$, and thus $1-1/k\le3/4<c_0$ for $k=3,4$, giving $1/(1+c_0)=4/(3+\sqrt{17})=(\sqrt{17}-3)/2$, and $1-1/k\ge4/5>c_0$ for $k\ge5$, giving $1/(2-1/k)=k/(2k-1)$.
\end{proof}
\Cref{cor:unconstrained} turns the unconstrained coefficients into an integral algorithm with regret bounds.

The following examples make the construction explicit in three cases. For five labels, the balanced and floor requirements both depend only on $s/c$, so their common tangent has a fixed slope and intercept in that scaled variable, which gives two explicit branches.
\begin{example}[Five-label common tangent and terminal branch]
\phantomsection\label[example]{thm:fivelabel}
Define $b(x)=x(2+3x)/((1+x)(2+x))$, and let $t$ be the unique positive root of
\begin{equation}
 t^5+26t^4+94t^3+70t^2-75t-80=0.
 \label{eq:fiveroot}
\end{equation}
Set
\begin{align}
 p&=\frac{t^4+13t^3+50t^2+60t+20}{5(t+1)^2(t+2)^2},\qquad
 M=p^2,\qquad V=5(1-p)^2,\\
 \zeta_L&=b(t),\qquad \zeta_U=5(1-p).
\end{align}
Then the following choices are feasible for every $3\leq k\leq5$:
\begin{equation}
 (c,\mu,v)=
 \begin{cases}
 \left(\dfrac{M\zeta}{\zeta-V},1-\dfrac{V}{\zeta},V\right),
       &\zeta_L\leq\zeta\leq\zeta_U,\\[7pt]
 (1-\zeta/5,\,1-\zeta/5,\,\zeta^2/5),
       &\zeta_U\leq\zeta\leq1.
 \end{cases}
\end{equation}
Both branches use $\omega=0$ and certify $1/(1+c)$. The constants evaluate to
\[
 p=0.8183102900\ldots,\qquad
 \zeta_L=0.7762242330\ldots,\qquad
 \zeta_U=0.9084485499\ldots.
\]
\end{example}

The proof in \cref{app:fiveprofile} uses the common tangent in the scaled variable followed by a floor tangent and bounds the pure-label requirement by two cubic inequalities; the choices remain feasible for fewer labels because the floor requirement is weaker.

Two rational parameter choices give guarantees that need no root computation; they are verified in \cref{app:elementaryimprovements} by completing squares and using the crossing signs of \cref{lem:concaveobstacles}.
\begin{example}[Half coefficient for up to twelve labels]\phantomsection\label[example]{thm:half}
If $3\le k\le12$ and $\zeta(P)\ge49/100$, the unweighted comparison has coefficient $\alpha=1/2$. At level $\zeta=49/100$, use
\[
c=1,\qquad\omega=0,\qquad\mu=\frac{17}{20},\qquad
\chi(s)=\frac{17}{20}s+\frac{147}{2000}.
\]
\end{example}

\begin{example}[A $45/86$ coefficient for up to ten labels]\phantomsection\label[example]{thm:high}
If $3\le k\le10$ and $\zeta(P)\ge9/10$, then the unweighted comparison has coefficient $\alpha=45/86$. At level $9/10$, the parameters are
\[
c=\frac{41}{45},\qquad\omega=0,\qquad\mu=\frac9{10},\qquad
\chi(s)=\frac9{10}s+\frac9{100}.
\]
\end{example}

At cost $c=1$, supporting lines give a half-coefficient family for every finite label count. The line is a common tangent of the pure-label and balanced requirements for up to five labels and of the pure-label and floor requirements for six through eleven labels; for twelve or more labels it is the tangent to the pure-label requirement that passes through the value of the floor requirement at full mass. With $r_1$ and $B_1$ the requirements \eqref{eq:unweightedobstacles} and \eqref{eq:balancedmass} at $c=1$, write
\begin{equation}\label{eq:halftangentdata}
 m(t)=r_1'(t)=\frac{6-8t+t^2}{(1+t)^2(2-t)^2},\qquad
 j(t)=r_1(t)-tm(t)=\frac{t^2(2t^2-6t+7)}{(1+t)^2(2-t)^2}.
\end{equation}
Thus $m(t)s+j(t)$ is the tangent to $r_1$ at mass $t$.
\begin{theorem}[A half-coefficient family for every label count]\phantomsection\label{thm:halfthreshold}
For each integer $k\ge3$, choose $\tau_k$ as follows.
\begin{enumerate}
 \item[\textup{(a)}] If $3\le k\le5$, let $(\tau_*,\sigma_*)$ be the unique solution in
 $(3/10,61/200)\times(1/2,1)$ of
 \begin{equation}\label{eq:halfpurebalanced}
 m(\tau_*)=B_1'(\sigma_*),\qquad
 j(\tau_*)=B_1(\sigma_*)-\sigma_* B_1'(\sigma_*),
 \end{equation}
 and put $\tau_k=\tau_*$.
 \item[\textup{(b)}] If $6\le k\le11$, let $\tau_k$ be the unique root in $(1/4,3/10)$ of
 \begin{equation}\label{eq:halfpurefloor}
 j(t)=k\bigl(1-\sqrt{m(t)}\bigr)^2.
 \end{equation}
 \item[\textup{(c)}] If $k\ge12$, let $\tau_k$ be the unique root in $(1/5,1/4)$ of
 \begin{equation}\label{eq:halfthresholdpoly}
 (k+2)t^4-(4k+6)t^3+(11k+8)t^2-(12k+8)t+2k+6=0.
 \end{equation}
\end{enumerate}
Define
\begin{equation}\label{eq:halfthreshold}
 \mu_k=m(\tau_k),\qquad
 \zeta_{1/2}(k)=\frac{j(\tau_k)}{1-m(\tau_k)}.
\end{equation}
Then $c=1$, $\omega=0$, and $\mu=\mu_k$ certify coefficient $1/2$ whenever $\zeta\ge\zeta_{1/2}(k)$. The sufficient levels are equal for $k=3,4,5$, increase strictly thereafter, and converge to one. Rounded to six decimal places, the first values are
\[
\begin{array}{c|rrrrrrr}
 k&3\text{--}5&6&7&8&9&10&11\\ \hline
 \zeta_{1/2}(k)&.402682&.408161&.421537&.435140&.448676&.462009&.475074
\end{array}
\]
and
\[
 \zeta_{1/2}(12)=0.4878056924\ldots,\qquad
 \zeta_{1/2}(13)=0.5000788601\ldots.
\]
\end{theorem}
For $k\ge12$, \eqref{eq:halfthresholdpoly} states that the pure-label tangent takes the value $q_k=k/(k+1)=f_{k,1}(1)$ at mass one, so \eqref{eq:halfthreshold} can also be written as $(q_k-\mu_k)/(1-\mu_k)$. \Cref{app:halfthreshold} derives the three regimes, checks the remaining requirement in each, and proves the rate of convergence to one.

\FloatBarrier
\section{Fractional Optimization and Online Learning}\label{sec:stochastic}
\Cref{cor:linearizable} is an upper-linearizability statement in the sense of \citet{PedramfarAggarwal24}: the nonlinear reward $F(p(x))$ is controlled by the linear function $y\mapsto\ip{g_F(x)}{y-x}$ of the comparator. Any online linear optimization method over $\K(P)$ that is fed estimates of $g_F(x_t)$ therefore yields approximation and regret guarantees with coefficient $\alpha$. \Cref{app:stochastic} contains the oracle implementations, their moment bounds, the transfer from linear regret, and the proofs.

\subsection{Access model and field estimators}
The discrete value oracle accepts every labeling, including labelings whose supports violate $\mathcal S$, and the value-based estimators below may query such labelings. The first-order oracle returns $\nabla F(x)$; the first-order estimator queries it only at feasible points (see below), and the field can also be estimated from discrete values alone. Convex projections are separate constraint oracles, and we count objective queries and projections separately in a real-arithmetic model, with exact projections; for the entropy map, feasible approximate projections with variational residuals $\varrho_0,\varrho_1,\ldots,\varrho_T$ add $(\varrho_0+\sum_t\varrho_t)/\eta$ to the linear-regret bound (\cref{lem:entropy}). Sampled differences are scaled by the singleton quantities
\begin{equation}\label{eq:scales}
 L_e=\max_i\bigl(f(e\mapsto i)-f(0)\bigr),\quad L=\max_e L_e,\quad\Lambda^2=\sum_eL_e^2.
\end{equation}
Set $L=\Lambda=0$ when $V$ is empty. By the auxiliary \cref{lem:marginal} in \cref{app:marginalproof}, $|\Delta_{e,i}f(s)|\le L_e$ whenever $s_e=0$, also for infeasible $s$, and $|\partial_{e,i}F|\le L_e$ on $\D$; if $L=0$ the objective is constant, and with feasible singletons $L\le\OPT-f(0)$. Offline, the $nk+1$ values of the empty labeling and the singletons determine these scales; online, a supplied uniform bound replaces them.

By \eqref{eq:gdef}, the field in block $e$ is a fixed matrix applied to one extension gradient:
\[
 g_{F,e}=M_e\nabla_eF(\Pi^e),\qquad M_{e,\ell j}=(\varphi_{e,\ell}-\varphi_{e,0})_j,\qquad \sum_j|M_{e,\ell j}|\le G.
\]
The block rules of \cref{sec:geometry} give the row bound
\begin{equation}\label{eq:feedbackbound}
 G=\begin{cases}
 (1+\omega)/(1-\omega),&k=2,\\[2pt]
 (1+1/(c+1))/(1-\omega),&k\ge3,
 \end{cases}
\end{equation}
whose exact zero-diagonal values are $(1+\omega_2)/(1-\omega_2)$ and $(1+1/(c_3+1))/(1-\omega_3)$, with the algebraic parameters of \cref{prop:zerodiag}; their decimal evaluations are $1.1753947\ldots$ and $1.6364205\ldots$. The matrices, the element order, and the lower chain are computed from the state before the objective is queried. The lower-chain point $\Pi^e$ agrees with the action $p(x)$ on the predecessors of $e$ and vanishes elsewhere, so $|\Pi^e|\le|p(x)|$ coordinatewise; since $p(x)\in\K(P)$ and $P$ is down-closed, $\Pi^e\in\K(P)$, and every first-order query is feasible. By \eqref{eq:expectedmarginal}, each gradient coordinate $\partial_{e,j}F(\Pi^e)$ is the expected marginal $\Delta_{e,j}f(R)$ of a random labeling $R$ that samples the predecessors of $e$ from the action and leaves $e$ and its successors unselected. An unbiased estimate of $g_F(x)$ is therefore obtained by choosing a uniformly random element $E$ and using one of the following: one extension gradient at $\Pi^E$ (\emph{one-gradient} feedback); the two values $f(R)$ and $f(R+E\mapsto I)$ for a uniform label $I$ (\emph{two-value} feedback); or one of these two values, chosen by an independent random sign (\emph{one-value} feedback). A cached scan of one sampled chain evaluates all $nk$ marginals at once (\emph{full-prefix} feedback, $nk+1$ values).

\subsection{Fractional optimization and online guarantees}
The algorithm maintains a state $x_t\in\K(P)$, plays the action $p_t=p(x_t)$, which lies in $\K(P)$ by \cref{cor:linearizable}, and feeds the field estimate $\widehat g_t$ to an online linear learner that updates the state. For entropy mirror ascent, lift $x$ to the categorical vector $q$ with $q_{e,0}=1-|x_e|$ and $q_{e,i}=x_{e,i}$ for $i\ge1$, let $\overline\K(P)$ be the region of these lifts for $x\in\K(P)$, and define
\[
 H(q)=\sum_{e\in V}\sum_{i=0}^kq_{e,i}\log q_{e,i},\qquad
 D_H(q,q')=\sum_{e\in V}\sum_{i=0}^kq_{e,i}\log\frac{q_{e,i}}{q'_{e,i}}.
\]
Use $0\log0=0$ and extend the field estimate to the categorical coordinates by $\widehat g_{e,0}=0$. Initialize $q_1$ at the minimizer of $H$ on the feasible relative face, and update with step size $\eta>0$ by
\[
 q_{t+1}=\argmax_{q\in\overline\K(P)}\bigl\{\eta\langle\widehat g_t,q\rangle-D_H(q,q_t)\bigr\}.
\]
The divergence from the initialization to any feasible comparator is at most $B_0=n\log(k+1)$. Coordinates fixed at zero are removed from these expressions, and the standard mirror-map properties give a well-defined positive minimizer on the remaining face \citep[Sections 4.1--4.3]{Bubeck15}. Each projection is a separable convex problem over the support masses, implemented by an explicit packing description or a matroid separation oracle. The Euclidean alternative projects $x_t+\eta\widehat g_t$ onto $\K(P)$.

In the online protocol, $f_t$ is the objective of round $t\in[T]$ and $F_t$ its multilinear extension. The learner commits $p_t$, or an objective-independent rounding of it, before obtaining current feedback, and all queries in round $t$ refer to $f_t$. These post-decision queries may be at labelings other than the rewarded action, so the protocol differs from full-bandit feedback, and the earned reward need not be observed. An oblivious adversary fixes the sequence in advance; a nonanticipating adaptive adversary fixes $f_t$ from the past before the round's fresh action and query randomness. A noisy response is centered conditional on the current objective, query, and all previous choices and responses, with conditional variance at most $\sigma^2$. For $\alpha\in(0,1]$, the expected approximate regret is
\begin{equation}\label{eq:regretdefinition}
 \operatorname{Reg}_{\alpha}(T)=
 \E\left[\alpha\max_{y\in\K(P)}\sum_{t=1}^TF_t(y)-\sum_{t=1}^TF_t(p_t)\right].
\end{equation}
For an oblivious sequence the maximum is deterministic, and integral plays replace $F_t(p_t)$ by the realized reward. For adaptive sequences assume $\E\sum_{t=1}^T f_t(0)<\infty$; the marginal bounds below then make the objective and regret terms integrable.

For a single objective, running the learner on the constant sequence $f_t=f$ and returning a uniformly chosen action gives a fractional approximation.
\begin{theorem}[Fractional optimization]\phantomsection\label[theorem]{thm:continuous}
Fix a nonempty compact convex down-closed support region $P$ containing $\zeta\one$, for which the entropy projections are available, and a triple in $\mathcal A_k(\zeta)$, with constants $\alpha$ and $C$ as in \cref{cor:linearizable} and row bound $G$ as in \eqref{eq:feedbackbound}. Suppose $\Lambda>0$; otherwise $f$ is constant by \cref{lem:marginal} and every feasible output is optimal. Run entropy mirror ascent with full-prefix feedback on the constant objective $f$ for $T$ iterations, with $\eta=\sqrt{2n\log(k+1)/(TG^2\Lambda^2)}$, and return $p_J$ for an independent uniform $J\in[T]$. Then $p_J\in\K(P)$ and, for every $y\in\K(P)$,
\begin{equation}\label{eq:continuous}
 \E F(p_J)\ge\alpha F(y)+\alpha C f(0)
 -\alpha G\sqrt{\frac{2n\Lambda^2\log(k+1)}T}.
\end{equation}
For every $\eps>0$, the horizon
$T=\max\{1,\lceil2\alpha^2G^2n^2\log(k+1)/\eps^2\rceil\}$ makes the last term at most $\eps L$, using at most $nk+1+Tnk$ value queries, $T$ mirror updates, and one entropy initialization. In particular,
\[
 \E F(p_J)\ge\alpha\mathrm{OPT}_F(P)+\alpha C f(0)-\eps L.
\]
If Euclidean projections onto $\K(P)$ are available, Euclidean projected ascent with $\eta=\sqrt{2n/(TkG^2\Lambda^2)}$ and the same queries per iteration gives \eqref{eq:continuous} with error term $\alpha G\sqrt{2nk\Lambda^2/T}$; this term is at most $\eps L$ once $T\ge2\alpha^2G^2n^2k/\eps^2$, a horizon larger by the factor $k/\log(k+1)$.
\end{theorem}
The proof is in \cref{app:continuousproof}. Suppose in addition that $L>0$ and every vector of $P$ has total mass at most $r\in[1,n]$, and put $B_r=\min\{n\log(k+1),r\log(\mathrm e nk/r)\}$. With $\eta=\sqrt{B_r/(TG^2rL^2)}$ and $T\ge B_r/r$, so that $\eta GL\le1$, \eqref{eq:entropysupportregret} of \cref{lem:entropy} and \cref{cor:transfer} replace the error term of \eqref{eq:continuous} by $2\alpha GL\sqrt{rB_r/T}$. The horizon $T\ge\max\{B_r/r,4\alpha^2G^2rB_r/\eps^2\}$ makes it at most $\eps L$ and, for $0<\eps\le1$, uses $O(nkr^2\log(\mathrm e nk/r)/\eps^2)$ value queries; the proof of \cref{thm:matroid} carries out this calculation.

For a sequence of objectives, the same learner gives approximate-regret bounds under each feedback model; only the second moments of the field estimates change.
\begin{theorem}[Online learning under limited feedback]\phantomsection\label[theorem]{thm:fewquery}
Fix $P$ and a triple in $\mathcal A_k(\zeta)$ as in \cref{thm:continuous}, with constants $\alpha$, $C$, and $G$, and let $L_e$ bound every marginal $|\Delta_{e,i}f_t|$ of every objective in the sequence, with $L=\max_eL_e$ and $\Lambda^2=\sum_eL_e^2$ as in \cref{lem:marginal}. For each feedback model $j\in\{\mathrm{full},\nabla,2,1\}$, define
\[
\begin{array}{lll}
\text{full prefix ($nk+1$ exact values):} & \mathcal V_{\mathrm{full}}=G^2\Lambda^2, & \mathcal U_{\mathrm{full}}=kG^2\Lambda^2,\\
\text{one exact gradient:} & \mathcal V_\nabla=nG^2\Lambda^2, & \mathcal U_\nabla=nkG^2\Lambda^2,\\
\text{two noisy values:} & \mathcal V_2=n^2k^2G^2(L^2+2\sigma^2), & \mathcal U_2=nk^2G^2(\Lambda^2+2n\sigma^2),\\
\text{one noisy value:} & \mathcal V_1=4n^2k^2G^2(f_{\max}^2+\sigma^2), & \mathcal U_1=\mathcal V_1,
\end{array}
\]
where the one-value model assumes $0\le f_t\le f_{\max}$ on the full labeling domain. For oblivious objectives, entropy mirror ascent with $\eta=\sqrt{2B_0/(T\mathcal V_j)}$ satisfies
\begin{equation}\label{eq:fewqueryregret}
 \operatorname{Reg}_{\alpha}(T)+\alpha C\sum_tf_t(0)
 \le\alpha\sqrt{2B_0\mathcal V_jT}.
\end{equation}
For a nonanticipating adaptive adversary, if Euclidean projections onto $\K(P)$ are available, Euclidean projected ascent on $\K(P)$, of Euclidean diameter $\mathsf D$, with $\eta=\mathsf D/\sqrt{T\mathcal U_j}$, gives
\begin{equation}\label{eq:fewqueryadaptive}
 \operatorname{Reg}_{\alpha}(T)+\alpha C\E\sum_tf_t(0)
 \le2\alpha\mathsf D\sqrt{\mathcal U_jT}.
\end{equation}
The same expected guarantees hold when each action is replaced by its objective-independent integral matroid rounding or exact-base completion (\cref{sec:matroid}). Zero moments or a singleton feasible region are handled by keeping the state fixed.
\end{theorem}
The constants $\mathcal V_j$ and $\mathcal U_j$ bound the conditional second moments of the field estimates (\cref{lem:oraclemoments}). The proof in \cref{app:fewqueryproof} combines the transfer from linear regret with the entropy and Euclidean estimates; the adaptive bound also controls the accumulated field error uniformly over the comparator chosen in hindsight.

With exact values, the two-value estimate is bounded on every outcome, and the entropy step can then be charged to the second moment weighted by the current state rather than to the largest coordinate of each block. In this weighted moment, the squared entries of a row of the coefficient matrix are averaged over the uniformly sampled label; they sum to at most $G^2$, so their average is at most $G^2/k$, which saves a factor of $k$.
\begin{corollary}[Local-norm bound for exact two-value feedback]\phantomsection\label[corollary]{cor:localentropy}
In the setting of \cref{thm:fewquery}, suppose that the objective sequence is oblivious, the two-value oracle returns exact values, $L>0$, and every support vector in $P$ has total mass at most $r\in[1,n]$. With $B_r=\min\{n\log(k+1),r\log(\mathrm e nk/r)\}$, put
\[
 \mathcal V_{\mathrm{loc}}=nkG^2\min\{rL^2,\Lambda^2\},\qquad \bar a=nkGL.
\]
Entropy mirror ascent with two-value feedback and step size $\eta=\min\{1/\bar a,\sqrt{B_r/(T\mathcal V_{\mathrm{loc}})}\}$ satisfies
\begin{equation}\label{eq:localentropyregret}
 \operatorname{Reg}_{\alpha}(T)+\alpha C\sum_tf_t(0)
 \le2\alpha\sqrt{B_r\mathcal V_{\mathrm{loc}}T}+\alpha\bar aB_r.
\end{equation}
For $r=n$, the leading term is at most $2\alpha GLn\sqrt{nk\log(k+1)T}$, while \eqref{eq:fewqueryregret} with exact values ($\sigma=0$) gives $\alpha GLnk\sqrt{2n\log(k+1)T}$; the ratio is $\sqrt{2/k}$. Since all quantities in both bounds are known before the run, choosing the step size of the smaller bound attains the minimum of \eqref{eq:fewqueryregret} and \eqref{eq:localentropyregret}. The query count, two post-decision values per round, and the matroid rounding transfer are those of \cref{thm:fewquery}.
\end{corollary}
The proof is in \cref{app:localentropy}. The additive term $\alpha\bar aB_r$ does not grow with $T$; for short horizons it can exceed the gain in the leading term, which is why the minimum is stated.

At the unrestricted endpoint, independent categorical sampling already produces a valid integral labeling, so the unconstrained coefficient incurs no rounding loss.
\begin{corollary}[Unconstrained optimization]\phantomsection\label[corollary]{cor:unconstrained}
For every $k\ge2$ and $0<\eps<b_k$, a polynomial-query algorithm returns a labeling with
$\E f(S)\ge(b_k-\eps)\OPT+(1-b_k+\eps)f(0)$, that is, $\E f(S)-f(0)\ge(b_k-\eps)(\OPT-f(0))$.
The first-order and value-feedback versions have the $b_k$-regret bounds of \cref{thm:fewquery}.
\end{corollary}
\begin{proof}
At $\zeta=1$, down-closedness gives $P=[0,1]^V$ and $\K(P)=\D$, so the entropy projection is blockwise normalization and every action is feasible. Use the parameters attaining $b_k$ in \cref{thm:exactendpoint}, namely the last piece of \cref{thm:twoclosed} at $\zeta=1$ for $k=2$ and \cref{thm:kprofile}(d) for $k\ge3$; for both, $\omega=0$ and $\alpha=1/(1+c)=b_k$; then $C=c$ and $\alpha C=c/(1+c)=1-1/(1+c)=1-b_k$. \Cref{thm:continuous} with accuracy $\eps$ (if $\Lambda=0$, $f$ is constant and any labeling is optimal) gives
\[
 \E F(p_J)\ge b_k\OPT_F(\D)+(1-b_k)f(0)-\eps L.
\]
Sample $S$ from $p_J$ independently across elements. Conditional on $p_J$, its expected value is $F(p_J)$ by \eqref{eq:extensionpoly}, and every labeling is feasible. Every labeling has $\one_s\in\D$ and $F(\one_s)=f(s)$, so $\OPT_F(\D)\ge\OPT$, and since all singletons are feasible, \cref{lem:marginal} gives $L\le\OPT-f(0)$, so $-\eps L\ge-\eps(\OPT-f(0))$. Hence
\[
 \E f(S)=\E F(p_J)\ge b_k\OPT+(1-b_k)f(0)-\eps(\OPT-f(0))=(b_k-\eps)\OPT+(1-b_k+\eps)f(0).
\]
Subtracting $f(0)$ from both sides gives the equivalent form $\E f(S)-f(0)\ge(b_k-\eps)(\OPT-f(0))$, which is exact for a constant objective. The number of value queries is $O(n^3k\log(k+1)/\eps^2)$ by \cref{thm:continuous}. Independent categorical sampling is objective-independent and retains $F(p_t)$ in expectation, so the rounding transfer of \cref{sec:rounding} (\cref{cor:roundtransfer} with $s=1$ and $\tau=0$) transfers the regret bounds of \cref{thm:fewquery}.
\end{proof}

\subsection{Coordinate capacities and inactive elements}\label{sec:normalization}
The diagonal level compares $P$ with the all-ones direction, so a single coordinate with a small maximum lowers $\zeta(P)$ even when the remaining coordinates are unconstrained. For a compact convex down-closed $P\subseteq[0,1]^V$, the \emph{coordinate capacity} of $e$ is $\rho_e=\max_{v\in P}v_e$. Every $v\in P$ vanishes on the elements with $\rho_e=0$, so we remove these elements, leaving the ground set $V_\rho=\{e\in V:\rho_e>0\}$; when $V_\rho$ is empty, return $0$, set $L_\rho=\Lambda_\rho=0$ and $\zeta_\rho(P)=1$, and omit the algorithmic steps below. For vectors $u\in\R^{V_\rho}$ and for $x\in\Delta_k^{V_\rho}$, write $(\rho\odot u)_e=\rho_eu_e$ and $(\rho\odot x)_e=\rho_ex_e$, with zero entries and blocks on $V\setminus V_\rho$. The \emph{normalized region} and the \emph{normalized diagonal level} are
\begin{equation}\label{eq:normalizedregion}
 P^\rho=\{u\in\R_+^{V_\rho}:\rho\odot u\in P\},\qquad
 \zeta_\rho(P)=\max\{t\in[0,1]:t\rho\in P\}.
\end{equation}
Since $\rho\le\one$ and $P$ is down-closed, $t\one\in P$ implies $t\rho\in P$, so $\zeta_\rho(P)\ge\zeta(P)$. Rescaling coordinates changes the field and its scales but not the comparison set.
\begin{proposition}[Coordinate-normalized profile]\phantomsection\label[proposition]{prop:normalized}
Let $P$ be a nonempty compact convex down-closed support region with capacities $\rho$ and $V_\rho\ne\varnothing$, and let $(c,\omega,\mu)\in\mathcal A_k(\zeta_\rho(P))$, with constants $\alpha$, $C$, and $G$. Suppose the entropy projections onto $\K(P^\rho)$ are available. Run the algorithm of \cref{thm:continuous} on the region $P^\rho$ at level $\zeta_\rho(P)$, with every sampled predecessor of an element drawn from $\rho\odot p$ instead of $p$ and every returned derivative estimate of element $e$ multiplied by $\rho_e$, and output $\rho\odot p_J$. For first-order feedback, query $\nabla F$ at $\rho\odot\Pi^E$ and multiply block $E$ by $\rho_E$. Then $\rho\odot p_J\in\K(P)$, and \eqref{eq:continuous} holds for every $y\in\K(P)$ with $\Lambda^2$ replaced by $\Lambda_\rho^2=\sum_{e\in V_\rho}\rho_e^2L_e^2$ and $n$ by $|V_\rho|$; in particular
\[
 \E F(\rho\odot p_J)\ge\alpha\,\mathrm{OPT}_F(P)+\alpha Cf(0)-\eps L_\rho,\qquad L_\rho=\max_{e\in V_\rho}\rho_eL_e\le L,
\]
with the value-query count of \cref{thm:continuous} on $|V_\rho|$ elements, in addition to computing the capacities. Since $\rho_eL_e\le\mathrm{OPT}_F(P)-f(0)$ for every $e\in V_\rho$, this gives
\[
 \E F(\rho\odot p_J)\ge(\alpha-\eps)\,\mathrm{OPT}_F(P)+(\alpha C+\eps)f(0).
\]
In the online protocol, playing $\rho\odot p_t$ gives the bounds of \cref{thm:fewquery} for the same regret \eqref{eq:regretdefinition} against $\K(P)$, with $n$ replaced by $|V_\rho|$ in $B_0$, with $n$ and $L_e$ replaced by $|V_\rho|$ and $\rho_eL_e$ in every moment constant $\mathcal V_j$ and $\mathcal U_j$, with $\sigma$ and $f_{\max}$ unchanged, and, when Euclidean projections onto $\K(P^\rho)$ are available, with $\mathsf D$ the diameter of $\K(P^\rho)$.
\end{proposition}
The proof in \cref{app:normalized} runs the unchanged algorithm on the objective $f_\rho(s)=\E f(s|_R)$, where $R$ keeps each element $e$ independently with probability $\rho_e$ and $s|_R$ unselects the elements outside $R$; this objective is nonnegative and $k$-submodular, its multilinear extension is $u\mapsto F(\rho\odot u)$, and $u\mapsto\rho\odot u$ maps $\K(P^\rho)$ onto $\K(P)$. The inequality $\rho_eL_e\le\mathrm{OPT}_F(P)-f(0)$ holds because $\rho_e\ee_e\in P$ by compactness and down-closedness, and the fractional labeling that gives element $e$ a label $i$ attaining $L_e$ with probability $\rho_e$ has value $f(0)+\rho_eL_e$; here $L_e\ge0$ by \eqref{eq:pairwise}. The first-order queries $\rho\odot\Pi^E$ are feasible in $\K(P)$, because $|\rho\odot\Pi^E|\le|\rho\odot p|$. The capacities $\rho$ and the level $\zeta_\rho(P)$, or a certified lower bound $\bar\zeta$ with $\bar\zeta\rho\in P$, are inputs to the proposition: they are either supplied or computed with additional geometric access, since a projection oracle alone does not provide them. In general the capacities cost $n$ linear optimizations over $P$. For an explicit packing description $\{v\in[0,1]^V:Av\le b\}$ of $P$ with $A,b\ge0$, down-closedness gives $\rho_e=\min\{1,\min_{m:A_{me}>0}b_m/A_{me}\}$ and
\[
 \zeta_\rho(P)=\min\Bigl(\{1\}\cup\Bigl\{\frac{b_m}{\sum_eA_{me}\rho_e}:\sum_eA_{me}\rho_e>0\Bigr\}\Bigr),
\] and $P^\rho=\{u\in[0,1]^{V_\rho}:A\operatorname{diag}(\rho)u\le b\}$ is again an explicit packing description, so the entropy projections onto $\K(P^\rho)$ are implemented in the same way as those onto $\K(P)$; a separation oracle for $P$ gives one for $P^\rho$ by testing $\rho\odot u$ and multiplying the coefficients of a violated inequality by $\rho$. A projection oracle for $\K(P)$ that is available only as a black box does not by itself give the projections onto $\K(P^\rho)$, because the entropy geometry changes under the rescaling.

A strict gain can occur when some coordinate capacity is below one; the capacities need not be unequal. No gain occurs when all capacities are one. After the preprocessing of \cref{sec:discrete}, every singleton is feasible, so the matroid and knapsack relaxations have $\rho=\one$ and $\zeta_\rho(P)=\zeta(P)$, and the integral guarantees of \cref{sec:rounding} are unchanged. For a box $P=\prod_e[0,\rho_e]$ with $0<\rho_e\le1$, on the other hand, $\zeta(P)=\min_e\rho_e$ can be arbitrarily small while $\zeta_\rho(P)=1$, so the proposition gives the unrestricted coefficients $b_k$ of \cref{thm:exactendpoint} in place of the coefficient at level $\min_e\rho_e$; the latter is the zero-diagonal value when $\min_e\rho_e$ lies in the plateau, $\min_e\rho_e\le t_2$ for $k=2$ and $\min_e\rho_e\le t_3$ for $k\ge3$.

The objective can also enlarge the diagonal level. Call $e$ \emph{inactive} if $L_e=0$.
\begin{remark}[Removing inactive elements]\phantomsection\label[remark]{rem:nullelements}
Let $D\subseteq V$, let $V'=V\setminus D$, and let $f_{V'}(s)=f(s,0_D)$ on labelings of $V'$, with the support family $\{Y\in\mathcal S:Y\subseteq V'\}$ and the region $P_{V'}=\{u\in[0,1]^{V'}:(u,0_D)\in P\}$. Then $f_{V'}$ is nonnegative and $k$-submodular, $P_{V'}$ is a compact convex down-closed relaxation of the restricted family (for a matroid it is the independence polytope of the restriction to $V'$, and for knapsacks it is the knapsack polytope of the items in $V'$), $\zeta(P_{V'})\ge\zeta(P)$, and the restricted optimum satisfies $\OPT_{V'}\le\OPT$. If $\mathcal S$ is down-closed, then also
\[
 \OPT_{V'}\ge\OPT-\sum_{e\in D}L_e.
\]
If $\mathcal S$ is down-closed and every element of $D$ is inactive, then $\OPT_{V'}=\OPT$, and every guarantee for the restricted instance transfers by extending its output by $0$ on $D$. For exact matroid bases, apply this to the independence family, whose optimum equals the base optimum (\cref{thm:matroid}), and complete the extended support in the original matroid, using randomized completion for expected and online guarantees and value-nondecreasing offline completion for high-probability guarantees. The scales $L_e$ come from the $nk+1$ initial values, so offline this preprocessing needs no further queries. For example, a uniform matroid of rank $r$ on $n$ elements, of which only $m\ge r$ are active, has $\zeta=r/n$ before and $r/m$ after the removal. Online, an element may be removed only if it is inactive for every objective of the sequence.
\end{remark}
The proof is in \cref{app:normalized}. More generally, offline one can remove the elements with $L_e\le\theta$ for each of the at most $n+1$ distinct thresholds $\theta\in\{0\}\cup\{L_e:e\in V\}$, run a solver on each restricted instance, and return the best of the exactly evaluated integral outputs; the returned value is at least each of the resulting guarantees.

\section{Integral Guarantees}\label{sec:rounding}
The fractional guarantees of \cref{sec:stochastic} yield integral ones through standard rounding. Conditional label sampling \eqref{eq:labelconditioning} reduces a fractional labeling to an ordinary submodular support function. \Citet[Lemmas~3.2 and~D.4]{ZhouHuangWang25} combine this reduction with randomized matroid rounding \citep{CalinescuEtAl11,ChekuriVondrakZenklusen10}, which keeps $F$ in expectation, and with small-weight rounding and enumeration of labeled large elements under a constant number of knapsacks, which keep $F$ up to a factor $1-\eps$; we use the same reduction. For one knapsack we use cost-preserving pipage rounding \citep{AgeevSviridenko04}. For several knapsacks we follow the framework of \citet[Theorem~2.6]{KulikShachnaiTamir13}, which turns an $\alpha$-approximation for the fractional problem into an $(\alpha-\eps)$-approximation under any constant number of knapsack budgets, combining partial enumeration \citep{Sviridenko04,LeeEtAl10} with contention resolution \citep{ChekuriVondrakZenklusen14,AdamczykWlodarczyk18}. The enumeration runs over labeled prefixes. The residual objective after an enumerated prefix $B$ is not monotone, and our residual guarantee keeps its value $f(B)$ only with the coefficient $\alpha C$; this suffices once the prefix has of order $d/(\delta C)$ elements, where $d$ is the number of budgets and $\delta$ bounds the relative cost of the retained items (\cref{lem:prefix}). Relative to these reductions, this section adds the coefficients at the diagonal level, with a residual condition on the enumerated branch; the explicit initial-value term; explicit enumeration lengths, of order $1/\eps$ for one budget and set by the explicit loss of \cref{lem:priority} for several budgets; rounding that does not query the objective for matroids and for one budget with small items, so that the online guarantees carry over; and high-probability and exact-base guarantees for matroids. The proofs are in \cref{app:matroid}.

\begin{corollary}[Rounding transfer; of \cref{cor:linearizable}]\phantomsection\label[corollary]{cor:roundtransfer}
Fix a triple in $\mathcal A_k(\zeta)$ and a region $P$ as in \cref{cor:linearizable}. Suppose a randomized rounding rule $\mathcal R$ returns a labeling that is feasible on every outcome and satisfies
$\E f(\mathcal R(p))\ge sF(p)+\tau f(0)$, with $s>0$, $\tau\ge0$.
Combining it with \eqref{eq:linearizable} gives
\begin{equation}\label{eq:roundtransfer}
 s\alpha F(y)+(s\alpha C+\tau)f(0)-\E f(\mathcal R(p(x)))
 \le s\alpha\langle g_F(x),y-x\rangle.
\end{equation}
If the rule is objective-independent, it can be executed before online feedback, and the regret bounds of \cref{thm:fewquery,cor:localentropy} hold for the rounded plays with coefficient $s\alpha$, initial-value coefficient $s\alpha C+\tau$, and right-hand side multiplied by $s$, using the same number of objective queries. For an objective-dependent rule, the same inequality gives an offline optimization consequence after the objective is available.
\end{corollary}

\subsection{Matroids}\label{sec:matroid}
For a matroid, fix conditional labels and apply randomized pipage rounding \citep[Section~2 and Appendix~B.2]{ChekuriVondrakZenklusen10}, which does not query the objective and retains the multilinear extension $G_\ell$ of the conditioned support function in expectation; averaging the labels by \eqref{eq:labelconditioning} gives $s=1$ and $\tau=0$ in \cref{cor:roundtransfer} and shows $\mathrm{OPT}_F(P(\mathcal M))=\OPT$, as in the matroid case of \citet[Lemma~D.4]{ZhouHuangWang25}. To complete an independent support to a base, add one element at a time and give it one of two fixed labels uniformly at random; by pairwise monotonicity \eqref{eq:pairwise} the mean gain is nonnegative, as in the unconstrained greedy algorithms \citep{WardZivny16,IwataTanigawaYoshida16}.
\begin{corollary}[Matroid approximation]\phantomsection\label[corollary]{thm:matroid}\label[corollary]{cor:bases}
Let $\alpha$ and $C$ come from a triple $(c,\omega,\mu)$ admissible at level $\zeta(P(\mathcal M))$. For every $0<\eps<\alpha$, there is a randomized $(\alpha-\eps)$-approximation for nonnegative $k$-submodular maximization under one matroid support constraint. For rank $r\ge1$, the expected guarantee uses $O(nkr^2\log(\mathrm e nk/r)/\eps^2)$ value queries, one entropy projection per iteration, and one final matroid rounding. More precisely, a fractional solve with additive error $\nu L$, followed by rounding and comparison with the stored empty and singleton labelings, gives
\begin{equation}\label{eq:matroidaffine}
 \E f(S)\ge\frac{\alpha}{1+\nu}\OPT+
 \frac{\alpha C+\nu}{1+\nu}f(0)\qquad(\nu>0).
\end{equation}
For every $\beta\in(0,1)$, rounding each iterate and returning the largest evaluated value gives
\begin{equation}\label{eq:matroidhp}
 \Pr\Bigl\{f(S)\ge(\alpha-\eps)\OPT+\Bigl(\alpha C+\frac{\eps}{1-\omega}\Bigr)f(0)\Bigr\}\ge1-\beta
\end{equation}
with $O(nkr^2[\log(\mathrm e nk/r)+\log(2/\beta)]/\eps^2)$ value queries and one rounding per iteration.

The independence and exact-base optima coincide. The expected ratio extends to exact bases by objective-independent completion. The high-probability guarantee extends by value-nondecreasing offline completion with at most $2r$ additional value queries. Online, round each committed action without selecting among candidates using the current objective; objective-independent matroid rounding and randomized base completion preserve every expected regret bound in \cref{thm:fewquery}.
\end{corollary}
The proofs are in \cref{app:matroidrounding,app:basecompletion}. For a uniform matroid of rank $r$, where $\zeta=r/n$, \cref{thm:halfthreshold} gives a $(1/2-\eps)$-approximation for every $k\ge3$ when $r/n\ge\zeta_{1/2}(k)$, and \cref{thm:high} gives a $(45/86-\eps)$-approximation for $3\le k\le10$ when $r/n\ge9/10$.

\subsection{Knapsacks}\label{sec:multiknapsack}
One budget admits a stronger rounding rule than several budgets. When each item costs at most a $\delta$ fraction of capacity, cost-preserving pipage rounding \citep{AgeevSviridenko04} of the scaled point $(1-\delta)p$ is objective-independent and loses only a linear factor.
\begin{lemma}[Small-item rounding for one budget]\phantomsection\label[lemma]{lem:singlepipage}
Let $g$ be a nonnegative submodular set function with multilinear extension $G$. Suppose $0\le c_e\le\delta b'$ and $\sum_ec_ez_e\le(1-\delta)b'$ for $z\in[0,1]^V$, where $b'\ge0$ and $0\le\delta<1$. An objective-independent randomized rule returns $R$ with
\[
 \sum_{e\in R}c_e\le b'\quad\hbox{on every outcome},\qquad \E g(R)\ge G(z).
\]
It uses at most $|V|$ coordinate or pair steps and no objective queries.
\end{lemma}
The rule leaves at most one positive-cost fractional coordinate, and the reserved capacity permits either final outcome. Applied after conditional labels, together with the ray inequality \eqref{eq:ray}, it retains
\[
 (1-\delta)F(p)+\delta f(0).
\]
This gives a linear, rather than square-root (\cref{lem:priority} below), loss in the item threshold and an $O(1/\eps)$ prefix exponent.

For $d$ budgets, sample the labels and the support, keep the longest prefix in a uniformly random priority order that satisfies all budgets, and discard elements whose marginal gain is negative. For a constant number of knapsack constraints, \citet{ChekuriVondrakZenklusen14} give contention-resolution schemes whose loss can be made arbitrarily small, and \citet{KulikShachnaiTamir13} give a related rounding for several budgets; random-order schemes are studied by \citet{AdamczykWlodarczyk18}; and \citet{ZhouHuangWang25} use small-weight rounding to keep $(1-\eps)F$ for $k$-submodular objectives. The following lemma is a variant with an explicit loss that also retains the value of the empty set.
\begin{lemma}[Random-priority prefix rounding]\phantomsection\label[lemma]{lem:priority}
Let $g:2^V\to\R_+$ be submodular, let $y\in[0,1]^V$ satisfy $\sum_e a_{m,e}y_e\le1$ for every row $m\in[d]$, where $d\ge1$ is an integer, and suppose $0\le a_{m,e}\le\delta<1$. Sample independent indicators $X_e$ with means $y_e$ and independent priorities $U_e$ uniform on $[0,1]$, and put $R_{<e}=\{w:X_w=1,\ U_w<U_e\}$. The \emph{feasible prefix} consists of the sampled elements $e$ for which $R_{<e}\cup\{e\}$ satisfies all budgets. The \emph{pruning scan} visits the feasible prefix in priority order and keeps $e$ if and only if $g(S_{<e}\cup\{e\})-g(S_{<e})\ge0$, where $S_{<e}$ is the set of elements kept before $e$; let $S$ be the set of kept elements. Write $G_g(y)=\E g(R_y)$ and
\begin{equation}\label{eq:priorityloss}
 \kappa_d=\int_0^\infty\min\{1,de^{-u^2/2}\}\,du,\qquad
 \lambda_d(\delta)=\delta+\kappa_d\sqrt\delta.
\end{equation}
If $\lambda_d(\delta)<1$, then $S$ satisfies all budgets on every outcome and
\begin{equation}\label{eq:priorityretention}
 \E g(S)\ge(1-\lambda_d(\delta))G_g(y)+\lambda_d(\delta)g(\varnothing).
\end{equation}
The pruning scan uses at most $|V|+1$ value queries.
\end{lemma}
With $a_d=\sqrt{2\log d}$, the constant equals $\kappa_d=a_d+d\int_{a_d}^\infty e^{-u^2/2}\,du$; in particular, $\kappa_1=\sqrt{\pi/2}$ and $\kappa_d^2=O(\log(2d))$. The proofs of \cref{lem:singlepipage,lem:priority} are in \cref{app:singleproof,app:priority}.

For general costs, enumerate a short feasible labeled prefix $B$, keep only the items whose cost is at most a $\delta$ fraction of every residual capacity, and solve for the residual objective $f_B(s)=f(B\cup s)$, with multilinear extension $F_B$ and with singleton scale $L_B$ given by \eqref{eq:scales} applied to $f_B$.

After conditioning on a feasible labeled prefix $B$, let $b'_m=1-\sum_{e\in\supp(B)}c_{m,e}$ be the residual capacities, let the residual ground set $U_B$ consist of the retained items, those $e\notin\supp(B)$ with $c_{m,e}\le\delta b'_m$ for every $m$, and $A_m(B)=\sum_{e\in U_B}c_{m,e}$. The residual support polytope $P_B=\{v\in[0,1]^{U_B}:\sum_{e\in U_B}c_{m,e}v_e\le b'_m\text{ for every }m\}$ has
\begin{equation}\label{eq:residualzeta}
 \zeta(P_B)=\min\left(\{1\}\cup\left\{\frac{b'_m}{A_m(B)}:m\in[d],\ A_m(B)>0\right\}\right).
\end{equation}
Conditioning spends budget and removes large items, so $\zeta(P_B)$ can be smaller than $\zeta(P)$. A triple chosen at a level $\zeta$ is therefore applied on a branch only when $\zeta(P_B)\ge\zeta$, which the algorithm checks with \eqref{eq:residualzeta}; on the other branches it stores only the prefix and its singleton extensions. The analysis uses the residual guarantee on a single branch. We say that the \emph{residual condition} holds at level $\zeta$, for a prefix length $\rho$ and a threshold $\delta$, if some optimal labeling has at most $\rho$ elements, or if some optimal labeling has a greedy prefix $B$ of $\rho$ elements, in the sense of \cref{lem:optimalprefix}, whose residual region satisfies $\zeta(P_B)\ge\zeta$. The algorithm need not identify this prefix. The residual condition holds trivially at $\zeta=0$, and it holds whenever every enumerated residual region has diagonal level at least $\zeta$.
In \cref{thm:single,thm:multi,cor:finiteprefix} the residual condition matters only for a triple chosen at a positive level; the zero-diagonal triples satisfy it on every instance.
\begin{corollary}[Single-knapsack approximation]\phantomsection\label[corollary]{thm:single}
Choose $\alpha,C$ from a triple $(c,\omega,\mu)\in\mathcal A_k(\zeta)$. For $0<\eps<\alpha$ and $0<\nu\le\eps$, put
\begin{equation}\label{eq:singleparams}
 s=\frac{\alpha-\eps}{\alpha-\nu(\alpha-\eps)},\qquad
 \delta=1-s,\qquad
 \rho=\left\lceil\frac1{\delta\bigl(C+\eps/[\alpha(\alpha-\eps)]\bigr)}\right\rceil.
\end{equation}
Suppose the residual condition holds at level $\zeta$ for this prefix length and threshold. Enumerate feasible labeled prefixes of at most $\rho$ elements, solve the accepted residual instances with error $\nu L_B$, and round $(1-\delta)p$ by conditional labels and cost-preserving pipage rounding. Return the best evaluated candidate, including the prefixes and their retained singleton extensions. The returned labeling is feasible on every outcome, and
\begin{equation}\label{eq:budgetaffine}
 \E f(S)\ge(\alpha-\eps)\OPT+\Bigl(\alpha C+\frac{\eps}{1-\omega}\Bigr)f(0).
\end{equation}
The running time is $(nk+1)^\rho\poly(n,k,1/\nu)$. In particular, for $\nu=\eps/2$,
\[
 \rho=\frac{\alpha}{C(1-\alpha/2)}\eps^{-1}+O(1),
\]
and $\nu=\eps^2$ gives $\rho=\alpha/(C\eps)+O(1)$.
\end{corollary}
The full rounding construction and proof are in \cref{app:singleproof}. Prefix selection and choosing the best evaluated output are offline; the small-item rounding itself uses no objective queries.

\begin{corollary}[Multiple-knapsack approximation]\phantomsection\label[corollary]{thm:multi}
Choose $\alpha,C$ from a triple $(c,\omega,\mu)\in\mathcal A_k(\zeta)$. For a fixed $d\ge2$, $0<\eps<\alpha$, and $0<\nu\le\eps$, put
\begin{equation}\label{eq:multparams}
 \begin{split}
 s&=\frac{\alpha-\eps}{\alpha-\nu(\alpha-\eps)},\qquad\lambda=1-s,\\
 \delta&=\left(\frac{2\lambda}{\sqrt{\kappa_d^2+4\lambda}+\kappa_d}\right)^2,
 \qquad
 \rho=\left\lceil\frac d{\delta(C+\nu/\alpha+(1-s)/(s\alpha))}\right\rceil,
 \end{split}
\end{equation}
where $\kappa_d$ is the constant in \eqref{eq:priorityloss}. Suppose the residual condition holds at level $\zeta$ for this prefix length and threshold. Prefix enumeration, residual optimization to error $\nu L_B$, and random-priority prefix rounding with marginal pruning (\cref{lem:priority}) produce the candidates. Return the best evaluated candidate, including the prefixes and their retained singleton extensions. The returned labeling is feasible on every outcome and satisfies \eqref{eq:budgetaffine}. With $\nu=\eps/2$, the running time is
\[
 (nk+1)^{O(d\log(2d)/\eps^2)}\poly(n,k,d,1/\eps),
\]
and for fixed $d$, as $\eps\downarrow0$,
\[
 \rho=\left(\frac{d\alpha^2\kappa_d^2}{C(1-\alpha/2)^2}+o(1)\right)\eps^{-2}.
\]
\end{corollary}
The proof is in \cref{app:multiproof}. Pruning queries the current objective, so this multiple-budget implementation is offline. The same construction also works for $d=1$, but \cref{thm:single} gives the better enumeration exponent.

The prefix length in \cref{thm:multi} grows like $1/\eps^2$. A shorter enumeration still gives an explicit coefficient, because every enumerated prefix is itself a stored candidate.
\begin{corollary}[Finite enumeration]\phantomsection\label[corollary]{cor:finiteprefix}
Let $\alpha$, $C$, and $\zeta$ come from an admissible triple. Fix an integer $\rho\ge1$, $0<\delta<1$, and $\nu>0$. For $d=1$ use $s=1-\delta$; for $d\ge2$ assume $\lambda_d(\delta)<1$ and use $s=1-\lambda_d(\delta)$, with $\lambda_d$ from \cref{lem:priority}. Put $\vartheta=d/(\delta\rho)$, and, for a triple chosen at level $\zeta$, assume the residual condition at $\zeta$ for these $\rho$ and $\delta$. The corresponding algorithm of \cref{thm:single,thm:multi}, run with these parameters and with residual accuracy $\nu$, returns a feasible labeling $S$ with
\begin{equation}\label{eq:finiteprefix}
 \E f(S)\ge\frac{A\,\OPT+\bigl(D_*+(B_*)_+\bigr)f(0)}{1+(-B_*)_+},\qquad
 A=\frac{s\alpha}{1+s\nu},\quad
 B_*=\frac{s\alpha(C-\vartheta)+s\nu+1-s}{1+s\nu},\quad
 D_*=\frac{s\alpha\vartheta}{1+s\nu}.
\end{equation}
If $B_*<0$, the right side equals $\alpha\bigl(\OPT+\vartheta f(0)\bigr)/(1+\alpha(\vartheta-C))$. The running time is $(nk+1)^\rho$ times a polynomial in $n$, $k$, $d$, and $1/\nu$.
\end{corollary}
The coefficient of $\OPT$ is positive for every $\rho\ge1$, and it equals $A$ once $\vartheta\le C+\nu/\alpha+(1-s)/(s\alpha)$, which is how $\rho$ is chosen in \eqref{eq:singleparams}, where $\nu/\alpha+(1-s)/(s\alpha)=\eps/[\alpha(\alpha-\eps)]$, and in \eqref{eq:multparams}.

When every cost is small, no enumeration is needed, and the coefficient uses the original region's diagonal level.
\begin{proposition}[Small items, offline and online]\phantomsection\label[proposition]{prop:smallitems}
Consider $d\ge1$ budgets normalized to one, with $c_{m,e}\le\delta<1$, and put
\[
 s=\begin{cases}1-\delta,&d=1,\\
 1-\delta-\kappa_d\sqrt\delta,&d\ge2.
 \end{cases}
\]
For $d\ge2$ assume $s>0$. Fix an admissible triple at level $\zeta(P)$ with constants $\alpha,C$. For every $\nu>0$, solve the original fractional problem to error $\nu L$, use conditional labels and the corresponding rounding rule, and return the best of the rounded labeling, the empty labeling, and all singletons. Then all budgets hold on every outcome and
\begin{equation}\label{eq:smallitems}
 \E f(S)\ge\frac{s\alpha\,\mathrm{OPT}_F(P)+(s\alpha C+s\nu+1-s)f(0)}{1+s\nu}.
\end{equation}
For $0<\eps<s\alpha$, taking $\nu=\eps/[s(s\alpha-\eps)]$ gives coefficient $s\alpha-\eps$. In addition to the fractional solver's calls, at most one value query is needed for $d=1$ and at most $n+2$ for $d\ge2$.

For one budget, the rounding is objective-independent. In an online round, round $(1-\delta)p_t$ directly, without choosing candidates using $f_t$. Every regret guarantee of \cref{thm:fewquery,cor:localentropy} transfers with coefficient $(1-\delta)\alpha$, initial-value coefficient $(1-\delta)\alpha C+\delta$, and its right-hand side multiplied by $1-\delta$. No additional value query is required for rounding.
\end{proposition}
The rounding loss is linear in $\delta$ for one budget and $\delta+\kappa_d\sqrt\delta$ for several budgets. Taking $\delta\to0$ and $\nu\to0$ gives coefficients approaching $\alpha$ at the original diagonal level, without any residual condition. The proofs of \cref{cor:finiteprefix,prop:smallitems} are in \cref{app:finiteprefix,app:smallitems}.

\section{Conclusion}\label{sec:conclusion}
The diagonal level of the support region links constrained and unconstrained $k$-submodular maximization quantitatively: a comparison valid for every comparator gives coefficients that are continuous and nondecreasing in this level, and three concave mass requirements reduce the choice of parameters to tangency equations and polynomial factorizations. The same field serves fractional optimization and online learning, and its guarantees hold at the possibly larger level obtained after rescaling each coordinate by its largest feasible value. Integral matroid rounding preserves the online protocol and gives a high-probability offline guarantee. For one knapsack, objective-independent pipage rounding gives a linear small-item loss and an $O(1/\eps)$ prefix exponent; for several budgets, the random-priority rule gives a $\delta+\kappa_d\sqrt\delta$ loss and an $O(1/\eps^2)$ exponent. The budget guarantees retain an explicit nonnegative initial-value term.


\noindent\textbf{Note added after the initial version.}
After the first version of this paper (2 October 2026),
several subsequent works improved the known approximation
guarantees, and one established a new oracle lower bound.

\emph{Achievability results.}
For nonnegative, not necessarily monotone objectives and every $k\ge 2$, \citet{ChenEtAl26later} obtained
a $1/2$-approximation under cardinality and single-knapsack
constraints, under their normalization $f(0)=0$.
\citet{WangWang26later} established an analytical
$0.48$-approximation under a single matroid constraint
for every $k\ge 2$, while \citet{ChenGaoJiangYang26}
obtained a $(1/2-\eps)$-approximation under any fixed
number of knapsack constraints for every $k\ge 2$.
These results improve the geometry-independent
offline guarantees, but do not provide the
objective-independent actions and comparator-uniform
fields required for the online guarantees of this paper.

\emph{Inapproximability results.}
In the discrete value-oracle model, \citet{Chen26oracle}
established that, for every fixed integer $k\ge 2$ and every
fixed $\eps>0$, a $(k/(2k-1)+\eps)$-approximation requires
exponentially many queries, even on nonnegative monotone
instances. This upper bound therefore applies to the
broader nonnegative class considered here. For every fixed
$k\ge 5$, the value-oracle algorithm of
\cref{cor:unconstrained} matches it, up to arbitrarily
small error. For $k=3,4$, its achievable coefficient
$(\sqrt{17}-3)/2$ remains below the respective upper
bounds $3/5$ and $4/7$, leaving an approximation gap.
For $k=2$, the coefficient $1/2$, which this algorithm
also attains, was known to be tight before this work: the
non-monotone oracle bound of \citet{FeigeMirrokniVondrak11}
transfers to $k=2$ by the reduction of \citet{WardZivny16}.

\section*{AI Disclosure}
The author thanks Hillclimb for compute and Generative AI credits. GPT and Claude were used to check and refine proofs, audit and revise the manuscript.

\appendix
\crefalias{section}{appendix}\crefalias{subsection}{appendix}\crefalias{subsubsection}{appendix}
\counterwithin{theorem}{section}\counterwithin{example}{section}\counterwithin{lemmactr}{section}\counterwithin{corollaryctr}{section}\counterwithin{remarkctr}{section}\counterwithin{propositionctr}{section}\counterwithin{definitionctr}{section}
\section{Proofs for \texorpdfstring{\cref{sec:setup}}{Section 2} and an Auxiliary Lemma}\label{app:base}
\subsection{Proof of \texorpdfstring{\cref{lem:extension}}{the extension lemma} (continuous structure of the extension)}\label{app:extensionproof}
Most of \cref{lem:extension} is in \citet{ZhouHuangWang25}: affinity of $F$ in each coordinate and in each block, the identities $\partial_{e,i}\partial_{e,j}F=0$, and the sign $\partial_{e,i}\partial_{u,j}F\le0$ for $u\ne e$ are the multilinearity and element-wise non-positive Hessian properties of their Lemma~C.1 and Eq.~(3); \eqref{eq:gradientpairwise} is the pairwise-monotonicity property of their Lemma~C.1; and the reduction \eqref{eq:labelconditioning}, together with submodularity of $g_\ell$, is their Claim~D.3 and Eq.~(11). For a submodular set function, the multilinear extension has zero pure second derivatives and nonpositive mixed second derivatives \citep{CalinescuEtAl11}. For completeness we derive every assertion from \eqref{eq:orthant} and \eqref{eq:pairwise}.
\begin{proof}[of \cref{lem:extension}]
\emph{Values.} Every product in \eqref{eq:extensionpoly} is nonnegative on $\D$ and $f\ge0$, so $F\ge0$. At $x=\one_s$, each $q_e$ is the point mass at $s_e$; hence $X=s$ almost surely and $F(\one_s)=f(s)$.

\emph{The derivative formula \eqref{eq:expectedmarginal}.} Fix $e$ and condition on the independent labels $X_{-e}$ of the other elements. Given $X_{-e}$, the element $e$ is unselected with probability $1-|x_e|$ and receives label $i$ with probability $x_{e,i}$, so
\begin{align}
 F(x)&=\E\Bigl[(1-|x_e|)f(X_{-e})+\sum_{i=1}^kx_{e,i}f(X_{-e}+e\mapsto i)\Bigr]\nonumber\\
 &=\E f(X_{-e})+\sum_{i=1}^kx_{e,i}\,\E\bigl[f(X_{-e}+e\mapsto i)-f(X_{-e})\bigr]\nonumber\\
 &=\E f(X_{-e})+\sum_{i=1}^kx_{e,i}\,\E\Delta_{e,i}f(X_{-e}),\label{eq:blockaffine}
\end{align}
where the second line collects the two terms containing $f(X_{-e})$, since $(1-|x_e|)=1-\sum_ix_{e,i}$. The law of $X_{-e}$ depends only on the blocks $x_u$ with $u\ne e$, so neither $\E f(X_{-e})$ nor $\E\Delta_{e,i}f(X_{-e})$ depends on $x_e$. Differentiating \eqref{eq:blockaffine} in $x_{e,i}$ therefore gives \eqref{eq:expectedmarginal}, and the right side of \eqref{eq:expectedmarginal} does not depend on $x_e$; that is, $\nabla_eF$ is independent of $x_e$.

\emph{Second derivatives and pairwise sums.} Since $\partial_{e,i}F$ does not depend on $x_e$, $\partial_{e,j}\partial_{e,i}F=0$. For $u\ne e$, let $Y$ agree with $X$ outside $\{e,u\}$ and leave $e$ and $u$ unselected. Conditioning \eqref{eq:expectedmarginal} on $Y$ as in \eqref{eq:blockaffine} shows that $\partial_{e,i}F$ is affine in $x_u$, and
\[
 \partial_{u,j}\partial_{e,i}F(x)=\E\bigl[\Delta_{e,i}f(Y+u\mapsto j)-\Delta_{e,i}f(Y)\bigr]\le0
\]
by \eqref{eq:orthant}, because $Y\preceq Y+u\mapsto j$ and both leave $e$ unselected. This is \eqref{eq:hessiansigns}. For $i\ne j$, \eqref{eq:expectedmarginal} and \eqref{eq:pairwise} give $\partial_{e,i}F+\partial_{e,j}F=\E[\Delta_{e,i}f(X_{-e})+\Delta_{e,j}f(X_{-e})]\ge0$, which is \eqref{eq:gradientpairwise}.

\emph{Antitone gradient.} Let $x\le y$ in $\D$ and put $z(\tau)=x+\tau(y-x)$ for $\tau\in[0,1]$; since $\D$ is convex, $z(\tau)\in\D$. By the chain rule,
\[
 \frac{d}{d\tau}\partial_{e,i}F(z(\tau))
 =\sum_{u\in V}\sum_{j=1}^k\partial_{u,j}\partial_{e,i}F(z(\tau))\,(y_{u,j}-x_{u,j}).
\]
The terms with $u=e$ vanish because $\partial_{e,j}\partial_{e,i}F=0$ by \eqref{eq:hessiansigns}. For $u\ne e$, the second derivative is nonpositive by \eqref{eq:hessiansigns} and the increment $y_{u,j}-x_{u,j}$ is nonnegative because $x\le y$. Hence the derivative of $\tau\mapsto\partial_{e,i}F(z(\tau))$ is nonpositive on $[0,1]$, and integrating from $0$ to $1$ gives $\partial_{e,i}F(y)\le\partial_{e,i}F(x)$.

\emph{Concavity along sign-definite directions.} Let $x$ and $x+\tau_0v$ lie in $\D$, where all coordinates of $v$ are nonnegative or all are nonpositive, and put $\phi(\tau)=F(x+\tau v)$ for $\tau\in[0,\tau_0]$; the segment lies in $\D$ by convexity. Since $F$ is a polynomial,
\[
 \phi''(\tau)=\sum_{e,u\in V}\sum_{i,j\in[k]}v_{e,i}v_{u,j}\,\partial_{e,i}\partial_{u,j}F(x+\tau v).
\]
Every product $v_{e,i}v_{u,j}$ is nonnegative because the two factors have the same sign. Every second derivative with $u=e$ is zero, and every second derivative with $u\ne e$ is nonpositive, by \eqref{eq:hessiansigns}. Hence every term is nonpositive, $\phi''\le0$, and $\phi$ is concave. For \eqref{eq:ray}, take $v=x$ on the segment from $0$ to $x$; its direction is nonnegative. Concavity of $\tau\mapsto F(\tau x)$ on $[0,1]$ gives $F(tx)=F(t\cdot x+(1-t)\cdot0)\ge tF(x)+(1-t)F(0)$, and $F(0)=f(0)$ by the value identity with $s=0$.

\emph{Label conditioning.} The function $g_\ell$ is nonnegative because $f$ is. If $A\subseteq B$ and $e\notin B$, then $\ell|_A\preceq\ell|_B$ and both leave $e$ unselected, so \eqref{eq:orthant} gives $g_\ell(A\cup\{e\})-g_\ell(A)=\Delta_{e,\ell_e}f(\ell|_A)\ge\Delta_{e,\ell_e}f(\ell|_B)=g_\ell(B\cup\{e\})-g_\ell(B)$; thus $g_\ell$ is submodular. Draw $R_{|x|}$ independently of $\ell$. Each element is selected with probability $|x_e|$ and, given that it is selected, carries label $i$ with probability $x_{e,i}/|x_e|$, independently across elements; hence $\ell|_{R_{|x|}}$ has the law of $X\sim x$, and $\E_\ell G_\ell(|x|)=\E f(\ell|_{R_{|x|}})=F(x)$, which is \eqref{eq:labelconditioning}.

\emph{Exchange convexity.} By \citet{CalinescuEtAl11}, the multilinear extension $G_\ell$ of the submodular function $g_\ell$ satisfies $\partial_e^2G_\ell=0$ for every $e$ and $\partial_e\partial_uG_\ell\le0$ for $u\ne e$. Along the segment $\tau\mapsto v+\tau(a\ee_e-b\ee_u)$ with $a,b\ge0$ and $u\ne e$, the second derivative is therefore
\[
 a^2\partial_e^2G_\ell-2ab\,\partial_e\partial_uG_\ell+b^2\partial_u^2G_\ell=-2ab\,\partial_e\partial_uG_\ell\ge0,
\]
so $G_\ell$ is convex along the segment.
\end{proof}

\subsection{Auxiliary uniform marginal-sample bound (\texorpdfstring{\cref{lem:marginal}}{the marginal-sample lemma}), used in \texorpdfstring{\cref{sec:stochastic,sec:rounding}}{Sections 4 and 5}}\label{app:marginalproof}
The algorithms of \cref{sec:stochastic,sec:rounding} scale their sampled differences by the singleton quantities \eqref{eq:scales}, which bound every marginal on the full query domain, including labelings with infeasible support.
\begin{lemma}[Uniform marginal-sample bound]\phantomsection\label[lemma]{lem:marginal}
Let $L_e$, $L$, and $\Lambda$ be as in \eqref{eq:scales}, with $L=0$ for $V=\varnothing$. Then $L_e\ge0$ and
\begin{equation}\label{eq:marginalbound}
 -L_e\le\Delta_{e,i}f(s)\le L_e\qquad(s_e=0).
\end{equation}
The same bounds hold for $\partial_{e,i}F(x)$. If $L=0$, the objective is constant. With feasible singletons, $L\le\OPT-f(0)$.
\end{lemma}
The upper bound follows from orthant submodularity at the empty labeling, and the lower bound from pairwise monotonicity combined with the upper bound for another label.
\begin{proof}[of \cref{lem:marginal}]
Fix $e\in V$. Since $f(e\mapsto i)=f(0+e\mapsto i)$, we have $L_e=\max_i\Delta_{e,i}f(0)$.

\emph{Nonnegativity of $L_e$.} Choose labels $i\ne j$, which exist because $k\ge2$. Pairwise monotonicity \eqref{eq:pairwise} at the empty labeling gives $\Delta_{e,i}f(0)+\Delta_{e,j}f(0)\ge0$. The larger of the two terms is at least their average, which is nonnegative, so $L_e\ge\max\{\Delta_{e,i}f(0),\Delta_{e,j}f(0)\}\ge0$.

\emph{Upper bound in \eqref{eq:marginalbound}.} Let $s_e=0$ and $i\in[k]$. Orthant submodularity \eqref{eq:orthant}, applied to the pair $0\preceq s$, gives
\[
 \Delta_{e,i}f(s)\le\Delta_{e,i}f(0)\le L_e.
\]

\emph{Lower bound in \eqref{eq:marginalbound}.} Choose any label $j\ne i$. Pairwise monotonicity \eqref{eq:pairwise} at $s$ and the upper bound just proved, applied to label $j$, give
\[
 \Delta_{e,i}f(s)\ge-\Delta_{e,j}f(s)\ge-L_e.
\]

\emph{Derivatives.} By \eqref{eq:expectedmarginal}, $\partial_{e,i}F(x)$ is the expectation of $\Delta_{e,i}f(X_{-e})$, and $X_{-e}$ leaves $e$ unselected on every outcome. Every such value lies in $[-L_e,L_e]$, and hence so does the expectation.

\emph{Constancy when $L=0$.} Then $L_e=0$ for every $e$, so \eqref{eq:marginalbound} gives $\Delta_{e,i}f(s)=0$ whenever $s_e=0$. For any labeling $s$, list its support as $e_1,\ldots,e_r$ and let $s^{(h)}$ assign $s_{e_1},\ldots,s_{e_h}$ to the first $h$ elements and $0$ elsewhere. Telescoping gives
\[
 f(s)-f(0)=\sum_{h=1}^r\Delta_{e_h,s_{e_h}}f(s^{(h-1)})=0.
\]

\emph{Comparison with $\OPT$.} Every singleton is feasible by assumption, so $f(e\mapsto i)\le\OPT$ for all $e,i$. Subtracting $f(0)$ and maximizing over $i$ and $e$ gives $L\le\OPT-f(0)$. For $V=\varnothing$, the only labeling is empty, so $\OPT=f(0)$ and both sides are zero.
\end{proof}
\section{Proofs for \texorpdfstring{\cref{sec:global,sec:geometry}}{the block comparison and construction sections}: Block Rules}\label{app:geometry}
All coefficient vectors below are deterministic functions of the input block and the parameters.

\subsection{Proof of \texorpdfstring{\cref{lem:pairing}}{the pairing cone lemma} (pairing cone)}\label{app:pairing}
\begin{proof}[of \cref{lem:pairing}]
Write $K=\operatorname{cone}\{\ee_i+\ee_j:i\ne j\}$ and $\widetilde Q=\{v\ge0:2v_i\le\sum_jv_j\text{ for all }i\}$, and let $K^*=\{a:\ip va\ge0\text{ for all }v\in K\}$ be the dual cone.

\emph{Dual cone.} Since $K$ is generated by the vectors $\ee_i+\ee_j$, a vector $a$ belongs to $K^*$ exactly when $\ip{\ee_i+\ee_j}a=a_i+a_j\ge0$ for all $i\ne j$, that is, when $a$ is pairwise monotone. This is the asserted duality.

\emph{$K\subseteq\widetilde Q$.} Each generator $\ee_i+\ee_j$ is nonnegative, has total $2$, and has every coordinate at most $1$, so it lies in $\widetilde Q$. The set $\widetilde Q$ is defined by finitely many homogeneous linear inequalities, so it is a convex cone and contains $K$.

\emph{$\widetilde Q\subseteq K$.} Let $v\in\widetilde Q$ and let $a$ be pairwise monotone. At most one coordinate of $a$ is negative, since two negative coordinates would violate $a_i+a_j\ge0$. If $a\ge0$, then $\ip va\ge0$ because $v\ge0$. If $a_j<0$, then $a_i\ge-a_j$ for every $i\ne j$, and hence
\[
 \ip va=a_jv_j+\sum_{i\ne j}a_iv_i\ge a_jv_j-a_j\sum_{i\ne j}v_i=a_j\Bigl(v_j-\sum_{i\ne j}v_i\Bigr)\ge0,
\]
because $a_j<0$ and $2v_j\le\sum_iv_i$. Since $a\in K^*$ was arbitrary, $\ip va\ge0$ for every $a\in K^*$; that is, $v\in K^{**}$. A finitely generated cone is closed, so the bipolar theorem gives $K^{**}=K$, and $v\in K$.

Consequently $\ip va\ge0$ for every $v\in Q_k$ and every pairwise-monotone $a$.
\end{proof}

\subsection{Proof of \texorpdfstring{\cref{lem:weightedlevel}}{the weighted level lemma} (weighted level comparison)}\label{app:weightedlevel}
\begin{proof}[of \cref{lem:weightedlevel}]
Throughout, $k\ge3$, the input is sorted so that $z_1\ge\cdots\ge z_k$, and hence $q_1\ge\cdots\ge q_k$. Steps~1--5 construct the level and verify the cone inclusions and the budget identity; Step~6 compares resources and demands at the target mass.

\emph{Step 1: an elementary pairing.} Let $y\in\R_+^k$ have a largest coordinate at index $i$, put $E=(2y_i-|y|)_+$, and let $\theta$ be a probability vector with $\theta_i=0$. We claim $y+E\theta\in Q_k$, where $Q_k$ is characterized in \eqref{eq:cone}. If $E=0$, then $2y_j\le2y_i\le|y|$ for every $j$, so $y\in Q_k$ (and $\theta$ plays no role). If $E>0$, put $v=y+E\theta$. Then $|v|=|y|+E|\theta|=|y|+E=2y_i$, because $\theta$ is a probability vector, and $v_i=y_i+E\theta_i=y_i$, so $2v_i=|v|$. For $j\ne i$, using $\theta_j\le1$,
\[
 v_j\le y_j+E=y_j+2y_i-|y|=y_i-\sum_{i'\ne i,j}y_{i'}\le y_i,
\]
so $2v_j\le2y_i=|v|$. Since also $v\ge0$, \eqref{eq:cone} gives $v\in Q_k$.

\emph{Step 2: existence of the level and the scan.} Write $\mathrm{Res}(\pi)=\sum_j\gamma_j$ and $\mathrm{Dem}(\pi)=(1-Z)D_0+\sum_\ell z_\ell E_\ell$. Since $\gamma_1=0$ and, for $j\ge2$, $\gamma_j=(c+z_j)\max\{q_j,\pi\}-z_j=(c+z_j)(\pi-q_j)_+$ because $(c+z_j)q_j=z_j$, the resource is a sum of functions $(c+z_j)(\pi-q_j)_+$ with positive coefficients; it is therefore continuous, piecewise linear, nondecreasing, unbounded, and zero at $\pi=0$. The tail sum $S(\pi)=\sum_{j\ge2}\max\{q_j,\pi\}$ is nondecreasing, so $D_0$ is nonincreasing. For $\ell\ge2$, $q_1-S+p_\ell=q_1-\sum_{j\ge2,\,j\ne\ell}p_j$ is nonincreasing, so $E_\ell$ is nonincreasing. For $E_1$, note that $2p_2-S=p_2-\sum_{j\ge3}p_j$. While $\pi\le q_2$, $p_2=q_2$ is constant and this expression is nonincreasing. When $\pi\ge q_2$, every tail coordinate equals $\pi$ (because $q_j\le q_2\le\pi$), so $2p_2-S=(3-k)\pi\le0$ and $E_1=0$. Thus every demand is continuous and nonincreasing. At $\pi=q_1$, all tail coordinates equal $q_1$, so $S=(k-1)q_1$, $q_1-S=(2-k)q_1$, $q_1-S+p_\ell=(3-k)q_1$ for $\ell\ge2$, and
\[
 D_0=((2-k)q_1-\omega)_+=0,\qquad E_\ell=((3-k)q_1)_+=0\ (\ell\ge2),\qquad E_1=0,
\]
using $k\ge3$. Hence $\mathrm{Res}-\mathrm{Dem}$ is continuous and nondecreasing, is at most zero at $\pi=0$ because $\mathrm{Res}(0)=0\le\mathrm{Dem}(0)$, and is nonnegative at $\pi=q_1$ because all demands vanish there. The least level $\pi^*$ with $\mathrm{Res}(\pi^*)\ge\mathrm{Dem}(\pi^*)$ therefore exists and lies in $[0,q_1]$. If $\pi^*>0$, minimality gives $\mathrm{Res}-\mathrm{Dem}<0$ on $[0,\pi^*)$, continuity gives $\mathrm{Res}(\pi^*)-\mathrm{Dem}(\pi^*)\le0$, and the defining inequality gives the reverse, so $\mathrm{Res}(\pi^*)=\mathrm{Dem}(\pi^*)$; if $\pi^*=0$, then $0\le\mathrm{Dem}(0)\le\mathrm{Res}(0)=0$. In both cases
\begin{equation}\label{eq:resdemequal}
 \mathrm{Res}(\pi^*)=\mathrm{Dem}(\pi^*),
\end{equation}
and if the resource is zero, so is the demand. Since $\pi^*\le q_1$ and $q_j\le q_1$, every $p_j=\max\{q_j,\pi^*\}$ is at most $p_1=q_1$; the head is a largest coordinate of the action.

For the running time, scan the breakpoints in increasing order of $\pi$ after sorting. Maintain the number and total input mass of raised tail coordinates, the floor sum $\sum q_j$ of the unraised tail coordinates, and the total mass and weighted floor sum of unraised coordinates whose label demand is positive. A tail coordinate becomes raised once and its demand becomes zero at most once; the empty-label demand and the first-label demand each vanish at most once. Since $E_\ell=(q_1-S+p_\ell)_+$ for $\ell\ge2$, at every level these demands are ordered like the action probabilities $p_\ell$: the raised coordinates share $p_\ell=\pi$ and lose their demands together, before any unraised coordinate, and the unraised coordinates, which have $p_\ell=q_\ell\ge\pi$, lose theirs in increasing order of $q_\ell$. The demand events can therefore be read off the sorted tail with one pointer. Between consecutive events both $\mathrm{Res}$ and $\mathrm{Dem}$ are affine, so their intersection is found by one linear equation. This gives $O(k)$ updates after the $O(k\log k)$ sort.

\emph{Step 3: the cone inclusions.} Put $\hat t=\omega t=\min\{\omega,(q_1-S)_+\}\ee_1$, where $S=S(\pi^*)$. Then $0\le\hat t_1\le(q_1-S)_+\le q_1=p_1$, so $p-\hat t\ge0$. We claim
\begin{equation}\label{eq:weightedcone}
 p-\hat t+D_0\theta\in Q_k,\qquad
 p-\ee_\ell+(1-p_\ell)\ee_\ell+E_\ell\theta\in Q_k\quad(\ell\in[k]).
\end{equation}
For the first inclusion, apply Step~1 to $y=p-\hat t$. If $q_1>S$, then $y_1=q_1-\min\{\omega,q_1-S\}\ge S\ge p_j$ for every $j\ge2$; if $q_1\le S$, then $\hat t=0$ and $y=p$. In both cases index $1$ carries a largest coordinate, and
\[
 2y_1-|y|=y_1-S=\begin{cases}
 q_1-S-\omega,&q_1-S\ge\omega,\\
 0,&0<q_1-S<\omega,\\
 q_1-S\le0,&q_1\le S,
 \end{cases}
\]
where $|y|=y_1+S$ because $y_j=p_j$ for $j\ge2$. In the second case $q_1-S-\omega<0$, and in the third $q_1-S-\omega\le q_1-S\le0$; hence in all three cases $(2y_1-|y|)_+=(q_1-S-\omega)_+=D_0$. The vector $\theta$ has $\theta_1=0$, because $\gamma_1=0$ and the fallback $\ee_k$ has $k\ge3$. Step~1 proves the first inclusion.

For the second inclusion, note that $p-\ee_\ell+(1-p_\ell)\ee_\ell=p-p_\ell\ee_\ell$ is the action with coordinate $\ell$ set to zero. If $\ell\ge2$, its largest coordinate is the head, because $p_1=q_1\ge p_j$ for every $j$ by Step~2; its total is $q_1+S-p_\ell$, and $2q_1-(q_1+S-p_\ell)=q_1-S+p_\ell$, whose positive part is $E_\ell$. Since $\theta_1=0$, Step~1 applies. If $\ell=1$, the vector is $(0,p_2,\ldots,p_k)$; its largest coordinate is $p_2$, because $p_j=\max\{q_j,\pi^*\}$ and $q_2\ge q_j$ for $j\ge2$; its total is $S$, and $(2p_2-S)_+=E_1$. When $E_1=0$, Step~1 needs no condition on $\theta$. When $E_1>0$, coordinate $2$ has not been raised: otherwise $\pi^*\ge q_2$, every tail coordinate would equal $\pi^*$, and $E_1=((3-k)\pi^*)_+=0$. Hence $p_2=q_2$, $\gamma_2=0$, and $\theta_2=0$, also for the fallback $\theta=\ee_k$. Step~1 again applies.

\emph{Step 4: the bracket inequalities.} Set $\psi_0=D_0\theta$ and $\psi_\ell=(1-p_\ell)\ee_\ell+E_\ell\theta$, so that $\varphi=u\psi$ in \eqref{eq:weightedphi}. These vectors are nonnegative because $p_\ell\le|p|\le1$ and $D_0,E_\ell,\theta\ge0$. Recall $u=1/(1-\omega)$, so $1+W=1+u\omega=u$ and $Wt=u\omega t=u\hat t$. Let $h\le a$ be pairwise monotone, and put $\bar a=(a+Wh)/(1+W)$. Since $\bar a=\frac1{1+W}a+\frac W{1+W}h$ is a convex combination of the pairwise-monotone vectors $a$ and $h$, it is pairwise monotone, and $h\le a$ gives $h\le\bar a\le a$. For \textup{(B0)},
\[
 -p\cdot a-W(p-t)\cdot h
 =-(1+W)\,p\cdot\bar a+Wt\cdot h
 =u\bigl(-p\cdot\bar a+\hat t\cdot h\bigr)
 \le-u(p-\hat t)\cdot\bar a,
\]
where the first equality uses $(1+W)\bar a=a+Wh$, the second uses $1+W=u$ and $Wt=u\hat t$, and the inequality uses $\hat t\cdot h\le\hat t\cdot\bar a$, which holds because $\hat t\ge0$ and $h\le\bar a$. By the first inclusion in \eqref{eq:weightedcone} and the duality between $Q_k$ and pairwise-monotone vectors (\cref{lem:pairing}), $(p-\hat t+\psi_0)\cdot\bar a\ge0$. Hence $-(p-\hat t)\cdot\bar a\le\psi_0\cdot\bar a\le\psi_0\cdot a$, the last step because $\psi_0\ge0$ and $\bar a\le a$. This proves $-p\cdot a-W(p-t)\cdot h\le u\psi_0\cdot a=\varphi_0\cdot a$, which is \textup{(B0)}. For \textup{(B1)}, the second inclusion and \cref{lem:pairing} give $(p-\ee_\ell+\psi_\ell)\cdot\bar a\ge0$, that is, $(\ee_\ell-p)\cdot\bar a\le\psi_\ell\cdot\bar a$. Using $a+Wh=(1+W)\bar a=u\bar a$, then this inequality, and finally $\psi_\ell\ge0$ with $\bar a\le a$,
\[
 (\ee_\ell-p)\cdot(a+Wh)=u(\ee_\ell-p)\cdot\bar a\le u\,\psi_\ell\cdot\bar a\le u\,\psi_\ell\cdot a=\varphi_\ell\cdot a.
\]

\emph{Step 5: the budget identity and the row bound.} Coordinate $j$ of $cp-(1-Z)\psi_0-\sum_\ell z_\ell\psi_\ell$ equals
\[
 cp_j-z_j(1-p_j)-\Bigl[(1-Z)D_0+\sum_\ell z_\ell E_\ell\Bigr]\theta_j
 =(c+z_j)p_j-z_j-\mathrm{Dem}(\pi^*)\theta_j=\gamma_j-\mathrm{Dem}(\pi^*)\theta_j,
\]
where for $j=1$ we used $(c+z_1)q_1-z_1=0=\gamma_1$. If $\mathrm{Res}(\pi^*)>0$, then $\theta=\gamma/\mathrm{Res}(\pi^*)$ and \eqref{eq:resdemequal} make this zero. If $\mathrm{Res}(\pi^*)=0$, then $\gamma=0$ and $\mathrm{Dem}(\pi^*)=0$, and it is again zero. Multiplying by $u$ shows that the vector in \textup{(B2)} is $0\in Q_k$; the budget remainder is zero. For the row bound, $\varphi_\ell-\varphi_0=u[(1-p_\ell)\ee_\ell+(E_\ell-D_0)\theta]$, so, since $\|\theta\|_1=1$ and $E_\ell,D_0\ge0$, $\|\varphi_\ell-\varphi_0\|_1\le u[(1-p_\ell)+|E_\ell-D_0|]\le u[1+\max\{E_\ell,D_0\}]$. Each demand is at most $q_1$: $D_0=(q_1-S-\omega)_+\le q_1$; for $\ell\ge2$, $E_\ell=(q_1-(S-p_\ell))_+\le q_1$ because $S\ge p_\ell$; and $E_1=(p_2-(S-p_2))_+\le p_2\le q_1$ because $S\ge p_2$. Since $s\mapsto s/(c+s)$ is increasing and $z_1\le1$, $q_1=z_1/(c+z_1)\le1/(c+1)$, and with $u=1/(1-\omega)$ this proves \eqref{eq:weightednorm}.

\emph{Step 6: the scalar mass condition.} Assume \eqref{eq:weightedfloors} and \eqref{eq:weightedgamma}. The action mass is $|p|=q_1+S(\pi^*)$, so it suffices to show $S(\pi^*)\le S_\chi=\chi(Z)-q_1$. By \eqref{eq:weightedfloors}, $S(0)=\sum_{j\ge2}q_j\le S_\chi$. Since $S$ is continuous, nondecreasing, and unbounded, the level $\pi_\chi=\max\{\pi\ge0:S(\pi)\le S_\chi\}$ exists: the set contains $0$, is closed by continuity, and is bounded because $S$ is unbounded. It satisfies $S(\pi_\chi)=S_\chi$, since $S(\pi_\chi)<S_\chi$ would, by continuity, allow a larger level. If $\pi_\chi\ge q_1$, then $\pi^*\le q_1\le\pi_\chi$ by Step~2. Otherwise we show $\mathrm{Res}(\pi_\chi)-\mathrm{Dem}(\pi_\chi)\ge\Gamma^\omega_\chi(z)\ge0$, which gives $\pi^*\le\pi_\chi$ by minimality of $\pi^*$. In either case $S(\pi^*)\le S(\pi_\chi)=S_\chi$.

At $\pi_\chi$, write $d=q_1-S_\chi$ and $\sigma=Z-z_1=\sum_{j\ge2}z_j$. Then, using $\sum_{j\ge2}p_j=S(\pi_\chi)=S_\chi$, $\mathrm{Res}=\sum_{j\ge2}[(c+z_j)p_j-z_j]=cS_\chi+\sum_{j\ge2}z_jp_j-\sigma$, while \eqref{eq:weighteddemands} with $q_1-S_\chi=d$ gives $D_0=(d-\omega)_+$ and $E_\ell=(d+p_\ell)_+$ for $\ell\ge2$, so
\[
 \mathrm{Dem}=(1-Z)(d-\omega)_++z_1(2p_2-S_\chi)_++\sum_{j\ge2}z_j(p_j+d)_+.
\]
Since $p_j-(p_j+d)_+=\min\{p_j,-d\}$ (both sides equal $-d$ if $p_j+d\ge0$ and $p_j$ otherwise),
\begin{equation}\label{eq:resminusdem}
 \mathrm{Res}-\mathrm{Dem}
 =cS_\chi-\sigma-(1-Z)(d-\omega)_+-z_1(2p_2-S_\chi)_++\sum_{j\ge2}z_j\min\{p_j,-d\}.
\end{equation}
We compare \eqref{eq:resminusdem} with \eqref{eq:weightedgamma} term by term. First, $(2p_2-S_\chi)_+=(2q_2-S_\chi)_+$: if coordinate $2$ is not raised, $p_2=q_2$; if it is raised, all tail coordinates equal $\pi_\chi$, so $2p_2-S_\chi=(3-k)\pi_\chi\le0$ and $2q_2-S_\chi\le2\pi_\chi-S_\chi\le0$, and both positive parts vanish. Second, if $d\ge0$, then $\min\{p_j,-d\}=-d$ because $p_j\ge0\ge-d$, so the last sum is $-\sigma d=-\sigma d_+$, while in \eqref{eq:weightedgamma} $\sum_jz_j\min\{q_j,(-d)_+\}=0$; the two expressions coincide. If $d<0$, then $\sigma d_+=0$ and $\min\{p_j,-d\}\ge\min\{q_j,(-d)_+\}$ because $p_j\ge q_j$. In every case $\mathrm{Res}-\mathrm{Dem}\ge\Gamma^\omega_\chi(z)\ge0$, as required. Finally $|p|\le\chi(Z)=\mu Z+(1-\mu)\zeta\le1$, so the hypothesis $|p|\le1$ of the first part holds.
\end{proof}

\subsection{Proof of \texorpdfstring{\cref{lem:twolabel}}{the two-label rule lemma} (two-label rule)}\label{app:twoweighted}
\begin{proof}[of \cref{lem:twolabel}]
Assume $z_1\ge z_2$, so $\kappa=z_1-z_2\in[0,Z]$; ties are broken by the fixed order.

\emph{The balance equation.} The left side of $(c+Z)d+(1-Z)(d-\omega)_+=\kappa$ is continuous in $d\ge0$, equals $0$ at $d=0$, and has slope at least $c+Z>0$, so it increases strictly to $+\infty$ and the equation has a unique solution. If $\kappa\le\omega(c+Z)$, the solution is $d=\kappa/(c+Z)\le\omega$; otherwise $d>\omega$ and $(c+Z)d+(1-Z)(d-\omega)=\kappa$ gives
\[
 d=\begin{cases}
 \kappa/(c+Z),&\kappa\le\omega(c+Z),\\
 [\kappa+(1-Z)\omega]/(c+1),&\kappa\ge\omega(c+Z).
 \end{cases}
\]
In the second case $(1-Z)(d-\omega)\ge0$, so $(c+Z)d\le\kappa$. Thus in both cases $0\le d\le\kappa/(c+Z)$, which is less than one when $Z>0$ because $\kappa\le Z<c+Z$; when $Z=0$, $d=\kappa=0$.

\emph{The action and its coefficient vectors.} Put $p_1=(r+d)/2$, $p_2=(r-d)/2$, $\hat t=\min\{d,\omega\}\ee_1$, and
\[
 \psi_0=(d-\omega)_+\ee_2,\qquad
 \psi_1=(1-d)\ee_1,\qquad \psi_2=(1+d)\ee_2.
\]
We claim
\begin{equation}\label{eq:twobudgetrows}
 cp_1=(1-d)z_1,\qquad
 cp_2=(1+d)z_2+(1-Z)(d-\omega)_+.
\end{equation}
Indeed, the two right sides add to $Z-\kappa d+(1-Z)(d-\omega)_+=cr=c(p_1+p_2)$ by the definition of $r$, and their difference is $\kappa-Zd-(1-Z)(d-\omega)_+=cd=c(p_1-p_2)$ by the balance equation. By \eqref{eq:twobudgetrows}, $cp_1=(1-d)z_1\ge0$ because $d<1$ or $z_1=0$, and $cp_2\ge0$ because each of its terms is nonnegative; hence both probabilities are nonnegative.

\emph{Cone inclusions.} In $Q_2=\{(v,v):v\ge0\}$, membership means equal nonnegative coordinates. Using $d=\min\{d,\omega\}+(d-\omega)_+$ and $p_1=p_2+d$,
\[
 p-\hat t+\psi_0=\bigl(p_1-\min\{d,\omega\},\ p_2+(d-\omega)_+\bigr)=\bigl(p_2+(d-\omega)_+,\ p_2+(d-\omega)_+\bigr)\in Q_2,
\]
\[
 p-\ee_1+\psi_1=(p_1-d,\ p_2)=(p_2,p_2)\in Q_2,\qquad
 p-\ee_2+\psi_2=(p_1,\ p_2+d)=(p_1,p_1)\in Q_2.
\]
These are the two-label analogues of \eqref{eq:weightedcone}. Also $p-\hat t\ge0$, because its first coordinate equals $p_2+(d-\omega)_+\ge0$; moreover $\hat t\ge0$, and $\psi_0,\psi_1,\psi_2\ge0$ because $d\le1$.

\emph{Block inequalities.} Set $\varphi=u\psi$, $w=W$, and $t=\hat t/\omega$ when $\omega>0$, or $t=0$ when $\omega=0$; then $t\in\Delta_2$ because $\min\{d,\omega\}/\omega\le1$, and $Wt=u\omega t=u\hat t$. Step~4 of the proof of \cref{lem:weightedlevel} used only these cone inclusions, nonnegativity, and $Wt=u\hat t$; repeating it verbatim proves \textup{(B0)} and \textup{(B1)}. By \eqref{eq:twobudgetrows},
\[
 cp-(1-Z)\psi_0-z_1\psi_1-z_2\psi_2
 =\bigl(cp_1-(1-d)z_1,\ cp_2-(1-Z)(d-\omega)_+-(1+d)z_2\bigr)=0,
\]
and multiplying by $u$ proves \textup{(B2)} with zero remainder. Finally $\psi_1-\psi_0=(1-d,-(d-\omega)_+)$ and $\psi_2-\psi_0=(0,1+d-(d-\omega)_+)$ have $\ell_1$-norms $1-d+(d-\omega)_+=1-\min\{d,\omega\}\le1$ and $1+d-(d-\omega)_+=1+\min\{d,\omega\}\le1+\omega$; multiplying by $u=1/(1-\omega)$ gives
\[
 \|\varphi_\ell-\varphi_0\|_1
 \le\frac{1+\omega}{1-\omega}.
\]

\emph{Mass bound.} The action mass is $|p|=r$. When $d\le\omega$, the first formula for $d$ gives $\kappa d=\kappa^2/(c+Z)$ and $(d-\omega)_+=0$, so $r=(Z-\kappa^2/(c+Z))/c\le Z/c$. When $d\ge\omega$, substituting the second formula for $d$ into the definition of $r$ gives
\[
 r=\frac{(c+1)Z+(1-\omega)(1-Z)\kappa-\kappa^2-(c+Z)\omega(1-Z)}{c(c+1)}.
\]
This is a concave quadratic in $\kappa$; its $\kappa$-derivative is $[(1-\omega)(1-Z)-2\kappa]/(c(c+1))$, so its unconstrained maximizer is $\kappa^*=(1-\omega)(1-Z)/2$. With $Z_b=(1-\omega)/(3-\omega)$, the inequality $\kappa^*\ge Z$, that is $(1-\omega)(1-Z)\ge2Z$, holds exactly when $Z\le Z_b$. For $Z\le Z_b$, the quadratic is nondecreasing on $\kappa\le\kappa^*$ and every admissible $\kappa$ satisfies $\kappa\le Z\le\kappa^*$, so $r$ is at most the value at $\kappa=Z$, namely
\[
 r=\mathcal T_{c,\omega}(Z):=
 \frac{Z(c+2-2Z)-\omega(1-Z)(c+2Z)}{c(c+1)}.
\]
For $Z\ge Z_b$, $r$ is at most its unconstrained maximum $r^*(Z)$, obtained by substituting $\kappa^*$:
\[
 c(c+1)r^*(Z)=(c+1)Z-(c+Z)\omega(1-Z)+\tfrac14(1-\omega)^2(1-Z)^2.
\] Viewed as a function of $Z$, $c(c+1)r^*(Z)$ is a quadratic with leading coefficient $(1+\omega)^2/4>0$, so $r^*$ is convex; moreover $r^*(Z_b)=\mathcal T_{c,\omega}(Z_b)$ and $r^*(1)=1/c$.

Now suppose that \eqref{eq:twoweightedtest} is nonnegative on $[0,1]$ and $\mu+(1-\mu)\zeta\ge1/c$. The quadratic \eqref{eq:twoweightedtest} equals $c(c+1)\bigl(\chi(Z)-\mathcal T_{c,\omega}(Z)\bigr)$, so its nonnegativity gives $\chi(Z)\ge\mathcal T_{c,\omega}(Z)$ on $[0,1]$. The affine function $\chi(Z)-Z/c$ is nonnegative at $Z=0$ and, by the side condition, at $Z=1$; hence $\chi(Z)\ge Z/c$ on $[0,1]$. For $Z\le Z_b$ the two cases give $r\le\max\{Z/c,\mathcal T_{c,\omega}(Z)\}\le\chi(Z)$. For $Z\ge Z_b$, convexity places $r^*$ below its chord from $(Z_b,\mathcal T_{c,\omega}(Z_b))$ to $(1,1/c)$, and the affine function $\chi$ lies above both chord endpoints, because $\chi(Z_b)\ge\mathcal T_{c,\omega}(Z_b)=r^*(Z_b)$ and $\chi(1)\ge1/c=r^*(1)$; an affine function that lies above both endpoints of a chord lies above the whole chord. Hence $r\le\max\{Z/c,r^*(Z)\}\le\chi(Z)$. Therefore $|p|=r\le\chi(Z)=\mu Z+(1-\mu)\zeta\le1$, the action is feasible, and the mass bound \eqref{eq:lifting} holds.

\end{proof}

\subsection{Proof of \texorpdfstring{\cref{thm:finitepoly}}{the univariate polynomial certificate} (univariate polynomial certificate)}\label{app:finiteverification}
The proof of \cref{thm:finitepoly} separates blocks according to whether the head floor $q_1$ exceeds half the target mass $\chi(Z)$. In the complementary case of a \emph{nondominant head}, $2q_1\le\chi(Z)$, the balanced requirement \eqref{eq:balancedmass} controls every term of \eqref{eq:weightedgamma}, including the truncated positive tail terms.
\begin{lemma}[Blocks with a nondominant head]\phantomsection\label[lemma]{lem:balancedblock}
Let $z\in\Delta_k$ be sorted decreasingly, $Z=|z|>0$, $c\ge c_0$, and $r\in[0,1]$. Suppose
$\sum_i z_i/(c+z_i)\le r$, $r\ge B_c(Z)$, and $2z_1/(c+z_1)\le r$. Then the surplus in \eqref{eq:weightedgamma}, with $\chi(Z)=r$, is nonnegative for every $\omega\ge0$.
\end{lemma}
\begin{proof}
\emph{Normalization.} Put ${\hat c}=c/Z$, which is at least $c\ge c_0$ because $Z\le1$, and write $a=z_1/Z$, $b=z_2/Z$, $q=a/({\hat c}+a)$, and $\delta=r-2q$; this rescales input masses by $1/Z$ and leaves the action mass unchanged. Then $q=(z_1/Z)/((c+z_1)/Z)=z_1/(c+z_1)=q_1$, and $\delta\ge0$ by the hypothesis $2q_1\le r$. Dividing the numerator and denominator of $B_c(Z)$ by $Z^2$ shows that the balanced requirement becomes
\[
 r\ge B_{\hat c}:=\frac{{\hat c}+3/2}{({\hat c}+1/2)({\hat c}+1)}.
\]
With $\chi(Z)=r$ in \eqref{eq:weightedgamma}, we have $S_\chi=r-q$ and $d=q-S_\chi=2q-r=-\delta\le0$. Hence $d_+=0$, $(d-\omega)_+=0$ for every $\omega\ge0$, $(-d)_+=\delta$, and $2q_2-S_\chi=2q_2-r+q$, so the surplus is
\begin{equation}\label{eq:balancedsurplus}
 \Gamma^\omega_\chi(z)=c(r-q)-\sigma-z_1\bigl(2q_2-r+q\bigr)_++\sum_{j\ge2}z_j\min\{q_j,\delta\},
\end{equation}
independently of $\omega$. Here $\sigma=Z-z_1$, $q_2=z_2/(c+z_2)=b/({\hat c}+b)$, and $q_j\le q_2\le q$ for $j\ge2$.

\emph{Case $q_2\le\delta$.} Then $q_j\le\delta$ for every $j\ge2$, so every minimum in \eqref{eq:balancedsurplus} equals $q_j$. Moreover $2q_2\le q+\delta$ because $q_2\le\min\{q,\delta\}$, so $2q_2-r+q\le2q+\delta-r=0$ and the positive part vanishes. Since $\sigma=\sum_{j\ge2}z_j$ and $z_j(1-q_j)=z_jc/(c+z_j)=cq_j$,
\[
 \Gamma^\omega_\chi(z)=c(r-q)-\sum_{j\ge2}z_j(1-q_j)=c\Bigl(r-\sum_{i=1}^k\frac{z_i}{c+z_i}\Bigr)\ge0.
\]

\emph{Case $q_2>\delta$.} Then $q\ge q_2>\delta=r-2q$, so $q>r/3$; moreover $r\ge B_c(Z)>0$ and $r/3<q\le r/2\le1/2$. Drop the nonnegative terms with $j\ge3$ from the last sum in \eqref{eq:balancedsurplus} and use $\min\{q_2,\delta\}=\delta$. Dividing by $Z$ and using $c/Z={\hat c}$, $\sigma/Z=1-a$, $z_1/Z=a$, $z_2/Z=b$, and $q_2=b/({\hat c}+b)$ gives $\Gamma^\omega_\chi(z)/Z\ge H(b)$, where
\begin{equation}\label{eq:balancedH}
 H(b)={\hat c}(r-q)-1+a+b\delta-a\left(\frac{2b}{{\hat c}+b}-r+q\right)_+.
\end{equation}
We regard $a$, $q$, $r$, and $\delta$ as fixed and bound $H$ over the admissible values of $b$. Since $z_2\le z_1$ and $z_1+z_2\le Z$, we have $b\le {\bar b}:=\min\{a,1-a\}$. The case condition $b/({\hat c}+b)>\delta$ is equivalent to $b>{b_0}:={\hat c}\delta/(1-\delta)$ (here $\delta<1$: $\delta=r-2q\le r\le1$, and equality would force $q=0$, contradicting $q>r/3>0$). Thus $b\in({b_0},{\bar b}]$. At $b={b_0}$ we have ${b_0}/({\hat c}+{b_0})=\delta$, so the expression inside the positive part equals $2\delta-r+q=r-3q$, which is negative because $q>r/3$.

\emph{Shape of $H$.} While the positive part is inactive, $H$ is affine in $b$ with slope $\delta\ge0$. Where it is active,
\[
 \partial_bH=\delta-\frac{2a{\hat c}}{({\hat c}+b)^2}\le\delta-\frac{2a{\hat c}}{({\hat c}+a)^2}=r-2q-2q(1-q)=r-4q+2q^2,
\]
using $b\le a$, $a/({\hat c}+a)=q$, and ${\hat c}/({\hat c}+a)=1-q$. The right side is a convex function of $q$, so on $(r/3,r/2]$ it is at most the larger endpoint value, namely
\[
 \max\Bigl\{\frac{r(2r-3)}9,\ \frac{r(r-2)}2\Bigr\}=\frac{r(2r-3)}9<0,
\]
because $(2r-3)/9-(r-2)/2=(12-5r)/18>0$ and $r>0$; the value is negative because $0<r\le1<3/2$. The argument of the positive part is increasing in $b$, so once active it stays active. Hence $H$ is nondecreasing and then decreasing on $[{b_0},{\bar b}]$, and
\[
 H(b)\ge\min\{H({b_0}),\,H_{\mathrm{act}}({\bar b})\}\qquad({b_0}<b\le {\bar b}),
\]
where $H_{\mathrm{act}}$ denotes \eqref{eq:balancedH} with the positive part replaced by its argument. If the positive part is active at ${\bar b}$, then $H({\bar b})=H_{\mathrm{act}}({\bar b})$; if it is inactive throughout $[{b_0},{\bar b}]$, then $H$ is nondecreasing and $H(b)\ge H({b_0})$. Hence it suffices to show $H({b_0})\ge0$ and $H_{\mathrm{act}}({\bar b})\ge0$.

\emph{The value $H({b_0})$.} At $b={b_0}$ the positive part is inactive. Using $a={\hat c}q/(1-q)$, ${b_0}\delta={\hat c}\delta^2/(1-\delta)$, and $-q+q/(1-q)=q^2/(1-q)$,
\[
 H({b_0})={\hat c}\left(r+\frac{q^2}{1-q}+\frac{\delta^2}{1-\delta}\right)-1.
\]
By the Cauchy--Schwarz inequality in the form $x_1^2/y_1+x_2^2/y_2\ge(x_1+x_2)^2/(y_1+y_2)$ with $x_1=2q$, $y_1=4(1-q)$, $x_2=\delta$, $y_2=1-\delta$, and then $4q+\delta=2q+r>5r/3$,
\[
 \frac{q^2}{1-q}+\frac{\delta^2}{1-\delta}
 \ge\frac{(2q+\delta)^2}{4(1-q)+(1-\delta)}
 =\frac{r^2}{5-4q-\delta}
 \ge\frac{3r^2}{5(3-r)}.
\]
The equality uses $2q+\delta=r$, and the last inequality uses $5-4q-\delta=5-(2q+r)<5-5r/3=5(3-r)/3$. Since $r+3r^2/(5(3-r))=r(15-2r)/(5(3-r))$, it follows that $H({b_0})\ge {\hat c}r(15-2r)/(5(3-r))-1$. The function $r\mapsto r(15-2r)/(3-r)$ is increasing on $[0,1]$, since its derivative has numerator $(15-4r)(3-r)+r(15-2r)=45-12r+2r^2>0$. As ${\hat c}>0$ and $B_{\hat c}\le r\le1$, it suffices to check $r=B_{\hat c}$. There, after multiplication by $5(3-B_{\hat c})>0$, nonnegativity is equivalent to positivity of the left side of the identity
\begin{equation}\label{eq:balancedinactiveidentity}
 {\hat c}B_{\hat c}(15-2B_{\hat c})-5(3-B_{\hat c})
 =\frac{2{\hat c}(6{\hat c}^2+3{\hat c}-4)}{({\hat c}+1)^2(2{\hat c}+1)^2}>0.
\end{equation}
The identity follows by substituting $B_{\hat c}$ and clearing denominators. Its right side is positive for ${\hat c}\ge c_0$: the factor $6{\hat c}^2+3{\hat c}-4$ is increasing for ${\hat c}>0$, and at ${\hat c}=c_0$ it equals $3(2c_0^2+c_0)-4=2$, because $c_0$ is the positive root of $2c^2+c-2$.

\emph{The value $H_{\mathrm{act}}({\bar b})$ when $a\ge1/2$.} Here ${\bar b}=1-a$. Substitute $b=1-a$ and $\delta=r-2q$ into $H_{\mathrm{act}}$. The coefficient of $r$ is ${\hat c}+(1-a)+a={\hat c}+1$, and the remaining terms are $a-1-q({\hat c}+2-a)-2a(1-a)/({\hat c}+1-a)$; with $q=a/({\hat c}+a)$, $a-q({\hat c}+2-a)=-2a(1-a)/({\hat c}+a)$. Hence
\[
 H_{\mathrm{act}}(1-a)=({\hat c}+1)r-1-2a(1-a)\left(\frac1{{\hat c}+a}+\frac1{{\hat c}+1-a}\right).
\]
With $x=a(1-a)\in[0,1/4]$, the subtracted term equals $2x(2{\hat c}+1)/({\hat c}^2+{\hat c}+x)$, because $({\hat c}+a)({\hat c}+1-a)={\hat c}^2+{\hat c}+x$. Its derivative in $x$ is $2(2{\hat c}+1)({\hat c}^2+{\hat c})/({\hat c}^2+{\hat c}+x)^2>0$, so it is at most its value at $x=1/4$, namely $2/(2{\hat c}+1)=1/({\hat c}+1/2)$. On the other hand, $r\ge B_{\hat c}$ gives $({\hat c}+1)r-1\ge({\hat c}+1)B_{\hat c}-1=({\hat c}+3/2)/({\hat c}+1/2)-1=1/({\hat c}+1/2)$. Hence $H_{\mathrm{act}}(1-a)\ge0$.

\emph{The value $H_{\mathrm{act}}({\bar b})$ when $a\le1/2$.} Here ${\bar b}=a$. Substitute $b=a$, so that $2b/({\hat c}+b)=2q$, and $\delta=r-2q$. The coefficient of $r$ is ${\hat c}+2a$, and the remaining terms are $-1+a-q({\hat c}+5a)$, where $a-q({\hat c}+5a)=-4a^2/({\hat c}+a)$ by $q=a/({\hat c}+a)$. Hence
\[
 H_{\mathrm{act}}(a)=({\hat c}+2a)r-1-\frac{4a^2}{{\hat c}+a}.
\]
This is increasing in $r$ because ${\hat c}+2a>0$, so it suffices to take $r=B_{\hat c}$. The case condition $q>r/3\ge B_{\hat c}/3$ is equivalent to $a(3-B_{\hat c})>{\hat c}B_{\hat c}$, that is, $a>\ell:={\hat c}B_{\hat c}/(3-B_{\hat c})$; so only $a\in(\ell,1/2]$ matters. After multiplication by ${\hat c}+a>0$, the expression $(({\hat c}+2a)B_{\hat c}-1)({\hat c}+a)-4a^2$ is a quadratic in $a$ with leading coefficient $2B_{\hat c}-4<0$, hence concave, and on $[\ell,1/2]$ (there is nothing to check if $\ell\ge1/2$) it is at least the smaller of its values at $a=\ell$ and $a=1/2$. At $a=1/2$ it equals $(({\hat c}+1)B_{\hat c}-1)({\hat c}+1/2)-1=0$, because $({\hat c}+1)B_{\hat c}-1=1/({\hat c}+1/2)$. At $a=\ell$, substituting $\ell={\hat c}B_{\hat c}/(3-B_{\hat c})$, the value before multiplication is
\[
 ({\hat c}+2\ell)B_{\hat c}-1-\frac{4\ell^2}{{\hat c}+\ell}=\frac{{\hat c}B_{\hat c}(9-B_{\hat c})}{3(3-B_{\hat c})}-1,
 \qquad
 {\hat c}B_{\hat c}(9-B_{\hat c})-3(3-B_{\hat c})
 =\frac{{\hat c}(8{\hat c}^2+6{\hat c}-3)}{({\hat c}+1)^2(2{\hat c}+1)^2}>0,
\]
where the first expression is nonnegative because $3(3-B_{\hat c})>0$ and the second is positive; positivity holds for ${\hat c}\ge c_0$ because $8c_0^2+6c_0-3=4(2c_0^2+c_0)+2c_0-3=5+2c_0>0$ and the numerator increases in ${\hat c}$. Hence $H_{\mathrm{act}}(a)\ge0$.

Combining the cases, $\Gamma^\omega_\chi(z)\ge ZH(b)\ge0$.
\end{proof}

For a dominant head, the joint polynomial \eqref{eq:jointpoly} enters. By \cref{lem:affinejoint}, an affine function lying above the balanced requirement and above the auxiliary function $u_c$ defined there bounds the joint term. The function $u_c$ lies below the pure-label requirement, since $r_c(s)-u_c(s)=s^2(1-s)/((c+1)(c+s)(c+1-s))\ge0$; \cref{lem:jointredundant} shows that the pure-label and floor tests, after an affine shift absorbing $\omega$, place the target above $u_c$.
\begin{lemma}[Affine domination of the joint penalty]\phantomsection\label[lemma]{lem:affinejoint}
Let $c>0$, and define
\[
 u_c(s)=\frac{s(c+2-s)}{(c+s)(c+1)}
       =\frac{2s}{c+s}-\frac{s}{c+1},
 \qquad
 B_c(s)=\frac{s(2c+3s)}{(c+s)(2c+s)}.
\]
If an affine function $\eta$ satisfies $\eta(s)\ge\max\{u_c(s),B_c(s)\}$ for every $s\in[0,1]$, then
\[
 J(a,b):=(c+1)\eta(a+b)-(a+b)
 -\frac{2a(1-a)}{c+a}-\frac{2ab}{c+b}\ge0
 \qquad(0\le b\le a,\ a+b\le1).
\]
\end{lemma}
\begin{proof}
Every denominator below is positive. At $b=0$, the inequality $J(a,0)\ge0$ is exactly $\eta(a)\ge u_c(a)$, because $(c+1)u_c(a)=a+2a(1-a)/(c+a)$. At $b=a$ (so $2a\le1$), the bound $\eta(2a)\ge B_c(2a)$ and $2a(1-a)+2a\cdot a=2a$ give
\[
 J(a,a)\ge(c+1)B_c(2a)-2a-\frac{2a}{c+a}
 =\frac{2a^2(1-2a)}{(c+a)(c+2a)}\ge0.
\]
On $a+b=1$, write $x=a(1-a)\le1/4$. The rational penalty equals
\[
 \frac{2a(1-a)}{c+a}+\frac{2a(1-a)}{c+1-a}
 =\frac{2x(2c+1)}{c^2+c+x}\le\frac1{c+1/2},
\]
where the identity uses $(c+a)(c+1-a)=c^2+c+x$, and the inequality holds because the middle expression increases in $x$ and equals $1/(c+1/2)$ at $x=1/4$. On the other hand, $(c+1)\eta(1)-1\ge(c+1)B_c(1)-1=1/(c+1/2)$. On this segment $J$ equals $(c+1)\eta(1)-1$ minus the rational penalty, so $J\ge0$ there. Hence $J\ge0$ on all three boundary segments $b=0$, $b=a$, and $a+b=1$.

Suppose for contradiction that $J$ is negative somewhere. Since $J$ is continuous on the compact triangle, it attains a negative minimum there; since $J\ge0$ on the boundary, the minimizer is an interior point, where $0<b<a$, $a+b<1$, and both partial derivatives vanish. Put $Z=a+b$, let $m=\eta'$ be the affine slope, and set ${\bar m}=(c+1)m-1$. Since $\eta(s)=\eta(0)+ms$, we have $\partial_bJ=(c+1)m-1-2ac/(c+b)^2$, and $\partial_aJ-\partial_bJ$ does not involve $m$; the stationarity equations therefore imply
\[
 {\bar m}=\frac{2ac}{(c+b)^2},
 \qquad
 0=\partial_aJ-\partial_bJ
 =2c\left(\frac{c+Z}{(c+b)^2}-\frac{c+1}{(c+a)^2}\right).
\]
Set $x=(c+a)/(c+b)$, which exceeds $1$ because $a>b$. The second equation gives $c+Z=(c+1)(c+b)^2/(c+a)^2=(c+1)/x^2$; together with $c+a=x(c+b)$, solving these two linear equations for $a$ and $b$ yields
\[
 a=\frac{c+1-cx}{x(x+1)},
 \qquad
 b=\frac{c+1-cx^3}{x^2(x+1)}.
\]
In particular, $b>0$ gives $cx^3<c+1$, that is, $0<c<1/(x^3-1)$, and
\[
 a=\frac{1-c(x-1)}{x(x+1)}<\frac12,
\]
where the inequality uses $c(x-1)>0$ and $x(x+1)>2$. Thus $2a\in(0,1)$ is an admissible argument for the balanced lower bound. Affinity gives $\eta(Z)=\eta(2a)+m(Z-2a)=\eta(2a)+m(b-a)$; with $(c+1)m={\bar m}+1$ this yields $(c+1)\eta(Z)-Z=(c+1)\eta(2a)-2a+{\bar m}(b-a)$, and $\eta(2a)\ge B_c(2a)$ gives
\begin{align*}
 J(a,b)&\ge (c+1)B_c(2a)-2a
 -\frac{2a(1-a)}{c+a}-\frac{2ab}{c+b}+{\bar m}(b-a)\\
 &=\frac{2(x-1)(cx-c-1)T(c,x)}
 {x(x+1)(cx^2+c+1)(cx^2-cx+2c+2)},
\end{align*}
where substitution of the displayed $a,b,{\bar m}$ and collection of terms give
\[
 T(c,x)=c^2(x-1)(x^4+x^2+2x+2)
       +c(2x^3+x^2-x-4)-x-2.
\]
For fixed $x>1$, this is a convex quadratic in $c$, since its $c^2$-coefficient $(x-1)(x^4+x^2+2x+2)$ is positive. At the endpoints of the admissible interval,
\[
 T(0,x)=-x-2<0,
 \qquad
 T\left(\frac1{x^3-1},x\right)
 =-\frac{x^3(x+1)^2}{(x^2+x+1)^2}<0.
\]
A convex function lies below its chord, so $T(c,x)<0$ for $0\le c\le1/(x^3-1)$. Also $cx-c-1=-a x(x+1)<0$ by the formula for $a$. The numerator $2(x-1)(cx-c-1)T(c,x)$ is therefore a product of one positive and two negative factors, hence positive, and the denominator factors $x$, $x+1$, $cx^2+c+1$, and $cx^2-cx+2c+2=c(x^2-x+2)+2$ are positive. Therefore $J(a,b)>0$, contradicting the choice of a negative minimum.
\end{proof}

\begin{lemma}[Redundancy of the weighted joint polynomial]\phantomsection\label[lemma]{lem:jointredundant}
Let $c>0$, $\omega\ge0$, and let $\chi$ be affine. If the polynomials $\mathcal F_k$, $\mathcal B$, and $\mathcal R$ defined in \eqref{eq:finitefloor}, \eqref{eq:balancedpoly}, and \eqref{eq:purepoly} are nonnegative on $[0,1]$ for some $k\ge1$, then $\mathcal P$ in \eqref{eq:jointpoly} is nonnegative on $0\le b\le a$, $a+b\le1$.
\end{lemma}
\begin{proof}
Define the affine function
\[
 \eta(s)=\chi(s)+\frac{\omega(1-s)}{c+1}.
\]
The balanced test $\mathcal B\ge0$ is equivalent to $\chi(s)\ge B_c(s)$, because $(c+s)(c+s/2)>0$; since $\omega(1-s)/(c+1)\ge0$, it follows that $\eta(s)\ge B_c(s)$. The floor test gives
\[
 \chi(s)\ge\frac{ks}{kc+s}\ge\frac{s}{c+s},
\]
where the second inequality is $c+s/k\le c+s$ for $k\ge1$. Expanding $\mathcal R$ from \eqref{eq:purepoly} gives the identity
\[
 (c+1)\eta(s)-s-\frac{2s(1-s)}{c+s}
 =\frac{\mathcal R(s)}{c+s}
   +s\left(\chi(s)-\frac{s}{c+s}\right)\ge0.
\]
Since $(c+1)u_c(s)=s+2s(1-s)/(c+s)$, this says $\eta(s)\ge u_c(s)$. Apply \cref{lem:affinejoint} to $\eta$. Because $(c+1)\eta(a+b)=(c+1)\chi(a+b)+(1-a-b)\omega$, multiplying its function $J$ by $(c+a)(c+b)$ reproduces \eqref{eq:jointpoly}:
\[
 (c+a)(c+b)J(a,b)=\mathcal P(a,b),
\]
and $(c+a)(c+b)>0$, so the joint polynomial is nonnegative on the entire triangle.
\end{proof}

For a dominant head, the proof below works with the head mass $a=z_1$ and the tail mass $b=Z-a$; the pure polynomial and \cref{lem:jointredundant} control the branches that the balanced bound does not.
\begin{proof}[of \cref{thm:finitepoly}]
Assume the three univariate tests. At $Z=1$, the balanced test gives $\chi(1)\ge B_c(1)$, and $\chi(1)=\mu+(1-\mu)\zeta\le1$; hence $B_c(1)=(c+3/2)/((c+1)(c+1/2))\le1$, which is equivalent to $c^2+c/2-1\ge0$, that is, $2c^2+c-2\ge0$. Since $c_0$ is the positive root of $2c^2+c-2$, this gives $c\ge c_0$, which \cref{lem:balancedblock} uses. By \cref{lem:weightedlevel}, it suffices to verify \eqref{eq:weightedfloors} and $\Gamma^\omega_\chi(z)\ge0$ for every $z\in\Delta_k$; the lemma then gives $|p(z)|\le\chi(|z|)\le1$ and \textup{(B0)--(B2)}.

\emph{Floors and the balanced bound.} The function $s\mapsto s/(c+s)$ is concave, so Jensen's inequality over the $k$ coordinates gives
\[
 \sum_{i=1}^k\frac{z_i}{c+z_i}\le k\,\frac{Z/k}{c+Z/k}=\frac{kZ}{kc+Z}.
\]
The condition $\mathcal F_k(Z)\ge0$ in \eqref{eq:finitefloor} is equivalent to $\chi(Z)\ge kZ/(kc+Z)$, so \eqref{eq:weightedfloors} holds. Since $(c+Z)(c+Z/2)>0$, the condition $\mathcal B(Z)\ge0$ in \eqref{eq:balancedpoly} is likewise equivalent to $\chi(Z)\ge B_c(Z)$.

\emph{Zero mass and nondominant heads.} If $Z=0$, then $z=0$, every $q_i$ is zero, and $\mathrm{Res}(0)=\mathrm{Dem}(0)=0$; the action is zero. Let $Z>0$. If the head floor satisfies $2q_1\le\chi(Z)$, apply \cref{lem:balancedblock} with $r=\chi(Z)\in[0,1]$: its three hypotheses are \eqref{eq:weightedfloors}, $\chi(Z)\ge B_c(Z)$, and $2q_1\le\chi(Z)$. It gives $\Gamma^\omega_\chi(z)\ge0$.

\emph{Dominant heads.} Suppose $2q_1>\chi(Z)$. Set $a=z_1$, $b=Z-a$, and $q=q_1=a/(c+a)$. Since $(c+3Z/2)>(c+Z)$,
\[
 2q_1>\chi(Z)\ge B_c(Z)>\frac{Z}{c+Z/2}=2\,\frac{Z/2}{c+Z/2},
\]
and monotonicity of $s/(c+s)$ gives $a>Z/2$. Hence $0\le b<a$ and $a+b\le1$. In \eqref{eq:weightedgamma}, $S_\chi=\chi(Z)-q$ and $d=2q-\chi(Z)>0$, so $(-d)_+=0$, every minimum in the last sum vanishes, $\sigma=b$, $\sigma d_+=bd$, and
\[
 \Gamma^\omega_\chi(z)=cS_\chi-b-bd-(1-Z)(d-\omega)_+-a\bigl(2q_2-S_\chi\bigr)_+.
\]
Since $z_2\le b$, we have $q_2\le q_b:=b/(c+b)$, and replacing $q_2$ by $q_b$ can only decrease the expression, because $x\mapsto(x)_+$ is nondecreasing and the positive part enters with the coefficient $-a\le0$. Because $-(x)_+=\min\{-x,0\}$, the resulting lower bound is the minimum of four expressions, obtained by keeping or omitting each of the two penalty terms. Substituting $S_\chi=\chi-q$, $d=2q-\chi$, and $cq=a(1-q)$, these four expressions are
\begin{align*}
 H_0&=(c+b)\chi-Z+(a-2b)q,\\
 H_1&=(c+Z)\chi-Z-2ab\left(\frac1{c+a}+\frac1{c+b}\right),\\
 J_0&=(c+1-a)\chi(Z)-Z+(3a-2)q+(1-Z)\omega,\\
 J_1&=(c+1)\chi(Z)-Z-\frac{2a(1-a)}{c+a}-\frac{2ab}{c+b}+(1-Z)\omega,
\end{align*}
where $\chi=\chi(Z)$. Here $H_0$ omits both penalties, $H_1$ keeps only the selected-label penalty $-a(2q_b-S_\chi)$, $J_0$ keeps only the empty-label penalty $-(1-Z)(d-\omega)$, and $J_1$ keeps both; it suffices to show that all four are nonnegative.

\emph{$H_0$.} With $q=a/(c+a)$, multiplying $H_0$ by $c+a>0$ gives $(c+a)(c+b)\chi-(cZ+3ab)$, because $Z(c+a)-a(a-2b)=cZ+3ab$. Since $(c+a)(c+b)=c(c+Z)+ab>0$, $H_0\ge0$ is equivalent to
\[
 \chi(Z)\ge\frac{cZ+3ab}{c(c+Z)+ab}.
\]
The right side is increasing in $ab$, since its derivative in $ab$ has the sign of $3c(c+Z)-cZ>0$, and $ab\le Z^2/4$. Hence it is at most
\[
 \frac{Z(c+3Z/4)}{(c+Z/2)^2}\le B_c(Z),
\]
where the last inequality is $(c+3Z/4)(c+Z)\le(c+3Z/2)(c+Z/2)$, that is, $c^2+\tfrac74cZ+\tfrac34Z^2\le c^2+2cZ+\tfrac34Z^2$. Since $\chi(Z)\ge B_c(Z)$, $H_0\ge0$.

\emph{$H_1$.} With $x=ab$, $(c+a)(c+b)=c^2+cZ+x$ and the penalty equals $2x(2c+Z)/(c^2+cZ+x)$, which is increasing in $x$. At $x=Z^2/4$ it equals $2Z^2/(2c+Z)$, so $H_1\ge(c+Z)\chi(Z)-Z-2Z^2/(2c+Z)$. This is nonnegative exactly when $\chi(Z)\ge Z(2c+3Z)/((c+Z)(2c+Z))=B_c(Z)$, which holds.

\emph{$J_0$.} Fix $a$ and regard $J_0$ as a function of $Z\in[a,\min\{2a,1\}]$, that is, of $b\in[0,\min\{a,1-a\}]$. Since $\chi$ is affine and $q$ depends only on $a$, $J_0$ is affine in $Z$; it suffices to check the endpoints. At $Z=a$, comparing with \eqref{eq:purepoly} gives $J_0=\mathcal R(a)/(c+a)\ge0$. A direct subtraction gives $J_0-H_0=(1-Z)(\chi(Z)-2q+\omega)$. At $Z=2a\le1$, that is $b=a$, we have $2q=Z/(c+Z/2)<B_c(Z)\le\chi(Z)$ by the display above, so $J_0-H_0=(1-Z)(\chi(Z)-2q+\omega)\ge0$ and $J_0\ge H_0\ge0$; here $H_0\ge0$ was proved for all $a,b$ with $\chi(Z)\ge B_c(Z)$. At $Z=1$, the factor $1-Z$ vanishes and $J_0=H_0\ge0$. Hence $J_0\ge0$ on the whole segment.

\emph{$J_1$.} By \eqref{eq:jointpoly}, $J_1=\mathcal P(a,b)/((c+a)(c+b))$. The three univariate tests imply $\mathcal P\ge0$ on this triangle by \cref{lem:jointredundant}, so $J_1\ge0$.

All four branches are nonnegative, so $\Gamma^\omega_\chi(z)\ge0$, and \cref{lem:weightedlevel} gives feasibility and the block inequalities.

\emph{Smaller label counts.} The function $k\mapsto kZ/(kc+Z)=Z/(c+Z/k)$ is nondecreasing, so $\mathcal F_{k_0}\ge0$ implies $\mathcal F_k\ge0$ for $3\le k\le k_0$. The remaining tests do not depend on $k$, which proves the last assertion.
\end{proof}

\subsection{Proof of \texorpdfstring{\cref{lem:labelcutoff}}{the floor certificate lemma} (floor certificate across label counts)}\label{app:labelcutoff}
\begin{proof}[of \cref{lem:labelcutoff}]
Write
\begin{equation}\label{eq:floorsplit}
 \mathcal F_k(s)=k\bigl(c\chi(s)-s\bigr)+s\chi(s).
\end{equation}
The second term is nonnegative on $[0,1]$ because $\mu,v>0$.

\emph{Case $c(\mu+v)\ge1$.} The affine function $c\chi(s)-s$ equals $cv>0$ at $s=0$ and $c(\mu+v)-1\ge0$ at $s=1$, so it is nonnegative on $[0,1]$. By \eqref{eq:floorsplit}, $\mathcal F_k\ge0$ for every $k$, in agreement with $\Kbar=+\infty$.

\emph{Case $c(\mu+v)<1$.} Since $v>0$, also $c\mu<1$, so $h=\sqrt{c\mu}\in(0,1)$. The affine term is $c\chi(s)-s=cv-(1-c\mu)s$, which is nonnegative exactly for $s\le s_0:=cv/(1-c\mu)$. Here $0<s_0<1$, the second inequality being equivalent to $c(\mu+v)<1$. On $[0,s_0]$ both terms of \eqref{eq:floorsplit} are nonnegative. On $(s_0,1]$ the affine term is negative, and $\mathcal F_k(s)\ge0$ is equivalent to
\[
 k\le\Phi(s):=\frac{s(\mu s+v)}{(1-c\mu)s-cv}.
\]
Thus $k\le\inf_{(s_0,1]}\Phi$ suffices to give $\mathcal F_k\ge0$ on $[0,1]$. Write $\Phi=n_\Phi/d_\Phi$ with $n_\Phi=\mu s^2+vs$ and $d_\Phi=(1-c\mu)s-cv$. Since $n_\Phi(s_0)=s_0\chi(s_0)$ is positive, while $d_\Phi$ vanishes at $s_0$ and is positive on $(s_0,1]$, $\Phi(s)\to+\infty$ as $s\downarrow s_0$. Moreover,
\[
 n_\Phi'd_\Phi-n_\Phi d_\Phi'=(2\mu s+v)\bigl((1-c\mu)s-cv\bigr)-(\mu s^2+vs)(1-c\mu)
 =\mu(1-c\mu)s^2-2c\mu vs-cv^2.
\]
This quadratic has positive leading coefficient $\mu(1-c\mu)$ and negative constant term $-cv^2$, so it has exactly one positive root, is negative before it, and positive after it. Since $\Phi'=(n_\Phi'd_\Phi-n_\Phi d_\Phi')/d_\Phi^2$, the sign of $\Phi'$ is the sign of this quadratic. Using $c=h^2/\mu$, its positive root is
\[
 \frac{c\mu v+v\sqrt{c^2\mu^2+c\mu(1-c\mu)}}{\mu(1-c\mu)}
 =\frac{v(h^2+h)}{\mu(1-h)(1+h)}=\frac{vh}{\mu(1-h)}=\frac{v\sqrt{c/\mu}}{1-h}=\tau.
\]
Moreover $s_0=vh^2/(\mu(1-h)(1+h))$, so $\tau/s_0=(1+h)/h>1$. Therefore $\Phi$ decreases on $(s_0,\min\{\tau,1\}]$ and increases on $[\tau,1]$ when $\tau<1$, and its infimum over $(s_0,1]$ is attained at $\min\{\tau,1\}$.

If $\tau\le1$, then $\chi(\tau)=\mu\tau+v=vh/(1-h)+v=v/(1-h)$, so
\[
 n_\Phi(\tau)=\tau\chi(\tau)=\frac{v^2h}{\mu(1-h)^2},\qquad
 d_\Phi(\tau)=\tau-c\chi(\tau)=\frac{vh}{\mu(1-h)}-\frac{h^2}{\mu}\cdot\frac{v}{1-h}=\frac{vh}{\mu},
\]
and $\Phi(\tau)=v/(1-h)^2$. If $\tau\ge1$, the minimum is $\Phi(1)=(\mu+v)/(1-c(\mu+v))$. At $\tau=1$ both expressions equal $\Phi(1)$. This proves \eqref{eq:labelcutoff}: every positive integer $k\le\lfloor\Kbar\rfloor$ satisfies the floor certificate.
\end{proof}

\subsection{Proof of \texorpdfstring{\cref{thm:uniformcubic}}{the uniform cubic certificate} (cubic certificate for every \texorpdfstring{$k\ge3$}{k>=3})}\label{app:uniformcubic}
The limiting floor bound $\chi(Z)\ge Z/c$ implies the floor requirement for every finite label count and also the balanced requirement, so only the pure polynomial remains to be checked.

\begin{proof}[of \cref{thm:uniformcubic}]
Suppose that $\chi(1)=\mu+(1-\mu)\zeta\ge1/c$ and that $\mathcal R\ge0$ on $[0,1]$. The affine function $\chi(Z)-Z/c$ equals $\chi(0)=(1-\mu)\zeta\ge0$ at $Z=0$ and is nonnegative at $Z=1$ by hypothesis, so $\chi(Z)\ge Z/c$ throughout the interval. Hence
\[
 \frac{kZ}{kc+Z}\le\frac Zc\le\chi(Z)
 \qquad(k\ge3),
\]
and
\[
 B_c(Z)=\frac{Z(c+3Z/2)}{(c+Z)(c+Z/2)}\le\frac Zc\le\chi(Z),
\]
because $(c+Z)(c+Z/2)-c(c+3Z/2)=Z^2/2\ge0$. Multiplying the first display by $kc+Z>0$ and the second by $(c+Z)(c+Z/2)>0$ shows $\mathcal F_k(Z)\ge0$ and $\mathcal B(Z)\ge0$. Since the floor and balanced tests hold for every $k$ and the pure test holds by assumption, \cref{thm:finitepoly} proves the mass bound for every finite $k\ge3$.

The block inequalities follow from \cref{lem:weightedlevel}. With $C=c/(1-\omega)$ and $W=\omega/(1-\omega)$, \cref{thm:certificate} then gives the comparison coefficient $(1-\omega)/(1+c)$ and preservation of $\K(P)$ exactly as stated.
\end{proof}

\section{Proofs for \texorpdfstring{\cref{sec:profile}}{Section 3.3}: Analytical Profiles}\label{app:parameters}
The parameters of \cref{sec:profile} are verified here by polynomial identities and sign arguments: first the zero-diagonal constants and the two-label profile, then the feasibility and regularity of the weighted transition for three or more labels, and finally the affine mass bounds obtained from tangent lines to the three mass requirements.

\subsection{Proof of \texorpdfstring{\cref{prop:zerodiag}}{the zero-diagonal proposition} (zero-diagonal parameters)}\label{app:threeconstants}
We first give a rational parameterization of the two-label test (\cref{lem:rationaltwofamily}). At one parameter value it gives part~(b) of \cref{prop:zerodiag}, and \cref{app:twoclosedproof} reuses the whole family for the two-label profile. Part~(a) is then proved directly. Throughout this subsection, $Q_\zeta(Z)$ denotes the quadratic in \eqref{eq:twoweightedtest}. Define
\begin{equation}\label{eq:rationaltwoquartic}
 \mathsf H(r)=r^4+6r^3+4r^2-16r-4,
\end{equation}
and let $r_*$ be its root in $(4/3,7/5)$, which is unique because
\[
 \mathsf H(4/3)=-68/81<0,\qquad \mathsf H(7/5)=1091/625>0,
\]
and $\mathsf H'(r)=4r^3+18r^2+8r-16$ is increasing for $r>0$, since $\mathsf H''(r)=4(3r^2+9r+2)>0$, so $\mathsf H'(r)\ge \mathsf H'(1)=14$ for $r\ge1$. Hence $\mathsf H$ is strictly increasing on $[1,\infty)$, and in particular $\mathsf H(r)\le0$ for $1\le r\le r_*$.
For $\sqrt5-1\le r\le r_*$, put
\begin{equation}\label{eq:rationaltwofamily}
\begin{aligned}
 z(r)&=\frac{2}{r^2+4r+2},
 &c(r)&=\frac{2r}{r^2+2r-2},\\
 w(r)&=\frac{2(r^2+2r-4)}{(r+4)(r^2+2r-2)},
 &m(r)&=\frac{(r+2)(r^3+5r^2+2r-6)}{(r+4)(r^2+4r-2)},\\
 A(r)&=\frac{r(r^2+4r+2)}{(r+4)(r^2+4r-2)}.
\end{aligned}
\end{equation}
\begin{lemma}[Rational two-label family]\label{lem:rationaltwofamily}
At level $\zeta=z(r)$, take $c=c(r)$, $\omega=w(r)$, and $\mu=m(r)$. This triple is admissible and satisfies
\begin{equation}\label{eq:rationaltwosquare}
 Q_{z(r)}(Z)=2(1-w(r))(Z-z(r))^2,
 \qquad\frac{1-w(r)}{1+c(r)}=A(r).
\end{equation}
At $r=r_*$, it equals $(c_2,\omega_2,1)$, and $z(r_*)=t_2$, $A(r_*)=\al$.
\end{lemma}
\begin{proof}
Throughout, $\sqrt5-1\le r\le r_*<7/5$, and we write
\[
 D=r^2+4r+2,\qquad E=r^2+2r-2,\qquad \widetilde D=r^2+4r-2,
\]
all three positive on this range (already $E>0$ for $r>\sqrt3-1$). We check the quadratic identity coefficient by coefficient. The formulas in \eqref{eq:rationaltwofamily} give $z=2/D$ and $c=2r/E$, and
\[
 1-w=\frac{(r+4)E-2(r^2+2r-4)}{(r+4)E}=\frac{rD}{(r+4)E},\qquad
 1+c=\frac{E+2r}E=\frac{\widetilde D}E,
\]
because $(r+4)E-2(r^2+2r-4)=r^3+4r^2+2r=rD$. Expanding \eqref{eq:twoweightedtest} with $\chi(Z)=mZ+(1-m)z$ gives
\[
 Q_z(Z)=2(1-w)Z^2+\bigl[c(c+1)m-c-2+w(2-c)\bigr]Z+c(c+1)(1-m)z+cw,
\]
since $-Z(c+2-2Z)=2Z^2-(c+2)Z$ and $w(1-Z)(c+2Z)=cw+w(2-c)Z-2wZ^2$. For the linear coefficient we use
\[
 c(c+1)=\frac{2r\widetilde D}{E^2},\qquad c+2=\frac{2(r^2+3r-2)}E,\qquad 2-c=\frac{2(r^2+r-2)}E,
\]
so that, after the factor $\widetilde D$ cancels against the denominator of $m$,
\[
 c(c+1)m-c-2+w(2-c)=\frac{L(r)}{(r+4)E^2},
\]
where
\[
 L(r)=2r(r+2)(r^3+5r^2+2r-6)+4(r^2+2r-4)(r^2+r-2)-2(r^2+3r-2)(r+4)E.
\]
The numerator $L(r)$ equals $-8rE$, so the linear coefficient is $-8r/((r+4)E)$, which equals $-4(1-w)z=-4\cdot\frac{rD}{(r+4)E}\cdot\frac2D$. For the constant coefficient, the formula for $1-m$ displayed below gives $c(c+1)(1-m)z=-4r\mathsf H/((r+4)E^2D)$, while $cw=4r(r^2+2r-4)/((r+4)E^2)$. Since $(r^2+2r-4)D-\mathsf H=2E$,
\[
 \begin{aligned}
 c(c+1)(1-m)z+cw&=\frac{4r\bigl[(r^2+2r-4)D-\mathsf H\bigr]}{(r+4)E^2D}=\frac{8r}{(r+4)ED}\\
 &=2\cdot\frac{rD}{(r+4)E}\cdot\frac4{D^2}=2(1-w)z^2.
 \end{aligned}
\]
The three coefficients $2(1-w)$, $-4(1-w)z$, and $2(1-w)z^2$ are exactly those of $2(1-w)(Z-z)^2$, which proves the first identity in \eqref{eq:rationaltwosquare}. Dividing the formulas for $1-w$ and $1+c$ gives $(1-w)/(1+c)=rD/((r+4)\widetilde D)=A(r)$.

For the parameter ranges, note $1<\sqrt5-1\le r<7/5<\sqrt2$. Since $E>0$, the inequality $c(r)=2r/E>1$ is equivalent to $2r>r^2+2r-2$, that is, $r^2<2$; and $c(r)<2$ is equivalent to $r<E$, that is, $(r+2)(r-1)>0$. Hence $1<c(r)<2$, and in particular $c(r)\ge c_0$. The numerator $r^2+2r-4$ of $w$ vanishes at $r=\sqrt5-1$ and increases with $r$, so $w\ge0$. Since $(r+4)E-8(r^2+2r-4)=16-(4-r)E$, the inequality $w<1/4$ is equivalent to
\[
 (4-r)(r^2+2r-2)<16;
\]
for $1<r<3/2$ its left side is less than $3\cdot(13/4)<16$, because $4-r<3$ and $E<9/4+3-2=13/4$. The cubic $r^3+5r^2+2r-6$ has derivative $3r^2+10r+2>0$ and equals $2$ at $r=1$, so it is positive for $r\ge1$, and the numerator defining $m$ is positive; hence $m>0$. Moreover,
\[
 1-m(r)=-\frac{\mathsf H(r)}{(r+4)(r^2+4r-2)}\ge0,
\]
because $\mathsf H(r)\le0$ for $1\le r\le r_*$, as shown after \eqref{eq:rationaltwoquartic}. Thus $0<m\le1$ and $0\le w<1/4$. The endpoint slack factors as
\begin{equation}
 m+(1-m)z-\frac1c
 =\frac{(r^2+2r-2)J(r)}{2r(r^2+4r-2)(r^2+4r+2)},
 \quad J(r)=(r^2+r-2)^2+r(8-5r)>0.
\end{equation}
Here $J(r)>0$ because $(r^2+r-2)^2\ge0$ and $r(8-5r)>0$ for $r<8/5$, and every other factor is positive; hence $m+(1-m)z>1/c$, which is the side condition. The quadratic \eqref{eq:twoweightedtest} equals $2(1-w)(Z-z)^2\ge0$ because $1-w=rD/((r+4)E)>0$. Together with $c\in[c_0,2]$, $w\in[0,1/4]$, and $m\in[0,1]$, these are all the conditions defining $\mathcal A_2(z(r))$ in \cref{sec:profile}, so the triple is admissible.

For the identification at $r_*$, the unnormalized parameters of \eqref{eq:normalizedparameters} follow from the formula for $1-w$:
\[
 C=\frac{c}{1-w}=\frac{2(r+4)}D,
 \qquad W=\frac{w}{1-w}=\frac{2(r^2+2r-4)}{rD}.
\]
Expansion gives
\[
 \mathsf Q_2(C)=
 \frac{2\mathsf H(r)(r^4+24r^3+164r^2+436r+396)}{D^4},
\]
\[
 \frac{C^3+6C^2+7C-16}{19}-W
 =-\frac{4(4r^3+30r^2+67r+38)\mathsf H(r)}{19rD^3}.
\]
The polynomial $\mathsf Q_2$ decreases strictly on $[1,6/5]$, from $\mathsf Q_2(1)=3$ to $\mathsf Q_2(6/5)=-634/625$: its second derivative $12C^2-36C-16$ is negative at both endpoints of $[1,6/5]$ and hence, being convex, on the whole interval, so $\mathsf Q_2'$ decreases there and $\mathsf Q_2'\le\mathsf Q_2'(1)=-16<0$. Thus $\Cs$ is the only zero of $\mathsf Q_2$ in $(1,6/5)$. The map $r\mapsto2(r+4)/D$ is decreasing, since its derivative is $-2(r^2+8r+14)/D^2<0$, and takes the values $48/41$ at $r=4/3$ and $270/239$ at $r=7/5$; because $4/3<r_*<7/5$, the value $C=2(r_*+4)/D$ lies in $(270/239,48/41)\subset(1,6/5)$. At $r=r_*$ the first identity gives $\mathsf Q_2(C)=0$, so $C=\Cs$, and the second gives $W=(\Cs^3+6\Cs^2+7\Cs-16)/19=\Ws$. Since $1+W=1/(1-w)$, we recover $c=C/(1+W)=c_2$ and $w=W/(1+W)=\omega_2$ from \eqref{eq:twozerodiag}. Also
\[
 \frac{2(1+W)-C^2+CW}{4(1+W)}-z(r)
 =\frac{\mathsf H(r)}{2(r^2+2r-2)(r^2+4r+2)}.
\]
At $r_*$ the right side vanishes, so $z(r_*)$ equals the expression defining $t_2$ in \eqref{eq:twozerodiag}. The relation $\mathsf H(r_*)=0$ also gives $1-m(r_*)=0$, so $m(r_*)=1$. Finally, $A(r_*)=(1-\omega_2)/(1+c_2)=\al$ by the second identity in \eqref{eq:rationaltwosquare}.
\end{proof}

\begin{proof}[of \cref{prop:zerodiag}]
For part~(b), use \cref{lem:rationaltwofamily} at $r=r_*$, where the triple is $(c_2,\omega_2,1)$ and $z(r_*)=t_2$. Because $\mu=1$, the target is $\chi(Z)=Z+0\cdot\zeta=Z$ for every $\zeta$, so neither the quadratic \eqref{eq:twoweightedtest} nor the side condition $\mu+(1-\mu)\zeta=1\ge1/c_2$ depends on $\zeta$. Hence, at every $\zeta\in[0,1]$ and in particular at $\zeta=0$, the quadratic equals $2(1-\omega_2)(Z-t_2)^2$ by \eqref{eq:rationaltwosquare}, and the triple lies in $\mathcal A_2(\zeta)$ because it lies in $\mathcal A_2(t_2)$. Its coefficient is $(1-\omega_2)/(1+c_2)=\al$; since $1+\Cs+\Ws=(1+c_2)/(1-\omega_2)$ by \eqref{eq:normalizedparameters}, this also equals $1/(1+\Cs+\Ws)$.

For part~(a), write $t_-=461/2500$ and $t_+=923/5000$. Substituting these rational numbers into the sextic gives
\[
 \mathsf P_3(t_-)>1/1000,\qquad \mathsf P_3(t_+)<-1/1000.
\]
It is strictly decreasing between them: on the larger interval $[9/50,1/5]\supset[t_-,t_+]$, dropping the nonpositive terms $-172t^3-81t^2$ and bounding each remaining positive monomial by its value at $t=1/5$,
\[
 \mathsf P_3'(t)=54t^5+165t^4-172t^3-81t^2+146t-34
 \le\frac{54}{3125}+\frac{165}{625}+\frac{146}{5}-34<0.
\]
A continuous, strictly decreasing function that is positive at $t_-$ and negative at $t_+$ has exactly one zero in $(t_-,t_+)$, so $t_3$ is uniquely defined by the stated rational bracket. To bound $c_3$ and $\omega_3$, note that every monomial $t^j$ increases on $[t_-,t_+]\subset(0,1)$. Replacing each monomial with a positive coefficient by its value at $t_-$, and each monomial with a negative coefficient by its value at $t_+$, gives a lower bound; the opposite choice gives an upper bound. This gives the explicit bounds
\begin{align*}
 c_3&\ge\frac{350-1417t_++810t_-^2+900t_-^3-945t_+^4-243t_+^5}{119}>
 \frac{101}{100},\\
 c_3&\le\frac{350-1417t_-+810t_+^2+900t_+^3-945t_-^4-243t_-^5}{119}<
 \frac{51}{50},\\
 \omega_3&\ge\frac{394-2324t_++1790t_-^2+1086t_-^3-2043t_+^4-486t_+^5}{357}>
 \frac2{25},\\
 \omega_3&\le\frac{394-2324t_-+1790t_+^2+1086t_+^3-2043t_-^4-486t_-^5}{357}<
 \frac9{100}.
\end{align*}
and hence the looser bounds
\begin{equation}\label{eq:zero_three_loose}
 \frac{101}{100}<c_3<\frac{51}{50},\qquad
 \frac2{25}<\omega_3<\frac9{100}.
\end{equation}
In particular, $c_3>1$, so the endpoint condition \eqref{eq:uniformside}, which reads $\mu+v=1\ge1/c_3$ for $\mu=1$ and $v=(1-\mu)\zeta=0$, holds.

To establish the double-root identity, expand both sides of \eqref{eq:newdoublefactor} in powers of $s$. From \eqref{eq:baselinecubic},
\[
 N(s)=-s^3+(3-\omega)s^2+\bigl(c^2-2+\omega(1-c)\bigr)s+\omega c,
\]
and, for any $t$,
\[
 (3-\omega-2t-s)(s-t)^2=-s^3+(3-\omega)s^2+(-6t+2\omega t+3t^2)s+(3-\omega-2t)t^2.
\]
Thus the cubic and quadratic coefficients agree for all $(c,\omega,t)$. For the remaining two coefficients, set
\[
 \begin{aligned}
 U_0(t)&=13122t^4+58077t^3-13662t^2-49887t+34475,\\
 U_1(t)&=729t^4+2538t^3-4125t^2-2542t+5320.
 \end{aligned}
\]
Direct expansion yields
\begin{align*}
 \omega(t)c(t)-(3-\omega(t)-2t)t^2
 &=\mathsf P_3(t)U_0(t)/42483,\\
 c(t)^2+\omega(t)(1-c(t))-2+6t-2\omega(t)t-3t^2
 &=9\mathsf P_3(t)U_1(t)/42483.
\end{align*}
The left sides are exactly the differences between the constant coefficients, and between the linear coefficients, of $N(s)$ and $(3-\omega-2t-s)(s-t)^2$ when $(c,\omega)=(c(t),\omega(t))$. At $t=t_3$ the factor $\mathsf P_3(t_3)=0$ makes both differences vanish, so \eqref{eq:newdoublefactor} holds. For $s\in[0,1]$, its first factor satisfies
\[
 3-\omega_3-2t_3-s\ge3-\frac9{100}-2\cdot\frac{923}{5000}-1>0
\]
by \eqref{eq:zero_three_loose} and $t_3<923/5000$, and $(s-t_3)^2\ge0$, so $N\ge0$ on $[0,1]$. For $\mu=1$ we have $v=0$ and $\chi(s)=s$ at every $\zeta$, and then the pure-label polynomial \eqref{eq:purepoly} equals $N$; hence \eqref{eq:uniformcubic} holds. With the endpoint condition, this verifies the hypotheses of \cref{thm:uniformcubic} at every diagonal level. The argument of \cref{app:uniformcubic} shows that \eqref{eq:uniformside} also gives $\mathcal F_k\ge0$ and $\mathcal B\ge0$, so \eqref{eq:finitecriterion} holds for every finite $k\ge3$. Since moreover $c_3\in(101/100,51/50)\subset[c_0,2]$, $\omega_3\in(2/25,9/100)\subset[0,1/4]$, and $\mu=1$, the triple lies in $\mathcal A_k(\zeta)$ for every $k\ge3$ and every $\zeta\in[0,1]$.
\end{proof}

\subsection{Proof of \texorpdfstring{\cref{thm:twoclosed}}{the two-label profile theorem}: the two-label profile}\label{app:twoclosedproof}
The rational family of \cref{lem:rationaltwofamily} supplies the plateau and the transition. The proof checks the mass bound on all four pieces and then matches values and first derivatives at the junctions.

\begin{proof}[of \cref{thm:twoclosed}]
For fixed $c>0$ and $\zeta\in(0,1)$, the target $\chi(Z)=\mu Z+(1-\mu)\zeta$ has $\chi(\zeta)=\zeta$ and $\chi'=\mu$. Evaluating \eqref{eq:twoweightedtest} and its derivative in $Z$ at $Z=\zeta$ therefore gives
\[
 Q_\zeta(\zeta)=c(c+1)\zeta-\zeta(c+2-2\zeta)+\omega(1-\zeta)(c+2\zeta)=\zeta[c^2-2(1-\zeta)]+\omega(1-\zeta)(c+2\zeta),
\]
\[
 Q_\zeta'(\zeta)=c(c+1)\mu-(c+2-4\zeta)+\omega\bigl(2(1-\zeta)-(c+2\zeta)\bigr)=c(c+1)\mu-c-2+4\zeta+\omega(2-c-4\zeta).
\]
The first expression is affine in $\omega$ and the second is affine in $\mu$. Setting the first to zero and solving for $\omega$, then setting the second to zero and solving for $\mu$, gives
\[
 \omega(c,\zeta)=\frac{\zeta\{2(1-\zeta)-c^2\}}{(1-\zeta)(c+2\zeta)},\qquad
 \mu(c,\zeta)=\frac{c+2-4\zeta-\omega(c,\zeta)(2-4\zeta-c)}{c(c+1)}.
\]
With these choices the quadratic vanishes together with its derivative at $Z=\zeta$, so $(Z-\zeta)^2$ divides it; its leading coefficient is $2(1-\omega)$, the coefficient of $Z^2$ in \eqref{eq:twoweightedtest}. Hence
\begin{equation}\label{eq:twoTransitionSquare}
 Q_\zeta(Z)=2(1-\omega(c,\zeta))(Z-\zeta)^2.
\end{equation}
Since
\[
 1-\omega(c,\zeta)=\frac{(1-\zeta)(c+2\zeta)-2\zeta(1-\zeta)+\zeta c^2}{(1-\zeta)(c+2\zeta)}=\frac{c(1-\zeta+\zeta c)}{(1-\zeta)(c+2\zeta)},
\]
the coefficient is $\alpha(c,\zeta)=(1-\omega(c,\zeta))/(1+c)=c(1-\zeta+\zeta c)/((1-\zeta)(1+c)(c+2\zeta))$. Differentiating this rational function, and using the formula for $\mu(c,\zeta)$ for the $\zeta$-derivative, gives
\begin{equation}\label{eq:twopartials}
 \partial_c\alpha=\frac{g_\zeta(c)}{(c+1)^2(c+2\zeta)^2(1-\zeta)},\qquad
 \partial_\zeta\alpha=\frac{c(1-\mu)}{(c+2\zeta)(1-\zeta)},
\end{equation}
where $g_\zeta(c)=(2\zeta^2+2\zeta-1)c^2+4\zeta^2c+2\zeta(1-\zeta)$.

\emph{Plateau $0\le\zeta\le t_2$.} Use $(c_2,\omega_2,1)$, which lies in $\mathcal A_2(\zeta)$ for every $\zeta$ by \cref{prop:zerodiag}(b) and certifies $\al$.

\emph{Transition $t_2\le\zeta\le\zeta_{\mathrm a}$.} Set
\[
 r=\sqrt{\frac{2(1+\zeta)}\zeta}-2.
\]
Since $D=r^2+4r+2=(r+2)^2-2=2/\zeta$, we get $z(r)=2/D=\zeta$. The map $r\mapsto z(r)$ is a decreasing bijection from $(0,\infty)$ onto $(0,1)$, with $z(r_*)=t_2$ by \cref{lem:rationaltwofamily} and $z(\sqrt5-1)=2/((\sqrt5+1)^2-2)=1/(2+\sqrt5)=\sqrt5-2=\zeta_{\mathrm a}$. Hence $\zeta\in[t_2,\zeta_{\mathrm a}]$ corresponds to $r\in[\sqrt5-1,r_*]$, and the triple of \cref{lem:rationaltwofamily} is admissible at level $\zeta$; write $\tilde c_\zeta=c(r)$, $\tilde\omega_\zeta=w(r)$, and $\tilde\mu_\zeta=m(r)$. Because $Q_\zeta$ has a double root at $Z=\zeta$ for this triple, $\tilde\omega_\zeta=\omega(\tilde c_\zeta,\zeta)$ and $\tilde\mu_\zeta=\mu(\tilde c_\zeta,\zeta)$. Moreover $1+\zeta=(r+2)^2/D$, so $2\zeta(1+\zeta)=4(r+2)^2/D^2$ and $\sqrt{2\zeta(1+\zeta)}=2(r+2)/D$; substituting this and $\zeta=2/D$ into \eqref{eq:twotransitionc} gives $2r/E=c(r)=\tilde c_\zeta$, and the same substitution gives $g_\zeta(c(r))=0$. The quadratic $g_\zeta$ has leading coefficient $2\zeta^2+2\zeta-1<0$, because $\zeta\le\zeta_{\mathrm a}=\sqrt5-2<(\sqrt3-1)/2$, and positive constant term $2\zeta(1-\zeta)$, so the product of its roots is negative and $\tilde c_\zeta$ is its unique positive root. The coefficient $\alpha(\tilde c_\zeta,\zeta)$ is the second piece of \eqref{eq:twoclosedprofile}, by the formula for $\alpha(c,\zeta)$ above. At $r=r_*$ the triple is the plateau triple. At $r=\sqrt5-1$, the numerator $r^2+2r-4$ of $w$ vanishes, so $\omega=0$, and $c=2r/(r^2+2r-2)=2r/2=\sqrt5-1$, at the level $\zeta=\sqrt5-2$.

\emph{Pure tangent $\zeta_{\mathrm a}\le\zeta\le \zeta_{\mathrm b}=1-1/\sqrt2$.} Put
\[
 \tilde c_P(\zeta)=\sqrt{2(1-\zeta)},\qquad
 \tilde\omega_P=0,\qquad
 \tilde\mu_P(\zeta)=\frac{\tilde c_P(\zeta)+2-4\zeta}{\tilde c_P(\zeta)(\tilde c_P(\zeta)+1)}.
\]
At $c=\tilde c_P(\zeta)$ the factor $2(1-\zeta)-c^2$ in $\omega(c,\zeta)$ vanishes, so $\omega(c,\zeta)=0$, and $\mu(c,\zeta)$ reduces to $\tilde\mu_P(\zeta)$. Hence \eqref{eq:twoTransitionSquare} becomes
\begin{equation}\label{eq:twoPureSquare}
 Q_\zeta(Z)=2(Z-\zeta)^2.
\end{equation}
Write $x=\tilde c_P(\zeta)$, so $\zeta=1-x^2/2$; since $x^2=2(1-\zeta)$ decreases in $\zeta$, with $x^2=2(3-\sqrt5)=(\sqrt5-1)^2$ at $\zeta_{\mathrm a}$ and $x^2=\sqrt2$ at $\zeta_{\mathrm b}$, we have $x\in[2^{1/4},\sqrt5-1]$ on this interval. Substituting $4\zeta=4-2x^2$ gives $\tilde\mu_P=(2x^2+x-2)/(x(x+1))$. This lies in $[0,1]$: the numerator is at least $1>0$ because $x\ge1$, and it is at most the denominator $x^2+x$ exactly when $x^2\le2$. Also $c_0<1\le x\le\sqrt5-1<2$ and $\omega=0$, and
\[
 \tilde\mu_P+(1-\tilde\mu_P)\zeta-\frac1{x}=\frac{x^4-2}{2x(x+1)},
\]
which is nonnegative exactly when $x^4\ge2$, that is, $\zeta\le \zeta_{\mathrm b}$. Thus the side condition holds precisely on the stated interval, with equality at $\zeta_{\mathrm b}$, and \eqref{eq:twoPureSquare} shows that the quadratic is nonnegative; the triple is admissible. Its coefficient is $(1-0)/(1+\tilde c_P)$, the third line of \eqref{eq:twoclosedprofile}.

\emph{Endpoint chord $\zeta_{\mathrm b}\le\zeta<1$.} Put $\rho_2=\sqrt2-1$ and
\[
 c_E(\zeta)=\sqrt{2\rho_2+\rho_2^2/\zeta},\qquad
 \omega_E=0,\qquad
 \mu_E(\zeta)=\frac{c_E(\zeta)^{-1}-\zeta}{1-\zeta}.
\]
By construction $\mu_E+(1-\mu_E)\zeta=\zeta+\mu_E(1-\zeta)=1/c_E$: the affine mass bound joins $(\zeta,\zeta)$ to $(1,1/c_E)$, and the side condition holds with equality. To verify the quadratic, abbreviate $c=c_E(\zeta)$ and $\mu=\mu_E(\zeta)$, so that $c(c+1)\mu=(c+1)(1-c\zeta)/(1-\zeta)$ and $1-\mu=(1-1/c)/(1-\zeta)$. With $\omega=0$, the linear and constant coefficients of \eqref{eq:twoweightedtest} are
\[
 c(c+1)\mu-c-2
 =\frac{(c+1)(1-c\zeta)-(c+2)(1-\zeta)}{1-\zeta}
 =\frac{-1+\zeta(2-c^2)}{1-\zeta}
 =-2(1-\rho_2)=-4\zeta_{\mathrm b},
\]
\[
 c(c+1)(1-\mu)\zeta
 =\frac{(c^2-1)\zeta}{1-\zeta}
 =\rho_2^2=2\zeta_{\mathrm b}^2.
\]
For the third equality in the first line, substitute $c^2\zeta=2\rho_2\zeta+\rho_2^2$ and $1+\rho_2^2=2-2\rho_2$ (from $\rho_2^2+2\rho_2=1$): the numerator becomes $-(1+\rho_2^2)+2\zeta(1-\rho_2)=-2(1-\rho_2)(1-\zeta)$. For the second line, $(c^2-1)\zeta=\rho_2^2-\zeta(1-2\rho_2)=\rho_2^2(1-\zeta)$, using $1-2\rho_2=\rho_2^2$. Finally $\zeta_{\mathrm b}=1-1/\sqrt2=\rho_2/\sqrt2$, so $4\zeta_{\mathrm b}=4-2\sqrt2=2(1-\rho_2)$ and $2\zeta_{\mathrm b}^2=\rho_2^2$. The quadratic coefficient is $2$, giving
\begin{equation}\label{eq:twoChordSquare}
 Q_\zeta(Z)=2(Z-\zeta_{\mathrm b})^2.
\end{equation}
Since $\zeta\le1$, $c_E^2\ge2\rho_2+\rho_2^2=(1+\rho_2)^2-1=1$, so $c_E\ge1$, hence $1/c_E\le1$ and $\mu_E\le1$. Since $c_E$ decreases in $\zeta$, also $c_E\le c_E(\zeta_{\mathrm b})=2^{1/4}<2$, where $c_E(\zeta_{\mathrm b})^2=2\rho_2+\sqrt2\rho_2=\rho_2(2+\sqrt2)=\sqrt2$. Moreover $c_E^2\zeta^2=2\rho_2\zeta^2+\rho_2^2\zeta\le2\rho_2+\rho_2^2=1$, so $1/c_E\ge\zeta$ and $\mu_E\ge0$. Thus $c_E\in[1,2]\subset[c_0,2]$, $\mu_E\in[0,1]$, the quadratic \eqref{eq:twoChordSquare} is nonnegative, and the side condition holds, so the triple is admissible for $\zeta_{\mathrm b}\le\zeta<1$. At $\zeta=1$, take the limiting parameters $c=1$, $\omega=0$, and $\mu=1-\rho_2^2/2=\sqrt2-1/2$; then $\chi(Z)=\mu Z+1-\mu$, and $Q_1(Z)=2Z^2-(3-2\mu)Z+2-2\mu=2(Z-\zeta_{\mathrm b})^2$ because $3-2\mu=1+\rho_2^2=2-2\rho_2=4\zeta_{\mathrm b}$ and $2-2\mu=\rho_2^2=2\zeta_{\mathrm b}^2$; the side condition reads $\mu+(1-\mu)=1\ge1/c$. The coefficient is $1/(1+c_E)$, the fourth line of \eqref{eq:twoclosedprofile}, and equals $1/(1+\sqrt{2\rho_2+\rho_2^2})=1/2$ at $\zeta=1$.

\emph{Junctions.} The values agree at the junctions: at $t_2$ because the transition triple at $r=r_*$ is the plateau triple; at $\zeta_{\mathrm a}$ because the transition triple at $r=\sqrt5-1$ has $\omega=0$ and $c=\sqrt5-1=\tilde c_P(\zeta_{\mathrm a})$, so both formulas equal $1/(1+(\sqrt5-1))=1/\sqrt5$; and at $\zeta_{\mathrm b}=\rho_2/\sqrt2$ one has $\tilde c_P(\zeta_{\mathrm b})^2=\sqrt2$ and $c_E(\zeta_{\mathrm b})^2=\sqrt2$, so both coefficients equal $1/(1+2^{1/4})$. For the derivatives, the transition branch is $\zeta\mapsto\alpha(\tilde c_\zeta,\zeta)$. Its derivative is $\partial_c\alpha\cdot\tilde c_\zeta'+\partial_\zeta\alpha$, and $\partial_c\alpha(\tilde c_\zeta,\zeta)=0$ by \eqref{eq:twopartials} because $g_\zeta(\tilde c_\zeta)=0$. So the derivative of the branch is $\partial_\zeta\alpha=\tilde c_\zeta(1-\tilde\mu_\zeta)/((\tilde c_\zeta+2\zeta)(1-\zeta))$. At $t_2$, $\tilde\mu_{t_2}=m(r_*)=1$ gives $\acf'(t_2^+)=0$, and the plateau is constant, so $\acf'(t_2^-)=0$. At $\zeta_{\mathrm a}$, the transition triple is the pure triple, so $\tilde\mu_\zeta=\tilde\mu_P$ and, with $x=\sqrt5-1$, $1-\tilde\mu_P=(x^2+x-2x^2-x+2)/(x(x+1))=(2-x^2)/(x(x+1))=2\zeta/(x(x+1))$; hence the transition derivative equals $2\zeta/((x+1)(x+2\zeta)(1-\zeta))$. The pure branch is $(1+x)^{-1}$ with $x=\sqrt{2(1-\zeta)}$ and $dx/d\zeta=-1/x$, so its derivative is $1/(x(1+x)^2)$. Using $1-\zeta=x^2/2$, the two derivatives agree exactly when $2\zeta(1+x)=x(x+2\zeta)/2$, that is, $x^2-2\zeta x-4\zeta=0$; at $\zeta=\zeta_{\mathrm a}$ and $x=\sqrt5-1$ this holds because $2\zeta(x+2)=2(\sqrt5-2)(\sqrt5+1)=6-2\sqrt5=x^2$, and both derivatives equal $1/((\sqrt5-1)\cdot5)=(1+\sqrt5)/20$. At $\zeta_{\mathrm b}$, the pure derivative is $1/(x(1+x)^2)$ with $x=2^{1/4}$. On the chord, $c_E^2=2\rho_2+\rho_2^2/\zeta$ gives $c_E'=-\rho_2^2/(2c_E\zeta^2)$, so the derivative of $(1+c_E)^{-1}$ is $\rho_2^2/(2c_E\zeta^2(1+c_E)^2)$; at $\zeta_{\mathrm b}$ it equals $1/(x(1+x)^2)$ because $c_E=2^{1/4}=x$ and $\rho_2^2/(2\zeta_{\mathrm b}^2)=1$. Hence
\[
 \acf'(t_2^-)=\acf'(t_2^+)=0,\qquad
 \acf'(\zeta_{\mathrm a}^-)=\acf'(\zeta_{\mathrm a}^+)=\frac{1+\sqrt5}{20},
 \qquad
 \acf'(\zeta_{\mathrm b}^-)=\acf'(\zeta_{\mathrm b}^+),
\]
and, since each piece is continuously differentiable on its own interval, the profile is continuously differentiable on $[0,1]$. It is nondecreasing because the values match at the junctions and each piece is nondecreasing: the plateau is constant; on the transition, the derivative computed above is $\tilde c_\zeta(1-\tilde\mu_\zeta)/((\tilde c_\zeta+2\zeta)(1-\zeta))\ge0$ because $\tilde\mu_\zeta=m(r)\le1$ by \cref{lem:rationaltwofamily}; on the third and fourth pieces, $\sqrt{2(1-\zeta)}$ and $\sqrt{2\rho_2+\rho_2^2/\zeta}$ decrease in $\zeta$, so $1/(1+c)$ increases. The displayed decimal values follow by evaluating the fourth line of \eqref{eq:twoclosedprofile} at $\zeta=0.3,0.5,0.8$, all of which exceed $\zeta_{\mathrm b}=0.2928\ldots$.

\end{proof}

\subsection{Proof of \texorpdfstring{\cref{thm:kprofile}(a),(d)}{the weighted and endpoint profile statements}: weighted regimes and unrestricted level}\label{app:kprofileproof}
The weighted transition is defined by one quadratic root and two rational parameter formulas. Its endpoint inequality follows from an explicit polynomial-division identity and elementary sign bounds on a rational interval. \Cref{app:transition} proves part~(a), and \cref{app:exactendpointdetail} proves part~(d), the unrestricted level; the unweighted tangent cases, parts~(b)--(c), are proved in \cref{app:closedprofiles}.

\subsubsection{Proof of \texorpdfstring{\cref{thm:kprofile}(a)}{the weighted-transition statement}: plateau and weighted transition}\label{app:transition}
Put $I=[9/50,43/200]$ and define
\begin{equation}
 f_z(c)=(z^2+2z-1)c^2+2z(z^2-z+1)c+z(1-z)(2-z^2).
\end{equation}
Its coefficients are $a_z$, $b_z$, and $h_z$ of \eqref{eq:transitionpolynomials}, because $z(1-z)(2-z^2)=z^4-z^3-2z^2+2z$; so the positive root of $f_z$ is the $c_z$ of \cref{def:curves}. On $I$ the leading coefficient $z^2+2z-1$ is negative, since $z<\sqrt2-1$, and the constant term $z(1-z)(2-z^2)$ is positive. The product of the two roots is therefore negative, so $f_z$ has exactly one positive root $c_z$ and one negative root $c_z^-$. Writing $f_z(c)=(z^2+2z-1)(c-c_z)(c-c_z^-)$ gives $\partial_c f_z(c_z)=(z^2+2z-1)(c_z-c_z^-)<0$; also $f_z>0$ on $[0,c_z)$ and $f_z<0$ on $(c_z,\infty)$. Regard the formulas of \eqref{eq:transitionparameters} as functions of arbitrary $(c,z)$:
\begin{equation}\label{eq:analytictransitionfunctions}
 \begin{aligned}
 w(c,z)&=\frac{z[(1-z)(2-z)-c^2]}{(c+z)(1-z)},\\
 m(c,z)&=\frac{c+2-5z+2z^2-w(c,z)(1-c-2z)}{c(c+1)+z-z^2}.
 \end{aligned}
\end{equation}
At $c=c_z$ these give $\omega_z$ and $\mu_z$. The denominator of $m$ factors as $c(c+1)+z-z^2=(c+z)(c+1-z)$, which is positive for $c>0$ and $0<z<1$.
Let $G(z)=4z^5-16z^4+12z^3+15z^2-13z+2$ be the quintic in \eqref{eq:junctionroot}. Then
\[
 G(1/5)=224/3125>0,\qquad
 G(43/200)=-1177192357/80000000000<0.
\]
and $G$ is strictly decreasing on $[1/5,43/200]$: there $G'(z)=20z^4-64z^3+36z^2+30z-13$, and dropping the nonpositive term $-64z^3$ and bounding the remaining positive monomials by their values at $43/200$ gives
\[
 G'(z)\le20(43/200)^4+36(43/200)^2+30(43/200)-13<0.
\]
Hence $\zeta_0$ is the unique root of $G$ in this bracket, and $t_3<\zeta_0<43/200<1/4$.

\begin{lemma}[Analytical endpoint inequality on the transition]
\label{lem:wa_three_sides}
For $t_3\le z\le \zeta_0$, the parameters \eqref{eq:analytictransitionfunctions} at $c=c_z$ satisfy
\begin{equation}
 1<c<6/5,\quad 0\le\omega<1/4,\quad0<\mu\le1,
 \qquad
 \mu+(1-\mu)z>1/c.
 \label{eq:wa_three_sides}
\end{equation}
Furthermore $c'(z)>0$, $\mu'(z)<0$, and $\omega'(z)<0$.
\end{lemma}

\begin{proof}
\emph{Endpoints.} For the initial endpoint, let
\[
 V(t)=6561t^6+40095t^5+30456t^4-82944t^3-9405t^2+76818t-30625.
\]
Substitution of the polynomial $c(t)$ from \eqref{eq:threeexact} gives
\[
 f_t(c(t))=\frac{\mathsf P_3(t)V(t)}{119^2}.
\]
Hence $f_{t_3}(c_3)=0$, and since $c_3>0$ is a positive root, $c_{t_3}=c_3$. The double-root identity \eqref{eq:newdoublefactor} also gives $N(t_3)=N'(t_3)=0$. From \eqref{eq:baselinecubic}, $N(t)=t[c^2-(1-t)(2-t)]+\omega(1-t)(c+t)$ and $N'(t)=c^2-2+6t-3t^2+\omega(1-c-2t)$, so these two equations read
\[
 t_3[c_3^2-(1-t_3)(2-t_3)]+\omega_3(1-t_3)(c_3+t_3)=0,
\]
\[
 c_3^2-2+6t_3-3t_3^2+\omega_3(1-c_3-2t_3)=0.
\]
Solving the first for $\omega_3$ gives $\omega_3=t_3[(1-t_3)(2-t_3)-c_3^2]/((1-t_3)(c_3+t_3))=w(c_3,t_3)$. For the second, the identity
\[
 \bigl(c(c+1)+z-z^2\bigr)-\bigl(c+2-5z+2z^2-w(1-c-2z)\bigr)=c^2-2+6z-3z^2+w(1-c-2z)
\]
shows that the left side of the second equation is $(c_3+t_3)(c_3+1-t_3)\bigl(1-m(c_3,t_3)\bigr)$; since the first factor is positive, $m(c_3,t_3)=1$. The initial parameters satisfy the rational bounds \eqref{eq:zero_three_loose}.
At the other endpoint, put $\widehat c_z=z(1-z)(3-2z)/(1-3z)$, which is positive for $0<z<1/3$. The identities
\[
 \widehat c_z^2-(1-z)(2-z)=\frac{(z-1)G(z)}{(3z-1)^2},\qquad
 f_z(\widehat c_z)=\frac{z(z-1)^2G(z)}{(3z-1)^2}
\]
together with $G(\zeta_0)=0$ show that $\widehat c_{\zeta_0}>0$ is a root of $f_{\zeta_0}$, hence $c_{\zeta_0}=\widehat c_{\zeta_0}$, and that the numerator of $w(c_{\zeta_0},\zeta_0)$ vanishes. Thus
\[
 c_{\zeta_0}^2=(1-\zeta_0)(2-\zeta_0),\qquad w(c_{\zeta_0},\zeta_0)=0.
\]

\emph{Monotonicity.} The expression
\[
 f_z(1)=z^4+z^3-3z^2+6z-1
\]
has derivative $4z^3+3z^2-6z+6\ge6-6/4=9/2$ for $0\le z\le1/4$, so it increases on $[t_3,\zeta_0]\subset[0,1/4]$. It is positive at $z=t_3$, because $f_{t_3}>0$ on $[0,c_{t_3})$ and $1<c_3=c_{t_3}$. Hence $f_z(1)>0$ for $t_3\le z\le\zeta_0$, which places $1$ to the left of the positive root: $c_z>1$ throughout the transition. For $c\ge1$ and $0\le z\le1/4$,
\[
 \partial_zf_z(c)
 =(2z+2)c^2+2(3z^2-2z+1)c+4z^3-3z^2-4z+2
 \ge61/16.
\]
Indeed, the first term is at least $2$; the coefficient $3z^2-2z+1$ decreases on $[0,1/4]$ to $11/16$, so the second term is at least $11/8\ge1$; and $4z^3-3z^2-4z+2$, whose derivative $2(6z^2-3z-2)$ is negative on $[0,1/4]$, is at least its value $7/8\ge13/16$ at $z=1/4$. Since $\partial_c f_z(c_z)<0$, the implicit function theorem shows that $z\mapsto c_z$ is differentiable, with
\[
 c'(z)=-\frac{\partial_z f_z(c_z)}{\partial_c f_z(c_z)}>0.
\]
Since $(1-z)(2-z)$ decreases for $z<3/2$ and $\zeta_0>1/5$, the endpoint formula gives $c_{\zeta_0}^2<(4/5)(9/5)=36/25$, that is, $c_{\zeta_0}<6/5$. As $c_z$ increases from $c_3>1$ to $c_{\zeta_0}$, we get $1<c_z<6/5$ on the transition.

For $c\ge1$ and $0\le z\le1/4$, both partial derivatives of $m$ are strictly negative. To see this, write
\[
 m_c=\frac{N_c}{(c+z)^3(1-z)(c+1-z)^2},\qquad
 m_z=\frac{cN_z}{(c+z)^3(1-z)^2(c+1-z)^2},
\]
where
\begin{align*}
 N_c={}&c^3(z^2-z-1)+c^2(3z^3-9z^2+9z-4)\\
 &+c(3z^4-9z^2+8z-2)+z^5-4z^4+7z^3-6z^2+2z,\\
 N_z={}&-c^4+c^3(z^2-3z)+c^2(3z^3-6z^2+4z-2)\\
 &+c(3z^4+2z^3-17z^2+18z-6)
   +z^5-6z^4+17z^3-24z^2+16z-4.
\end{align*}
The three denominators are positive for $c\ge1$ and $0\le z\le1/4$, so the signs of $m_c$ and $m_z$ are those of $N_c$ and $N_z$. For $0\le z\le1/4$, the first two coefficients of $N_c$ obey
\[
 z^2-z-1\le-1,\qquad
 3z^3-9z^2+9z-4\le3z^3+9z-4\le3/64+9/4-4=-109/64,
\]
where the first uses $z^2\le z$. Its coefficient of $c$ is at most $8z-2\le0$, since $3z^4-9z^2=3z^2(z^2-3)\le0$. Its constant term is at most $z^5+7z^3+2z\le1/1024+7/64+1/2=625/1024$, after dropping the negative terms. Consequently, for $c\ge1$,
\[
 N_c\le-c^3-(109/64)c^2+625/1024\le-1-109/64+625/1024<0.
\]
For $N_z$, the coefficient $z(z-3)$ of $c^3$ is nonpositive. Dropping the negative terms in each remaining coefficient and evaluating the positive ones at $z=1/4$ gives
\[
 3z^3-6z^2+4z-2\le3/64+1-2=-61/64,
\]
\[
 3z^4+2z^3-17z^2+18z-6
 \le3/256+2/64+18/4-6=-373/256,
\]
\[
 z^5-6z^4+17z^3-24z^2+16z-4
 \le1/1024+17/64+4-4=273/1024.
\]
It follows that $N_z\le-c^4-(61/64)c^2-(373/256)c+273/1024\le-1+273/1024<0$ for $c\ge1$. Since $\mu_z=m(c_z,z)$ with $c_z\ge1$, the chain rule gives $\mu'=m_cc'+m_z<0$, because $m_c<0$, $c'>0$, and $m_z<0$. Its initial value is $\mu_{t_3}=1$, and its endpoint value is positive because, at $\omega=0$, the numerator in \eqref{eq:analytictransitionfunctions} is $c+(2-z)(1-2z)>0$ and the denominator is positive. This proves $0<\mu\le1$.

Finally
\[
 w_c=-\frac{z[c^2+2cz+(1-z)(2-z)]}{(c+z)^2(1-z)}<0,
 \qquad w_z=-\frac{(1-m)(c+1-z)}{1-z}\le0.
\]
In $w_c$ every term in the bracket is positive, and in $w_z$ the factor $1-m$ is nonnegative at $c=c_z$ because $\mu_z\le1$. Hence $\omega'=w_cc'+w_z<0$, so $\omega$ decreases from $\omega_3<9/100$ at $t_3$ to $0$ at $\zeta_0$, and $0\le\omega<1/4$.

\emph{Endpoint inequality.} Put
\[
 S_1=m+(1-m)z-1/c,\qquad
 N_1=c(c+z)^2(c+1-z)S_1.
\]
The multiplier $c(c+z)^2(c+1-z)$ is positive, so $S_1$ and $N_1$ have the same sign. After substituting \eqref{eq:analytictransitionfunctions} and collecting powers of $c$, the denominator cancels and gives
\begin{align*}
 N_1={}&z(1-z)c^3+(-2z^3+5z^2-5z+1)c^2\\
 &+(-z^4+z^3+z^2-2z)c+z^3-z^2.
\end{align*}
Put $a=z^2+2z-1$ and define
\[
 P(z)=5-28z+44z^2-28z^3+3z^4,
\]
\[
 T(z)=10-105z+343z^2-349z^3+98z^4+32z^5-15z^6,
\]
\[
 Q(c,z)=cz(1-z)a-3z^4+11z^3-16z^2+7z-1.
\]
Polynomial division gives the exact identity
\begin{equation}\label{eq:wa_directdivision}
 a^2N_1(c,z)=Q(c,z)f_z(c)+\frac{z(1-z)}5
 \bigl[T(z)+zP(z)(6-5c)\bigr].
\end{equation}
The first term vanishes at $c=c_z$. Write $z_+=43/200$. On $I=[9/50,z_+]$,
\[
 P'(z)\le-28+88z_++12z_+^3<0,
 \qquad P(z_+)=\frac{1187257203}{1600000000}>0.
\]
Furthermore
\[
 T''(z)\ge686-2094z_+-450z_+^4>0,
 \qquad T'(z_+)=-\frac{54758336787}{32000000000}<0,
\]
\[
 T(z_+)=\frac{439077901893}{12800000000000}>0.
\]
Thus $P$ and $T$ are positive throughout $I$: $P$ is decreasing, and $T'$ is increasing but still negative at the right endpoint, so $T$ is decreasing as well. On the transition $1<c_z<6/5$, so $6-5c_z>0$, while $z(1-z)>0$ and $a^2>0$. Equation~\eqref{eq:wa_directdivision} therefore proves $N_1(c_z,z)>0$, hence $S_1>0$, directly on the whole transition.

\end{proof}

\begin{proof}[of \cref{thm:kprofile}(a)]
The plateau follows from \cref{prop:zerodiag}(a), whose triple is admissible at every diagonal level. On the transition, keep $z$ fixed, put $v=(1-\mu)z$, so that $\chi(s)=\mu s+v$ satisfies $\chi(z)=z$ and $\chi'=\mu$, and regard $\mathcal R$ in \eqref{eq:purepoly} as a polynomial in $s$. Evaluating it and its derivative at $s=z$, and using $(c+z)(c+1-z)-(c+2-2z)=c^2+c+z-z^2-c-2+2z=c^2-(1-z)(2-z)$, gives
\[
 \mathcal R(z)=z[c^2-(1-z)(2-z)]+\omega(1-z)(c+z),
\]
\[
 \mathcal R'(z)=[c(c+1)+z-z^2]\mu-c-2+5z-2z^2
                 +\omega(1-c-2z).
\]
The first expression vanishes exactly when $\omega=w(c,z)$, and, given that, the second vanishes exactly when $\mu=m(c,z)$; so for the transition parameters $(s-z)^2$ divides $\mathcal R(s)$. Expanding \eqref{eq:purepoly}, the cubic coefficient of $\mathcal R$ is $-\mu$ and the quadratic coefficient is $2-\omega+\mu-v$. Writing $\mathcal R(s)=(a-\mu s)(s-z)^2$, the quadratic coefficient of the right side is $a+2z\mu$, so $a=2-\omega+\mu-v-2z\mu=2-\omega-z+(1-z)\mu$:
\begin{equation}
 \mathcal R(s)=[2-\omega-z+(1-z)\mu-\mu s](s-z)^2.
\end{equation}
For $s\in[0,1]$ the first factor is at least its value $2-\omega-z-z\mu\ge2-3/4>0$ at $s=1$, since $z<1/4$, $\omega<1/4$, and $0<\mu\le1$; hence $\mathcal R\ge0$ on $[0,1]$, which is \eqref{eq:uniformcubic}. The last inequality of \eqref{eq:wa_three_sides} is \eqref{eq:uniformside} with $v=(1-\mu)z$. Together with the ranges $1<c<6/5$, $0\le\omega<1/4$, and $0<\mu\le1$ from \cref{lem:wa_three_sides}, this verifies every hypothesis of \cref{thm:uniformcubic}; by \cref{app:uniformcubic} the floor and balanced tests then also hold, so the construction is admissible for every finite $k\ge3$.

The endpoint identities in the proof of \cref{lem:wa_three_sides} show that the transition starts at $(c_3,\omega_3,1)$. At $\zeta_0$ they give $c_{\zeta_0}=\sqrt{(1-\zeta_0)(2-\zeta_0)}=c_P(\zeta_0)$ and $\omega=0$, and then $\mu=m(c_P(\zeta_0),\zeta_0)$ with $w=0$ is exactly $\mu_P(\zeta_0)$ in \eqref{eq:pureclosed}; so the transition ends at $(c_P(\zeta_0),0,\mu_P(\zeta_0))$. To prove the derivative formula, write $\alpha(c,z)=(1-w(c,z))/(1+c)$. Direct differentiation gives
\[
 \partial_c\alpha(c,z)=\frac{f_z(c)}{(c+1)^2(c+z)^2(1-z)}.
\]
At $c=c_z$ this vanishes, because $f_z(c_z)=0$. Since $\partial_z\alpha=-w_z/(1+c)$, the formula for $w_z$ established above gives
\[
 \partial_z\alpha(c,z)
 =\frac{(c+1-z)(1-m(c,z))}{(c+1)(1-z)}.
\]
The chain rule $\alpha_{\mathrm{tr}}'(z)=\partial_c\alpha(c_z,z)c_z'+\partial_z\alpha(c_z,z)$ therefore reduces to
\[
 \alpha_{\mathrm{tr}}'(z)=\frac{(c_z+1-z)(1-\mu_z)}{(c_z+1)(1-z)}.
\]
It is zero at $t_3$ and positive afterward because $\mu_{t_3}=1$ and $\mu_z'<0$ give $1-\mu_z>0$ for $z>t_3$, while the other factors are positive. Differentiating the product once more at $t_3$, the terms containing the factor $1-\mu_{t_3}=0$ vanish, and
\[
 \alpha_{\mathrm{tr}}''(t_3)
 =-\frac{c_3+1-t_3}{(c_3+1)(1-t_3)}\mu'(t_3)>0.
\]
At $\zeta_0$, the pure branch is $z\mapsto\alpha(c_P(z),z)$, because $w(c_P(z),z)=0$ makes $\alpha(c_P(z),z)=1/(1+c_P(z))=\alpha_{\mathrm{pu}}(z)$. Its derivative is $\partial_c\alpha(c_P(z),z)c_P'(z)+\partial_z\alpha(c_P(z),z)$. At $z=\zeta_0$, $c_P(\zeta_0)=c_{\zeta_0}$, so $\partial_c\alpha=0$ there and both branches have derivative $\partial_z\alpha(c_{\zeta_0},\zeta_0)$; their values agree because they are evaluated at the same triple.

\end{proof}

\subsubsection{Proof of \texorpdfstring{\cref{thm:kprofile}(d)}{the unrestricted endpoint statement}: unrestricted level}\label{app:exactendpointdetail}
\begin{proof}[of \cref{thm:kprofile}(d)]
Take $\omega=0$, $\mu=0$, $\zeta=1$, and $c\in[\max\{c_0,1-1/k\},1]$, which is nonempty because $c_0<1$ and $1-1/k<1$. Then $v=(1-\mu)\zeta=1$ and $\chi(Z)=0\cdot Z+1\equiv1$. The ranges $c\in[c_0,2]$, $\omega\in[0,1/4]$, and $\mu\in[0,1]$ hold, so it suffices to verify the three tests of \cref{thm:finitepoly}.

\emph{Floor polynomial.} $\mathcal F_k(Z)=(kc+Z)-kZ=kc-(k-1)Z$ is decreasing, and its value at $Z=1$ is $kc-(k-1)\ge0$ because $c\ge1-1/k$.

\emph{Balanced polynomial.} $\mathcal B(Z)=(c+Z)(c+Z/2)-Z(c+3Z/2)=c^2+cZ/2-Z^2$ is concave, so it suffices to check the endpoints: $\mathcal B(0)=c^2>0$, and $\mathcal B(1)=c^2+c/2-1\ge0$ is equivalent to $2c^2+c-2\ge0$, that is, $c\ge c_0$, the positive root $(\sqrt{17}-1)/4$ of $2c^2+c-2$.

\emph{Pure polynomial.} $\mathcal R(Z)=(c+Z)(c+1-Z)-Z(c+2-2Z)=Z^2-(c+1)Z+c(c+1)$. Its minimum over $\R$ is attained at the vertex $Z=(c+1)/2$, which lies in $[0,1]$ because $c\le1$, and equals $c(c+1)-(c+1)^2/4=(c+1)(3c-1)/4>0$, because $c\ge c_0>1/3$.

Hence \cref{thm:finitepoly} shows that the triple is admissible with coefficient $1/(1+c)$. This applies at $c=\max\{c_0,1-1/k\}$, whose coefficient is evaluated in the proof of \cref{thm:exactendpoint}, and at $c=1$, where $(1,0,0)$ is admissible at $\zeta=1$ for every finite $k\ge3$ simultaneously and certifies $1/2$.

\end{proof}

\subsection{Proofs of \texorpdfstring{\cref{lem:concaveobstacles} and \cref{thm:unweightedenvelope}}{the concavity lemma and the supporting-line theorem}: the supporting-line construction}\label{app:analyticprofile}
This subsection proves \cref{lem:concaveobstacles,thm:unweightedenvelope}. The subsubsections that follow prove \cref{thm:kprofile}(b)--(c), \cref{cor:explicitranges}, and \cref{thm:fivelabel,thm:half,thm:high}, all of which use the concavity and contact analysis established here.

\begin{proof}[of \cref{lem:concaveobstacles}]
\emph{Concavity.} In \eqref{eq:purepartialfractions}, put $A_c=c(3c+2)/(2c+1)$ and $D_c=c(c+1)/(2c+1)$; both are positive for $c>0$. Differentiating \eqref{eq:purepartialfractions} twice removes the constant $2$ and gives
\[
 r_c''(s)=-\frac{2A_c}{(c+s)^3}
          -\frac{2D_c}{(c+1-s)^3}<0\qquad(0\le s\le1),
\]
because $c+s>0$ and $c+1-s\ge c>0$ on $[0,1]$. Writing $B_c(s)=s(2c+3s)/((c+s)(2c+s))$ and differentiating directly,
\begin{align}
 B_c'(s)&=\frac{c(7s^2+12sc+4c^2)}{(s+c)^2(s+2c)^2},\\
 B_c''(s)&=-\frac{2sc(7s^2+18sc+12c^2)}{(s+c)^3(s+2c)^3}\leq0,\\
 f_{k,c}''(s)&=-\frac{2k^2c}{(kc+s)^3}<0.
\end{align}
The numerator of $B_c''$ is positive for $s>0$ and its denominator is positive, so $B_c''<0$ on $(0,1]$. Hence all three requirements are concave on $[0,1]$, and strictly concave for finite $k$.

\emph{Decrease in $c$.} For the floor requirements, $\partial_cf_{k,c}(s)=-k^2s/(kc+s)^2<0$ and $\partial_c(s/c)=-s/c^2<0$ for $s>0$. For the pure requirement,
\[
 \partial_c r_c(s)=
 -\frac{s\{c^2+4c(1-s)+(1-s)(2-s)\}}
 {(c+s)^2(c+1-s)^2}<0\qquad(s>0),
\]
since $c^2>0$ and $4c(1-s)$ and $(1-s)(2-s)$ are nonnegative for $0\le s\le1$. For the balanced requirement, dividing the numerator and denominator by $c^2$ and setting $x=s/c$ gives $B_c(s)=b(x)$ with $b(x)=x(2+3x)/((1+x)(2+x))$. The formula for $B_c'$ at $c=1$ gives $b'(x)=(7x^2+12x+4)/((1+x)^2(2+x)^2)>0$ for $x\ge0$, so $b$ is strictly increasing. At fixed $s>0$, $x=s/c$ strictly decreases as $c$ increases, and therefore $B_c(s)=b(s/c)$ strictly decreases.

\emph{Sign identities.} Bringing each difference to a common denominator gives
\[
 \begin{aligned}
 r_c(s)-B_c(s)&=\frac{s\bigl(s^2-(4c+1)s+2c\bigr)}{(c+s)(2c+s)(c+1-s)},\\
 r_c(s)-f_{k,c}(s)&=\frac{s\bigl((k-2)s^2-((2k-1)c+k-2)s+kc\bigr)}{(c+s)(kc+s)(c+1-s)},\\
 B_c(s)-f_{k,c}(s)&=\frac{2s^2\bigl(c+(3-k)s/2\bigr)}{(c+s)(2c+s)(kc+s)}.
 \end{aligned}
\]
For $0<s\le1$, every factor of the denominators and the prefactors $s$ and $2s^2$ are positive, which proves \eqref{eq:obstaclecrossings}.

\emph{One crossing.} The quadratic $q_1(s)=s^2-(4c+1)s+2c$ has $q_1(0)=2c>0$ and $q_1(1)=-2c<0$. The quadratic $q_2(s)=(k-2)s^2-((2k-1)c+k-2)s+kc$ has leading coefficient $k-2>0$ for $k\ge3$, $q_2(0)=kc>0$, and $q_2(1)=-(k-1)c<0$. A quadratic with positive leading coefficient that is positive at $0$ and negative at $1$ has exactly one root in $(0,1)$; its other root exceeds one, because the quadratic is positive again for large $s$. Hence each of the first two differences changes sign exactly once in $(0,1)$, from positive to negative.

\emph{Domination for $k=3,4$.} For $k\in\{3,4\}$ and $0\le s\le1$, $c+(3-k)s/2\ge c-s/2\ge c-1/2>0$, because $c\ge c_0=(\sqrt{17}-1)/4>1/2$. The third sign identity therefore gives $B_c\ge f_{k,c}$ on $[0,1]$.

\emph{All-label floor.} The same denominator clearing gives
\[
 f_{\infty,c}(s)-B_c(s)=\frac{s^3}{c(c+s)(2c+s)}\ge0,\qquad
 r_c(s)-f_{\infty,c}(s)=\frac{s\bigl(s^2-(2c+1)s+c\bigr)}{c(c+s)(c+1-s)}.
\]
The quadratic $s^2-(2c+1)s+c$ equals $c>0$ at $s=0$ and $-c<0$ at $s=1$, so the argument for $q_1$ gives exactly one crossing in $(0,1)$.
\end{proof}

\begin{proof}[of \cref{thm:unweightedenvelope}]
\emph{Existence of the minimum.} By \eqref{eq:supportprofile}, $M_H(c,\mu)\leq(1-\mu)\zeta$ says $H_c(s)\le\mu s+(1-\mu)\zeta=\chi(s)$ for every $s\in[0,1]$, with $\omega=0$ and $v=(1-\mu)\zeta$. For $H\in\{r,B,f_k\}$ these are the three inequalities in \eqref{eq:threeobstaclecriterion}, so $\mathcal E_k(c,\zeta)\leq0$ holds exactly when some slope $\mu\in[0,1]$ passes the sufficient mass tests of \cref{thm:finitepoly} at cost $c$.

The slope $\mu=1$ is feasible at $c=\sqrt2$. For $\mu=1$ the target is $\chi(s)=s$. Each requirement is concave with $H_c(0)=0$, so $H_c(s)\le H_c(0)+H_c'(0)s=H_c'(0)s$, and it suffices that $H_c'(0)\le1$. The derivative formulas in the proof of \cref{lem:concaveobstacles} give $B_c'(0)=f_{k,c}'(0)=1/c$, and direct differentiation gives $r_c'(0)=(c+2)/(c(c+1))$. At $c=\sqrt2$, $1/c<1$, and $(c+2)/(c(c+1))\le1$ is equivalent to $c+2\le c^2+c$, that is, to $c^2\geq2$, which holds with equality. Hence $\sqrt2$ belongs to the set of costs $c\in[c_0,\sqrt2]$ with $\mathcal E_k(c,\zeta)\le0$.

This set is closed. Each $M_H(c,\mu)$ is the maximum over the compact interval $[0,1]$ of a function continuous in $(s,c,\mu)$, so it is continuous; $\mathcal E_k(c,\zeta)$ is then the minimum over the compact interval $\mu\in[0,1]$ of a continuous function, so it is continuous in $c$. The set $\{c\in[c_0,\sqrt2]:\mathcal E_k(c,\zeta)\le0\}$ is therefore a nonempty closed subset of a compact interval, its minimum is attained, and \cref{thm:finitepoly} certifies the coefficient $1/(1+c)$ at the minimizing cost.

\emph{Contacts.} Fix $c$ and $\mu$. Since $H_c$ is concave, so is $s\mapsto H_c(s)-\mu s$, and its first-order conditions on $[0,1]$ are necessary and sufficient for a maximizer $s^\ast$: $H_c'(0)\leq\mu$ if $s^\ast=0$, $H_c'(1)\geq\mu$ if $s^\ast=1$, and $H_c'(s^\ast)=\mu$ if $0<s^\ast<1$. Strict concavity makes the maximizer unique for finite $k$. For the floor, $f_{k,c}'(s)=k^2c/(kc+s)^2$ is decreasing in $s$, and for $\mu>0$ the equation $k^2c/(kc+s)^2=\mu$ with $kc+s>0$ gives $kc+s=k\sqrt{c/\mu}$, that is, $s=k(\sqrt{c/\mu}-c)$; clipping to $[0,1]$ gives \eqref{eq:floorcontact}. For $\mu=0$, $f_{k,c}'>0$ on $[0,1]$, so the maximizer is $s=1$.

\emph{Completeness of the contact types.} Fix $c$ and put $G_H(\mu)=M_H(c,\mu)-(1-\mu)\zeta$. Since
\[
 G_H(\mu)=\max_{0\le s\le1}\bigl\{H_c(s)-\zeta+\mu(\zeta-s)\bigr\}
\]
is a maximum of affine functions of $\mu$, it is convex. By Danskin's theorem its one-sided derivatives are attained by the numbers $\zeta-s$ at the maximizing contacts $s$, and where the contact $s_H$ is unique $G_H$ is differentiable with derivative $\zeta-s_H$. The function $\max_HG_H$ is convex, and at an interior minimizer $\mu^\ast\in(0,1)$ zero belongs to its subdifferential, which is the convex hull of the derivatives $\zeta-s_H$ of the active requirements (those with $G_H(\mu^\ast)=\max_{H'}G_{H'}(\mu^\ast)$). If one requirement is active, its derivative vanishes and its contact is $\zeta$. If several are active, zero is a convex combination of the numbers $\zeta-s_H$, so some active contact is at most $\zeta$ and some is at least $\zeta$; the smallest and the largest active contacts then bracket $\zeta$, which also covers a triple contact. For the linear requirement $f_{\infty,c}$, the subdifferential is obtained from the endpoint contacts. An active pair from this requirement alone could arise only at $\mu=1/c$, and a zero gap with $\zeta>0$ would then force $\mu=1$ and $c=1$, which violates the pure requirement; the case $\zeta=0$ is the slope boundary treated below. Thus the three pairs of distinct requirement types suffice.

Suppose the selected cost $c=c_k^{\mathrm u}(\zeta)$ exceeds $c_0$, the minimizing slope satisfies $\mu<1$, and $\zeta>0$. If $\mathcal E_k(c,\zeta)$ were negative, continuity of $\mathcal E_k$ in $c$ would give $\mathcal E_k(c',\zeta)<0$ for some $c'\in[c_0,c)$, contradicting minimality. Hence $\mathcal E_k(c,\zeta)=0$, so every active requirement satisfies $H_c(s)-\mu s=(1-\mu)\zeta$ at its contact, which is \eqref{eq:contactvalues}, and the interior first-order conditions give \eqref{eq:contactslopes}.

\emph{Boundaries.} The remaining cases are the boundaries, including $c=c_0$ when it is feasible. At $\mu=0$, use an active equation $H_c(s)=\zeta$, with $H_c'(s)=0$ for an interior maximizing contact, or its endpoint derivative inequality. At $\mu=1$, the equation $H_c(0)=0$ holds for every $c$, and the derivative calculation in the second paragraph gives the boundary candidate $c=\sqrt2$; this also handles $\zeta=0$. Clearing the positive denominators turns every case into one of finitely many types of polynomial equations and inequalities, which specify the parameters, including the cost and slope boundaries. For the pure and balanced requirements, the interior slope conditions $r_c'(s)=\mu$ and $B_c'(s)=\mu$ become
\[
 c(c^2-4cs+3c+s^2-4s+2)=\mu(c+s)^2(c+1-s)^2,\qquad
 c(4c^2+12cs+7s^2)=\mu(c+s)^2(2c+s)^2,
\]
both of degree four in $s$, and the value conditions $H_c(s)=\mu s+(1-\mu)\zeta$ become cubics in $s$; this is the degree bound stated in \cref{thm:unweightedenvelope}.

\emph{Monotonicity in $\zeta$.} Since $M_H$ does not depend on $\zeta$ and $(1-\mu)\zeta$ is nondecreasing in $\zeta$ for $\mu\le1$, a pair $(c,\mu)$ feasible at $\zeta$ is feasible at every $\zeta'\ge\zeta$. The feasible set of costs therefore grows with $\zeta$, so $c_k^{\mathrm u}$ is nonincreasing and $\alpha_k^{\mathrm u}=1/(1+c_k^{\mathrm u})$ is nondecreasing.

\emph{Continuity.} Let $\zeta_n\to\zeta$, and let $(c_n,\mu_n)$ be minimizing pairs at $\zeta_n$. First choose a subsequence along which $c_n$ tends to $\liminf_j c_j$. By compactness of $[c_0,\sqrt2]\times[0,1]$, a further subsequence converges to some $(\bar c,\bar\mu)$ with $\bar c=\liminf_j c_j$, and the inequalities $M_H(c_n,\mu_n)\le(1-\mu_n)\zeta_n$ pass to the limit by continuity of $M_H$. Thus $\bar c$ is feasible at $\zeta$, and $c_k^{\mathrm u}(\zeta)\le\liminf_nc_k^{\mathrm u}(\zeta_n)$.

For the reverse inequality, let $(c,\mu)$ be a minimizing pair at $\zeta$. If $c=\sqrt2$, the pair $(\sqrt2,1)$ is feasible at every level, so $c_k^{\mathrm u}(\zeta_n)\le\sqrt2$ for all $n$. Otherwise fix $\delta>0$ with $c+\delta\le\sqrt2$.
\begin{itemize}
 \item If $\mu<1$ and $\zeta>0$, then $\chi(s)=\mu s+(1-\mu)\zeta\ge H_c(s)$ on $[0,1]$. By the strict decrease in $c$ from \cref{lem:concaveobstacles}, $H_{c+\delta}(s)<H_c(s)\le\chi(s)$ for $s>0$, and $H_{c+\delta}(0)=0<(1-\mu)\zeta=\chi(0)$. There are only three requirements, so the common slack $\eta:=\min_{H\in\{r,B,f_k\}}\min_{s\in[0,1]}(\chi(s)-H_{c+\delta}(s))$ is strictly positive by continuity and compactness. Changing the level from $\zeta$ to $\zeta'$ changes $\chi$ by $(1-\mu)(\zeta'-\zeta)$, so $(c+\delta,\mu)$ is feasible at every $\zeta'$ with $(1-\mu)|\zeta'-\zeta|\le\eta$.
 \item If $\mu=1$, the target $\chi(s)=s$ does not depend on the level, so $(c,1)$ is feasible at every level.
 \item If $\zeta=0$, every nearby level satisfies $\zeta'\ge0$, and the monotonicity just proved gives $c_k^{\mathrm u}(\zeta')\le c_k^{\mathrm u}(0)$.
\end{itemize}
In each case $\limsup_nc_k^{\mathrm u}(\zeta_n)\le c+\delta$, and $\delta>0$ is arbitrary. Hence $c_k^{\mathrm u}$ is continuous, and so is $\alpha_k^{\mathrm u}$.

\emph{Dependence on $k$.} Since $\partial_kf_{k,c}(s)=s^2/(kc+s)^2\ge0$ and $f_{k,c}(s)=ks/(kc+s)\le s/c=f_{\infty,c}(s)$, the floor requirement increases with $k$ while the other two requirements do not depend on $k$. The feasible sets of costs are therefore nested, decreasing in $k$, so $c_k^{\mathrm u}$ is nondecreasing and $\alpha_k^{\mathrm u}$ is nonincreasing in $k$. For $k=3,4$, \cref{lem:concaveobstacles} gives $B_c\ge f_{k,c}$ on $[0,1]$, so the floor inequality is implied by the balanced one; the feasible sets for $k=3$ and $k=4$ both coincide with the set defined by the pure and balanced inequalities alone, proving $\alpha_3^{\mathrm u}=\alpha_4^{\mathrm u}$.
\end{proof}

\subsubsection{Proof of \texorpdfstring{\cref{thm:kprofile}(b)--(c)}{the tangency statements}: pure and balanced single-contact branches}\label{app:closedprofiles}
\begin{proof}[of \cref{thm:kprofile}(b)--(c)]
\emph{The costs.} A line through $(\zeta,\zeta)$ tangent to a requirement at mass $\zeta>0$ touches it at $(\zeta,\zeta)$, so the requirement must equal $\zeta$ there. For the pure requirement, $r_c(\zeta)=\zeta$ reads $\zeta(c+2-2\zeta)=\zeta(c+\zeta)(c+1-\zeta)$. Canceling $\zeta>0$ and expanding $(c+\zeta)(c+1-\zeta)=c^2+c+\zeta-\zeta^2$ gives
\[
 c+2-2\zeta=c^2+c+\zeta-\zeta^2
 \quad\Longleftrightarrow\quad
 c^2=2-3\zeta+\zeta^2=(1-\zeta)(2-\zeta).
\]
The positive root is $c_P(\zeta)$. For the balanced requirement, $B_c(\zeta)=\zeta$ reads $\zeta(c+3\zeta/2)=\zeta(c+\zeta)(c+\zeta/2)$; canceling $\zeta$ and expanding $(c+\zeta)(c+\zeta/2)=c^2+\frac32c\zeta+\frac12\zeta^2$ gives
\[
 c^2+\left(\frac32\zeta-1\right)c
       +\frac12\zeta^2-\frac32\zeta=0.
\]
For $0<\zeta<3$, the product of the roots, $\zeta(\zeta-3)/2$, is negative, so there is exactly one positive root. The discriminant is $(\frac32\zeta-1)^2-2\zeta^2+6\zeta=(\zeta^2+12\zeta+4)/4$, and the positive root is
\[
 c=\frac12\left(1-\frac32\zeta\right)+\frac14\sqrt{\zeta^2+12\zeta+4}=c_B(\zeta).
\]
These are the formulas in \eqref{eq:pureclosed}--\eqref{eq:balancedclosed}.

\emph{The slopes.} Differentiating $r_c$ gives
\[
 r_c'(s)=\frac{c\bigl(c^2-4cs+3c+s^2-4s+2\bigr)}{(c+s)^2(c+1-s)^2}.
\]
At $s=\zeta$ and $c=c_P(\zeta)$, the identity $(c+\zeta)(c+1-\zeta)=c+2-2\zeta$ from the previous paragraph and the reduction $c^2=(1-\zeta)(2-\zeta)$ turn this expression into the formula for $\mu_P(\zeta)$ in \eqref{eq:pureclosed}; thus $\mu_P(\zeta)=r'_{c_P(\zeta)}(\zeta)$. By definition, $\mu_B(\zeta)=B'_{c_B(\zeta)}(\zeta)$.

\emph{The three tests.} For $H\in\{P,B\}$, the target $\chi(s)=\mu_Hs+(1-\mu_H)\zeta$ takes the value $\zeta$ at $s=\zeta$ and has slope $\mu_H$, so it is the tangent line of the selected requirement at its point $(\zeta,\zeta)$. A tangent line of a concave function majorizes it, so $\chi$ majorizes the selected requirement on $[0,1]$. The support inequality in \eqref{eq:pureacceptance} or \eqref{eq:balancedacceptance} states $H_c(s)-\mu_Hs\le v_H(\zeta)$ for every $s\in[0,1]$ for the other of $r_c$ and $B_c$, which is the second mass inequality. The condition $k\le K_H(\zeta)$ gives $\mathcal F_k\ge0$ on $[0,1]$ by \cref{lem:labelcutoff}. Hence all three tests of \cref{thm:finitepoly} hold.

\emph{Monotonicity.} Both costs decrease strictly with $\zeta$ wherever the displayed formulas have $c\ge c_0$. For the balanced cost,
\[
 c_B'(\zeta)=\frac14\left(-3+\frac{\zeta+6}{\sqrt{\zeta^2+12\zeta+4}}\right)<0
 \qquad(\zeta>0),
\]
because $\zeta+6$ and $\sqrt{\zeta^2+12\zeta+4}$ are positive and $9(\zeta^2+12\zeta+4)-(\zeta+6)^2=8\zeta(\zeta+12)>0$, so $(\zeta+6)/\sqrt{\zeta^2+12\zeta+4}<3$. The pure cost $c_P(\zeta)=\sqrt{(1-\zeta)(2-\zeta)}$ decreases because both factors are positive and decreasing on $[0,1)$. Since $c\mapsto1/(1+c)$ is decreasing, both coefficient formulas increase strictly.

\emph{The junction at $\zeta_0$.} By \cref{thm:kprofile}(a), the transition parameters at $\zeta_0$ are $\omega_{\zeta_0}=0$ and $c_{\zeta_0}=c_P(\zeta_0)$, and with $\omega=0$ the formula for $\mu_\zeta$ in \eqref{eq:transitionparameters} coincides with $\mu_P$ in \eqref{eq:pureclosed}. Hence $\alpha_{\mathrm{tr}}(\zeta_0)=1/(1+c_P(\zeta_0))=\alpha_{\mathrm{pu}}(\zeta_0)$. The first derivatives agree by the envelope argument at the end of \cref{app:transition}: both branches have the form $z\mapsto(1-w(c,z))/(1+c)$ with $c=c_z$ or $c=c_P(z)$, the partial derivative in $c$ vanishes at the common cost $c_{\zeta_0}=c_P(\zeta_0)$, and so both derivatives equal the partial derivative in $z$ there.
\end{proof}

\subsubsection{Proof of \texorpdfstring{\cref{cor:explicitranges}}{the explicit-interval corollary}: explicit pure and balanced intervals}\label{app:explicitranges}
The proof of \cref{cor:explicitranges} uses rational changes of variables, derivative signs, and completion of squares.

\begin{proof}[of \cref{cor:explicitranges}]
We use the following elementary floor comparison. For $0<M<1$ and $V\ge0$, the affine function $Mx+V$ dominates $kx/(k+x)$ on the entire half-line $x\ge0$ whenever
\begin{equation}\label{eq:floorhalflinetest}
 V\ge k(1-\sqrt M)^2.
\end{equation}
Indeed, differentiating $kx/(k+x)-Mx$ gives its unique maximum at $x=k(M^{-1/2}-1)$, with value $k(1-\sqrt M)^2$. In particular, if $U=k(1+M)-V>0$, the exact rational inequality $4k^2M-U^2\ge0$ proves \eqref{eq:floorhalflinetest}.

\emph{Pure branch: rational coordinates.}
Put $t=c_P(z)/(1-z)$. Solving $c_P(z)^2=(1-z)(2-z)$ gives
\[
\begin{gathered}
 z=\frac{t^2-2}{t^2-1},\qquad c=\frac t{t^2-1},\qquad
 \mu=-\frac{t(t^3-t^2-3t+1)}{(t-1)(t+1)(t+2)},\\
 V=\frac{(t^2-2)^2(t^2+1)}{(t-1)^2(t+1)^2(t+2)}.
\end{gathered}
\]
Here $V=(1-\mu)z$. On $3/2\le t\le25/16$, the cost lies in $(1,6/5]$, and $0<\mu<1$. For the latter statement, differentiation gives
\[
 -\mu'(t)=\frac{t^6+4t^5-2t^4-8t^3+7t^2+12t-2}
 {(t-1)^2(t+1)^2(t+2)^2}>0.
\]
The numerator after $t=3/2+u$ has positive coefficients
\[
 1,\quad13,\quad247/4,\quad275/2,\quad2479/16,\quad1581/16,\quad2329/64
\]
in descending powers of $u$. The endpoint values are $\mu(3/2)=57/70$ and $\mu(25/16)=236975/336528$.

In the scaled mass $x=s/c$, write the target as $L_t(x)=M(t)x+V(t)$, where
\[
 M(t)=c\mu=-\frac{t^2(t^3-t^2-3t+1)}{(t-1)^2(t+1)^2(t+2)}.
\]
Then $M'<0$ and $\partial_tL_t(x)=M'(t)(x-\eta(t))$, with
\[
 \eta(t)=\frac{(t^2-2)(t^6+4t^5-4t^3+9t^2+16t-2)}
 {t(3t^5+2t^4-4t^3+10t^2+17t-4)}<\frac13.
\]
Here the denominator is positive: the numerator of $-M'$ is $t(3t^5+2t^4-4t^3+10t^2+17t-4)$, whose expansion at $3/2$ has positive coefficients. The numerator of $1/3-\eta$, after removal of the positive denominator $3t(3t^5+2t^4-4t^3+10t^2+17t-4)$, is
\[
 E(t)=-3t^8-12t^7+9t^6+38t^5-31t^4-62t^3+77t^2+92t-12.
\]
At $t=3/2+u$, it equals
\[
 \frac{10809}{256}-\frac{8143}{16}u-1525u^2-\frac{4375}{2}u^3
 -\frac{15383}{8}u^4-1015u^5-306u^6-48u^7-3u^8.
\]
It is decreasing for $u\ge0$, and $E(25/16)=16718789373/4294967296>0$. Thus $L_t(x)$ is decreasing in $t$ for every $x\ge1/3$.

For $0\le x\le1/3$, the crossing identities in \cref{lem:concaveobstacles} show that $r_c(cx)$ dominates both $B_c(cx)$ and $f_{k,c}(cx)$. More explicitly, the relevant sign polynomials, after positive factors are removed, are
\[
 cx^2-(4c+1)x+2,\qquad
 (k-2)cx^2-((2k-1)c+k-2)x+k.
\]
Both decrease on this interval. At $x=1/3$, they are at least $1/5$ and $4/5$, respectively, using $c\le6/5$. Since $L_t$ is tangent to the concave pure requirement, it already majorizes that requirement.

It remains to check two fixed right-end targets for $x\ge1/3$. Put $q=139/89$ and $q'=120/77$, both in $[3/2,25/16]$. They give
\[
 z(q)=\frac{3479}{11400}>\frac{61}{200},\qquad
 z(q')=\frac{2542}{8471}>\frac3{10}.
\]
At $q$ the scaled slope and intercept are
\[
 M=\frac{15766303099}{20598660000},\qquad
 V=\frac{164860969861}{1833280740000}.
\]
For $b(x)=x(2+3x)/((1+x)(2+x))$, define
$P(x)=(x+1)(x+2)(Mx+V)-x(2+3x)$. Writing $h=x-2/3$ gives
\[
 P(x)=Mh^3+a_2 h^2+a_1 h+a_0.
\]
Its coefficients satisfy the exact rational inequalities
\[
 a_2-\frac23M>\frac25,\qquad |a_1|<\frac1{250},\qquad a_0>\frac1{2000}.
\]
The three quantities on the left equal
\[
 \frac{1118334467563}{2749921110000},\qquad
 \left|\frac{970717069}{366656148000}\right|,\qquad
 \frac{215905841}{247492899900},
\]
respectively. As $h\ge-2/3$, the cubic and quadratic terms are at least $(a_2-2M/3)h^2$. This quadratic is strictly positive after adding $a_1h+a_0$, because
$4(2/5)(1/2000)>(1/250)^2$.
Hence $L_q\ge b$ on $x\ge0$.

For the floor with $k=5$ at $q$, and with $k=6$ at $q'$, the rational test following \eqref{eq:floorhalflinetest} has, respectively, the positive values
\[
 4\cdot5^2M(q)-[5(1+M(q))-V(q)]^2
 =\frac{6336119606988820932133}{31119613626434700000000}>0,
\]
\[
 4\cdot6^2M(q')-[6(1+M(q'))-V(q')]^2
 =\frac{14198233306146634765424}{573008826273547668454681}>0.
\]
In each case $0<M<1$, $V\ge0$, and $k(1+M)-V>0$. These identities prove the floor inequalities globally. The monotonicity of $L_t$ on $x\ge1/3$, and the pure comparison on $x\le1/3$, propagate all these inequalities to the asserted smaller parameter values. Since $z(t)$ increases and $\zeta_0>1/5=z(3/2)$, this proves the pure interval through $61/200$ for $k\le5$, and through $3/10$ for $k\le6$.

\emph{Balanced branch: rational coordinates.}
Now put $t=z/c_B(z)$. The identity $B_c(z)=z$ yields
\begin{equation}
\begin{gathered}
 z=\frac{t(2+3t)}{(1+t)(2+t)},\qquad c=\frac{2+3t}{(1+t)(2+t)},\\
 \mu=\frac{7t^2+12t+4}{(t+1)(t+2)(3t+2)},\qquad
 V=\frac{t^3(3t+4)}{(t+1)^2(t+2)^2}.
\end{gathered}
\end{equation}
The map $t\mapsto z$ is increasing. The interval $17/25\le z\le1$ corresponds to $t_-\le t\le t_+$, where $58t_-^2-t_--34=0$ and $t_+=(1+\sqrt{17})/4$. On this range $c_0\le c<9/10$, and $0<\mu<1$, the latter because
$1-\mu=t^2(3t+4)/[(t+1)(t+2)(3t+2)]>0$.

We first check the pure requirement. For $s\ge1/2$, it follows directly from the balanced requirement: the sign polynomial $s^2-(4c+1)s+2c$ is negative at $1/2$ and decreasing through $s=1$ when $c\ge3/4$.
For $0\le s\le1/2$, the gap $\mu(t)s+V(t)-r_{c(t)}(s)$ is strictly increasing in $t\ge3/4$ as long as $c(t)\ge3/4$.
To prove this, note first that
\[
 0\le-\partial_c r_c(s)\le-\partial_c r_c(1/2)
 =\frac{2(2c+3)}{(2c+1)^3}\le\frac{72}{125}<\frac34.
\]
To verify the first comparison, put $a=c-3/4\ge0$ and $y=1/2-s\ge0$. The numerator of $\partial_s(-\partial_c r_c(s))$, over the positive denominator $(c+s)^3(c+1-s)^3$, is
\[
 y^4+(8a+9)y^3+\frac{3(4a+3)(4a+5)}8y^2
 +\frac{(4a+5)(32a^2+68a+27)}{16}y
 +\frac{(4a+1)(4a+5)^3}{256}>0.
\]
The displayed value at $s=1/2$ decreases with $c$. Also $c'<0$, $\mu'<0$, and
\[
 \frac{V'+\mu'/2}{-c'}-\frac34
 =\frac{381t^5+1239t^4+1142t^3-148t^2-648t-224}
 {4(t+1)(t+2)(3t+2)^2(3t+4)}>0.
\]
The numerator at $t=3/4+u$ has positive coefficients
\[
 381,\quad10671/4,\quad56017/8,\quad262735/32,\quad960177/256,\quad175075/1024.
\]
Consequently
\[
 \partial_t(\mu s+V-r_c(s))
 \ge(-c')\left(\frac34-\frac{72}{125}\right)>0.
\]

It therefore suffices to check a lower anchor $q=24/31<t_-$, since $z(q)=1608/2365<17/25$. At this anchor,
\[
 c=\frac{2077}{2365},\qquad\mu=\frac{130231}{158455},\qquad
 V=\frac{677376}{5593225}.
\]
Let $R(s)=(c+s)(c+1-s)(\mu s+V)-s(c+2-2s)$ and set $h=s-3/10$. Then
$R=-\mu h^3+a_2 h^2+a_1 h+a_0$. For $0\le s\le1/2$, we have $h\le1/5$, and the exact coefficients give
\[
\begin{gathered}
 a_2-\frac\mu5=\frac{1346616653}{749492150}>\frac32,\qquad
 |a_1|=\frac{10169472589}{3545097869500}<\frac1{250},\\
 a_0=\frac{726832714677}{3353662584547000}>\frac1{5000}.
\end{gathered}
\]
Thus $R\ge(a_2-\mu/5)h^2+a_1h+a_0>0$ by completion of squares, because
$4(3/2)(1/5000)>(1/250)^2$. This establishes the pure inequality on the entire balanced interval.

For $k\le4$, the balanced requirement dominates the floor on $[0,1]$, by \cref{lem:concaveobstacles} and $c\ge c_0>1/2$. For $k=5$, use the scaled mass $x=s/c$ again. Now $M=c\mu=b'(t)$, $V=b(t)-tb'(t)$, so
\[
 \partial_t(Mx+V)=b''(t)(x-t),\qquad b''(t)<0.
\]
On $x\le1$, $b(x)\ge5x/(5+x)$, so the balanced tangent already covers the floor. On $x\ge1$, the target decreases with $t$ as long as $t<1$. Choose $q=41/45<1$, for which $z(q)-31/40=37/225320>0$. At this anchor,
\[
 M=\frac{85064175}{126922756},\quad V=\frac{20883063}{126922756},\quad
 100M-[5(1+M)-V]^2=\frac{1858871805971}{1006836624414721}>0.
\]
The floor test \eqref{eq:floorhalflinetest} applies, and the target monotonicity extends it to every $t$ corresponding to $17/25\le z\le31/40$. The three tests in \cref{thm:finitepoly} are now verified on every asserted interval. The coefficient and junction statements follow from \cref{thm:kprofile}.
\end{proof}

\subsubsection{Proof of \texorpdfstring{\cref{thm:fivelabel}}{the five-label example}: five-label common tangent and terminal branch}\label{app:fiveprofile}
\begin{proof}[of \cref{thm:fivelabel}]
\emph{The root $t$.} The coefficients of the quintic $Q(t)=t^5+26t^4+94t^3+70t^2-75t-80$ in \eqref{eq:fiveroot} have the sign pattern $+,+,+,+,-,-$ with one sign change, so by Descartes' rule of signs $Q$ has exactly one positive root. Since $Q(0)=-80<0$ and $Q(1)=36>0$, that root lies in $(0,1)$. Direct rational substitution gives $Q(41/45)=-107884204/184528125<0$ and $Q(23/25)=26659868/9765625>0$, so
\[
 \frac{41}{45}<t<\frac{23}{25}.
\]

\emph{The formula for $p$.} Write the common slope as $M=p^2$ with $p>0$. Tangency to $b$ at $x=t$ gives $p^2=b'(t)$, and equality of the intercept with the floor support intercept $5(1-p)^2$ (computed below) gives $b(t)-tp^2=5(1-p)^2$. Expanding the right side, the latter equation is $b(t)-tp^2=5-10p+5p^2$, that is,
\[
 10p=5-b(t)+(t+5)p^2.
\]
Substituting $p^2=b'(t)$ gives
\[
 p=\frac{5-b(t)+(t+5)b'(t)}{10},
\]
and the identities
\begin{align}
 b'(t)&=\frac{7t^2+12t+4}{(t+1)^2(t+2)^2},\\
 b(t)-tb'(t)&=\frac{t^3(3t+4)}{(t+1)^2(t+2)^2}
\end{align}
expand this to $p=(t^4+13t^3+50t^2+60t+20)/(5(t+1)^2(t+2)^2)$, the expression in the theorem. With this $p$,
\begin{equation}
 p^2-b'(t)=
 \frac{t^3(t^5+26t^4+94t^3+70t^2-75t-80)}
 {25(t+1)^4(t+2)^4},
\end{equation}
which vanishes at the root $t$ of \eqref{eq:fiveroot}. Hence both tangency conditions hold: $p^2=b'(t)$, and then $b(t)-tp^2=5(1-p)^2=V$ by the linear relation above.

\emph{Bounds on $p$.} By the proof of \cref{lem:concaveobstacles}, $b''(x)=-2x(7x^2+18x+12)/((1+x)^3(2+x)^3)<0$ for $x>0$, so $b'$ is decreasing. From $41/45<t<23/25$,
\[
 \left(\frac{49}{60}\right)^2<b'(23/25)<p^2=b'(t)<b'(41/45)<\left(\frac{131}{160}\right)^2,
\]
where the outer two inequalities follow by direct rational evaluation. Taking positive square roots gives $49/60<p<131/160$.

\emph{Tangency to the floor.} The floor requirement $f_5(x)=5x/(5+x)$ has $f_5'(x)=25/(5+x)^2$. The equation $f_5'(x)=p^2$ with $5+x>0$ gives $5+x=5/p$, so the floor contact is
\[
 x_F=\frac{5(1-p)}p.
\]
There $f_5(x_F)=5(1-p)$ and $f_5(x_F)-p^2x_F=5(1-p)-5p(1-p)=5(1-p)^2=V$. Therefore $Mx+V$ is tangent to $b$ at $x=t$ and to $f_5$ at $x_F$. Both functions are concave on $[0,\infty)$ (their second derivatives are negative), so each lies below its tangent line, and the line majorizes both. Moreover $t<23/25<1$, and $x_F>1$ is equivalent to $p<5/6$, which holds because $p<131/160<5/6$.

\emph{Scaling.} In the mass variable $s=cx$, dividing the numerator and denominator of each requirement by the relevant power of $c$ gives
\[
 B_c(s)=b(x),\qquad f_{5,c}(s)=\frac{5cx}{5c+cx}=\frac{5x}{5+x},\qquad
 \chi(s)=\mu cx+V=Mx+V
\]
when $\mu=M/c$ and $v=V$. Thus the single tangent in the scaled variable proves the balanced and floor inequalities for every cost on the first branch. The relation $v=(1-\mu)\zeta$ reads $V=(1-M/c)\zeta$, that is, $c(\zeta-V)=M\zeta$, which solves to $c=M\zeta/(\zeta-V)$.

\emph{The branch interval.} Since $b(t)-tM=V$ and $V+Mx_F=5(1-p)^2+5p(1-p)=5(1-p)$,
\[
 \zeta_L=b(t)=V+Mt,\qquad \zeta_U=5(1-p)=V+Mx_F.
\]
For $\zeta_L\le\zeta\le\zeta_U$, the contact mass $ct=M\zeta t/(\zeta-V)$ is at most $\zeta$ exactly when $Mt\le\zeta-V$, that is, when $\zeta\ge\zeta_L$; likewise $cx_F\ge\zeta$ exactly when $\zeta\le\zeta_U$. So the contact masses $ct$ and $cx_F$ bracket $\zeta$.

These masses lie in $[0,1]$. The cost $c(\zeta)=M\zeta/(\zeta-V)$ has derivative $-MV/(\zeta-V)^2<0$, so on the branch it decreases from $c(\zeta_L)=M(V+Mt)/(Mt)=b(t)/t$ to $c(\zeta_U)=M\cdot5(1-p)/(Mx_F)=p$. The function $b(x)/x=(2+3x)/((1+x)(2+x))$ is decreasing for $x>0$, since its derivative is $-x(3x+4)/((1+x)^2(2+x)^2)$, so $b(t)/t<b(41/45)/(41/45)=9585/11266<43/50$. Also $x_F=5(1-p)/p<5(1-49/60)/(49/60)=55/49$. Hence $cx_F<(43/50)(55/49)=473/490<1$, and $ct\le cx_F$. Finally $p\le c\le b(t)/t<1$ and $0<\mu=M/c\le p^2/p=p<1$.

\emph{The balanced inequality on the terminal branch.} For a slope $m\in(0,1)$, consider the support values $\sup_{x\ge0}\{b(x)-mx\}$ and $\sup_{x\ge0}\{f_5(x)-mx\}$. Both maximizing contacts are unique by strict concavity. By \eqref{eq:obstaclecrossings} with $k=5$ and $c=1$, the difference $b-f_5$ has the sign of $1-x$, so it is positive for $0<x<1$ and negative for $x>1$. At a slope where the two support values are equal, the balanced contact lies below one and the floor contact above one: at each contact its own function equals the common supporting line, and the other function lies on or below that line, which by the sign of $b-f_5$ is possible only on the stated side of one. The two contacts cannot both equal one, since $b'(1)=23/36\ne25/36=f_5'(1)$. The derivative in $m$ of a support value is minus its contact, so the derivative of the floor support minus the balanced support at any such equality is (balanced contact) $-$ (floor contact) $<0$. As $m$ decreases to zero the difference tends to $\sup f_5-\sup b=5-3=2>0$, and the common tangent constructed above supplies a zero at $m=p^2$. Every zero is a strict crossing in the same direction, so there is no other zero, and the floor support dominates the balanced support whenever $0<m\leq p^2$.

On the terminal branch, $c=1-\zeta/5$ and $\zeta\ge\zeta_U=5(1-p)$ give $c\le p$. In the scaled variable the target is $\chi(cx)=c\cdot cx+\zeta^2/5=c^2x+5(1-c)^2$, because $\zeta=5(1-c)$. By the floor-tangency computation with $p$ replaced by $c$, this is exactly the tangent to $f_5$ with slope $m=c^2\le p^2$, so its intercept is the floor support value at slope $c^2$, which dominates the balanced support value by the previous paragraph. Hence $c^2x+5(1-c)^2\ge b(x)$ for every $x\ge0$, which is the balanced inequality; the same line majorizes $f_5$ because it is a tangent of the concave function $f_5$.

\emph{The pure requirement on the first branch.} First,
\begin{equation}
 r_{81/100}(s)\leq\frac{81}{100}s+\frac4{25}\qquad(0\leq s\leq1).
 \label{eq:fixed-pure-bound-one}
\end{equation}
Multiplying by the positive denominator $(c+s)(c+1-s)$ with $c=81/100$, the gap $(c+s)(c+1-s)(\frac{81}{100}s+\frac4{25})-s(c+2-2s)$ equals
\[
 T(s)=-\frac{81}{100}s^3+\frac{53}{20}s^2
      -\frac{1462459}{10^6}s+\frac{14661}{62500}.
\]
For $0\le s\leq2/5$, $T(s)$ minus the quadratic
\[
 \frac{1163}{500}s^2-\frac{1462459}{10^6}s+\frac{14661}{62500}
\]
equals $-\frac{81}{500}s^2(5s-2)\ge0$. This quadratic has positive leading coefficient and $4AC-B^2=43708777319/10^{12}>0$, so it is positive for every real $s$; hence $T>0$ on $[0,2/5]$. For $2/5\le s\le1$, $T''' =-243/50<0$, so $T'$ is concave; its endpoint values are $T'(2/5)=268741/10^6>0$ and $T'(1)=1407541/10^6>0$, so $T'>0$ on $[2/5,1]$ and $T(s)\ge T(2/5)>0$ there. This proves \eqref{eq:fixed-pure-bound-one}.

At $c=p$ the first branch has $\mu=M/p=p$ and $v=V$, so its target is $ps+V$. Since $p>81/100$ and $V=5(1-p)^2>5(29/160)^2=841/5120>4/25$, the decrease of $r_c$ in $c$ (\cref{lem:concaveobstacles}) and \eqref{eq:fixed-pure-bound-one} give
\[
 r_p(s)\le r_{81/100}(s)\le\frac{81}{100}s+\frac4{25}\le ps+V\qquad(0\le s\le1).
\]
For a fixed scaled variable $x$, direct differentiation gives
\[
 \frac{\partial}{\partial c}r_c(cx)
 =-\frac{x}{(1+x)(1+c(1-x))^2}<0.
\]
Every cost $c$ on the first branch satisfies $c\geq p$; if also $cx\leq1$, then $px\leq1$. Hence $r_c(cx)\leq r_p(px)\leq p\cdot px+V=Mx+V=\chi(cx)$, which proves the pure inequality on the entire first branch.

\emph{The pure requirement on the terminal branch.} Here $\zeta\le1$ gives $c=1-\zeta/5\geq4/5$ and $\mu=c\geq4/5$, and $\zeta\ge\zeta_U$ gives $v=\zeta^2/5\ge\zeta_U^2/5=V>841/5120>41/250$. By the decrease of $r_c$ in $c$, it therefore suffices to prove
\begin{equation}
 r_{4/5}(s)\leq\frac45s+\frac{41}{250}\qquad(0\le s\le1),
 \label{eq:fixed-pure-bound-two}
\end{equation}
since then $r_c(s)\le r_{4/5}(s)\le\frac45s+\frac{41}{250}\le cs+v$. Its cleared gap is
\[
 U(s)=-\frac45s^3+\frac{659}{250}s^2-\frac{371}{250}s+\frac{738}{3125}.
\]
For $0\le s\leq3/8$, $U(s)$ minus
$\frac{292}{125}s^2-\frac{371}{250}s+\frac{738}{3125}$ equals $-\frac1{10}s^2(8s-3)\ge0$, and this quadratic has positive leading coefficient and $4AC-B^2=6911/1562500>0$, so $U>0$ there. For $3/8\le s\le1$, $U'''=-24/5<0$, so $U'$ is concave with endpoint values $U'(3/8)=311/2000>0$ and $U'(1)=347/250>0$; hence $U$ increases and $U(s)\ge U(3/8)>0$. This proves \eqref{eq:fixed-pure-bound-two}, so all three requirements hold on both branches. For $3\le k<5$ the floor requirement is weaker, since $f_{k,c}\le f_{5,c}$ by the monotonicity in $k$ noted in the proof of \cref{thm:unweightedenvelope}.
\end{proof}

\subsubsection{Proofs of \texorpdfstring{\cref{thm:half,thm:high}}{the rational-parameter examples}: rational parameter choices}\label{app:elementaryimprovements}
\begin{proof}[of \cref{thm:half}]
At $\zeta=49/100$, the parameters $c=1$, $\omega=0$, $\mu=17/20$ give $v=(1-\mu)\zeta=\frac3{20}\cdot\frac{49}{100}=\frac{147}{2000}$ and the target $\chi(s)=17s/20+147/2000$.

\emph{Floor.} For $k=12$ and $c=1$, the floor polynomial \eqref{eq:finitefloor} is
\[
 (12+s)\chi(s)-12s=\frac{1700s^2-3453s+1764}{2000}>0.
\]
Its numerator has positive leading coefficient and discriminant $3453^2-4\cdot1700\cdot1764=-71991<0$, so it is positive for every real $s$. For $3\le k<12$ the floor requirement is weaker, because $f_{k,c}$ increases with $k$.

\emph{Pure.} With $c=1$ and $\omega=0$, the pure polynomial \eqref{eq:purepoly} is $R(s)=(1+s)(2-s)\chi(s)-s(3-2s)$. Put $h=s-1/4\in[-1/4,3/4]$. Expanding gives the exact identity
\[
 R(s)=\left(\frac{2139}{1000}-\frac{17}{20}h\right)h^2
       +\frac{19}{8000}h+\frac1{1600}.
\]
For $h\le3/4$, the coefficient of $h^2$ is at least $\frac{2139}{1000}-\frac{17}{20}\cdot\frac34=\frac{3003}{2000}>\frac32$, and $h^2\ge0$, so
\[
 R(s)\ge\frac32h^2+\frac{19}{8000}h+\frac1{1600}
 =\frac32\left(h+\frac{19}{24000}\right)^2
              +\frac{239639}{384000000}>0.
\]
This proves $\chi\ge r_1$ on $[0,1]$.

\emph{Balanced.} The balanced requirement follows from the other two. By \eqref{eq:obstaclecrossings} with $c=1$ and $k=12$,
\[
\operatorname{sign}(B_1(s)-f_{12,1}(s))
=\operatorname{sign}\left(1-\frac92s\right)\quad(s>0).
\]
Thus $B_1\le f_{12,1}\le\chi$ when $s\ge2/9$. When $0<s\le2/9$, \eqref{eq:obstaclecrossings} gives that $r_1-B_1$ has the sign of $s^2-5s+2$; this quadratic decreases on $[0,5/2]\supset[0,2/9]$ and equals $76/81>0$ at $2/9$, so $B_1\le r_1\le\chi$ there. At $s=0$ all three requirements vanish and $\chi(0)>0$. Hence \cref{thm:finitepoly} applies, and by \cref{prop:zetamonotone} the same parameters remain valid at every larger diagonal level.
\end{proof}

\begin{proof}[of \cref{thm:high}]
At $\zeta=9/10$, the parameters $c=41/45$, $\omega=0$, $\mu=9/10$ give $v=\frac1{10}\cdot\frac9{10}=\frac9{100}$ and $\chi(s)=\frac9{10}s+\frac9{100}$.

\emph{Floor.} For ten labels, the floor polynomial \eqref{eq:finitefloor} is
\[
(10c+s)\chi(s)-10s=\frac{90s^2-171s+82}{100}>0,
\]
since its numerator has positive leading coefficient and discriminant $171^2-4\cdot90\cdot82=-279<0$. For $3\le k<10$ the floor requirement is weaker.

\emph{Pure.} Expanding \eqref{eq:purepoly} with $c=41/45$ and $\omega=0$ gives
\[
R(s)=-\frac9{10}s^3+\frac{281}{100}s^2-\frac{627}{500}s+\frac{1763}{11250}.
\]
For $h=s-1/4\in[-1/4,3/4]$, it becomes
\begin{align*}
R(s)&=\left(\frac{427}{200}-\frac9{10}h\right)h^2-\frac{71}{4000}h+\frac{3437}{720000}\\
&\ge\frac{73}{50}h^2-\frac{71}{4000}h+\frac{3437}{720000}
=\frac{73}{50}\left(h-\frac{71}{11680}\right)^2
+\frac{3969047}{840960000}>0,
\end{align*}
where the inequality uses $h\le3/4$, so the coefficient of $h^2$ is at least $\frac{427}{200}-\frac9{10}\cdot\frac34=\frac{73}{50}$, and $h^2\ge0$.

\emph{Balanced.} By \eqref{eq:obstaclecrossings} with $k=10$, $B_c-f_{10,c}$ has the sign of $c-\frac72s$, so $B_c(s)\le f_{10,c}(s)\le\chi(s)$ for $s\ge2c/7=82/315$. Below this threshold, \eqref{eq:obstaclecrossings} gives that $r_c-B_c$ has the sign of
\[
q(s)=s^2-(4c+1)s+2c=s^2-\frac{209}{45}s+\frac{82}{45},
\]
whose vertex is at $209/90>82/315$; hence $q$ is decreasing on $[0,82/315]$, with endpoint value $q(82/315)=67568/99225>0$. Thus $B_c\le r_c\le\chi$ on $[0,82/315]$, the balanced requirement follows from the other two, \cref{thm:finitepoly} applies, and $1/(1+c)=45/86$.
\end{proof}

\subsection{Proof of \texorpdfstring{\cref{thm:halfthreshold}}{the half-coefficient theorem}: a half-coefficient family for every label count}\label{app:halfthreshold}
Throughout this proof the cost is $c=1$. Write $r=r_1$ and $B=B_1$, and use the tangent slope $m$ and intercept $j$ from \eqref{eq:halftangentdata}. At $c=1$, \eqref{eq:purepartialfractions} reads
\[
 r(s)=2-\frac5{3(1+s)}-\frac2{3(2-s)},
\]
so $r''(s)=-\frac{10}{3(1+s)^3}-\frac{4}{3(2-s)^3}<0$ on $[0,1]$. Since $m=r'$, this gives $m'(s)=r''(s)<0$. Differentiating $j(t)=r(t)-tm(t)$ gives $j'(t)=r'(t)-m(t)-tm'(t)=-tm'(t)$, and therefore
\begin{equation}\label{eq:halftangentderivatives}
 j'(t)=-tm'(t),\qquad
 L_t(s):=m(t)s+j(t),\qquad
 \partial_t L_t(s)=m'(t)s+j'(t)=m'(t)(s-t).
\end{equation}
Since $L_t$ is the tangent line of the concave function $r$ at $t$, every line $L_t$ majorizes $r$ on $[0,1]$; the contact $t$ is chosen below so that $L_t$ also majorizes the balanced and floor requirements.

\paragraph{A rational tangent that controls the balanced requirement.}
Evaluating \eqref{eq:halftangentdata} at $t_0=3/10$ gives
\[
 m_0:=m(t_0)=\frac{36900}{48841},\qquad
 j_0:=j(t_0)=\frac{4842}{48841}.
\]
To prove $L_{t_0}(s)>B(s)$ on $[0,1]$, multiply the gap by the positive denominator $(1+s)(2+s)$, using $B(s)=s(2+3s)/((1+s)(2+s))$:
\[
 P(s)=(1+s)(2+s)(m_0s+j_0)-s(2+3s).
\]
Set $h=s-7/10$, so $-7/10\le h\le3/10$. Expanding around this rational point gives
\[
 P(s)=\frac{36900}{48841}h^3+\frac{46509}{48841}h^2
       +\frac{7568}{244205}h+\frac{3593}{287300}.
\]
The first two terms equal $h^2\bigl(m_0h+\frac{46509}{48841}\bigr)$. Since $h\ge-7/10$ and $h^2\ge0$, they are at least $h^2\bigl(-\frac7{10}m_0+\frac{46509}{48841}\bigr)=Ah^2$, where
\[
 A=\frac{20679}{48841},\qquad
 D=\frac{7568}{244205},\qquad E=\frac{3593}{287300}.
\]
The resulting quadratic is positive for every real $h$:
\[
 Ah^2+Dh+E
 =A\left(h+\frac{D}{2A}\right)^2+\frac{4AE-D^2}{4A},
 \qquad 4AE-D^2=\frac{48232775}{2385443281}>0.
\]
Hence $P>0$ on $[0,1]$, and $L_{t_0}>B$ there.

Every pure tangent with $0<t\le3/10$ also majorizes $B$. For $0<s\le3/10$, the sign identity \eqref{eq:obstaclecrossings} at $c=1$ says that $r(s)-B(s)$ has the sign of $s^2-5s+2$; this quadratic decreases on $[0,5/2]$ and equals $59/100>0$ at $s=3/10$, so it is positive on $(0,3/10]$. Hence $L_t(s)\ge r(s)\ge B(s)$ there. For $s\ge3/10$ and every contact $t'\in[t,3/10]$, \eqref{eq:halftangentderivatives} gives $\partial_{t'}L_{t'}(s)=m'(t')(s-t')\le0$, because $m'<0$ and $s\ge3/10\ge t'$. Integrating in $t'$ from $t$ to $3/10$ gives $L_t(s)\ge L_{3/10}(s)>B(s)$ on the rest of the interval as well.

\paragraph{Three through five labels: a pure and balanced common tangent.}
For $a\in[3/10,61/200]$, let $\sigma(a)$ solve
\[
 B'(\sigma(a))=m(a).
\]
This solution is unique and lies in $(1/2,1)$. Indeed, $B''<0$ on $(0,1]$ (\cref{lem:concaveobstacles}), so $B'$ is strictly decreasing; $m$ is decreasing because $m'<0$; and substitution of the rational endpoints gives
\[
 B'(1)=\frac{23}{36}<m(61/200)
 \le m(a)\le m(3/10)<B'(1/2)=\frac{188}{225},
\]
with $m(61/200)=1948280000/2609511147$ and $m(3/10)=36900/48841$. The intermediate value theorem then gives a unique $\sigma(a)\in(1/2,1)$.

Because $B$ is concave and $B'(\sigma(a))=m(a)$, the point $\sigma(a)$ maximizes $B(s)-m(a)s$, so $B(\sigma(a))-\sigma(a)m(a)$ is the intercept of the balanced supporting line of slope $m(a)$. The difference between the pure intercept and this balanced support intercept is
\[
 \Delta_B(a)=j(a)-\bigl(B(\sigma(a))-\sigma(a)m(a)\bigr).
\]
Differentiating,
\[
 \Delta_B'(a)=j'(a)-B'(\sigma(a))\sigma'(a)+\sigma'(a)m(a)+\sigma(a)m'(a)=-am'(a)+\sigma(a)m'(a),
\]
because the two terms containing $\sigma'(a)$ cancel by $B'(\sigma(a))=m(a)$ and $j'(a)=-am'(a)$ by \eqref{eq:halftangentderivatives}. Thus
\[
 \Delta_B'(a)=(\sigma(a)-a)m'(a)<0,
\]
since $\sigma(a)>1/2>61/200\ge a$ and $m'(a)<0$. The rational tangent above gives $L_{3/10}(\sigma(3/10))>B(\sigma(3/10))$, that is, $\Delta_B(3/10)>0$. At the other endpoint, evaluation at the single mass $s=7/10$ gives
\[
 L_{61/200}(7/10)-B(7/10)
 =-\frac{36169609}{44361689499}<0,
\]
so $j(61/200)<B(7/10)-\frac7{10}m(61/200)\le B(\sigma_1)-\sigma_1m(61/200)$ with $\sigma_1=\sigma(61/200)$, because $\sigma_1$ maximizes $B(s)-m(61/200)s$; hence $\Delta_B(61/200)<0$. Since $\Delta_B$ is continuous and strictly decreasing, it has exactly one root $\tau_*\in(3/10,61/200)$. Setting $\sigma_*=\sigma(\tau_*)\in(1/2,1)$, the equalities $B'(\sigma_*)=m(\tau_*)$ and $\Delta_B(\tau_*)=0$ are the two equations \eqref{eq:halfpurebalanced}; any solution in the rectangle has $\sigma=\sigma(\tau)$ by uniqueness of $\sigma(\tau)$, so the solution is unique.

At $\tau_*$, the line $L_{\tau_*}$ is the tangent of $r$ at $\tau_*$ and the supporting line of $B$ at $\sigma_*$, so $L_{\tau_*}\ge r$ and $L_{\tau_*}\ge B$ on $[0,1]$. For the floor, \eqref{eq:obstaclecrossings} at $c=1$ and $k=5$ says that $B_1(s)-f_{5,1}(s)$ has the sign of $1-s$, so $B_1\ge f_{5,1}\ge f_{k,1}$ on $[0,1]$ for every $3\le k\le5$, the last inequality because $f_{k,1}$ increases with $k$. The common tangent therefore majorizes all three requirements.

\paragraph{Six through eleven labels: a pure and floor common tangent.}
For a slope $m=p^2$ with $0<p<1$, the function $s\mapsto f_{k,1}(s)-ms=ks/(k+s)-ms$ is concave on $s\ge0$ with derivative $k^2/(k+s)^2-m$. Solving $k^2/(k+s)^2=p^2$ with $k+s>0$ gives $k+s=k/p$, so the maximum occurs at
\begin{equation}\label{eq:halffloorsupport}
 s_k=k(1/p-1),\qquad
 f_{k,1}(s_k)-ms_k=k(1-p)^2.
\end{equation}
For the intercept, $s_k/(k+s_k)=1-p$ gives $f_{k,1}(s_k)=k(1-p)$, and $ms_k=p^2k(1-p)/p=kp(1-p)$; their difference is $k(1-p)^2$. A line that majorizes the floor on the half-line $s\ge0$ also majorizes it on $[0,1]$.

For $1/4\le a\le3/10$, put
\[
 \Delta_k(a)=j(a)-k(1-\sqrt{m(a)})^2.
\]
The second term is the floor support intercept from \eqref{eq:halffloorsupport} with $p=\sqrt{m(a)}$, and the floor contact is $s_k(a)=k(1/\sqrt{m(a)}-1)$. Since $m$ is decreasing,
\[
 0<m(a)\le m(1/4)=\frac{208}{245}<\left(\frac{93}{100}\right)^2,
\]
so $1/\sqrt{m(a)}>100/93$, and $k\ge6$ gives $s_k(a)>6(100/93-1)=14/31>3/10\ge a$. The derivative of the support intercept is
\[
 \frac{d}{da}\,k\bigl(1-\sqrt{m(a)}\bigr)^2=-k\bigl(1-\sqrt{m(a)}\bigr)\frac{m'(a)}{\sqrt{m(a)}}=-s_k(a)m'(a),
\]
and with $j'(a)=-am'(a)$ this gives
\[
 \Delta_k'(a)=(s_k(a)-a)m'(a)<0.
\]
For the endpoint signs, define
\[
 u_k(a)=\frac{k(1+m(a))-j(a)}{2k}.
\]
It is positive because $j(a)=r(a)-am(a)<r(a)<1\le k$ and $m(a)>0$. Expanding $k(1-\sqrt m)^2=k-2k\sqrt m+km$ gives $\Delta_k(a)=2k(\sqrt{m(a)}-u_k(a))$; since $\sqrt{m(a)}$ and $u_k(a)$ are positive, the sign of $\Delta_k(a)$ is the sign of $m(a)-u_k(a)^2$. At the two endpoint cases,
\[
 \begin{aligned}
 u_{11}(1/4)&=\frac{993}{1078},&
 m(1/4)-u_{11}(1/4)^2&=\frac{2683}{5810420}>0,\\
 u_6(3/10)&=\frac{42467}{48841},&
 m(3/10)-u_6(3/10)^2&=-\frac{1213189}{2385443281}<0.
 \end{aligned}
\]
Since $\partial_k\Delta_k(a)=-(1-\sqrt{m(a)})^2<0$, $\Delta_k(a)$ decreases with $k$. For $6\le k\le11$, this gives $\Delta_k(1/4)\ge\Delta_{11}(1/4)>0$ and $\Delta_k(3/10)\le\Delta_6(3/10)<0$, and the strict decrease in $a$ gives a unique root $\tau_k\in(1/4,3/10)$. At that root, $L_{\tau_k}$ is the tangent of $r$ at $\tau_k$ and the supporting line of the floor of slope $m(\tau_k)$, so it majorizes both; since $\tau_k\le3/10$, the balanced bound for pure tangents proved above supplies the remaining inequality.

To write \eqref{eq:halfpurefloor}, whose square root is the positive one, as a polynomial equation, put
\[
 D(t)=(1+t)^2(2-t)^2,\quad N(t)=t^2-8t+6,\quad
 J(t)=t^2(2t^2-6t+7),
\]
so that $m=N/D$ and $j=J/D$ by \eqref{eq:halftangentdata}. Expanding the square, \eqref{eq:halfpurefloor} is $j=k(1+m)-2k\sqrt m$, that is, $2k\sqrt m=k(1+m)-j$. Both sides are positive on the stated interval, so squaring is an equivalence there, and multiplying by $D^2$ gives
\[
 [k(D+N)-J]^2=4k^2ND.
\]
The unsquared quantity $k(D+N)-J=D[k(1+m)-j]$ is positive because $k\ge6$, $m>0$, and $j<r<1$, so squaring introduces no additional root in the stated interval.

\paragraph{Twelve or more labels: a pure tangent through the floor endpoint.}
Let $g(t)=L_t(1)$ be the value at mass one of the pure tangent with contact $t$. Since $j(t)=r(t)-tm(t)$,
\[
 g(t)=m(t)+j(t)=r(t)+(1-t)m(t).
\]
By \eqref{eq:halftangentderivatives}, $g'(t)=m'(t)+j'(t)=(1-t)m'(t)<0$ for $t<1$, and
\[
 g(1/5)=\frac{487}{486}>1,\qquad
 g(1/4)=\frac{226}{245}<\frac{12}{13}.
\]
For $k\ge12$, put $q_k=k/(k+1)\in[12/13,1)$. Since $g$ is continuous and strictly decreasing on $[1/5,1/4]$ with $g(1/5)>1>q_k\ge12/13>g(1/4)$, there is a unique $\tau_k\in(1/5,1/4)$ with $g(\tau_k)=q_k$. Multiplying $g(t)-q_k$ by the positive quantity $(k+1)D(t)$ gives
\[
 (k+1)(1+t)^2(2-t)^2(g(t)-q_k)
 =(k+2)t^4-(4k+6)t^3+(11k+8)t^2-(12k+8)t+2k+6,
\]
which proves the quartic statement \eqref{eq:halfthresholdpoly}.

To compare the slope of this line with the floor slope at one, use $q=g(t)$ as a coordinate on $1/5<t\le1/4$, which is possible because $g$ is strictly decreasing, and write $\widehat m(q)=m(g^{-1}(q))$. The chain rule gives
\[
 \widehat m'(q)=\frac{m'(t)}{g'(t)}=\frac{m'(t)}{(1-t)m'(t)}=\frac1{1-t}.
\]
For $226/245\le q\le q_k$, the corresponding contact satisfies $t\le1/4$, so $1/(1-t)\le4/3$ and
\[
 \frac{d}{dq}(q^2-\widehat m(q))
 =2q-\frac1{1-t}\ge\frac{452}{245}-\frac43=\frac{376}{735}>0.
\]
At the lower endpoint $q=226/245$, which corresponds to $t=1/4$,
\[
 \left(\frac{226}{245}\right)^2-m(1/4)=\frac{116}{60025}>0.
\]
Hence $q^2-\widehat m(q)>0$ on $[226/245,q_k]$; at $q=q_k$ this reads $\mu_k=m(\tau_k)<q_k^2$. The floor has $f_{k,1}(1)=k/(k+1)=q_k$ and $f_{k,1}'(1)=k^2/(k+1)^2=q_k^2$. Concavity of the floor, then $s-1\le0$ with $\mu_k<q_k^2$, then $g(\tau_k)=q_k$ give
\[
 f_{k,1}(s)\le q_k+q_k^2(s-1)
 \le q_k+\mu_k(s-1)=L_{\tau_k}(1)+\mu_k(s-1)=L_{\tau_k}(s),\qquad 0\le s\le1.
\]
Since $0<\tau_k<1/4<3/10$, the balanced bound for pure tangents with contact at most $3/10$ also applies, and $L_{\tau_k}\ge r$ because it is a tangent of $r$; so all three requirements are majorized.

\paragraph{Diagonal levels and their limit.}
All the selected contacts belong to $(1/5,61/200)$: $\tau_*\in(3/10,61/200)$, $\tau_k\in(1/4,3/10)$ for $6\le k\le11$, and $\tau_k\in(1/5,1/4)$ for $k\ge12$. On this interval $0<m(t)<1$: the numerator $N(t)=t^2-8t+6$ of $m$ is positive for $t<4-\sqrt{10}$, and since $m$ is decreasing, $m(t)<m(1/5)=925/972<1$. Also
\[
 r(t)-t=\frac{t(1-3t+t^2)}{(1+t)(2-t)}>0,
\]
because $1-3t+t^2>0$ for $t<(3-\sqrt5)/2$. The tangent line meets the diagonal $L_t(s)=s$ at
\[
 z(t)=\frac{j(t)}{1-m(t)},\qquad\text{so that}\qquad
 L_t(s)=m(t)s+(1-m(t))z(t).
\]
Hence $L_t$ is the target $\chi(s)=\mu s+(1-\mu)\zeta$ with slope $\mu=m(t)$ at level $\zeta=z(t)$, which is the level in \eqref{eq:halfthreshold}; it is positive because $J(t)=t^2(2t^2-6t+7)>0$ (the quadratic factor has negative discriminant) gives $j(t)>0$. Differentiating and using $j'=-tm'$ from \eqref{eq:halftangentderivatives},
\[
 z'(t)=\frac{j'(t)(1-m(t))+j(t)m'(t)}{(1-m(t))^2}
 =\frac{m'(t)(z(t)-t)}{1-m(t)}<0,
\]
where $z(t)>t$ is equivalent to $j(t)>t(1-m(t))$, that is, to $r(t)-tm(t)>t-tm(t)$, which is $r(t)>t$. For contacts $t\ge1/4$, this gives $z(t)\le z(1/4)=18/37<1$. For $k\ge12$, the line has $L_{\tau_k}(1)=m+(1-m)z=q_k<1$, so $z(\tau_k)=(q_k-m(\tau_k))/(1-m(\tau_k))<1$. In every case the parameters $c=1\ge c_0$, $\omega=0$, $\mu=m(\tau_k)\in(0,1)$ at level $z(\tau_k)\in(0,1)$ pass the three tests of \cref{thm:finitepoly}, giving the coefficient $1/(1+c)=1/2$ at this level, and \cref{prop:zetamonotone} extends it to every larger level.

The sign of $z'$ also gives the ordering of the sufficient levels. For $6\le k\le10$, $\Delta_{k+1}<\Delta_k$ pointwise; at the root $\tau_k$ of $\Delta_k$ this gives $\Delta_{k+1}(\tau_k)<0$, and since $\Delta_{k+1}$ is decreasing its root satisfies $\tau_{k+1}<\tau_k$. For $k\ge12$, $q_{k+1}>q_k$ and $g$ is decreasing, so $\tau_{k+1}<\tau_k$. The separating intervals give $\tau_*>3/10>\tau_6$ and $\tau_{11}>1/4>\tau_{12}$. Because $z$ is strictly decreasing throughout these contact intervals, the levels $z(\tau)$ are equal for $k=3,4,5$ (all use $\tau_*$) and increase strictly thereafter.

Finally, multiplying $g(t)-1$ by $D(t)>0$ gives $D(t)(g(t)-1)=t^4-4t^3+11t^2-12t+2$. Let $\tau_\infty$ be the unique root in $(1/5,1/4)$ of
\[
 t^4-4t^3+11t^2-12t+2=0,
\]
equivalently $g(t)=1$; it exists and is unique because $g$ is strictly decreasing with $g(1/5)>1>g(1/4)$. Put $\mu_\infty=m(\tau_\infty)<1$. Since $q_k\to1$ and $g$ is continuous and strictly decreasing, its inverse is continuous, so $\tau_k=g^{-1}(q_k)\to\tau_\infty$ and $\mu_k=m(\tau_k)\to\mu_\infty$. For $k\ge12$, $g(\tau_k)=q_k$ gives $j(\tau_k)=q_k-\mu_k$, so $\zeta_{1/2}(k)=(q_k-\mu_k)/(1-\mu_k)$ and, using $1-q_k=1/(k+1)$,
\[
 1-\zeta_{1/2}(k)=\frac{1-q_k}{1-\mu_k}=\frac1{(k+1)(1-\mu_k)}
 \sim\frac1{(1-\mu_\infty)k},
\]
which proves convergence to one at the rate $1-\Theta(1/k)$.

\FloatBarrier
\section{Proofs for \texorpdfstring{\cref{sec:stochastic}}{the optimization and learning section}: Fractional Optimization and Online Learning}\label{app:stochastic}
The field oracles and their moments (\cref{app:oracles}), the transfer from linear regret (\cref{app:transfer}), and the entropy estimates (\cref{app:entropy}) combine with the coefficients of \cref{cor:linearizable} to prove \cref{thm:continuous,thm:fewquery,cor:localentropy}; the proof of \cref{thm:matroid} also uses them.

\subsection{Field oracles and their moments (\texorpdfstring{\cref{lem:oraclemoments}}{the field-oracle lemma}), used for \texorpdfstring{\cref{thm:continuous,thm:fewquery,cor:localentropy,thm:matroid}}{the fractional, online, local-norm, and matroid bounds}}\label{app:oracles}
Let $p=p(x)$ be the action probabilities defining the prefix gradients in \eqref{eq:gdef}, and let $M_e$ and $G$ be as in \cref{sec:stochastic}; the coefficient matrices, the order, and the prefix probabilities are computed before the current objective is queried. Write $\widehat f(s)$ for the value-oracle response at labeling $s$; it equals $f(s)$ when queries are exact.

Choose $E$ uniformly from $V$. With a first-order oracle, query the gradient once at $\Pi^E$ and return the vector supported on block $E$ given by
\begin{equation}\label{eq:onegradientfield}
 \widehat g^{(\nabla)}_E=nM_E\nabla_EF(\Pi^E).
\end{equation}
With a discrete value oracle, also choose $I$ uniformly from $[k]$ and sample the predecessors of $E$ from $p$, obtaining a labeling $R$. Put $R^+=R+E\mapsto I$. Two responses give
\begin{equation}\label{eq:twoqueryfield}
 \widehat g^{(2)}_{E,\ell}=nkM_{E,\ell I}
       \bigl(\widehat f(R^+)-\widehat f(R)\bigr).
\end{equation}
One response suffices by choosing an independent uniform sign $\beta\in\{-1,1\}$, querying $R^\beta=R^+$ if $\beta=1$ and $R^\beta=R$ otherwise, and returning
\begin{equation}\label{eq:onequeryfield}
 \widehat g^{(1)}_{E,\ell}=2nk\beta M_{E,\ell I}\widehat f(R^\beta).
\end{equation}
All other blocks, and the coordinate indexed by $0$ in the active block, have zero feedback. The exact finite difference in \eqref{eq:expectedmarginal} makes both value estimators unbiased. \Cref{alg:fieldoracle} summarizes the oracle, whose sampling choices do not depend on the current replies.
\begin{algorithm2e}[t]
\DontPrintSemicolon
\KwIn{State $x$, matrices $M_e$, order, action probabilities $p$, and mode $j\in\{\nabla,1,2\}$.}
Draw $E$ uniformly from the $n$ elements; initialize $\widehat g=0$\;
\eIf{$j=\nabla$}{
 Query $\nabla F(\Pi^E)$; set $\widehat g_E=nM_E\nabla_EF(\Pi^E)$\;
}{
 Draw $I$ uniformly from $[k]$; sample the predecessors of $E$ from $p$ to form $R$\;
 Set $R^+=R+E\mapsto I$\;
 \eIf{$j=2$}{
  Query $\widehat f(R)$ and $\widehat f(R^+)$; set $\widehat\Delta=\widehat f(R^+)-\widehat f(R)$\;
 }{
  Draw a uniform sign $\beta$; query $R^\beta$ and set $\widehat\Delta=2\beta\widehat f(R^\beta)$\;
 }
 Set $\widehat g_{E,\ell}=nkM_{E,\ell I}\widehat\Delta$ for every $\ell\in[k]$\;
}
Return $\widehat g$ with zero feedback in every coordinate indexed by $0$\;
\caption{Unbiased oracle for the comparison field}\label{alg:fieldoracle}
\end{algorithm2e}

\begin{algorithm2e}[t]
\DontPrintSemicolon
\KwIn{Current objective $f$, state-dependent action probabilities $p$, coefficient matrices $M_e$, and element order.}
Draw $X\sim p$ independently across elements; set $R=0$ and $v_0=f(0)$\;
\For{each element $e$ in the given order}{
 Query $v_i=f(R+e\mapsto i)$ for every $i\in[k]$\;
 Set $\widehat m_{e,i}=v_i-v_0$ and $\widehat g_{e,\ell}=\sum_iM_{e,\ell i}\widehat m_{e,i}$\;
 \If{$X_e\ne0$}{Set $R\leftarrow R+e\mapsto X_e$ and $v_0\leftarrow v_{X_e}$\;}
 Set $\widehat g_{e,0}=0$\;
}
Return $\widehat g$\;
\caption{Full-prefix oracle with cached chain values}\label{alg:continuous}
\end{algorithm2e}
The full-prefix oracle (\cref{alg:continuous}) samples every element's label once, which determines all the chain prefixes; the value of the next prefix is among the values just queried, so one empty-labeling value initializes the entire scan. The sampled predecessors of element $e$ have the product distribution represented by $\Pi^e$ in \eqref{eq:gdef}. Therefore \cref{lem:oraclemoments} below applies and gives
$|\widehat g_{e,\ell}|\le GL_e$,
$\mathcal V_{\mathrm{full}}=G^2\Lambda^2$, and $\mathcal U_{\mathrm{full}}=kG^2\Lambda^2$.
Offline, $f(0)$ is computed once together with the singleton scales, so each call of \cref{alg:continuous} needs only $nk$ new values; online, $f_t(0)$ is queried in every round, giving $nk+1$ values.
\begin{lemma}[Field-oracle moments]\phantomsection\label[lemma]{lem:oraclemoments}
Condition on the current objective and internal state, before fresh sampling. The estimators above have conditional mean $g_F(x)$. Under $|\Delta_{e,i}f|\le L_e$, $L=\max_eL_e$, and the conditional variance assumption of \cref{sec:stochastic}, their second moments satisfy
\begin{equation}
 \E\sum_e\|\widehat g^{(j)}_e\|_\infty^2\le \mathcal V_j,
 \qquad \E\|\widehat g^{(j)}\|_2^2\le \mathcal U_j,
\end{equation}
with
\[
\begin{array}{c|c|c}
\text{Oracle} & \mathcal V_j & \mathcal U_j\\ \hline
\text{Full prefix, exact values} & G^2\Lambda^2 & kG^2\Lambda^2\\
\text{One exact gradient of }F & nG^2\Lambda^2 & nkG^2\Lambda^2\\
\text{Two noisy values} & n^2k^2G^2(L^2+2\sigma^2) & nk^2G^2(\Lambda^2+2n\sigma^2)\\
\text{One noisy value} &4n^2k^2G^2(f_{\max}^2+\sigma^2)&4n^2k^2G^2(f_{\max}^2+\sigma^2)
\end{array}
\]
The last row assumes $0\le f\le f_{\max}$ on the full labeling domain; the gradient and two-value rows require only the marginal bounds. With exact values, complete prefix feedback uses $nk$ new values offline, or $nk+1$ per online round, and has the full-prefix moments in the table.
\end{lemma}
\begin{proof}
Throughout, $e$ denotes a fixed element and $E$ the sampled one. By \eqref{eq:gdef},
\[
 g_{e,\ell}=\sum_jM_{e,\ell j}\,\partial_{e,j}F(\Pi^e).
\]
 The bounds use three facts. First, by \cref{lem:marginal}, every marginal $\Delta_{e,j}f(s)$ with $s_e=0$ and every derivative $\partial_{e,j}F$ lies in $[-L_e,L_e]$. Second, every row of $M_e$ has $\ell_1$-norm at most $G$, so $|M_{e,\ell j}|\le G$ and $\|M_e\|_F^2=\sum_\ell\sum_jM_{e,\ell j}^2\le\sum_\ell(\sum_j|M_{e,\ell j}|)^2\le kG^2$. Third, the sampled predecessor labeling $R$ assigns independent labels with laws $p_u$ to the elements $u$ preceding $E$ and leaves all other elements, including $E$, unselected. Its law is that of $X_{-E}$ under $\Pi^E$, so \eqref{eq:expectedmarginal} gives
\begin{equation}\label{eq:prefixmean}
 \E[\Delta_{E,j}f(R)\mid E]=\partial_{E,j}F(\Pi^E).
\end{equation}

\emph{Conditional means.} For one gradient, block $e$ of $\widehat g^{(\nabla)}$ equals $nM_e\nabla_eF(\Pi^e)$ on the event $\{E=e\}$, which has probability $1/n$, and is zero otherwise; hence block $e$ of $\E\widehat g^{(\nabla)}$ equals $n^{-1}\cdot nM_e\nabla_eF(\Pi^e)=g_{F,e}$. For two values, write $\widehat f(R^+)-\widehat f(R)=\Delta_{E,I}f(R)+\xi$, where $\xi$ is the difference of the two centered noise terms, so $\E[\xi\mid E,I,R]=0$. Coordinate $(e,\ell)$ of $\widehat g^{(2)}$ equals $\mathbf 1\{E=e\}\,nkM_{e,\ell I}(\Delta_{e,I}f(R)+\xi)$. Conditioning on $E=e$ and $I=j$ and averaging over the noise gives $nkM_{e,\ell j}\Delta_{e,j}f(R)$; averaging over $R$ with \eqref{eq:prefixmean} gives $nkM_{e,\ell j}\partial_{e,j}F(\Pi^e)$; averaging over the uniform label $I$ and then over the uniform element $E$ gives
\[
 \E\widehat g^{(2)}_{e,\ell}=\frac1n\sum_{j=1}^k\frac1k\,nkM_{e,\ell j}\,\partial_{e,j}F(\Pi^e)=\sum_{j=1}^kM_{e,\ell j}\,\partial_{e,j}F(\Pi^e)=g_{e,\ell}.
\]
For one value, the sign $\beta$ is independent of $(E,I,R)$ and equals $1$ or $-1$ with probability $\tfrac12$ each. Since the noise is centered, $\E[\widehat f(R^\beta)\mid E,I,R,\beta]=f(R^\beta)$, and therefore
\[
 \E[2\beta\widehat f(R^\beta)\mid E,I,R]=\tfrac12\cdot2f(R^+)-\tfrac12\cdot2f(R)=f(R^+)-f(R)=\Delta_{E,I}f(R).
\]
This is the conditional mean of the two-value difference, so the remaining averages over $R$, $I$ and $E$ are the same as for two values. The full-prefix oracle of \cref{alg:continuous} uses, for every $e$, the sampled prefix before $e$, which has the law in \eqref{eq:prefixmean}; hence each of its blocks has mean $g_{F,e}$.

\emph{Second moments of the scalar responses.} For two values, let $\xi_1,\xi_2$ be the noise in the first and second responses, with signs chosen so that the queried difference is $\Delta_{E,I}f(R)+\xi_2-\xi_1$. Conditional on $(E,I,R)$, the noise satisfies $\E\xi_i=0$ and $\E\xi_i^2\le\sigma^2$, and $\E[\xi_1\xi_2]=\E[\xi_1\E(\xi_2\mid\xi_1)]=0$ because the second response is centered conditional on the first. Expanding the square, the three cross terms $2\Delta_{E,I}f(R)\E\xi_2$, $-2\Delta_{E,I}f(R)\E\xi_1$ and $-2\E[\xi_1\xi_2]$ vanish, so
\[
 \E\bigl[(\Delta_{E,I}f(R)+\xi_2-\xi_1)^2\mid E,I,R\bigr]
 =\Delta_{E,I}f(R)^2+\E\xi_1^2+\E\xi_2^2\le L_E^2+2\sigma^2,
\]
where $|\Delta_{E,I}f(R)|\le L_E$ by \cref{lem:marginal}, because $R$ leaves $E$ unselected. For one value, $(2\beta\widehat f(R^\beta))^2=4\widehat f(R^\beta)^2$ since $\beta^2=1$, and $\E[\widehat f(R^\beta)^2\mid R^\beta]=f(R^\beta)^2+\operatorname{Var}(\widehat f(R^\beta)\mid R^\beta)\le f_{\max}^2+\sigma^2$, using $0\le f\le f_{\max}$ and the variance bound; hence the conditional second moment of $2\beta\widehat f(R^\beta)$ is at most $4(f_{\max}^2+\sigma^2)$.

\emph{Block-maximum bounds $\mathcal V_j$.} Only block $E$ is nonzero for the sampled oracles, so $\sum_e\|\widehat g_e\|_\infty^2=\|\widehat g_E\|_\infty^2$. For one gradient, the first two facts give
\[
 \|\widehat g^{(\nabla)}_E\|_\infty=n\max_\ell\Bigl|\sum_jM_{E,\ell j}\partial_{E,j}F(\Pi^E)\Bigr|\le n\max_\ell\sum_j|M_{E,\ell j}|\,L_E\le nGL_E,
\]
so, averaging over the uniform $E$, $\E\sum_e\|\widehat g_e\|_\infty^2\le n^{-1}\sum_en^2G^2L_e^2=nG^2\Lambda^2$. For two values, every entry of $\widehat g^{(2)}_E$ is $nkM_{E,\ell I}(\widehat f(R^+)-\widehat f(R))$ and $|M_{E,\ell I}|\le G$, so $\|\widehat g^{(2)}_E\|_\infty^2\le n^2k^2G^2(\widehat f(R^+)-\widehat f(R))^2$; by the scalar bound and $L_E\le L$, its mean is at most $n^2k^2G^2(L^2+2\sigma^2)$. For one value, $\|\widehat g^{(1)}_E\|_\infty^2\le n^2k^2G^2(2\beta\widehat f(R^\beta))^2$, and the scalar bound $4(f_{\max}^2+\sigma^2)$ gives $4n^2k^2G^2(f_{\max}^2+\sigma^2)$. For exact full-prefix feedback, $\widehat m_{e,i}=f(R+e\mapsto i)-f(R)=\Delta_{e,i}f(R)$ with $R_e=0$, so $|\widehat m_{e,i}|\le L_e$ by \cref{lem:marginal}; hence every block satisfies $|\widehat g_{e,\ell}|\le\sum_i|M_{e,\ell i}||\widehat m_{e,i}|\le GL_e$ on every outcome, giving $\sum_e\|\widehat g_e\|_\infty^2\le G^2\sum_eL_e^2=G^2\Lambda^2$.

\emph{Euclidean bounds $\mathcal U_j$.} For one gradient, block $E$ has $k$ coordinates, each at most $nGL_E$ in absolute value by the preceding display, so $\|\widehat g^{(\nabla)}\|_2^2\le k(nGL_E)^2$, whose mean over the uniform $E$ is $n^{-1}\sum_ekn^2G^2L_e^2=nkG^2\Lambda^2$. For two values, average the sampled column before bounding it. Conditional on $E$ and $I$, the returned vector is $nk$ times column $I$ of $M_E$, $M_{E,:,I}$, multiplied by a scalar with second moment at most $L_E^2+2\sigma^2$, so $\E[\|\widehat g^{(2)}\|_2^2\mid E,I]\le n^2k^2\|M_{E,:,I}\|_2^2(L_E^2+2\sigma^2)$. Averaging the squared column norm over the uniform $I$ gives $k^{-1}\sum_j\|M_{E,:,j}\|_2^2=\|M_E\|_F^2/k$, and averaging over $E$ and using $\|M_e\|_F^2\le kG^2$ gives
\[
 \E\|\widehat g^{(2)}\|_2^2
 \le n^2k^2\cdot\frac1n\sum_e\frac{\|M_e\|_F^2}{k}(L_e^2+2\sigma^2)
 \le nk\sum_e kG^2(L_e^2+2\sigma^2)
 =nk^2G^2(\Lambda^2+2n\sigma^2).
\]
The same calculation with the scalar bound $4(f_{\max}^2+\sigma^2)$ in place of $L_e^2+2\sigma^2$ gives
\[
 \begin{aligned}
 \E\|\widehat g^{(1)}\|_2^2&\le n^2k^2\cdot\frac1n\sum_e\frac{\|M_e\|_F^2}k\cdot4(f_{\max}^2+\sigma^2)\\
 &=nk\cdot4(f_{\max}^2+\sigma^2)\sum_e\|M_e\|_F^2\le4n^2k^2G^2(f_{\max}^2+\sigma^2),
 \end{aligned}
\]
because $\sum_e\|M_e\|_F^2\le nkG^2$. For full-prefix feedback, each of the $k$ coordinates of block $e$ is at most $GL_e$ in absolute value, so $\|\widehat g\|_2^2\le\sum_ekG^2L_e^2=kG^2\Lambda^2$.

\emph{Query counts.} \Cref{alg:continuous} queries the $k$ extensions of the current prefix at each element; the value of the next prefix is one of these values, so after $f(0)$ is known each scan uses $nk$ new values. Offline, $f(0)$ is evaluated once during initialization; online it is queried in every round.
\end{proof}
Sampling the pair $(e,j)$ nonuniformly can reduce the Euclidean moment. Put $a_{e,j}=\|M_{e,:,j}\|_2$, let $\varsigma_e$ bound the root second moment of the scalar response, and define $b_{e,j}=a_{e,j}\varsigma_e$ and $S=\sum_{e,j}b_{e,j}$. If $S=0$, return the zero field without a query. Otherwise sample only pairs with $b_{e,j}>0$, with probability $\pi_{e,j}=b_{e,j}/S$, and return the sampled column multiplied by the scalar response and divided by $\pi_{e,j}$. A pair with $b_{e,j}=0$ contributes zero almost surely: either its column is zero or its scalar response has zero second moment. Thus omitting these pairs preserves unbiasedness. The Euclidean moment is at most
\[
 \sum_{b_{e,j}>0}\frac{b_{e,j}^2}{\pi_{e,j}}=S^2.
\]
For any other probability vector positive on this support, Cauchy--Schwarz gives $S^2\le\sum b_{e,j}^2/\pi_{e,j}$. For one value use $\varsigma_e=2\sqrt{f_{\max}^2+\sigma^2}$; for two values use $\varsigma_e=\sqrt{L_e^2+2\sigma^2}$.

\subsection{Transfer from linear regret (\texorpdfstring{\cref{cor:transfer}}{the linear-regret transfer}), used for \texorpdfstring{\cref{thm:continuous,thm:fewquery,cor:localentropy,thm:matroid}}{the fractional, online, local-norm, and matroid bounds}}\label{app:transfer}
Feeding the field estimates to an online linear optimizer is the upper-linearizable reduction of \citet{PedramfarAggarwal24}; \cref{cor:transfer} states it with the initial-value term retained.
\begin{corollary}[Transfer from linear regret]\phantomsection\label[corollary]{cor:transfer}
Suppose \eqref{eq:linearizable} holds and an online linear learner, supplied with a conditionally unbiased field estimate, has expected linear regret at most $R_T$ against each fixed feasible comparator. For an oblivious objective sequence,
\begin{equation}\label{eq:mirrorreduction}
 \operatorname{Reg}_{\alpha}(T)+\alpha C\sum_{t=1}^Tf_t(0)\le\alpha R_T.
\end{equation}
For a single objective, let $J$ be an independent uniform index in $[T]$. For every fixed $y\in\K(P)$, returning the action $p_J$ gives
$\E F(p_J)\ge\alpha F(y)+\alpha Cf(0)-\alpha R_T/T$.
Choosing a maximizer of $F$ gives the same bound with $F(y)=\mathrm{OPT}_F(P)$.
\end{corollary}
\begin{proof}
Let $\mathcal H_t$ denote the history before the fresh sampling of round $t$, including the state $x_t$; for an oblivious sequence, $f_t$ is deterministic. Fix a maximizer $y^\star$ of $\sum_tF_t$ over the compact set $\K(P)$; it is deterministic because the sequence is fixed in advance. Applying \eqref{eq:linearizable} to $F_t$ at $x=x_t$ and $y=y^\star$ and summing over $t$,
\begin{equation}\label{eq:transfersum}
 \alpha\sum_{t=1}^TF_t(y^\star)+\alpha C\sum_{t=1}^Tf_t(0)-\sum_{t=1}^TF_t(p_t)
 \le\alpha\sum_{t=1}^T\ip{g_{F_t}(x_t)}{y^\star-x_t}.
\end{equation}
Since $x_t$ and $y^\star$ are $\mathcal H_t$-measurable and $\E[\widehat g_t\mid\mathcal H_t]=g_{F_t}(x_t)$, the tower property gives
\[
 \E\ip{\widehat g_t}{y^\star-x_t}=\E\bigl[\ip{\E[\widehat g_t\mid\mathcal H_t]}{y^\star-x_t}\bigr]=\E\ip{g_{F_t}(x_t)}{y^\star-x_t}.
\]
The learner receives $\widehat g_t$, so the expected linear regret bound gives $\E\sum_t\ip{\widehat g_t}{y^\star-x_t}\le R_T$. Taking expectations in \eqref{eq:transfersum} therefore gives
\[
 \alpha\sum_{t=1}^TF_t(y^\star)+\alpha C\sum_{t=1}^Tf_t(0)-\E\sum_{t=1}^TF_t(p_t)\le\alpha R_T.
\]
For an oblivious sequence the maximum in \eqref{eq:regretdefinition} is deterministic and attained at $y^\star$, so the first and third terms on the left equal $\operatorname{Reg}_\alpha(T)$; this proves \eqref{eq:mirrorreduction}.

For a single objective, apply \eqref{eq:linearizable} with any fixed $y\in\K(P)$ and sum over $t$. Since $f_t=f$, the left side becomes $\alpha TF(y)+\alpha CTf(0)-\sum_tF(p_t)$; taking expectations, bounding the right side by $\alpha R_T$ exactly as above, and dividing by $T$ gives
\[
 \frac1T\sum_{t=1}^T\E F(p_t)\ge\alpha F(y)+\alpha Cf(0)-\alpha R_T/T.
\]
Since $J$ is uniform on $[T]$ and independent of the run, $\E F(p_J)=T^{-1}\sum_t\E F(p_t)$.
\end{proof}
\Cref{alg:linearizable} implements this reduction for each feedback model: the learner updates $x_t$ from the sampled field and plays $p(x_t)$.

\begin{algorithm2e}[t]
\DontPrintSemicolon
\KwIn{Feasible region $\K(P)$, an admissible triple in \cref{cor:linearizable}, a linear learner, and a field oracle.}
Initialize the linear learner and its state $x_1$ independently of the objective\;
\For{$t=1,\ldots,T$}{
 Compute the action $p_t=p(x_t)$, coefficient matrices, and field-prefix probabilities\;
 Commit $p_t$, or an objective-independent feasible rounding of $p_t$, for reward\;
 Obtain $\widehat g_t$ using \eqref{eq:onegradientfield}, \eqref{eq:twoqueryfield}, \eqref{eq:onequeryfield}, or the full-prefix oracle\;
 Feed $\widehat g_t$ to the online linear learner to obtain $x_{t+1}$\;
}
For offline optimization, use a constant objective and return an independent uniformly chosen action\;
\caption{Linear optimization with the comparison-field oracle}\label{alg:linearizable}
\end{algorithm2e}

\subsection{Entropy estimates (\texorpdfstring{\cref{lem:entropy}}{the entropy estimates}), used for \texorpdfstring{\cref{thm:continuous,thm:fewquery,cor:localentropy,thm:matroid}}{the fractional, online, local-norm, and matroid bounds}}\label{app:entropy}
Let $\overline\K(P)$ be the categorical lift of $\K(P)$, with $q_{e,0}=1-|x_e|$ and $q_{e,i}=x_{e,i}$ for $i\ge1$. Remove coordinates identically zero on this region. Standard entropy mirror-map properties give existence and uniqueness of the entropy minimizer and of each update, with positive values on all remaining coordinates; the Bregman projection and its variational inequality apply on this relative face \citep[Sections 4.1--4.3]{Bubeck15}. \Cref{lem:entropy} combines the standard pathwise estimate with a support-sensitive one.
\begin{lemma}[Entropy estimates]\phantomsection\label[lemma]{lem:entropy}
Let $a_1,\ldots,a_T$ be categorical feedback vectors. For $q_1=\argmin_{\overline\K(P)}H(q)$ and
$q_{t+1}=\argmax_{q\in\overline\K(P)}\{\eta\langle a_t,q\rangle-D_H(q,q_t)\}$, every realization and every feasible categorical comparator $q$ satisfy
\begin{equation}\label{eq:entropygeneral}
 \sum_{t=1}^T\langle a_t,q-q_t\rangle
 \le\frac{n\log(k+1)}\eta+\frac\eta2\sum_{t,e}\|a_{t,e}\|_\infty^2.
\end{equation}
If every vector of $P$ has total mass at most $r\in[1,n]$, put $B_r=\min\{n\log(k+1),r\log(\mathrm e nk/r)\}$. If also $a_{t,e,0}=0$ and $\eta|a_{t,e,i}|\le1$ for all $t$, $e$, and $i$, then
\begin{equation}\label{eq:entropylocal}
 \sum_{t=1}^T\langle a_t,q-q_t\rangle
 \le\frac{B_r}\eta+\eta\sum_{t=1}^T\sum_{e\in V}\sum_{i=1}^kq_{t,e,i}a_{t,e,i}^2.
\end{equation}
In particular, if $a_{t,e,0}=0$, $|a_{t,e,i}|\le GL$, and $\eta GL\le1$, then
\begin{equation}\label{eq:entropysupportregret}
 \sum_{t=1}^T\langle a_t,q-q_t\rangle
 \le B_r/\eta+\eta TG^2rL^2.
\end{equation}
A feasible update with variational residual $\varrho_t$ adds $\sum_t\varrho_t/\eta$; an initialization residual $\varrho_0$ adds $\varrho_0/\eta$.
\end{lemma}
\begin{proof}
Each block of a lifted point is a probability vector on $\{0,1,\ldots,k\}$, so $D_H(q,q')=\sum_e\mathrm{KL}(q_e\|q'_e)$ is a sum of blockwise relative entropies, and $\nabla H(q)=(\log q_{e,i}+1)_{e,i}$ on the remaining coordinates.

\emph{Initialization.} The first-order optimality condition of $q_1=\argmin_{\overline\K(P)}H$ is $\ip{\nabla H(q_1)}{q-q_1}\ge0$ for every feasible $q$. Since $D_H(q,q_1)=H(q)-H(q_1)-\ip{\nabla H(q_1)}{q-q_1}$,
\begin{equation}\label{eq:entropyinitial}
 D_H(q,q_1)\le H(q)-H(q_1)\le-H(q_1)\le n\log(k+1).
\end{equation}
The first inequality uses the optimality condition, which makes the subtracted inner product nonnegative. The second holds because every coordinate lies in $[0,1]$, so $q_{e,i}\log q_{e,i}\le0$ and $H(q)\le0$. For the last, $-H(q_1)=\sum_e\bigl(-\sum_iq_{1,e,i}\log q_{1,e,i}\bigr)$ is a sum of $n$ entropies of probability vectors on at most $k+1$ points, each at most $\log(k+1)$.

\emph{General pathwise estimate.} By the blockwise Pinsker inequality, $D_H(q,q')=\sum_e\mathrm{KL}(q_e\|q'_e)\ge\frac12\sum_e\|q_e-q'_e\|_1^2$, so $H$ is $1$-strongly convex on $\overline\K(P)$ with respect to the norm $\|v\|=(\sum_e\|v_e\|_1^2)^{1/2}$, whose dual norm is $\|a\|_*=(\sum_e\|a_e\|_\infty^2)^{1/2}$. The proof of \citet[Theorem~4.2]{Bubeck15} uses the loss vectors only through their dual norms and therefore gives, for arbitrary vectors $a_1,\ldots,a_T$ and every feasible $q$,
\[
 \sum_{t=1}^T\ip{a_t}{q-q_t}\le\frac{D_H(q,q_1)}\eta+\frac\eta2\sum_{t=1}^T\|a_t\|_*^2.
\]
Substituting \eqref{eq:entropyinitial} and $\|a_t\|_*^2=\sum_e\|a_{t,e}\|_\infty^2$ proves \eqref{eq:entropygeneral}. The bound holds for each realized sequence $a_1,\ldots,a_T$, without an almost-sure bound on the feedback.

\emph{Support-sensitive diameter.} Let $q$ be feasible with selected mass $m=\sum_{e,i\ge1}q_{e,i}\le r$; if $m=0$, the selected terms below vanish. For $m>0$, the vector $\pi=(q_{e,i}/m)_{e,i\ge1}$ is a probability vector on $nk$ points, so its entropy is at most $\log(nk)$, and
\[
 \begin{aligned}
 -\sum_{e,i\ge1}q_{e,i}\log q_{e,i}&=-\sum_{e,i\ge1}m\pi_{e,i}\bigl(\log m+\log\pi_{e,i}\bigr)\\
 &=-m\log m+m\Bigl(-\sum_{e,i\ge1}\pi_{e,i}\log\pi_{e,i}\Bigr)\le m\log(nk/m).
 \end{aligned}
\]
Also $-(1-z)\log(1-z)\le z$ for $0\le z\le1$: at $z=1$ the left side is $0$, and for $z<1$, $\log x\le x-1$ with $x=1/(1-z)$ gives $-\log(1-z)\le z/(1-z)$, which we multiply by $1-z\ge0$. Applying this with $z=|x_e|$, so that $1-z=q_{e,0}$, and summing over blocks gives $-\sum_eq_{e,0}\log q_{e,0}\le\sum_e|x_e|=m$. Thus $-H(q)\le m\log(nk/m)+m=m\log(\mathrm enk/m)$. The function $m\mapsto m\log(\mathrm enk/m)$ has derivative $\log(nk/m)\ge0$ for $0<m\le nk$, and $m\le r\le n\le nk$, so $-H(q)\le r\log(\mathrm enk/r)$. Applying this to the feasible point $q_1$ in \eqref{eq:entropyinitial}, together with the bound $n\log(k+1)$ already there, gives $D_H(q,q_1)\le B_r$.

\emph{Support-sensitive step.} Let $\widetilde q_{e,i}=q_{t,e,i}\exp(\eta a_{t,e,i})/Z_e$ with $Z_e=\sum_iq_{t,e,i}\exp(\eta a_{t,e,i})$. For every lifted $q$, $\log(q_{e,i}/\widetilde q_{e,i})=\log(q_{e,i}/q_{t,e,i})-\eta a_{t,e,i}+\log Z_e$; multiplying by $q_{e,i}$, summing, and using $\sum_iq_{e,i}=1$ in each block gives $D_H(q,\widetilde q)=D_H(q,q_t)-\eta\ip{a_t}{q}+\sum_e\log Z_e$, that is,
\[
 \eta\ip{a_t}{q}-D_H(q,q_t)=\sum_e\log Z_e-D_H(q,\widetilde q).
\]
The first term on the right does not depend on $q$, so $q_{t+1}$, the maximizer of the left side over $\overline\K(P)$, is the relative-entropy projection of $\widetilde q$ onto $\overline\K(P)$. The generalized Pythagorean inequality for this projection gives $D_H(q,\widetilde q)\ge D_H(q,q_{t+1})+D_H(q_{t+1},\widetilde q)\ge D_H(q,q_{t+1})$ for every feasible $q$. Substituting this into the identity and subtracting $\eta\ip{a_t}{q_t}$ from both sides gives
\begin{equation}\label{eq:entropyonestep}
 \eta\langle a_t,q-q_t\rangle
 \le D_H(q,q_t)-D_H(q,q_{t+1})+
 \sum_e\log Z_e-\eta\langle a_t,q_t\rangle.
\end{equation}
Assume $|\eta a_{t,e,i}|\le1$ for every coordinate. Then $\mathrm e^v\le1+v+v^2$ for $|v|\le1$ gives, since $\sum_iq_{t,e,i}=1$,
\[
 Z_e\le\sum_iq_{t,e,i}\bigl(1+\eta a_{t,e,i}+\eta^2a_{t,e,i}^2\bigr)=1+\eta\ip{a_{t,e}}{q_{t,e}}+\eta^2\sum_iq_{t,e,i}a_{t,e,i}^2.
\]
Using $\log x\le x-1$ for each block and summing over $e$,
\[
 \sum_e\log Z_e-\eta\ip{a_t}{q_t}\le\eta^2\sum_{e}\sum_{i=0}^kq_{t,e,i}a_{t,e,i}^2=\eta^2\sum_{e}\sum_{i=1}^kq_{t,e,i}a_{t,e,i}^2,
\]
where the equality uses $a_{t,e,0}=0$, so the unselected probabilities $q_{t,e,0}$ carry no squared feedback. Summing \eqref{eq:entropyonestep} over $t=1,\ldots,T$, the divergence terms telescope to $D_H(q,q_1)-D_H(q,q_{T+1})\le D_H(q,q_1)\le B_r$, so $\eta\sum_t\ip{a_t}{q-q_t}\le B_r+\eta^2\sum_t\sum_{e,i\ge1}q_{t,e,i}a_{t,e,i}^2$; dividing by $\eta$ proves \eqref{eq:entropylocal}. If $|a_{t,e,i}|\le GL$ and $\eta GL\le1$, then $|\eta a_{t,e,i}|\le1$, and
\[
 \sum_{e,i\ge1}q_{t,e,i}a_{t,e,i}^2\le G^2L^2\sum_{e,i\ge1}q_{t,e,i}\le G^2rL^2,
\]
because the feasible point $q_t$ has selected mass at most $r$; substituting this into \eqref{eq:entropylocal} for each $t$ proves \eqref{eq:entropysupportregret}.

\emph{Inexact steps.} Say that a feasible update has variational residual $\varrho_t\ge0$ if its optimality condition holds up to an additive $\varrho_t$, that is, $\ip{\eta a_t-\nabla H(q_{t+1})+\nabla H(q_t)}{q-q_{t+1}}\le\varrho_t$ for every feasible $q$; the initialization has residual $\varrho_0\ge0$ if $\ip{\nabla H(q_1)}{q-q_1}\ge-\varrho_0$ for every feasible $q$. In each argument above, the exact optimality condition enters exactly once per step: through the first-order condition used in \citet[Theorem~4.2]{Bubeck15} and through the generalized Pythagorean inequality, which is derived from the same condition. Replacing it by its inexact form adds $\varrho_t$ to the right side of the step-$t$ inequality, and replacing the optimality condition of $q_1$ adds $\varrho_0$ to the first inequality in \eqref{eq:entropyinitial}. Summing over $t$ and dividing by $\eta$ adds $\sum_t\varrho_t/\eta$ and $\varrho_0/\eta$ to the bounds, which proves the error assertion.
\end{proof}

The entropy projection reduces to a separable convex problem over support masses. Let $\widetilde q$ be the exponentiated point from the proof of \cref{lem:entropy} and put $b_e=\sum_{i\ge1}\widetilde q_{e,i}$. Feasibility depends only on the masses $z_e=\sum_{i\ge1}q_{e,i}$, and the log-sum inequality $\sum_{i\ge1}q_{e,i}\log(q_{e,i}/\widetilde q_{e,i})\ge z_e\log(z_e/b_e)$, with equality exactly when $q_{e,i}/\widetilde q_{e,i}$ does not depend on $i$, shows that for fixed $z_e$ $\sum_{i\ge1}q_{e,i}\log(q_{e,i}/\widetilde q_{e,i})$, subject to $\sum_{i\ge1}q_{e,i}=z_e$, is minimized by $q_{e,i}=z_e\widetilde q_{e,i}/b_e$, with minimum value $z_e\log(z_e/b_e)$. The projection therefore preserves the proportions of the nonzero labels and reduces to
\begin{equation}
 \min_{z\in P}\sum_e\left[z_e\log\frac{z_e}{b_e}
 +(1-z_e)\log\frac{1-z_e}{\widetilde q_{e,0}}\right].
\end{equation}
Under the limiting convention $0\log0=0$, a zero reference coordinate fixes the corresponding mass, and the accuracy of an approximate feasible solution is measured by the variational residual of \cref{lem:entropy}.

\subsection{Proof of \texorpdfstring{\cref{thm:continuous}}{the fractional optimization theorem} (fractional optimization)}\label{app:continuousproof}
\begin{proof}[of \cref{thm:continuous}]
Run \cref{alg:linearizable} with the constant objective $f_t=f$ and entropy mirror ascent. Because the field has zero coordinate $0$, the lifted pairing satisfies $\ip{\widehat g_t}{q-q_t}=\ip{\widehat g_t}{y-x_t}$ for the lifts $q,q_t$ of $y,x_t$. The full-prefix estimates satisfy $\sum_e\|\widehat g_{t,e}\|_\infty^2\le G^2\Lambda^2$ on every outcome by \cref{lem:oraclemoments}, so the pathwise estimate \eqref{eq:entropygeneral} gives linear regret at most $n\log(k+1)/\eta+\eta TG^2\Lambda^2/2$ against every fixed $y$. For the stated $\eta$, each of the two terms equals $\sqrt{n\log(k+1)TG^2\Lambda^2/2}$, so
\[
 \frac{n\log(k+1)}\eta+\frac{\eta}{2}TG^2\Lambda^2
 =G\sqrt{2n\log(k+1)\Lambda^2T}
 \qquad\text{for }\eta=\sqrt{\frac{2n\log(k+1)}{TG^2\Lambda^2}}.
\]
The constant-objective part of \cref{cor:transfer} with $R_T=G\sqrt{2n\log(k+1)\Lambda^2T}$, applied to the fixed comparator $y$, gives $\E F(p_J)\ge\alpha F(y)+\alpha Cf(0)-\alpha R_T/T$, and $R_T/T=G\sqrt{2n\Lambda^2\log(k+1)/T}$; this is \eqref{eq:continuous}. Every state $x_t$ lies in $\K(P)$ and the action preserves $\K(P)$ by \cref{cor:linearizable}, so $p_J\in\K(P)$.

For the horizon, $\Lambda^2=\sum_eL_e^2\le nL^2$, so the error term in \eqref{eq:continuous} is at most
\[
 \alpha G\sqrt{2n\cdot nL^2\log(k+1)/T}=\alpha GL\sqrt{2n^2\log(k+1)/T}.
\]
This is at most $\eps L$ exactly when $\alpha^2G^2\cdot2n^2\log(k+1)/T\le\eps^2$, that is,
\[
 \alpha GL\sqrt{2n^2\log(k+1)/T}\le\eps L\qquad\text{whenever }T\ge2\alpha^2G^2n^2\log(k+1)/\eps^2,
\]
and the horizon in \cref{thm:continuous} satisfies this inequality.
 Taking $y$ to be a maximizer of the continuous function $F$ on the compact set $\K(P)$ gives the displayed bound with $\OPT_F(P)$. Initialization evaluates $f(0)$ and the $nk$ singletons, and each iteration uses $nk$ new values by \cref{lem:oraclemoments}; this gives at most $nk+1+Tnk$ values.

\emph{Euclidean version.} Let $\mathsf D$ be the Euclidean diameter of $\K(P)$. Every $x,y\in\D$ satisfy $\|x_e-y_e\|_2^2=\|x_e\|_2^2+\|y_e\|_2^2-2\ip{x_e}{y_e}\le\|x_e\|_2^2+\|y_e\|_2^2\le|x_e|^2+|y_e|^2\le2$, because the coordinates are nonnegative, so the inner product is nonnegative and $\|x_e\|_2\le\|x_e\|_1=|x_e|\le1$; summing over the $n$ blocks gives $\mathsf D^2\le2n$. The full-prefix estimates satisfy $\|\widehat g_t\|_2^2\le \mathcal U_{\mathrm{full}}=kG^2\Lambda^2$ on every outcome. With $\mathsf D^2\le2n$ and $\eta=\sqrt{2n/(T\mathcal U_{\mathrm{full}})}=\sqrt{2n/(TkG^2\Lambda^2)}$, the projected-ascent estimate \eqref{eq:projectedregret} of \cref{app:fewqueryproof} bounds the linear regret by
\[
 \frac{\mathsf D^2}{2\eta}+\frac{\eta T\mathcal U_{\mathrm{full}}}2\le\frac n\eta+\frac{\eta T\mathcal U_{\mathrm{full}}}2=2\sqrt{\frac{nT\mathcal U_{\mathrm{full}}}2}=\sqrt{2n}\sqrt{T\mathcal U_{\mathrm{full}}}=G\sqrt{2nkT\Lambda^2}.
\] The constant-objective part of \cref{cor:transfer} then gives the error $\alpha G\sqrt{2nk\Lambda^2/T}$, with $nk$ new values per iteration as before. Since $\Lambda^2\le nL^2$, this error is at most $\alpha GL\sqrt{2n^2k/T}$, which is at most $\eps L$ once $T\ge2\alpha^2G^2n^2k/\eps^2$.
\end{proof}

\subsection{Proof of \texorpdfstring{\cref{thm:fewquery}}{the online-learning theorem} (online learning under limited feedback)}\label{app:fewqueryproof}
The proof uses only the moments of the supplied field estimates (\cref{lem:oraclemoments}), so it also covers the full-prefix oracle.

\begin{proof}[of \cref{thm:fewquery}]
Let $\mathcal H_t$ contain the past, the current state $x_t$, and the current objective $f_t$, before the round's fresh sampling.
\emph{Oblivious sequences.} Fix a maximizer $y^\star$ of $\sum_tF_t$ on $\K(P)$; it is deterministic. Let $q^\star$ and $q_t$ be the lifts of $y^\star$ and $x_t$. Since $\widehat g_{t,e,0}=0$, $\ip{\widehat g_t}{q^\star-q_t}=\ip{\widehat g_t}{y^\star-x_t}$. The pathwise estimate \eqref{eq:entropygeneral} with $a_t=\widehat g_t$ gives
\[
 \sum_t\ip{\widehat g_t}{y^\star-x_t}\le\frac{B_0}\eta+\frac\eta2\sum_{t,e}\|\widehat g_{t,e}\|_\infty^2.
\]
Taking expectations and using the conditional bound $\E[\sum_e\|\widehat g_{t,e}\|_\infty^2\mid\mathcal H_t]\le \mathcal V_j$ of \cref{lem:oraclemoments} with the tower property, so that $\E\sum_{t,e}\|\widehat g_{t,e}\|_\infty^2\le T\mathcal V_j$, the expected linear regret is at most $B_0/\eta+\eta T\mathcal V_j/2$. With $\eta=\sqrt{2B_0/(T\mathcal V_j)}$, each of the two terms equals $\sqrt{B_0T\mathcal V_j/2}$, so their sum is $\sqrt{2B_0\mathcal V_jT}$. \Cref{cor:transfer} with $R_T=\sqrt{2B_0\mathcal V_jT}$ proves \eqref{eq:fewqueryregret}.

\emph{Euclidean regret for every comparator.} Let $x_{t+1}$ be the Euclidean projection of $x_t+\eta\widehat g_t$ onto $\K(P)$. The standard analysis of online projected gradient ascent \citep{Zinkevich03} gives, for every realized sequence and simultaneously for every $y\in\K(P)$,
\begin{equation}\label{eq:projectedregret}
 \sum_t\langle\widehat g_t,y-x_t\rangle
 \le\frac{\mathsf D^2}{2\eta}+\frac\eta2\sum_t\|\widehat g_t\|_2^2.
\end{equation}

\emph{Field error for adaptive sequences.} For a nonanticipating adaptive sequence, put $g_t=g_{F_t}(x_t)$ and $\mathcal E_t=g_t-\widehat g_t$. Both $f_t$ and $x_t$ are $\mathcal H_t$-measurable, so $g_t$ is too, and \cref{lem:oraclemoments} gives $\E[\widehat g_t\mid\mathcal H_t]=g_t$, hence $\E[\mathcal E_t\mid\mathcal H_t]=0$. Expanding $\|g_t-\widehat g_t\|_2^2$ and using this conditional mean,
\[
 \E[\|\mathcal E_t\|_2^2\mid\mathcal H_t]=\|g_t\|_2^2-2\ip{g_t}{\E[\widehat g_t\mid\mathcal H_t]}+\E[\|\widehat g_t\|_2^2\mid\mathcal H_t]=\E[\|\widehat g_t\|_2^2\mid\mathcal H_t]-\|g_t\|_2^2\le \mathcal U_j.
\]
For $s<t$, $\mathcal E_s$ is $\mathcal H_t$-measurable, so $\E\ip{\mathcal E_s}{\mathcal E_t}=\E\ip{\mathcal E_s}{\E[\mathcal E_t\mid\mathcal H_t]}=0$. Hence $\E\|\sum_t\mathcal E_t\|_2^2=\sum_t\E\|\mathcal E_t\|_2^2\le T\mathcal U_j$, and Jensen's inequality gives $\E\|\sum_t\mathcal E_t\|_2\le\sqrt{T\mathcal U_j}$. Fix a deterministic $x^\circ\in\K(P)$. For every $y\in\K(P)$, adding and subtracting $x^\circ$, and then applying the Cauchy--Schwarz inequality and $\|y-x^\circ\|\le\mathsf D$, give
\[
 \sum_t\langle \mathcal E_t,y-x_t\rangle
 =\Bigl\langle\sum_t\mathcal E_t,y-x^\circ\Bigr\rangle-\sum_t\langle \mathcal E_t,x_t-x^\circ\rangle
 \le\mathsf D\Big\|\sum_t\mathcal E_t\Big\|_2-\sum_t\langle \mathcal E_t,x_t-x^\circ\rangle.
\]
The right side does not depend on $y$, and the last sum has mean zero: $x_t-x^\circ$ is $\mathcal H_t$-measurable, so $\E\ip{\mathcal E_t}{x_t-x^\circ}=\E\ip{\E[\mathcal E_t\mid\mathcal H_t]}{x_t-x^\circ}=0$.

\emph{Combination.} Summing \eqref{eq:linearizable} over $t$ for an arbitrary $y\in\K(P)$ gives the bound $\alpha\sum_t\ip{g_t}{y-x_t}$ on the right. Splitting $g_t=\widehat g_t+\mathcal E_t$, bounding the $\widehat g_t$ part by \eqref{eq:projectedregret} and the $\mathcal E_t$ part by the preceding display,
\[
 \alpha\sum_tF_t(y)+\alpha C\sum_tf_t(0)-\sum_tF_t(p_t)
 \le\alpha\Bigl[\frac{\mathsf D^2}{2\eta}+\frac\eta2\sum_t\|\widehat g_t\|_2^2+\mathsf D\Big\|\sum_t\mathcal E_t\Big\|_2-\sum_t\langle \mathcal E_t,x_t-x^\circ\rangle\Bigr].
\]
The right side does not depend on $y$, so we may take the maximum over $y$ on the left before taking expectations, which covers the comparator chosen in hindsight after the adaptive sequence is realized. Taking expectations, and using $\E\sum_t\|\widehat g_t\|_2^2\le T\mathcal U_j$ (from \cref{lem:oraclemoments} and the tower property), $\E\|\sum_t\mathcal E_t\|_2\le\sqrt{T\mathcal U_j}$, and $\E\sum_t\ip{\mathcal E_t}{x_t-x^\circ}=0$,
\[
 \operatorname{Reg}_\alpha(T)+\alpha C\,\E\sum_tf_t(0)\le\alpha\Bigl[\frac{\mathsf D^2}{2\eta}+\frac{\eta T\mathcal U_j}2+\mathsf D\sqrt{T\mathcal U_j}\Bigr].
\]
With $\eta=\mathsf D/\sqrt{T\mathcal U_j}$, the first two terms are $\mathsf D^2/(2\eta)=\mathsf D\sqrt{T\mathcal U_j}/2$ and $\eta T\mathcal U_j/2=\mathsf D\sqrt{T\mathcal U_j}/2$, which together contribute $\mathsf D\sqrt{T\mathcal U_j}$, and the third contributes another $\mathsf D\sqrt{T\mathcal U_j}$, proving \eqref{eq:fewqueryadaptive}.

\emph{Integral rewards and degenerate cases.} Objective-independent matroid rounding and base completion are drawn from fresh randomness after $\mathcal H_t$ is fixed and satisfy $\E[f_t(S_t)\mid\mathcal H_t]\ge F_t(p_t)$ (\cref{app:matroidrounding,app:basecompletion}). Replacing each fractional reward by its integral reward therefore preserves these expected bounds, as in \cref{cor:roundtransfer} with $s=1$ and $\tau=0$. If $\mathcal V_j=0$ or $\mathcal U_j=0$, the corresponding second moment in \cref{lem:oraclemoments} is zero, so the estimates vanish almost surely; keeping the state fixed then gives zero linear regret, and the transfer arguments above give the bounds with right side zero. A singleton region has zero diameter, every state equals its unique point, and the same conclusion holds.

\end{proof}

\subsection{Proof of \texorpdfstring{\cref{cor:localentropy}}{the local-norm corollary} (local-norm bound for exact two-value feedback)}\label{app:localentropy}
The proof replaces the block-maximum moment of \cref{lem:oraclemoments} by the state-weighted moment in \eqref{eq:entropylocal}.

\begin{proof}[of \cref{cor:localentropy}]
Let $\mathcal H_t$ contain the past and the current state $x_t$, before the round's fresh sampling; for an oblivious sequence $f_t$ is deterministic. Let $q_t$ be the lift of $x_t$, so that $q_{t,e,i}=x_{t,e,i}$ for $i\ge1$, and let $\widehat g_t$ be the two-value estimate \eqref{eq:twoqueryfield} with exact values. Its block $e$ vanishes unless $e=E$, its coordinate $0$ is zero, and
\[
 \widehat g_{t,E,\ell}=nkM_{E,\ell I}\,\Delta_{E,I}f_t(R)\qquad(\ell\in[k]).
\]

\emph{Deterministic cap.} Every row of $M_E$ has $\ell_1$-norm at most $G$, so $|M_{E,\ell I}|\le G$, and $R$ leaves $E$ unselected, so $|\Delta_{E,I}f_t(R)|\le L_E\le L$ by \cref{lem:marginal}. Hence $|\widehat g_{t,e,\ell}|\le nkGL=\bar a$ on every outcome. The step size satisfies $\eta\le1/\bar a$, so $|\eta\widehat g_{t,e,\ell}|\le1$ for all $t$, $e$, and $\ell$.

\emph{Pathwise estimate.} Fix a deterministic $y^\star\in\K(P)$; its lift $q^\star$ is a feasible categorical comparator. Because $\widehat g_{t,e,0}=0$, $\ip{\widehat g_t}{q^\star-q_t}=\ip{\widehat g_t}{y^\star-x_t}$. The hypotheses of \eqref{eq:entropylocal} hold, so on every outcome
\begin{equation}\label{eq:localpathwise}
 \sum_{t=1}^T\ip{\widehat g_t}{y^\star-x_t}\le\frac{B_r}\eta+\eta\sum_{t=1}^T\sum_{e\in V}\sum_{\ell=1}^kx_{t,e,\ell}\widehat g_{t,e,\ell}^2.
\end{equation}

\emph{State-weighted moment.} The state $x_t$ is $\mathcal H_t$-measurable, and $E$, $I$, and $R$ are drawn afresh. Conditioning on $E=e$ and $I=j$, which has probability $1/(nk)$, and using $\Delta_{e,j}f_t(R)^2\le L_e^2$,
\[
 \E\Bigl[\sum_{e,\ell}x_{t,e,\ell}\widehat g_{t,e,\ell}^2\Bigm|\mathcal H_t\Bigr]
 \le\sum_{e\in V}\sum_{j=1}^k\frac1{nk}\sum_{\ell=1}^kx_{t,e,\ell}\,n^2k^2M_{e,\ell j}^2L_e^2
 =nk\sum_{e\in V}L_e^2\sum_{\ell=1}^kx_{t,e,\ell}\sum_{j=1}^kM_{e,\ell j}^2.
\]
For each row, $\sum_jM_{e,\ell j}^2\le(\sum_j|M_{e,\ell j}|)^2\le G^2$, and $\sum_\ell x_{t,e,\ell}=|x_{t,e}|$, so the right side is at most $nkG^2\sum_eL_e^2|x_{t,e}|$. Since $x_t\in\K(P)$, its mass vector $(|x_{t,e}|)_e$ lies in $P$ and has total mass at most $r$; also $|x_{t,e}|\le1$. Therefore
\[
 \sum_eL_e^2|x_{t,e}|\le\min\Bigl\{L^2\sum_e|x_{t,e}|,\ \sum_eL_e^2\Bigr\}\le\min\{rL^2,\Lambda^2\},
\]
and the conditional moment is at most $\mathcal V_{\mathrm{loc}}$. Taking expectations in \eqref{eq:localpathwise} and using the tower property gives the expected linear regret bound
\[
 \E\sum_{t=1}^T\ip{\widehat g_t}{y^\star-x_t}\le\frac{B_r}\eta+\eta T\mathcal V_{\mathrm{loc}}.
\]
Both $\mathcal V_{\mathrm{loc}}$ and $\bar a$ are positive, because $L>0$ forces some $L_e>0$ and $r\ge1$.

\emph{Step size.} If $\sqrt{B_r/(T\mathcal V_{\mathrm{loc}})}\le1/\bar a$, then $\eta=\sqrt{B_r/(T\mathcal V_{\mathrm{loc}})}$ and each of the two terms equals $\sqrt{B_rT\mathcal V_{\mathrm{loc}}}$, so their sum is $2\sqrt{B_rT\mathcal V_{\mathrm{loc}}}$. Otherwise $\eta=1/\bar a$ and $\sqrt{B_r/(T\mathcal V_{\mathrm{loc}})}>1/\bar a$, that is, $T\mathcal V_{\mathrm{loc}}<\bar a^2B_r$. The first term is $\bar aB_r$, and the second satisfies
\[
 \frac{T\mathcal V_{\mathrm{loc}}}{\bar a}=\sqrt{T\mathcal V_{\mathrm{loc}}}\cdot\frac{\sqrt{T\mathcal V_{\mathrm{loc}}}}{\bar a}<\sqrt{T\mathcal V_{\mathrm{loc}}}\sqrt{B_r}.
\]
In both cases the expected linear regret is at most $R_T=2\sqrt{B_r\mathcal V_{\mathrm{loc}}T}+\bar aB_r$.

\emph{Transfer.} The estimate is conditionally unbiased by \cref{lem:oraclemoments}, and the bound above holds against every fixed feasible comparator $y^\star$. \Cref{cor:transfer} therefore applies with this $R_T$ and gives $\operatorname{Reg}_\alpha(T)+\alpha C\sum_tf_t(0)\le\alpha R_T$, which is \eqref{eq:localentropyregret}.

\emph{Comparison at $r=n$.} Since $k+1\le\mathrm ek$, $\log(k+1)\le\log(\mathrm ek)$, so $B_n=n\log(k+1)$; also $\Lambda^2\le nL^2$, so $\mathcal V_{\mathrm{loc}}\le n^2kG^2L^2$. The leading term is then at most $2\alpha\sqrt{n\log(k+1)\cdot n^2kG^2L^2\cdot T}=2\alpha GLn\sqrt{nk\log(k+1)T}$. With $\sigma=0$, $\mathcal V_2=n^2k^2G^2L^2$ and $B_0=n\log(k+1)$, so \eqref{eq:fewqueryregret} equals $\alpha\sqrt{2n\log(k+1)\cdot n^2k^2G^2L^2\cdot T}=\alpha GLnk\sqrt{2n\log(k+1)T}$. The ratio of the first to the second is $2\sqrt k/(\sqrt2\,k)=\sqrt{2/k}$.
\end{proof}

\subsection{Proofs for \texorpdfstring{\cref{sec:normalization}}{the normalization subsection} (coordinate capacities and inactive elements)}\label{app:normalized}
For $R\subseteq V$ and a labeling $s$, let $s|_R$ agree with $s$ on $R$ and assign label $0$ outside $R$.

\emph{Restriction preserves $k$-submodularity.} Fix $R$. For $a,b\in\{0,1,\ldots,k\}$, the operations of \eqref{eq:ksubdef} satisfy $\operatorname{min}_0(0,0)=0$ and $\operatorname{max}_0(0,0)=\max\{0,0\}=0$. Hence, coordinate by coordinate,
\[
 \operatorname{min}_0(s,t)|_R=\operatorname{min}_0(s|_R,t|_R),\qquad
 \operatorname{max}_0(s,t)|_R=\operatorname{max}_0(s|_R,t|_R):
\]
on $R$ both sides apply the same operation to the same labels, and outside $R$ both sides equal $0$. Applying \eqref{eq:ksubdef} to the labelings $s|_R$ and $t|_R$ therefore gives
\begin{equation}\label{eq:restrictedksub}
 f(s|_R)+f(t|_R)\ge f(\operatorname{min}_0(s,t)|_R)+f(\operatorname{max}_0(s,t)|_R).
\end{equation}
Thus $s\mapsto f(s|_R)$ is $k$-submodular, and it is nonnegative because $f$ is.

\begin{proof}[of \cref{prop:normalized}]
Throughout, labelings and fractional points are indexed by $V_\rho$, and they are extended by the label $0$ and by zero blocks on $V\setminus V_\rho$ when $f$ or $F$ is evaluated. Since every $v\in P$ has $v_e\le\rho_e$, every $v\in P$ vanishes on $V\setminus V_\rho$.

\emph{Step 1: the region.} The set $P^\rho$ is the intersection of $\R_+^{V_\rho}$ with the preimage of $P$ under the linear map $u\mapsto\rho\odot u$, so it is closed and convex, and it contains $0$. If $u\in P^\rho$, then $\rho_eu_e=(\rho\odot u)_e\le\rho_e$, and dividing by $\rho_e>0$ gives $u_e\le1$; hence $P^\rho\subseteq[0,1]^{V_\rho}$ is compact. If $0\le u'\le u\in P^\rho$, then $0\le\rho\odot u'\le\rho\odot u\in P$, so $\rho\odot u'\in P$ by down-closedness; thus $P^\rho$ is down-closed. The set $\{t\in[0,1]:t\rho\in P\}$ is closed and contains $0$, so the maximum in \eqref{eq:normalizedregion} exists, and $\zeta_\rho(P)\one\in P^\rho$ because $\rho\odot(\zeta_\rho(P)\one)=\zeta_\rho(P)\rho\in P$.

\emph{Step 2: the objective.} Let $R\subseteq V_\rho$ contain each $e$ independently with probability $\rho_e$, and define $f_\rho(s)=\E_Rf(s|_R)$ for labelings $s$ of $V_\rho$. Averaging \eqref{eq:restrictedksub} over $R$ shows that $f_\rho$ is $k$-submodular, and it is nonnegative. Its value at the empty labeling is $f_\rho(0)=f(0)$, and for a singleton,
\[
 f_\rho(e\mapsto i)=\rho_ef(e\mapsto i)+(1-\rho_e)f(0),
\]
so its singleton scales \eqref{eq:scales} are $\max_i(f_\rho(e\mapsto i)-f_\rho(0))=\rho_eL_e$. They are computed from the $|V_\rho|k+1$ initial values of $f$ on the empty labeling and the singletons of $V_\rho$.

\emph{Step 3: the extension.} Let $F_\rho$ be the multilinear extension of $f_\rho$ and let $u\in\Delta_k^{V_\rho}$. Draw $X\sim u$ and $R$ independently. For each $e$, the label $(X|_R)_e$ equals $i\ge1$ with probability $\rho_eu_{e,i}$ and equals $0$ otherwise, and these labels are independent across elements. Hence $X|_R$ has the law of a sample from $\rho\odot u\in\Delta_k^{V_\rho}$, and
\[
 F_\rho(u)=\E_X\E_Rf(X|_R)=F(\rho\odot u).
\]
By the chain rule, $\partial_{e,i}F_\rho(u)=\rho_e\,\partial_{e,i}F(\rho\odot u)$.

\emph{Step 4: the comparison sets.} The map $\phi(u)=\rho\odot u$ sends $\K(P^\rho)$ into $\K(P)$, because $\rho\odot u\in\Delta_k^{V_\rho}$ when $\rho\le\one$ and $|\rho\odot u|=\rho\odot|u|\in P$. Conversely, for $x\in\K(P)$ we have $|x_e|\le\rho_e$, so $u_e=x_e/\rho_e$ lies in $\Delta_k$ for every $e\in V_\rho$, and $\rho\odot|u|=|x|\in P$; thus $u\in\K(P^\rho)$ and $\phi(u)=x$. Hence $\phi$ is a bijection from $\K(P^\rho)$ onto $\K(P)$ with $F_\rho=F\circ\phi$, and for every sequence of objectives,
\[
 \max_{u\in\K(P^\rho)}\sum_tF_{\rho,t}(u)=\max_{y\in\K(P)}\sum_tF_t(y),\qquad
 \mathrm{OPT}_{F_\rho}(P^\rho)=\mathrm{OPT}_F(P).
\]

\emph{Step 5: the estimators.} Apply \cref{thm:continuous,thm:fewquery} to $f_\rho$, the region $P^\rho$, and the level $\zeta_\rho(P)$; the triple lies in $\mathcal A_k(\zeta_\rho(P))$ by hypothesis, and the action map and the field matrices $M_e$ depend only on the state and the parameters. The theorems need estimators of the field of $f_\rho$ whose conditional means and second moments obey \cref{lem:oraclemoments} for $f_\rho$. Let $\Pi^e$ be the lower-chain point of the state, so the field block is $M_e\nabla_eF_\rho(\Pi^e)=\rho_eM_e\nabla_eF(\rho\odot\Pi^e)$ by Step~3. The point $\rho\odot\Pi^e$ assigns the blocks $\rho_up_u$ to the predecessors $u$ of $e$ and zero blocks to all other elements. Run each estimator of \cref{app:oracles} for $f$ at this point, which means sampling the predecessors from $\rho\odot p$, querying $\nabla F(\rho\odot\Pi^E)$ in the first-order model, and multiplying the returned block of element $e$ by $\rho_e$. By the proof of \cref{lem:oraclemoments}, the unscaled block has conditional mean $M_e\nabla_eF(\rho\odot\Pi^e)$, so the scaled block has mean $M_e\nabla_eF_\rho(\Pi^e)$. In that proof, each moment bound is obtained from a per-element scalar bound: $|\partial_{e,j}F|\le L_e$ and $|\Delta_{e,j}f|\le L_e$ for the gradient and full-prefix oracles, $L_e^2+2\sigma^2$ for two values, and $4(f_{\max}^2+\sigma^2)$ for one value. Multiplication by $\rho_e\in(0,1]$ turns these into $\rho_eL_e$, $\rho_e^2L_e^2+2\rho_e^2\sigma^2\le(\rho_eL_e)^2+2\sigma^2$, and $4\rho_e^2(f_{\max}^2+\sigma^2)\le4(f_{\max}^2+\sigma^2)$. The rest of that proof is unchanged with $|V_\rho|$ elements, so the table of \cref{lem:oraclemoments} holds with $n$ replaced by $|V_\rho|$ and $L_e$ by $\rho_eL_e$, and with $\sigma$ and $f_{\max}$ unchanged.

\emph{Step 6: conclusion.} If $\Lambda_\rho=0$, then $\rho_eL_e=0$ for every $e\in V_\rho$, so $f_\rho$ is constant by \cref{lem:marginal}; then $F=F_\rho\circ\phi^{-1}$ is constant on $\K(P)$ and every feasible output is optimal. Otherwise \cref{thm:continuous}, applied to $f_\rho$ on $P^\rho$, gives for every $u\in\K(P^\rho)$
\[
 \E F_\rho(p_J)\ge\alpha F_\rho(u)+\alpha Cf_\rho(0)-\alpha G\sqrt{\frac{2|V_\rho|\Lambda_\rho^2\log(k+1)}T}.
\]
Substituting $F_\rho=F\circ\phi$, $f_\rho(0)=f(0)$, and $u=\phi^{-1}(y)$ for an arbitrary $y\in\K(P)$ gives \eqref{eq:continuous} for the output $\phi(p_J)=\rho\odot p_J\in\K(P)$, with $\Lambda_\rho$ and $|V_\rho|$ in place of $\Lambda$ and $n$. Since $\Lambda_\rho^2\le|V_\rho|L_\rho^2$, the horizon of \cref{thm:continuous} with $|V_\rho|$ in place of $n$ makes the error term at most $\eps L_\rho$. In the online protocol, the plays $\rho\odot p_t$ earn $F_t(\rho\odot p_t)=F_{\rho,t}(p_t)$, and by Step~4 the comparator terms of the two regrets coincide, so the regret \eqref{eq:regretdefinition} against $\K(P)$ equals the regret of the transformed learner against $\K(P^\rho)$. \Cref{thm:fewquery} with the moments of Step~5 gives the stated bounds; the entropy diameter is $B_0=|V_\rho|\log(k+1)$, and the Euclidean bound uses the diameter of $\K(P^\rho)$.
\end{proof}

\begin{proof}[of the claims in \cref{rem:nullelements}]
\emph{Objective and region.} With $R=V'$ in \eqref{eq:restrictedksub}, the function $s\mapsto f(s|_{V'})$ is $k$-submodular and nonnegative; its restriction to labelings of $V'$ is $f_{V'}$, and $f_{V'}(0)=f(0)$. The set $P_{V'}$ is the intersection of $[0,1]^{V'}$ with the preimage of $P$ under the linear map $u\mapsto(u,0_D)$, so it is compact and convex. If $0\le u'\le u\in P_{V'}$, then $(u',0_D)\le(u,0_D)\in P$, so $u'\in P_{V'}$. For $Y\in\mathcal S$ with $Y\subseteq V'$, the vector $(\one_Y,0_D)=\one_Y$ lies in $P$, so $P_{V'}$ relaxes the restricted family. If $t\one\in P$, then $(t\one_{V'},0_D)\le t\one$ lies in $P$ by down-closedness, so $t\one_{V'}\in P_{V'}$; hence $\zeta(P_{V'})\ge\zeta(P)$. For a matroid, a point $(u,0_D)\ge0$ satisfies $\sum_{e\in S}u_e\le r_{\mathcal M}(S)$ for every $S\subseteq V$ exactly when it does so for every $S\subseteq V'$, because $\sum_{e\in S}u_e=\sum_{e\in S\cap V'}u_e\le r_{\mathcal M}(S\cap V')\le r_{\mathcal M}(S)$; by Edmonds' theorem, $P_{V'}$ is the independence polytope of the restriction to $V'$. For knapsacks, setting $v_D=0$ in the budget rows leaves the rows of the items in $V'$.

\emph{Optimum.} A feasible labeling of the restricted instance, extended by $0$ on $D$, is a feasible labeling of the original instance with the same value, so $\OPT_{V'}\le\OPT$. Conversely, suppose that $\mathcal S$ is down-closed and let $s$ be an optimal labeling of the original instance and list the elements of $\supp(s)\cap D$ as $e_1,\ldots,e_h$. Let $s^{(0)}=s$ and let $s^{(j)}$ unselect $e_1,\ldots,e_j$. Each $s^{(j)}$ has support contained in $\supp(s)$, so it is feasible by down-closedness, and $s^{(h)}$ is supported in $V'$. Since $s^{(j-1)}=s^{(j)}+e_j\mapsto s_{e_j}$, \cref{lem:marginal} gives
\[
 f(s^{(j-1)})-f(s^{(j)})=\Delta_{e_j,s_{e_j}}f(s^{(j)})\le L_{e_j}.
\]
Summing over $j$ gives $\OPT_{V'}\ge f(s^{(h)})\ge f(s)-\sum_{j=1}^hL_{e_j}\ge\OPT-\sum_{e\in D}L_e$, using $L_e\ge0$. If moreover every element of $D$ is inactive, the two bounds give $\OPT_{V'}=\OPT$.

\emph{Transfer.} The restricted instance satisfies the hypotheses of this paper with the region $P_{V'}$, the level $\zeta(P_{V'})$, the same value $f(0)$, and the same singleton scales on $V'$. A guarantee for its output therefore bounds the value of the extended output in terms of $\OPT_{V'}$, which equals $\OPT$ when $D$ is inactive, and the extended output is feasible for the original down-closed instance. For an exact-base constraint, complete this independent support in the original matroid as in \cref{app:basecompletion}; use value-nondecreasing offline completion when retaining a high-probability guarantee. The fractional and online guarantees also transfer when $D$ is inactive. For $e\in D$ and every label $i$, orthant submodularity and $L_e=0$ give $\Delta_{e,i}f(s)\le0$ in every context $s$, and pairwise monotonicity gives $\Delta_{e,i}f(s)+\Delta_{e,j}f(s)\ge0$ for $j\ne i$, so each marginal of $e$ is zero, since $k\ge2$. Hence $f(s)=f(s_{V'},0_D)$ for every labeling $s$, and $F(x)=F(x_{V'},0_D)$ for every fractional point $x$. It follows that $\mathrm{OPT}_F(P_{V'})=\mathrm{OPT}_F(P)$ and that every comparator in $\K(P)$ has the same value as its restriction, which lies in $\K(P_{V'})$.
\end{proof}

\section{Proofs for \texorpdfstring{\cref{sec:rounding}}{Section 5}: Rounding}\label{app:matroid}
Conditional labels reduce support rounding to rounding for ordinary submodular functions, and every proof below keeps the initial value in the residual solver guarantee.

\subsection{Proof of \texorpdfstring{\cref{cor:roundtransfer}}{the rounding-transfer corollary} (rounding transfer)}\label{app:roundtransferproof}
\begin{proof}[of \cref{cor:roundtransfer}]
Fix $x,y$ and write $p=p(x)$. Multiplying \eqref{eq:linearizable} by $s>0$ gives
\[
 s\alpha F(y)+s\alpha Cf(0)-sF(p)\le s\alpha\ip{g_F(x)}{y-x}.
\]
The retention inequality, whose expectation is over the rounding randomness with $p$ fixed, states $-sF(p)\ge\tau f(0)-\E f(\mathcal R(p))$. Substituting this lower bound for $-sF(p)$ on the left gives
\[
 s\alpha F(y)+(s\alpha C+\tau)f(0)-\E f(\mathcal R(p))
 \le s\alpha F(y)+s\alpha Cf(0)-sF(p)
 \le s\alpha\ip{g_F(x)}{y-x},
\]
which is \eqref{eq:roundtransfer}.

For online use, let $\mathcal H_t$ contain the past, the state $x_t$, and the objective $f_t$, but not the fresh randomness of round $t$. An objective-independent rule draws $\mathcal R(p_t)$ from fresh randomness and the action $p_t$, so it can be executed before feedback. Conditional on $\mathcal H_t$, it is feasible on every outcome and satisfies $\E[f_t(\mathcal R(p_t))\mid\mathcal H_t]\ge sF_t(p_t)+\tau f_t(0)$. Applying \eqref{eq:roundtransfer} to $F_t$ at $x_t$ and repeating the proof of \cref{cor:transfer} with $\alpha$ replaced by $s\alpha$, $\alpha C$ replaced by $s\alpha C+\tau$, and $F_t(p_t)$ replaced by the realized reward $f_t(\mathcal R(p_t))$ gives the regret bound with scale $s\alpha$ for oblivious sequences; for adaptive sequences, the same substitution in the proof of \cref{thm:fewquery} in \cref{app:fewqueryproof} applies, using the tower property over $\mathcal H_t$ for the realized-reward term. The field oracle and its query budget are unchanged because the rounding uses no objective values. An objective-dependent rule can only be executed after the objective is available, so the same inequality then gives only an offline consequence.
\end{proof}

\subsection{Auxiliary prefix and candidate bounds (\texorpdfstring{\cref{lem:optimalprefix,lem:candidates}}{the prefix and candidate lemmas}), used for \texorpdfstring{\cref{lem:prefix,thm:single,thm:multi,cor:finiteprefix,prop:smallitems}}{the knapsack bounds}}\label{app:auxround}
Both lemmas serve the knapsack proofs (\cref{lem:prefix}, \cref{thm:single,thm:multi,cor:finiteprefix}, and \cref{prop:smallitems}): the first bounds the remaining marginals after a greedy prefix of an optimal labeling, and the second converts an additive residual error into a multiplicative loss by keeping evaluated singleton candidates.

Greedy prefixes of an optimal labeling are the standard ingredient of partial enumeration \citep{Sviridenko04,KulikShachnaiTamir13}. Restricted to sub-labelings of an optimal labeling, the objective is monotone even when $f$ is not, so the largest remaining marginal decreases along a greedy prefix.
\begin{lemma}[Optimal-prefix marginals]\phantomsection\label[lemma]{lem:optimalprefix}
Let a downward-closed support family admit an optimal labeling $o$ with support $O$. The function $\psi(A)=f(o|_A)$ on subsets of $O$ is monotone and submodular. Let $B=o|_{B_h}$ be the restriction of $o$ to the first $h$ elements $B_h$ of a greedy ordering of $O$ under these fixed labels, where $1\le h<|O|$. Every remaining element satisfies
\begin{equation}\label{eq:optimalprefixtail}
 \Delta_{e,o_e}f(B)\le\frac{f(B)-f(0)}h.
\end{equation}
\end{lemma}
\begin{proof}
For $A\subseteq O$, the labeling $o|_A$ agrees with $o$ on $A$ and leaves all other elements unselected.

\emph{Submodularity.} Let $A\subseteq A'\subseteq O$ and $e\in O\setminus A'$. Then $o|_A\preceq o|_{A'}$, element $e$ is unselected in $o|_{A'}$, and $\psi(A\cup\{e\})-\psi(A)=\Delta_{e,o_e}f(o|_A)$. Orthant submodularity \eqref{eq:orthant} gives $\Delta_{e,o_e}f(o|_A)\ge\Delta_{e,o_e}f(o|_{A'})$, which is submodularity of $\psi$.

\emph{Monotonicity.} Let $A\subseteq O$ and $e\in O\setminus A$. The support $O\setminus\{e\}$ is feasible because the family is downward closed, so optimality of $o$ gives $\psi(O)-\psi(O\setminus\{e\})=f(o)-f(o|_{O\setminus\{e\}})\ge0$. By submodularity, applied with $A\subseteq O\setminus\{e\}$,
\[
 \psi(A\cup\{e\})-\psi(A)\ge\psi(O)-\psi(O\setminus\{e\})\ge0.
\]

\emph{Greedy marginals.} Let $e_i$ maximize the remaining marginal at $B_{i-1}=\{e_1,\ldots,e_{i-1}\}$, and put $d_i=\psi(B_i)-\psi(B_{i-1})$. Monotonicity gives $d_i\ge0$. Submodularity and the preceding greedy choice give $d_{i+1}\le d_i$. For every $e\in O\setminus B_h$, submodularity and the choice at step $h$ imply
\[
 \Delta_{e,o_e}f(B)\le\psi(B_{h-1}\cup\{e\})-\psi(B_{h-1})
 \le d_h\le\frac1h\sum_{i=1}^h d_i=\frac{f(B)-f(0)}h.
\]
This proves \eqref{eq:optimalprefixtail}.
\end{proof}

A residual solve may lose an additive term proportional to its largest singleton gain; keeping the singleton candidates already evaluated turns this loss into the multiplicative comparison below, which every knapsack proof uses; the matroid proof makes the same comparison directly.

\begin{lemma}[Retaining feasible singleton candidates]\phantomsection\label[lemma]{lem:candidates}
Let $B$ be a feasible prefix, $b=f(B)$, let $f_B(s)=f(B\cup s)$ be the residual objective on the retained elements (defined in \cref{app:prefix}) with multilinear extension $F_B$ and singleton scale $L_B$, and let a deterministic pool of feasible candidates contain $B$ and every singleton extension evaluated in its residual initialization. Let $M$ be the largest value in that pool. Suppose a residual fractional output satisfies
\[
 \E F_B(p)\ge\alpha \Phi+\alpha Cb-\nu L_B,
\]
where $\Phi$ is a deterministic comparator value and $\nu>0$. Suppose the value $X$ of its rounded candidate, combined with $B$, satisfies
$\E[X\mid p]\ge sF_B(p)+(1-s)b$ for some $s\in(0,1]$. Returning the best of this candidate and the pool gives a value $Y$ satisfying
\begin{equation}\label{eq:candidatecomparison}
 \E Y\ge
 \frac{s\alpha \Phi+(s\alpha C+s\nu+1-s)b}{1+s\nu}.
\end{equation}
\end{lemma}
\begin{proof}
\emph{The residual scale.} By definition, $L_B$ is the largest residual singleton gain $f_B(e\mapsto i)-f_B(0)=f(B+e\mapsto i)-b$ over the retained residual elements $e$ and labels $i$. Each labeling $B+e\mapsto i$ is evaluated in the residual initialization and belongs to the pool, so its value is at most $M$. Hence $L_B\le M-b$.

\emph{The rounded candidate.} Taking expectations in the retention inequality and then using the residual guarantee,
\[
 \E X\ge s\,\E F_B(p)+(1-s)b\ge s\alpha \Phi+s\alpha Cb-s\nu L_B+(1-s)b.
\]
Substituting $L_B\le M-b$ gives
\[
 \E X\ge A-s\nu M,\qquad A:=s\alpha \Phi+(s\alpha C+s\nu+1-s)b.
\]

\emph{The best candidate.} The returned value $Y$ is at least the value of every candidate, so pointwise $Y\ge M$ and $Y\ge X$. Taking expectations, $\E Y\ge M$ because $M$ is deterministic, and $\E Y\ge\E X\ge A-s\nu M$; hence $\E Y\ge\max\{M,A-s\nu M\}$. A maximum of two numbers is at least every convex combination of them. Using the weights $s\nu/(1+s\nu)$ on $M$ and $1/(1+s\nu)$ on $A-s\nu M$,
\[
 \E Y\ge\frac{s\nu M+A-s\nu M}{1+s\nu}
 =\frac A{1+s\nu},
\]
which is \eqref{eq:candidatecomparison}. The deterministic pool may contain candidates from all enumerated prefixes; the inequality then applies to any distinguished prefix whose residual guarantee is being analyzed, because adding candidates only increases $M$ and $Y$, and the argument used only $L_B\le M-b$ and $Y\ge\max\{M,X\}$.
\end{proof}

\subsection{Proof of \texorpdfstring{\cref{thm:matroid}}{the matroid corollary} (matroid approximation)}\label[appendix]{app:matroidrounding}
Conditional label sampling reduces the objective to an ordinary submodular support function, whose multilinear extension pipage rounding preserves in expectation; the stored singleton candidates absorb the optimization error.

\begin{proof}[of \cref{thm:matroid}]
Let $\mathcal M=(V,\mathcal I)$ have rank $r$. For $r=0$, return the empty labeling. If $L=0$, then $f$ is constant by \cref{lem:marginal}; return the empty labeling, whose value $f(0)=\OPT$ satisfies \eqref{eq:matroidaffine} because $\alpha(1+C)=(1+c-\omega)/(1+c)\le1$, and \eqref{eq:matroidhp} because its two coefficients sum to $(\alpha-\eps)(1+C)+\eps/\alpha\le(\alpha-\eps)/\alpha+\eps/\alpha=1$ for $0<\eps<\alpha$. Assume $r\ge1$ and $L>0$ below. Use randomized pipage rounding \citep[Section~2]{ChekuriVondrakZenklusen10}, \citep{CalinescuEtAl11}, applied to a point $z\in P(\mathcal M)$ after adding dummy elements as in \citet[Appendix~B.2]{ChekuriVondrakZenklusen10} and removing them at the end. It returns a random independent set $R$ whose law depends only on $z$ and $\mathcal M$, and for every submodular $g:2^V\to\R$,
\begin{equation}\label{eq:swapping}
 \E g(R)\ge G_g(z)
\end{equation}
\citep[Section~2]{ChekuriVondrakZenklusen10}. The rounding uses no values of $g$.

For the point $p$ returned by the fractional procedure, let $z_e=|p_e|$. Independently choose a label $\ell_e$ with probabilities $p_{e,i}/z_e$ when $z_e>0$, using an arbitrary label at a zero block. Conditional on these labels, $g_\ell(A)=f(\ell|_A)$ is nonnegative and submodular by \cref{lem:extension}. Round $z=|p|\in P(\mathcal M)$ to an independent set $R$ by the procedure above, independently of $\ell$, and return $S=\ell|_R$, whose support $R$ is independent. Conditional on $p$ and $\ell$, the returned value is $f(S)=g_\ell(R)$, and since the law of $R$ depends only on $z$ and the matroid, \eqref{eq:swapping} applied to $g_\ell$ gives $\E[f(S)\mid p,\ell]\ge G_\ell(|p|)$. Averaging over $\ell$ and using \eqref{eq:labelconditioning},
\[
 \E[f(S)\mid p]
 \ge\E_\ell G_\ell(|p|)=F(p).
\]
Every vector of $P(\mathcal M)$ has total mass at most $r$. Run entropy mirror ascent with the full-prefix oracle, step size $\eta=\sqrt{B_r/(TG^2rL^2)}$, and a horizon $T\ge B_r/r$, where $B_r=\min\{n\log(k+1),r\log(\mathrm enk/r)\}$. Full-prefix feedback satisfies $|\widehat g_{t,e,i}|\le GL$ and $\widehat g_{t,e,0}=0$ on every outcome (\cref{lem:oraclemoments}), and $\eta GL=\sqrt{B_r/(Tr)}\le1$. Hence \eqref{eq:entropysupportregret} bounds the linear regret against every comparator, on every realization, by $B_r/\eta+\eta TG^2rL^2=2GL\sqrt{rB_rT}$. \Cref{cor:transfer}, applied with $y$ a maximizer of $F$ over $\K(P(\mathcal M))$, gives $\E F(p_J)\ge\alpha\OPT_F(P(\mathcal M))+\alpha Cf(0)-2\alpha GL\sqrt{rB_r/T}$. Choose the horizon so that this error is at most $\nu L$, where $\nu>0$; the output $p_J$ then satisfies $\E F(p_J)\ge\alpha\OPT_F(P(\mathcal M))+\alpha Cf(0)-\nu L$. Rounding $p_J$ as above and taking expectations over $p_J$ gives the value $Y=f(S)$ of a feasible labeling with
\[
 \E Y\ge\E F(p_J)\ge\alpha\OPT_F(P(\mathcal M))+\alpha Cf(0)-\nu L.
\]
Let $M$ be the largest value among the empty labeling and the feasible singleton labelings, all evaluated at initialization; since every singleton is feasible, \eqref{eq:scales} gives $L=M-f(0)$. Return the better of the rounded output and the stored candidate of value $M$, with value $\widehat Y=\max\{M,Y\}$. Pointwise $\widehat Y\ge M$, so $\E\widehat Y\ge M$; pointwise $\widehat Y\ge Y$, so
\begin{align*}
 \E\widehat Y&\ge\alpha\OPT_F(P(\mathcal M))+\alpha Cf(0)-\nu(M-f(0))\\
 &=\alpha\OPT_F(P(\mathcal M))+(\alpha C+\nu)f(0)-\nu M\\
 &\ge\alpha\OPT_F(P(\mathcal M))+(\alpha C+\nu)f(0)-\nu\E\widehat Y,
\end{align*}
where the last step uses $M\le\E\widehat Y$ and $\nu>0$. Moving $\nu\E\widehat Y$ to the left and dividing by $1+\nu$ gives
\[
 \E\widehat Y\ge\frac{\alpha}{1+\nu}\OPT_F(P(\mathcal M))+\frac{\alpha C+\nu}{1+\nu}f(0).
\]
For a target loss $\eps\in(0,\alpha)$, choose $\nu=\eps/(\alpha-\eps)>0$; then $1+\nu=\alpha/(\alpha-\eps)$ and $\alpha/(1+\nu)=\alpha-\eps$. The coefficient of $f(0)\ge0$ is nonnegative and $\OPT_F(P(\mathcal M))\ge\OPT$ (shown below), so $\E\widehat Y\ge(\alpha-\eps)\OPT$. The error condition $2\alpha GL\sqrt{rB_r/T}\le\nu L$ is equivalent to $T\ge4\alpha^2G^2rB_r/\nu^2=4\alpha^2G^2(\alpha-\eps)^2rB_r/\eps^2$, and the step-size condition requires $T\ge B_r/r$. Thus the sufficient support-dependent horizon is
\begin{equation}\label{eq:matroidcandidatehorizon}
 T=\left\lceil\max\left\{B_r/r,
 \frac{4\alpha^2G^2(\alpha-\eps)^2rB_r}{\eps^2}\right\}\right\rceil.
\end{equation}
Its value-query count is at most $nk+2+Tnk$: $nk+1$ values for the empty labeling and the singletons at initialization, $nk$ new values per iteration of the full-prefix oracle (\cref{alg:continuous}), and one evaluation of the rounded output. Label sampling and pipage rounding add no objective queries. Since $B_r\le r\log(\mathrm enk/r)$, $r\ge1$, and $\eps<\alpha\le1$, both terms in \eqref{eq:matroidcandidatehorizon} are $O(r^2\log(\mathrm enk/r)/\eps^2)$ for fixed $\alpha$ and $G$, which gives the query bound $O(nkr^2\log(\mathrm enk/r)/\eps^2)$ of \cref{thm:matroid}.

The same argument proves equality of the fractional and integral optima. Every feasible labeling $s$ has $|\one_s|=\one_{\supp(s)}\in P(\mathcal M)$, so $\one_s\in\K(P(\mathcal M))$ and $F(\one_s)=f(s)$; hence $\OPT_F\ge\OPT$. Conversely, for every $p\in\K(P(\mathcal M))$, the rounding above returns a feasible labeling $S$ with $\E[f(S)\mid p]\ge F(p)$, and $f(S)\le\OPT$ on every outcome. Hence $F(p)\le\OPT$ for every feasible $p$, and maximizing over $p$ gives $\OPT_F\le\OPT$, so $\OPT_F=\OPT$.

\emph{High-probability implementation.} Fix $\beta\in(0,1)$, $0<\eps<\alpha$, and an optimal labeling $o^*$ with encoding $y^*$. Round each action $p_t$ independently of its feedback sample conditional on the past, evaluate the feasible labeling $S_t$, and retain it together with the empty labeling and all feasible singletons. Put
\begin{equation}\label{eq:matroidhphorizon}
 \begin{gathered}
 \ell_\beta=\log(2/\beta),\qquad
 A_\beta=2\alpha G\sqrt{rB_r}+\sqrt2\,r(2\alpha G+1)\sqrt{\ell_\beta},
 \qquad \nu=\frac{\eps}{\alpha-\eps},\\
 T=\left\lceil\max\{B_r/r,A_\beta^2/\nu^2\}\right\rceil,
 \qquad \eta=\sqrt{B_r/(TG^2rL^2)}.
 \end{gathered}
\end{equation}
These parameters do not use $\OPT$. Let $\mathcal H_t$ contain the history and the state before fresh feedback and rounding, set $g_t=g_F(x_t)$, and let $Z_t=\langle g_t-\widehat g_t,y^*-x_t\rangle$, which has conditional mean zero. Since $\|y^*-x_t\|_1\le2r$ and every coordinate of the full-prefix estimate lies in $[-GL,GL]$, its conditional range has length at most $4GrL$.

For conditionally centered increments of conditional range length at most $b$, Hoeffding's lemma \citep{Hoeffding63} and iterated conditioning give
\[
 \E\exp\!\left(\theta\sum_{t=1}^T Z_t\right)
 \le\exp(\theta^2b^2T/8),\qquad
 \Pr\!\left\{\sum_t Z_t>b\sqrt{T\log(1/q)/2}\right\}\le q.
\]
The second inequality follows from the first by Markov's inequality with $\theta=4a/(b^2T)$ at the threshold $a=b\sqrt{T\log(1/q)/2}$. For $b=4GrL$ and $q=\beta/2$, it follows that, with probability at least $1-\beta/2$,
\[
 \sum_t Z_t\le2\sqrt2\,GrL\sqrt{T\ell_\beta}.
\]
The entropy estimate \eqref{eq:entropysupportregret} is pathwise. Summing the comparison for $y^*$ and using it therefore gives, on this event,
\[
 \frac1T\sum_tF(p_t)\ge\alpha\OPT+\alpha Cf(0)
 -2\alpha GL\sqrt{rB_r/T}-2\sqrt2\,\alpha GrL\sqrt{\ell_\beta/T}.
\]
Let $Y_t=f(S_t)$ and $m_t=\E[Y_t\mid\mathcal H_t]\ge F(p_t)$. Every rounded independent support has size at most $r$, and the marginal bound gives $|Y_t-f(0)|\le rL$. Thus $m_t-Y_t$ is conditionally centered with range length at most $2rL$. The same exponential argument gives, with probability at least $1-\beta/2$,
\[
 \sum_t(F(p_t)-Y_t)\le\sum_t(m_t-Y_t)
 \le\sqrt2\,rL\sqrt{T\ell_\beta}.
\]
A union bound, which does not require independence of the two events, yields both inequalities with probability at least $1-\beta$. On their intersection the largest stored value $\widehat Y=\max\{M,\max_tY_t\}$ is at least the average of the $Y_t$, so
\[
 \widehat Y\ge\alpha\OPT+\alpha Cf(0)-A_\beta L/\sqrt T
 \ge\alpha\OPT+\alpha Cf(0)-\nu L.
\]
Since $L=M-f(0)\le\widehat Y-f(0)$,
\[
 \widehat Y\ge\frac\alpha{1+\nu}\OPT+\frac{\alpha C+\nu}{1+\nu}f(0)
 =(\alpha-\eps)\OPT+
 \left(C(\alpha-\eps)+\frac\eps\alpha\right)f(0).
\]
The last baseline coefficient equals $\alpha C+\eps(1/\alpha-C)=\alpha C+\eps/(1-\omega)$, because $1/\alpha-C=(1+c)/(1-\omega)-c/(1-\omega)=1/(1-\omega)$. This proves \eqref{eq:matroidhp}. The horizon in \eqref{eq:matroidhphorizon} is $O(r^2[\log(\mathrm e nk/r)+\log(2/\beta)]/\eps^2)$. Initialization and the at most $nk+1$ new values per iteration give at most $(nk+1)(T+1)$ queries. Each iteration uses one matroid rounding and one entropy update. The expectation-only implementation still rounds only one randomly selected iterate.

\end{proof}

\subsection{Proof of \texorpdfstring{\cref{cor:bases}}{the base-completion corollary} (exact bases)}\label[appendix]{app:basecompletion}
An independent support can be completed to a base one element at a time, and pairwise monotonicity \eqref{eq:pairwise} supplies an objective-independent randomized label choice whose conditional expected gain is nonnegative.

\begin{proof}[of \cref{cor:bases}]
Start with any independently supported labeling $S_0=S$. Given $S_h$ whose support has fewer than $r$ elements, the augmentation property, applied to $\supp(S_h)$ and a base (which has $r$ elements), supplies an element $e\notin\supp(S_h)$ such that $\supp(S_h)\cup\{e\}$ is independent. A fixed index rule selects such an element and two distinct labels $i,j\in[k]$ (they exist because $k\ge2$), and the algorithm sets $S_{h+1}=S_h+e\mapsto I$ with $I$ uniform on $\{i,j\}$. Conditional on $S_h$, the expected increase is
\[
 \E[f(S_{h+1})-f(S_h)\mid S_h]=\frac12\bigl(\Delta_{e,i}f(S_h)+\Delta_{e,j}f(S_h)\bigr)\ge0
\]
by pairwise monotonicity \eqref{eq:pairwise}, since $e$ is unselected in $S_h$. Each step adds one element to the support and keeps it independent, so after $m=r-|\supp(S)|\le r$ steps the support is an independent set with $r$ elements, that is, a base. Taking expectations of the displayed inequality and summing over the $m$ steps (tower property) gives $\E f(S_m)\ge f(S)$. All support and label choices use only the matroid and fresh randomness, so they are objective-independent.

Apply the construction to an optimal independent labeling $o$. Its completion $S_m$ is a base labeling with $\E f(S_m)\ge f(o)$, so some outcome, which is a feasible base labeling, has value at least the independent optimum. Every base is independent, so the base optimum is at most the independent optimum; hence the two optima are equal. If $S$ is the random output of \cref{thm:matroid} and $S'$ its completion, then $\E f(S')=\E[\E[f(S')\mid S]]\ge\E f(S)\ge(\alpha-\eps)\OPT$, where $\OPT$ is the common optimum; this proves the approximation guarantee for exact bases. In an online round, the label choices are drawn after conditioning on the current objective and state, so the same conditional inequality holds; this is the property used in \cref{thm:fewquery}.
For the offline high-probability guarantee, complete the selected labeling after all candidates have been compared. At each augmentation, evaluate its two fixed-label extensions and keep the better one. Pairwise monotonicity makes at least one marginal nonnegative, so the value never decreases on any outcome. At most $2r$ additional queries preserve \eqref{eq:matroidhp}. This value-dependent completion is not used online.

\end{proof}

\subsection{Auxiliary conditioned-prefix bound (\texorpdfstring{\cref{lem:prefix}}{the conditioned-prefix lemma}), used for \texorpdfstring{\cref{thm:single,thm:multi,cor:finiteprefix}}{the knapsack corollaries}}\label[appendix]{app:prefix}
For a fixed labeling $B$, the residual objective is $f_B(s)=f(B\cup s)$ on the other elements, where $B\cup s$ agrees with $B$ on $\supp(B)$ and with $s$ elsewhere. Row $m\in[d]$ has residual capacity $b'_m=1-\sum_{e\in\supp(B)}c_{m,e}$; retain precisely the elements with $c_{m,e}\le\delta b'_m$ in every row, so a row with zero residual capacity retains only zero-cost elements. The residual objective is again nonnegative and $k$-submodular: $\Delta_{e,i}f_B(s)=\Delta_{e,i}f(B\cup s)$, and $s\preceq t$ implies $B\cup s\preceq B\cup t$, so \eqref{eq:orthant} and \eqref{eq:pairwise} for $f$ give the same conditions for $f_B$.

\begin{lemma}[Conditioned-prefix bound]\phantomsection\label[lemma]{lem:prefix}
Let $o^*$ be an optimal labeling for $d$ support budgets, with $|\supp(o^*)|>\rho$, where $\rho\ge1$ is an integer, and let $0<\delta<1$. For every greedy prefix $B\preceq o^*$ of $\rho$ elements, in the sense of \cref{lem:optimalprefix}, some retained, feasible residual labeling $o''$ satisfies
\begin{equation}\label{eq:prefix}
 f_B(o'')\ge\OPT-\frac d{\delta\rho}\bigl(f(B)-f(0)\bigr).
\end{equation}
Moreover, let $F_B$ be the multilinear extension of $f_B$, let $s\in(0,1]$ and $\nu>0$, and put $\vartheta=d/(\delta\rho)$. Suppose the residual solve at $B$ uses a comparison valid on the residual region with constants $\alpha$ and $C$, and outputs $p$ with $\E F_B(p)\ge\alpha f_B(o'')+\alpha Cf(B)-\nu L_B$; suppose its rounding retains $sF_B(p)+(1-s)f(B)$; and suppose the algorithm returns the best of the rounded candidate, the enumerated prefixes, and their retained singleton extensions. Then the returned labeling $S$ satisfies
\begin{equation}\label{eq:compensation}
 \E f(S)\ge\frac{s\alpha}{1+s\nu}\OPT
 \quad\text{whenever}\quad
 \vartheta\le C+\frac\nu\alpha+\frac{1-s}{s\alpha}.
\end{equation}
\end{lemma}
If instead $|\supp(o^*)|\le\rho$, the enumeration contains $o^*$ itself and the returned labeling has value at least $\OPT$. The retained conditioned value, the stored singleton candidates, and the rounding baseline contribute the three terms $C$, $\nu/\alpha$, and $(1-s)/(s\alpha)$ of the threshold in \eqref{eq:compensation}. The comparison is needed only at the greedy prefix $B$ under analysis. It is available on every residual region for the zero-diagonal parameters, and for a triple chosen at a level $\zeta>0$ whenever $\zeta(P_B)\ge\zeta$, so that $\zeta\one\in P_B$ and \cref{cor:linearizable} applies on $P_B$.

\begin{proof}[of \cref{lem:prefix}]
\emph{The prefix.} Let $O=\supp(o^*)$. If $|O|\le\rho$, the enumeration includes $o^*$ itself. Assume $|O|>\rho$. The family defined by nonnegative budget costs is downward closed, so \cref{lem:optimalprefix} applies with $h=\rho<|O|$. Fix any greedy ordering of $O$ under the fixed labels of $o^*$ and let $B$ be the restriction of $o^*$ to its first $\rho$ elements, so that $B\preceq o^*$ is feasible. Write $b=f(B)$ and $b'_m=1-\sum_{e\in\supp(B)}c_{m,e}$; since costs are nonnegative and $o^*$ is feasible, $\sum_{e\in\supp(B)}c_{m,e}\le\sum_{e\in O}c_{m,e}\le1$, so $b'_m\ge0$. Let $o'$ be the restriction of $o^*$ to $O\setminus\supp(B)$, viewed as a residual labeling. It is feasible for the residual capacities, because $\sum_{e\in\supp(o')}c_{m,e}=\sum_{e\in O}c_{m,e}-\sum_{e\in\supp(B)}c_{m,e}\le1-\sum_{e\in\supp(B)}c_{m,e}=b'_m$ for every row, and $B\cup o'=o^*$ gives $f_B(o')=f(o^*)=\OPT$.

\emph{Marginals of the remaining optimal elements.} For every $e\in\supp(o')$, \eqref{eq:optimalprefixtail} gives
\begin{equation}\label{eq:prefixmarginalbound}
 \Delta_{e,o^*_e}f(B)\le\frac{b-f(0)}\rho.
\end{equation}
The right side is nonnegative because $\psi$ is monotone (\cref{lem:optimalprefix}), so $b=\psi(\supp(B))\ge\psi(\varnothing)=f(0)$.

\emph{Excluded elements.} Let $H_m$ be the elements of $\supp(o')$ excluded by row $m$, namely those with $c_{m,e}>\delta b'_m$. If $b'_m=0$, all elements of $o'$ have zero cost in this row, because costs are nonnegative and $o'$ is feasible; hence $H_m=\varnothing$. If $b'_m>0$ and $H_m\ne\varnothing$, each element of $H_m$ has cost above $\delta b'_m$, and residual feasibility of $o'$ gives
\[
 |H_m|\,\delta b'_m<\sum_{e\in H_m}c_{m,e}\le\sum_{e\in\supp(o')}c_{m,e}\le b'_m,
 \qquad\text{so}\qquad |H_m|<1/\delta.
\]
Thus $H=\bigcup_{m=1}^dH_m$ has $|H|\le\sum_m|H_m|\le d/\delta$ elements. Deleting $H$ from $o'$ gives a labeling $o''$ whose elements all satisfy $c_{m,e}\le\delta b'_m$ in every row, so they are retained, and whose support, as a subset of the residually feasible support $\supp(o')$, is residually feasible because costs are nonnegative.

\emph{Value of the deleted elements.} List $H=\{h_1,\ldots,h_q\}$ in any fixed order and put $L_0=B\cup o''$ and $L_j=L_{j-1}+h_j\mapsto o^*_{h_j}$, so that $L_q=B\cup o'=o^*$. Each $L_{j-1}$ satisfies $B\preceq L_{j-1}$ and leaves $h_j$ unselected, so orthant submodularity \eqref{eq:orthant} with $s=B$ and $t=L_{j-1}$, followed by \eqref{eq:prefixmarginalbound}, gives
\[
 f(L_j)-f(L_{j-1})=\Delta_{h_j,o^*_{h_j}}f(L_{j-1})\le\Delta_{h_j,o^*_{h_j}}f(B)\le\frac{b-f(0)}\rho.
\]
Telescoping over the $q=|H|$ steps and using $|H|\le d/\delta$ and $b-f(0)\ge0$ gives
\[
 \OPT-f_B(o'')
 =f(B\cup o')-f(B\cup o'')
 \le |H|\,\frac{b-f(0)}\rho
 \le\frac d{\delta\rho}(b-f(0)).
\]
This proves \eqref{eq:prefix}, including rows with zero residual capacity.

\emph{Residual scale.} The residual objective satisfies $f_B(0)=b$. On the retained ground set every singleton is residually feasible, since its cost in each row is at most $\delta b'_m\le b'_m$; hence $B+e\mapsto i$ is feasible for the original budgets, because $\sum_{u\in\supp(B)}c_{m,u}+c_{m,e}\le(1-b'_m)+b'_m=1$. Every residual singleton value $f_B(e\mapsto i)=f(B+e\mapsto i)$ is therefore at most $\OPT$, so $L_B=\max_{e,i}(f_B(e\mapsto i)-f_B(0))\le\OPT-b$, and \cref{lem:marginal} applied to the nonnegative $k$-submodular function $f_B$ gives $L_B\ge0$. Since $b\ge0$,
\[
 0\le L_B\le\OPT-b\le\OPT.
\]

\emph{Compensation.} Let $\vartheta=d/(\delta\rho)$ and let $M$ be the largest value among all retained prefixes and their feasible residual singleton extensions. Suppose a residual solve outputs $p$ with $\E F_B(p)\ge\alpha F_B(\one_{o''})+\alpha Cf_B(0)-\nu L_B=\alpha f_B(o'')+\alpha Cb-\nu L_B$, and its rounding retains $sF_B(p)+(1-s)b$, where $s\in(0,1]$. Apply \cref{lem:candidates} with the deterministic comparator value $\Phi=f_B(o'')$; the pool contains $B$ and its retained singleton extensions, and the returned labeling $S$ has value $Y$. Then substitute \eqref{eq:prefix}, that is, $\Phi\ge\OPT-\vartheta(b-f(0))$. The numerator in \eqref{eq:candidatecomparison} is increasing in $\Phi$ because $s\alpha>0$, and
\[
 s\alpha\bigl(\OPT-\vartheta(b-f(0))\bigr)+(s\alpha C+s\nu+1-s)b
 =s\alpha\OPT+\bigl[s\alpha(C-\vartheta)+s\nu+1-s\bigr]b+s\alpha\vartheta f(0),
\]
which gives the full bound
\begin{equation}\label{eq:budgetcandidatefull}
 \E f(S)\ge
 \frac{s\alpha\OPT+\bigl[s\alpha(C-\vartheta)+s\nu+1-s\bigr]b
       +s\alpha\vartheta f(0)}{1+s\nu}.
\end{equation}
The coefficient of $b$ equals $s\alpha\bigl(C+\nu/\alpha+(1-s)/(s\alpha)-\vartheta\bigr)$, so it is nonnegative exactly when $\vartheta\le C+\nu/\alpha+(1-s)/(s\alpha)$. The coefficient $s\alpha\vartheta$ of $f(0)$ is also nonnegative, and $b,f(0)\ge0$. Dropping these two terms proves \eqref{eq:compensation}. Retaining the baseline gives a stronger statement: for the distinguished greedy prefix, $b\ge f(0)$ by \cref{lem:optimalprefix}. When the coefficient of $b$ is nonnegative, replacing $b$ by $f(0)$ in \eqref{eq:budgetcandidatefull} yields
\begin{equation}\label{eq:compensationbaseline}
 \E f(S)\ge\frac{s\alpha}{1+s\nu}\OPT+
 \frac{s\alpha C+s\nu+1-s}{1+s\nu}f(0).
\end{equation}
The two coefficients sum to at most one because $\alpha(1+C)\le1$, so a branch that enumerates an optimum also meets this bound. The fractional guarantee alone may give $b$ a negative coefficient; the calculation adds the rounded baseline and the stored singleton candidates before collecting the terms in $b$. A residual with $L_B=0$ is constant by \cref{lem:marginal}, and the same calculation applies to its direct output.
\end{proof}

\subsection{Proofs of \texorpdfstring{\cref{lem:singlepipage} and \cref{thm:single}}{the one-budget rounding lemma and the single-knapsack corollary}: cost-preserving rounding}\label[appendix]{app:singleproof}
For one budget, saving capacity for the final fractional item avoids the square-root loss of random-priority rounding. The cost-preserving pipage direction of \citet{AgeevSviridenko04} supplies an objective-independent construction.

\begin{proof}[of \cref{lem:singlepipage}]
If $b'=0$ or $\delta=0$, all costs are zero, and independent Bernoulli rounding is feasible and preserves $G(z)$ in expectation. Otherwise round zero-cost fractional coordinates independently according to their current probabilities. Affinity preserves conditional expected value and total cost.

While two positive-cost fractional coordinates $v,w$ remain, use the cost-preserving direction $u=\ee_v/c_v-\ee_w/c_w$, let $\tau_+,\tau_->0$ be the maximal step lengths along $\pm u$ that stay in the cube, and move to $z+\tau_+u$ with probability $\tau_-/(\tau_++\tau_-)$ and to $z-\tau_-u$ otherwise. The conditional mean remains $z$, and at least one of the two coordinates becomes integral. Exchange convexity gives
\[
 \frac{d^2}{d\tau^2}G(z+\tau u)
 =-\frac2{c_vc_w}\partial_v\partial_wG(z+\tau u)\ge0.
\]
Jensen's inequality therefore preserves or increases conditional expected value. All earlier integral coordinates remain fixed.

At termination at most one positive-cost coordinate $e$ is fractional. Round it with its current Bernoulli probability, preserving conditional expected value by affinity. Rounding it up increases cost by at most $c_e\le\delta b'$, while the current total cost is at most $(1-\delta)b'$. Thus either outcome is feasible. The tower property gives the value guarantee. Every choice uses only the fractional point, costs, and fresh randomness.
\end{proof}

\begin{proof}[of \cref{thm:single}]
Write $D=\alpha-\nu(\alpha-\eps)$. Since $0<\nu\le\eps<\alpha<1$, we have $D>\alpha-\eps>0$, hence $s\in(0,1)$ and $\delta=1-s\in(0,1)$. Direct algebra gives
\begin{equation}\label{eq:singlecompidentity}
 \frac{1-s}{s}=\frac\eps{\alpha-\eps}-\nu,
 \qquad
 \frac\nu\alpha+\frac{1-s}{s\alpha}=\frac\eps{\alpha(\alpha-\eps)},
 \qquad \frac{s\alpha}{1+s\nu}=\alpha-\eps.
\end{equation}
The prefix choice \eqref{eq:singleparams} therefore gives $1/(\delta\rho)\le C+\nu/\alpha+(1-s)/(s\alpha)$.

Enumerate feasible labeled prefixes of size at most $\rho$, retaining each and all its feasible residual singleton extensions. On a branch meeting the diagonal test, run \cref{thm:continuous} on $f_B$ with error $\nu L_B$, producing $p$. Scale to $sp$. By \eqref{eq:ray},
\[
 F_B(sp)\ge sF_B(p)+(1-s)f(B).
\]
Draw conditional labels with probabilities $p_{e,i}/|p_e|$; these proportions are unchanged by scaling. The support vector $z=s|p|$ has cost at most $sb'=(1-\delta)b'$, and every retained item costs at most $\delta b'$. Applying \cref{lem:singlepipage} and then averaging over labels gives a feasible residual labeling whose conditional expected value is at least $F_B(sp)$. Thus the required retention is $sF_B(p)+(1-s)f(B)$, also when $b'=0$.

If an optimum has at most $\rho$ elements, it is enumerated. Otherwise the residual condition supplies a greedy prefix of an optimum on which the solver is run. Apply \cref{lem:prefix} on this branch with $d=1$, and use \eqref{eq:singlecompidentity}. This proves the $(\alpha-\eps)$ guarantee. In fact the strengthened baseline calculation after \eqref{eq:budgetcandidatefull} gives the coefficient $(s\alpha C+s\nu+1-s)/(1+s\nu)$ of $f(0)$. By \eqref{eq:singlecompidentity}, $(1+s\nu)/s=\alpha/(\alpha-\eps)$ and $(1-s+s\nu)/s=\eps/(\alpha-\eps)$, so this coefficient equals $C(\alpha-\eps)+\eps/\alpha=\alpha C+\eps/(1-\omega)$, the coefficient in \eqref{eq:budgetaffine}; short optimal supports satisfy it too because the two coefficients sum to at most one. Other branches only add feasible candidates.

There are at most $(nk+1)^\rho$ prefixes. Each solve uses polynomially many value queries and entropy projections; conditional labels and pipage rounding use no value queries, and one output evaluation suffices per branch. For $\nu=\eps/2$,
\[
 \delta=\frac{1-\alpha/2}{\alpha}\eps+O(\eps^2),\qquad
 C+\frac\eps{\alpha(\alpha-\eps)}=C+O(\eps),
\]
so $\rho=\alpha\eps^{-1}/[C(1-\alpha/2)]+O(1)$. For $\nu=\eps^2$, the expansions give $\delta=\eps/\alpha+O(\eps^2)$ and $\rho=\alpha/(C\eps)+O(1)$. These are algebraic asymptotics, with no numerical parameter search.
\end{proof}

\subsection{Proof of \texorpdfstring{\cref{lem:priority}}{the random-priority lemma} (random-priority prefix rounding)}\label[appendix]{app:priority}
We use the notation of \cref{lem:priority}. At each priority, both the acceptance indicator and the positive marginal are decreasing functions of the earlier independent arrivals, and their positive correlation, integrated over the priority, gives the bound.

Ties among the priorities have probability zero and are resolved by index. Since costs are nonnegative, once a prefix violates a budget, every longer prefix violates that budget as well; so the feasible prefix is an initial segment of the priority order.

\begin{proof}[of \cref{lem:priority}]
If $\delta=0$, every cost is zero, so the prefix is the entire sampled set $R=\{e:X_e=1\}$. The pruning scan never lowers its value. Indeed, let $D=R\setminus S$ be the skipped elements, listed in scan order as $e_1,\ldots,e_q$; when $e_j$ was skipped, its marginal with respect to the retained set $S_{<e_j}$ was negative. Adding $e_1,\ldots,e_q$ to $S$ in this order, the marginal of $e_j$ is taken with respect to $S\cup\{e_1,\ldots,e_{j-1}\}\supseteq S_{<e_j}$, so by submodularity it is at most the negative marginal at $S_{<e_j}$. Telescoping gives $g(R)-g(S)\le0$. Thus $\E g(S)\ge\E g(R)=G_g(y)$, which is \eqref{eq:priorityretention} because $\lambda_d(0)=0$. Assume $0<\delta<1$ below.

\emph{Pathwise bound.} Let $A_e=(g(R_{<e}\cup\{e\})-g(R_{<e}))_+$. In the pruning scan, the retained elements before $e$ form a subset $S_{<e}\subseteq R_{<e}$, and $e$ is kept exactly when its marginal $g(S_{<e}\cup\{e\})-g(S_{<e})$ is nonnegative. Telescoping over the scan,
\[
 g(S)=g(\varnothing)+\sum_{e\text{ in the prefix}}\bigl(g(S_{<e}\cup\{e\})-g(S_{<e})\bigr)_+.
\]
By submodularity, each summand is at least $A_e$. Therefore, pathwise,
\begin{equation}\label{eq:prioritypathwise}
 g(S)\ge g(\varnothing)+\sum_e X_e\mathbf1\{e\text{ belongs to the feasible prefix}\}A_e.
\end{equation}

\emph{Acceptance probability at a fixed priority.} Condition on $X_e=1$ and $U_e=t$. The indicators $\mathbf1\{w\in R_{<e}\}$, $w\ne e$, are independent Bernoulli variables with means $ty_w$. Fix a row $m$ and let $\Upsilon_m=\sum_{w\ne e}a_{m,w}\mathbf1\{w\in R_{<e}\}$, whose mean is $t\sum_{w\ne e}a_{m,w}y_w\le t$. For $u>0$, convexity of $x\mapsto\mathrm e^{ux}$ on $[0,\delta]$ gives $\mathrm e^{ua}-1\le(a/\delta)(\mathrm e^{u\delta}-1)$ for $a\in[0,\delta]$ (the chord from $0$ to $\delta$ lies above the function). By independence, then $1+x\le\mathrm e^x$, then this chord bound, and finally $\sum_wa_{m,w}y_w\le1$,
\begin{align*}
 \E\,\mathrm e^{u\Upsilon_m}
 &=\prod_{w\ne e}\bigl(1+ty_w(\mathrm e^{ua_{m,w}}-1)\bigr)
 \le\exp\Bigl(\sum_{w\ne e}ty_w(\mathrm e^{ua_{m,w}}-1)\Bigr)\\
 &\le\exp\Bigl(\frac{t}{\delta}(\mathrm e^{u\delta}-1)\sum_{w\ne e}a_{m,w}y_w\Bigr)
 \le\exp\Bigl(\frac{t}{\delta}(\mathrm e^{u\delta}-1)\Bigr).
\end{align*}
Put $a_0=1-\delta$ and let $0<t<a_0$. Markov's inequality gives $\Pr\{\Upsilon_m>a_0\}\le\mathrm e^{-ua_0}\E\,\mathrm e^{u\Upsilon_m}\le\exp\bigl(\frac t\delta(\mathrm e^{u\delta}-1)-ua_0\bigr)$. With $u=\log(a_0/t)/\delta>0$, so that $\mathrm e^{u\delta}=a_0/t$, the exponent equals $(a_0-t-a_0\log(a_0/t))/\delta$, and therefore
\[
 \Pr\{\Upsilon_m>a_0\}
 \le\exp\left(-\frac{a_0\log(a_0/t)-a_0+t}{\delta}\right)
 \le\exp\left(-\frac{(1-\delta-t)^2}{2\delta}\right).
\]
For the last step write $x=t/a_0\in(0,1)$; then $a_0\log(a_0/t)-a_0+t=a_0(x-1-\log x)\ge a_0(1-x)^2/2=(a_0-t)^2/(2a_0)\ge(a_0-t)^2/2$, where $x-1-\log x\ge(1-x)^2/2$ on $(0,1]$ because the difference vanishes at $x=1$ and has derivative $-(1-x)^2/x\le0$. The case $t=0$ is immediate because then $\Upsilon_m=0$. If $\Upsilon_m\le a_0$ for every row, then $R_{<e}\cup\{e\}$ is feasible, because its load in row $m$ is $\Upsilon_m+a_{m,e}\le(1-\delta)+\delta=1$. For $t=a_0$ the bound below equals one, and for $t>a_0$ it is set to one. A union bound over the $d$ rows shows that the conditional acceptance probability is at least $1-\gamma(t)$, where
\[
 \gamma(t)=\begin{cases}
 \min\{1,d\exp(-(1-\delta-t)^2/(2\delta))\},&0\le t\le1-\delta,\\
 1,&1-\delta<t\le1.
 \end{cases}
\]

\emph{Positive correlation.} For fixed $t$, both the acceptance indicator and $A_e$ are nonincreasing functions of the independent indicators of $R_{<e}$: adding elements increases every load, and by submodularity it decreases the marginal of $e$. Two nonincreasing functions $u,v$ of independent Bernoulli variables are positively correlated. This follows by induction on the number of variables: for one Bernoulli variable with mean $\pi$, $\operatorname{Cov}(u,v)=\pi(1-\pi)(u(1)-u(0))(v(1)-v(0))\ge0$; for several, condition on the last variable, apply the induction hypothesis to the conditional covariance, and apply the one-variable case to the conditional means, which are again nonincreasing. Conditional on $X_e=1$ and $U_e=t$, the set $R_{<e}$ has the law of $R_{ty,-e}$ defined below, so $\E[A_e\mid X_e=1,U_e=t]=h_e(t)$, and the acceptance probability is at least $1-\gamma(t)$. Hence
\begin{align*}
 \E[\mathbf1\{e\text{ accepted}\}A_e\mid X_e=1,U_e=t]
 &\ge\E[\mathbf1\{e\text{ accepted}\}\mid X_e=1,U_e=t]\,h_e(t)\\
 &\ge(1-\gamma(t))h_e(t),\qquad
 h_e(t):=\E\bigl[(\Delta_e g(R_{ty,-e}))_+\bigr],
\end{align*}
where $R_{ty,-e}$ contains each $w\ne e$ independently with probability $ty_w$ and $\Delta_eg(A)=g(A\cup\{e\})-g(A)$.

\emph{Integration over priorities.} Coupling the sets $R_{ty,-e}=\{w\ne e:V_w<ty_w\}$ for all $t$ through common uniform variables $V_w$, $R_{ty,-e}$ grows with $t$, so $(\Delta_eg(R_{ty,-e}))_+$ is pointwise nonincreasing in $t$ by submodularity, and so is its mean $h_e(t)$. The function $1-\gamma(t)$ is also nonincreasing, because $(1-\delta-t)^2$ decreases on $[0,1-\delta]$ and $\gamma=1$ afterward. For any two nonincreasing functions $u,v$ on $[0,1]$,
\[
 \int_0^1u(t)v(t)\,dt-\int_0^1u(t)\,dt\int_0^1v(t)\,dt
 =\frac12\int_0^1\!\int_0^1(u(t)-u(s))(v(t)-v(s))\,dt\,ds\ge0,
\]
because $(u(t)-u(s))(v(t)-v(s))\ge0$ for all $s,t$. Taking expectations in \eqref{eq:prioritypathwise}, conditioning on $X_e$ (which equals one with probability $y_e$) and on the uniform priority $U_e$, applying the preceding bound, and then the last display with $u=1-\gamma$ and $v=h_e$, we obtain, with $\bar\gamma=\int_0^1\gamma(t)\,dt$,
\[
 \E g(S)\ge g(\varnothing)+\sum_ey_e\int_0^1(1-\gamma(t))h_e(t)\,dt\ge g(\varnothing)+(1-\bar\gamma)\sum_e y_e\int_0^1h_e(t)\,dt.
\]

\emph{Comparison with the multilinear extension.} The multilinear extension is affine in each coordinate with $\partial_eG_g(z)=\E[\Delta_eg(R_{z,-e})]$, so the chain rule gives $\frac{d}{dt}G_g(ty)=\sum_ey_e\,\E[\Delta_eg(R_{ty,-e})]$, and $(\cdot)_+$ dominates the signed marginal. Hence, using $G_g(0)=g(\varnothing)$,
\[
 \sum_e y_e\int_0^1h_e(t)\,dt\ge\int_0^1\frac{d}{dt}G_g(ty)\,dt
 =G_g(y)-g(\varnothing).
\]
This sum is also nonnegative, because every $h_e\ge0$. The next step shows $\bar\gamma\le\lambda_d(\delta)$, so $1-\bar\gamma\ge1-\lambda_d(\delta)>0$; multiplying the nonnegative sum by the smaller factor and then bounding the sum below by $\max\{0,G_g(y)-g(\varnothing)\}$ gives
\[
 \E g(S)\ge g(\varnothing)+(1-\lambda_d(\delta))\max\{0,G_g(y)-g(\varnothing)\}\ge(1-\lambda_d(\delta))G_g(y)+\lambda_d(\delta)g(\varnothing),
\]
including when $G_g(y)<g(\varnothing)$.

\emph{The constant.} The interval $(1-\delta,1]$, where $\gamma=1$, contributes exactly $\delta$ to $\bar\gamma$. On $[0,1-\delta]$, substitute $u=(1-\delta-t)/\sqrt\delta$, so that $dt=-\sqrt\delta\,du$ and $u$ runs from $(1-\delta)/\sqrt\delta$ down to $0$, and extend the nonnegative integral to the half-line:
\[
 \bar\gamma\le\delta+\sqrt\delta\int_0^\infty\min\{1,de^{-u^2/2}\}\,du
 =\delta+\kappa_d\sqrt\delta=\lambda_d(\delta).
\]
With $a_d=\sqrt{2\log d}$, the integrand equals one for $u\le a_d$ and $d\mathrm e^{-u^2/2}$ afterward, which evaluates the constant exactly as $\kappa_d=a_d+d\int_{a_d}^\infty e^{-u^2/2}\,du$; in particular $\kappa_1=\int_0^\infty\mathrm e^{-u^2/2}du=\sqrt{\pi/2}$. Splitting instead at $b_d=\sqrt{2\log(2d)}$, bounding the integrand by one on $[0,b_d]$, and using $\int_{b}^\infty\mathrm e^{-u^2/2}du\le\int_b^\infty(u/b)\mathrm e^{-u^2/2}du=\mathrm e^{-b^2/2}/b$ for $b>0$, together with $d\,\mathrm e^{-b_d^2/2}=1/2$, gives
$\kappa_d\le b_d+d\,\mathrm e^{-b_d^2/2}/b_d=\sqrt{2\log(2d)}+[2\sqrt{2\log(2d)}]^{-1}$,
and squaring, with $b_d^2=2\log(2d)$, gives $\kappa_d^2\le2\log(2d)+1+1/(8\log(2d))=O(\log(2d))$.
This proves \eqref{eq:priorityretention}.

\emph{Feasibility and queries.} The accepted prefix is feasible on every outcome, and pruning only deletes elements, so $S$ is feasible. Maintaining the current value of the scan requires $g(\varnothing)$ and at most one extension query per prefix element, hence at most $|V|+1$ value queries.
\end{proof}

\subsection{Proof of \texorpdfstring{\cref{thm:multi}}{the multiple-budget corollary} (knapsack approximation)}\label[appendix]{app:multiproof}
The residual instance retains precisely the elements whose cost in every positive-capacity row is at most a $\delta$ fraction of that capacity, and a zero-capacity row retains only its zero-cost elements; thus every retained singleton extension is feasible. After dividing each positive-capacity row by its residual capacity, the retained costs satisfy $a_{m,e}=c_{m,e}/b'_m\le\delta$, the residual fractional support satisfies $\sum_ea_{m,e}y_e\le1$, and zero-capacity rows impose no constraint on retained elements; these are the hypotheses of \cref{lem:priority}.

\begin{proof}[of \cref{thm:multi}]
\emph{Parameters.} Fix $0<\nu\le\eps$ (below, $\nu=\eps/2$), write $\kappa=\kappa_d$ and $D=\alpha-\nu(\alpha-\eps)$, and use the parameters of \eqref{eq:multparams}.
The definition of $\delta$ gives $\sqrt\delta=2\lambda/(\sqrt{\kappa^2+4\lambda}+\kappa)=(\sqrt{\kappa^2+4\lambda}-\kappa)/2$, after multiplying numerator and denominator by $\sqrt{\kappa^2+4\lambda}-\kappa$ and using $(\kappa^2+4\lambda)-\kappa^2=4\lambda$; this is the positive root of $x^2+\kappa x-\lambda=0$. With $D=\alpha-\nu(\alpha-\eps)$, the bounds $0<\nu\le\eps<\alpha<1$ give $D>0$ and $D-(\alpha-\eps)=\eps-\nu(\alpha-\eps)>0$, so $s=(\alpha-\eps)/D\in(0,1)$. Hence $\lambda>0$, and substituting $x=\sqrt\delta$ in the quadratic gives $\lambda_d(\delta)=\delta+\kappa\sqrt\delta=\lambda=1-s<1$; in particular $0<\delta\le\lambda<1$, since $\kappa\sqrt\delta\ge0$.

\emph{Retention.} Fix a feasible labeled prefix $B$ and a fractional residual action $p$. Draw conditional labels as in \cref{lem:extension}; for each realization $\ell$ the set function $g_\ell(A)=f_B(\ell|_A)$ is nonnegative and submodular with $g_\ell(\varnothing)=f_B(0)=f(B)$. By the normalization above, $y=|p|$ and the normalized costs satisfy the hypotheses of \cref{lem:priority}, and $\lambda_d(\delta)=1-s$. Conditional on $\ell$, the lemma gives $\E[g_\ell(S)\mid\ell]\ge sG_\ell(|p|)+(1-s)f(B)$. Averaging over $\ell$ and using \eqref{eq:labelconditioning} for $F_B$, the rounded residual output satisfies
\[
 \E f_B(S)\ge sF_B(p)+(1-s)f(B).
\]

\emph{Feasibility on every branch.} The algorithm enumerates every feasible labeled prefix of at most $\rho$ elements. On each branch it computes $\zeta(P_B)$ by \eqref{eq:residualzeta} and runs the residual solver only if the chosen parameters are valid on $P_B$, that is, always for the zero-diagonal parameters, and for a triple chosen at level $\zeta>0$ exactly when $\zeta(P_B)\ge\zeta$. Every stored candidate is feasible: the prefixes by construction, their retained singleton extensions by the residual-scale step in the proof of \cref{lem:prefix}, and the rounded outputs on every outcome by \cref{lem:priority}, whatever the residual solver returns.

\emph{Approximation.} If some optimal labeling has at most $\rho$ elements, enumeration contains it and the returned value is at least $\OPT$. Otherwise, let $o^*$ be an optimal labeling and $B$ a greedy prefix of $\rho$ elements of $o^*$ in the sense of \cref{lem:optimalprefix}; for a triple chosen at level $\zeta>0$, take the pair supplied by the residual condition, so that $\zeta(P_B)\ge\zeta$. \Cref{lem:prefix} applies to this greedy prefix. With $\vartheta=d/(\delta\rho)$, its retained comparator obeys $f_B(o'')\ge\OPT-\vartheta(f(B)-f(0))$. On this branch the chosen parameters are valid on $P_B$, so the fractional solver of \cref{thm:continuous}, run on the residual instance with additive error $\nu L_B$, gives
\[
 \E F_B(p)\ge\alpha f_B(o'')+\alpha C f(B)-\nu L_B.
\]
The prefix, its feasible singleton extensions, and the rounded output are stored. The other branches enter the analysis only through their stored candidates, which can only increase the returned value. Since $\rho\ge d/\bigl(\delta(C+\nu/\alpha+(1-s)/(s\alpha))\bigr)$ by \eqref{eq:multparams},
\[
 \vartheta=\frac d{\delta\rho}\le C+\nu/\alpha+(1-s)/(s\alpha),
\]
so \cref{lem:candidates} and \eqref{eq:compensation} show that their best value is at least $s\alpha\OPT/(1+s\nu)$ in expectation. Keeping the baseline by \eqref{eq:compensationbaseline} yields \eqref{eq:budgetaffine}, with its coefficient of $f(0)$ simplified as in the proof of \cref{thm:single}, because $s$ is given by the same formula. Finally, $D+\nu(\alpha-\eps)=\alpha$, so $1+s\nu=(D+\nu(\alpha-\eps))/D=\alpha/D$, and with $s=(\alpha-\eps)/D$ this gives $s\alpha/(1+s\nu)=\frac{\alpha-\eps}{D}\cdot\alpha\cdot\frac D\alpha=\alpha-\eps$.

\emph{Running time.} There are at most $(nk+1)^\rho$ prefixes. Each residual optimization uses the polynomial query bound in \cref{thm:continuous}, and each rounding uses at most $n+1$ additional value queries and $O(nd+n\log n)$ arithmetic operations. For $\nu=\eps/2$, the identity $1-s=(D-(\alpha-\eps))/D=\eps(1-(\alpha-\eps)/2)/D$ with $D=\alpha+O(\eps)$, and the expansion $\delta=\bigl((\sqrt{\kappa^2+4\lambda}-\kappa)/2\bigr)^2=\lambda^2/\kappa^2+O(\lambda^3)$, give
\[
 \lambda=(1/\alpha-1/2)\eps+O(\eps^2),\qquad
 \delta=\lambda^2/\kappa^2+O(\lambda^3),\qquad
 \kappa^2\le2\log(2d)+1+\frac1{8\log(2d)}.
\]
The denominator in the definition of $\rho$ is $\delta(C+O(\eps))$, so $\rho=d\kappa^2/(C\lambda^2)\,(1+O(\eps))=O(d\log(2d)/\eps^2)$. Every stored candidate and every returned labeling is feasible on every realization.
\end{proof}

\subsection{Proof of \texorpdfstring{\cref{cor:finiteprefix}}{the finite-enumeration corollary} (finite enumeration)}\label[appendix]{app:finiteprefix}
\begin{proof}[of \cref{cor:finiteprefix}]
Write $Y=f(S)$. The feasibility arguments of \cref{app:singleproof,app:multiproof} do not use the accuracy-based prefix choices, so every stored candidate, and hence $S$, satisfies all budgets on every outcome.

\emph{Short optimal labelings.} Suppose some optimal labeling has at most $\rho$ elements. The enumeration contains it, so $Y\ge\OPT$ on every outcome. The right side of \eqref{eq:finiteprefix} has the form $a\,\OPT+a'f(0)$ with $a,a'\ge0$, and $a+a'\le1$ because $\alpha(1+C)=(1+c-\omega)/(1+c)\le1$: if $B_*\ge0$, then $a+a'=A+B_*+D_*=(s\alpha(1+C)+s\nu+1-s)/(1+s\nu)\le1$; if $B_*<0$, then, by the second form below, $a+a'=\alpha(1+\vartheta)/(1+\alpha(\vartheta-C))\le1$. Since $0\le f(0)\le\OPT$, the right side is at most $\OPT\le\E Y$.

\emph{Long optimal labelings.} Otherwise, choose $o^*$ and its greedy prefix $B$ as in the approximation step of the proof of \cref{thm:multi}, and write $b=f(B)$. The applicable retention argument uses only $s=1-\delta$ for one budget or $\lambda_d(\delta)=1-s<1$ for several budgets, so the rounded residual output retains $sF_B(p)+(1-s)b$. The residual solver on this branch gives $\E F_B(p)\ge\alpha f_B(o'')+\alpha Cb-\nu L_B$. These are the hypotheses of the second part of \cref{lem:prefix}, whose proof gives \eqref{eq:budgetcandidatefull} without any condition on $\vartheta$:
\[
 \E Y\ge\frac{s\alpha\OPT+\bigl[s\alpha(C-\vartheta)+s\nu+1-s\bigr]b+s\alpha\vartheta f(0)}{1+s\nu}.
\]
That is, $\E Y\ge A\,\OPT+B_*b+D_*f(0)$. The prefix $B$ is a stored candidate, so $Y\ge b$ on every outcome and $\E Y\ge b$; and $b\ge f(0)$, because the function $\psi$ of \cref{lem:optimalprefix} is monotone.

If $B_*\ge0$, then $B_*b\ge B_*f(0)$, and $\E Y\ge A\,\OPT+(B_*+D_*)f(0)$. If $B_*<0$, multiplying $b\le\E Y$ by $B_*$ reverses the inequality, so $B_*b\ge B_*\E Y$ and
\[
 \E Y\ge A\,\OPT+D_*f(0)+B_*\E Y,\qquad\text{that is,}\qquad (1-B_*)\E Y\ge A\,\OPT+D_*f(0).
\]
Since $1-B_*=1+|B_*|>0$, dividing gives $\E Y\ge(A\,\OPT+D_*f(0))/(1-B_*)$. The two cases together are \eqref{eq:finiteprefix}.

\emph{The second form.} Expanding,
\[
 (1+s\nu)(1-B_*)=1+s\nu-s\alpha(C-\vartheta)-s\nu-1+s=s\bigl(1+\alpha(\vartheta-C)\bigr).
\]
When $B_*<0$ the left side is positive, so $1+\alpha(\vartheta-C)>0$, and
\[
 \frac A{1-B_*}=\frac{s\alpha}{(1+s\nu)(1-B_*)}=\frac{s\alpha}{s(1+\alpha(\vartheta-C))}=\frac{\alpha}{1+\alpha(\vartheta-C)},\qquad
 \frac{D_*}{1-B_*}=\frac{\alpha\vartheta}{1+\alpha(\vartheta-C)}.
\]

\emph{Running time.} There are at most $(nk+1)^\rho$ labeled prefixes. On each branch, \eqref{eq:residualzeta} takes $O(nd)$ operations; the residual solver of \cref{thm:continuous} with error $\nu L_B$ uses a number of value queries and entropy projections polynomial in $n$, $k$, and $1/\nu$; and the rounding uses no value queries for one budget and at most $n+1$ for several budgets. Its arithmetic work is polynomial by \cref{lem:singlepipage,lem:priority}.
\end{proof}

\subsection{Proof of \texorpdfstring{\cref{prop:smallitems}}{the small-item proposition} (small items)}\label[appendix]{app:smallitems}
\begin{proof}[of \cref{prop:smallitems}]
\emph{Constant objectives.} If $\Lambda=0$, then $f$ is constant by \cref{lem:marginal}, so the returned labeling has value $f(0)=\OPT=\mathrm{OPT}_F(P)$. The right side of \eqref{eq:smallitems} is then $f(0)\bigl(s\alpha(1+C)+s\nu+1-s\bigr)/(1+s\nu)$. Since $\alpha=(1-\omega)/(1+c)$ and $\alpha C=c/(1+c)$, we have $\alpha(1+C)=(1+c-\omega)/(1+c)\le1$, so the right side is at most $f(0)(s+s\nu+1-s)/(1+s\nu)=f(0)$, and \eqref{eq:smallitems} holds. Assume $\Lambda>0$ below.

\emph{Feasibility.} All singletons and the empty labeling are feasible. The point $s|p|=(1-\delta)|p|$ and the item costs satisfy the hypotheses of \cref{lem:singlepipage} when $d=1$, and the point $|p|$ and the costs satisfy those of \cref{lem:priority} when $d\ge2$, so every rounded candidate satisfies every budget on every outcome.

\emph{Fractional guarantee.} The region $P$ is a nonempty compact convex down-closed polytope with an explicit packing description, so its entropy projections are available, and it contains $\zeta(P)\one$ by the definition of the diagonal level. \Cref{thm:continuous}, applied with the chosen triple in $\mathcal A_k(\zeta(P))$ and with $\eps$ replaced by $\nu$, gives
\[
 \E F(p)\ge\alpha\,\mathrm{OPT}_F(P)+\alpha Cf(0)-\nu L.
\]

\emph{Retention.} For one budget, the ray inequality and \cref{lem:singlepipage}, followed by conditional label averaging, give $\E[X\mid p]\ge F(sp)\ge sF(p)+(1-s)f(0)$. For several budgets the same inequality follows from \cref{lem:priority} and conditional label averaging, with $s=1-\lambda_d(\delta)$.

\emph{Comparison with the stored candidates.} Apply \cref{lem:candidates} with the empty prefix $B=0$. Then $b=f(0)$ and every residual capacity equals one, so every element is retained, because $c_{m,e}\le\delta=\delta b'_m$; hence $f_B=f$, $F_B=F$, and $L_B=L$. The pool consists of the empty labeling and the singletons, which are feasible and are evaluated when $L$ is computed from \eqref{eq:scales}. With the deterministic comparator value $\Phi=\mathrm{OPT}_F(P)$, the two displays above are the hypotheses of \cref{lem:candidates}, and \eqref{eq:candidatecomparison} is \eqref{eq:smallitems}.

\emph{The coefficient.} Let $0<\eps<s\alpha$ and $\nu=\eps/(s(s\alpha-\eps))>0$. Then $s\nu=\eps/(s\alpha-\eps)$, so $1+s\nu=s\alpha/(s\alpha-\eps)$ and $s\alpha/(1+s\nu)=s\alpha-\eps$. The coefficient of $f(0)$ in \eqref{eq:smallitems} is nonnegative and $f(0)\ge0$; moreover $\mathrm{OPT}_F(P)\ge\OPT$ (\cref{sec:access}). Hence
\[
 \E f(S)\ge(s\alpha-\eps)\OPT.
\]

\emph{Queries.} The $nk+1$ values of the empty labeling and the singletons are among the queries of \cref{thm:continuous}. One-budget pipage rounding uses no value queries, followed by one output evaluation. For several budgets the pruning scan uses at most $n+1$ value queries and an output evaluation costs at most one more. Label sampling uses no objective values.

\emph{Online transfer for one budget.} Conditional on each round's objective and history, the cost-preserving rule applied to $(1-\delta)p_t$ gives $\E[f_t(S_t)\mid\mathcal H_t]\ge(1-\delta)F_t(p_t)+\delta f_t(0)$ before any current feedback is obtained. Apply \cref{cor:roundtransfer} with $s=1-\delta$ and $\tau=\delta$; the learner and field queries remain unchanged. Choosing the best current-objective candidate is not part of this online rule.
\end{proof}

\end{document}